\documentclass[journal]{IEEEtran}

\usepackage{bm}
\usepackage{amsmath}
\usepackage{epsf}
\usepackage{graphics}
\usepackage{ amssymb }
\usepackage[dvips]{graphicx}
\usepackage{mathtools}
\usepackage{epsfig}
\usepackage{cite}
\usepackage{longtable}
\usepackage{colortbl}
\usepackage{color}
\usepackage{enumitem}
\usepackage{soul,xcolor}

\usepackage{bm}
\usepackage{amsmath}
\usepackage{epsf}
\usepackage{graphics}
\usepackage{ amssymb }
\usepackage[dvips]{graphicx}
\usepackage{epsfig}
\usepackage{cite}
\usepackage{graphicx}
\usepackage{epsfig}
\usepackage{latexsym}
\usepackage{amsfonts}
\usepackage{here}
\usepackage{rawfonts}
\usepackage[utf8]{inputenc}
\usepackage[english]{babel}
\usepackage{amsmath}
\usepackage{amsfonts}
\usepackage{amssymb}
\usepackage{color}
\usepackage{bm}
\usepackage{listings}
\usepackage{caption}
    \usepackage{multicol}

\usepackage{tabularx,booktabs}
\newcolumntype{Y}{>{\centering\arraybackslash}X}
\makeatletter
\newcommand\notsotiny{\@setfontsize\notsotiny{6.31415}{7.1828}}
\makeatother
\usepackage{amssymb}
\usepackage{amsthm}
\usepackage{graphicx}
\usepackage{epstopdf}
\usepackage{listings}
\usepackage{float}
\usepackage{amsmath}
\usepackage{amssymb}
\usepackage{amsfonts}
\usepackage{epstopdf}

\usepackage{multirow}
\usepackage{amscd}
\usepackage{mathrsfs}
\usepackage{graphicx}
\usepackage{color}
\usepackage{url}
\usepackage{bm}
\usepackage{algorithm}
\usepackage{setspace}
\usepackage{footnote}
\usepackage{algorithmic}
\usepackage{xcolor}
\usepackage{svg}
\usepackage{dsfont}
\DeclareMathOperator*{\argmin}{arg\,min}
\newtheorem{theorem}{Theorem}
\newtheorem{lemma}{Lemma}

\newtheorem{definition}{Definition}
\newtheorem{proposition}{Proposition}
\newtheorem{corollary}{Corollary}

\newtheorem{remark}{Remark}

\newtheorem{assumption}{Assumption}

\addto\captionsenglish{}
\usepackage{bbm}

\usepackage{multicol}
\usepackage{mathtools}

\usepackage{lipsum,graphicx,subcaption}

\usepackage{diagbox}
\usepackage{hyperref}
\usepackage[protrusion=true,expansion=true]{microtype}
\pdfoutput=1
\usepackage[font=footnotesize]{caption}
\allowdisplaybreaks[4]

\usepackage{graphicx}
\usepackage{grffile}
\usepackage{tabularx}
\newcounter{term}[section]

\renewcommand\theterm{\alph{term}}
\makeatletter
\newcommand{\vast}{\bBigg@{4}}

\newcommand{\Vast}{\bBigg@{5}}
\makeatother

\usepackage{soul,xcolor}

\usepackage[most]{tcolorbox}
\usepackage{nicefrac, xfrac}
\DeclareMathOperator{\atantwo}{atan2}

\usepackage[most]{tcolorbox}
\usepackage{subcaption}
\tcbset{enhanced,
      colframe=green!20!black,
      colback=green!10!white,
      boxsep=0pt}
        \allowdisplaybreaks
        
        \newcommand\sbullet[1][.5]{\mathbin{\vcenter{\hbox{\scalebox{#1}{$\bullet$}}}}}

\definecolor{GD}{RGB}{218, 224, 104} 
\definecolor{LT}{RGB}{95, 29, 122}
\definecolor{IT}{RGB}{5, 0, 236}
\definecolor{LA}{RGB}{38,14,17}
\definecolor{LD}{RGB}{222,48,40}
\definecolor{GA}{RGB}{104, 217, 224}
\definecolor{redMod}{RGB}{135, 60, 85}
\definecolor{mygreen}{HTML}{1e285e}

\newtcbox{\rshadeGD}{on line,
  colback=GD!50!white,   
  colframe=GD!50!white,  
  arc=3pt,           
  boxrule=0pt,       
  left=1.1pt, right=1.1pt, top=1.1pt, bottom=1.1pt}

  \newtcbox{\rshadeLT}{on line,
  colback=LT!25!white,   
  colframe=LT!25!white,  
  arc=3pt,           
  boxrule=0pt,       
   left=1.1pt, right=1.1pt, top=1.1pt, bottom=1.1pt}

  \newtcbox{\rshadeIT}{on line,
  colback=IT!25!white,   
  colframe=IT!25!white,  
  arc=3pt,           
  boxrule=0pt,       
   left=1.1pt, right=1.1pt, top=1.1pt, bottom=1.1pt}

  \newtcbox{\rshadeLA}{on line,
  colback=LA!32!white,   
  colframe=LA!32!white,  
  arc=3pt,           
  boxrule=0pt,       
   left=1.1pt, right=1.1pt, top=1.1pt, bottom=1.1pt}

    \newtcbox{\rshadeLD}{on line,
  colback=LD!31!white,   
  colframe=LD!31!white,  
  arc=3pt,           
  boxrule=0pt,       
  left=1.1pt, right=1.1pt, top=1.1pt, bottom=1.1pt}

    \newtcbox{\rshadeGA}{on line,
  colback=GA!35!white,   
  colframe=GA!35!white,  
  arc=3pt,           
  boxrule=0pt,       
  left=1.1pt, right=1.1pt, top=1.1pt, bottom=1.1pt}

\begin{document} 

\title{{Graph Theory Meets Federated Learning over  Satellite Constellations: Spanning Aggregations, Network Formation, and Performance Optimization}\vspace{-0.5mm}}

\author{\vspace{-1.5mm}Fardis \hspace{-0.5mm}Nadimi$^*$,~\IEEEmembership{Student~Member,~IEEE}, Payam \hspace{-0.5mm}Abdisarabshali$^*$,~\IEEEmembership{Student~Member,~IEEE},
Jacob \hspace{-0.1mm}Chakareski, \IEEEmembership{Senior Member,~IEEE}, Nicholas \hspace{-0.1mm}Mastronarde,~\IEEEmembership{Senior~Member,~IEEE}, and Seyyedali \hspace{-0.1mm}Hosseinalipour,~\IEEEmembership{Senior Member,~IEEE}\vspace{-8mm}
\thanks{* Authors contributed equally to the paper.\\
F. Nadimi, P. Abdisarabshali, N. Mastronarde, and S. Hosseinalipour are with the University at Buffalo---SUNY, NY, USA (Emails: \{ fardisna,payamabd,nmastron,alipour@buffalo.edu\}). J. Chakareski is with New Jersey Institute of Technology (Email: \{jacobcha@njit.edu\}).}
}
\maketitle
\setulcolor{red}
\setul{red}{2pt}
\setstcolor{red}
\setlength{\abovedisplayskip}{4pt}
\setlength{\belowdisplayskip}{4pt}
\setlength{\skip\footins}{6pt}
\setlength{\footnotesep}{0pc}
\begin{abstract} 
In this work, we introduce {Fed-Span}: \textit{\underline{fed}erated learning with  \underline{span}ning aggregation over low Earth orbit (LEO) satellite constellations}. {Fed-Span} aims to address critical challenges inherent to distributed learning in dynamic satellite networks, including intermittent satellite connectivity, heterogeneous computational capabilities of satellites, and time-varying satellites' datasets. At its core, {Fed-Span} leverages minimum spanning tree (MST) and minimum spanning forest (MSF) topologies to introduce spanning model aggregation and dispatching processes for distributed learning. To formalize {Fed-Span}, we offer a fresh perspective on MST/MSF topologies by formulating them through a set of continuous constraint representations (CCRs), thereby integrating these topologies into a distributed learning framework for satellite networks. Using these CCRs, we obtain the energy consumption and latency of operations in {Fed-Span}. Moreover, we derive novel convergence bounds for {Fed-Span}, accommodating its key system characteristics and degrees of freedom (i.e., tunable parameters). Finally, we propose a comprehensive optimization problem that jointly minimizes model prediction loss, energy consumption, and latency of {Fed-Span}. We unveil that this problem is NP-hard and develop a systematic approach to transform it into a geometric programming formulation, solved via successive convex optimization with performance guarantees. Through evaluations on real-world datasets, we demonstrate that {Fed-Span} outperforms existing methods, with faster model convergence, greater energy efficiency, and reduced latency. 
\end{abstract}
\vspace{-0.75mm}
\begin{IEEEkeywords}
Low-earth orbit satellites, Federated learning, Graph theory, Minimum spanning tree, Non-convex optimization.
\end{IEEEkeywords}
\vspace{-5.85mm}
\section{Introduction}\label{sec:intro}
\vspace{-0.5mm}
\noindent  Satellite networks have experienced tremendous advancements over the past few years. As an integral component of these networks, low Earth orbit (LEO) satellites collect vast amounts of data through Earth observation \cite{gomes2020overview}. Using machine learning (ML) techniques, this data can be leveraged in intelligent decision-making for disaster management, climate change, and urban planning~\cite{salcedo2020machine}.  
Nevertheless, traditional ML methods rely on a centralized learning structure, requiring the distributedly collected satellite data to be transferred to a central training unit (e.g., a cloud server). Such data centralization can pose logistical challenges due to the large volume of collected data, privacy restrictions, and transmission limitations of satellites.

Federated learning (FedL)~\cite{mcmahan2017communication}, an emerging distributed ML approach, offers an alternative to centralized ML methods. In its traditional client-to-server setting, resembling a \textit{star topology} (Fig.~\ref{fig:simpleFL}(a)), data collection units (called \textit{clients}) perform local model training on their own data and only transmit their trained models to a central server. The server aggregates these local models (e.g., via averaging) into a global model, which is then dispatched to the clients for the next training round.

\vspace{-1mm}
\subsection{Motivation and Dimensions of Innovation}
Although recent studies have begun to explore FedL over LEO satellites (we collectively refer to them as {Fed-LS}) \cite{9749193,10021101, chen2022satellite ,nguyen2022federated, 10216376, 9674028, tang2022federated,10121575,10039157, elmahallawy2023optimizing, elmahallawy2023one}, there remain numerous unexplored research directions due to the emerging nature of this field. In this work, we propose \textit{\underline{fed}erated learning with  \underline{span}ning aggregation over LEO satellite constellations} ({Fed-Span}), which leverages graph-theoretic principles to empower hierarchical Fed-LS.
{Fed-Span} consists of three major aspects: \textit{(i) use of emerging communication links (i.e., optical communications), (ii) explicit consideration of model convergence in satellite network orchestration, and (iii) incorporation of graph-theoretic formulations into the design of distributed learning topologies for satellite networks}. These aspects are motivated by the innate characteristics of LEO satellite constellations and the limitations of the existing studies, as detailed below.


\subsubsection{Motivation and Related Work} \label{sec:Hurdels}
A key feature of LEO satellite constellations is their vast operational region, making the presence of \textit{a single aggregation point}, as assumed in conventional FedL, an unnatural fit. To this end, Fed-LS architectures typically employ hierarchical FedL, where satellites' local models are aggregated by \textit{terrestrial/non-terrestrial regional aggregators}, such as ground stations (GSs), unmanned aerial vehicles, or high-altitude platforms~\cite{9674028, tang2022federated,10121575,10039157}. 
In these hierarchical Fed-LS architectures, radio-frequency (RF) links are often used for model transmissions between the satellites and regional aggregators \cite{kodheli2020satellite}. However, these links are often \textit{time-constrained} and \textit{intermittent} due to the satellites' orbital movements and limited windows of observation. These limitations make transmitting large ML models over long-distance RF links \textit{time-consuming} and \textit{energy-intensive}, particularly with the added challenges such as atmospheric attenuation and propagation delays.
Moreover, heterogeneous communication patterns between satellites and aerial/ground nodes complicate implementing \textit{synchronized} FedL~\cite{mcmahan2017communication} for Fed-LS, where global aggregations only occur after the reception of the local models of all satellites at an aggregation point. 

Several solutions have been proposed to address the above challenges. AsyncFLEO~\cite{10021101}, an asynchronous Fed-LS methodology, performs global aggregations whenever a satellite model arrives at GSs. However, it struggles with the stale models of \textit{straggler} satellites (i.e., those with infrequent communications to the GSs), which can degrade the performance. As a solution, FedSat~\cite{9674028} aims to better align the satellites' model transmission frequencies/periods by assuming that GSs are located at the North Pole. Additionally, buffered asynchronous Fed-LS~\cite{nguyen2022federated} triggers global aggregations after receiving a predefined number of satellite models at GSs. Although these strategies facilitate model aggregations in Fed-LS, an important question remains open: \textit{Can Fed-LS reduce its dependence on long-distance RF links between satellites and aerial/ground nodes?}

\begin{figure}[t]
\centering
\includegraphics[width=0.48\textwidth, trim= 49 1 1 1, clip]{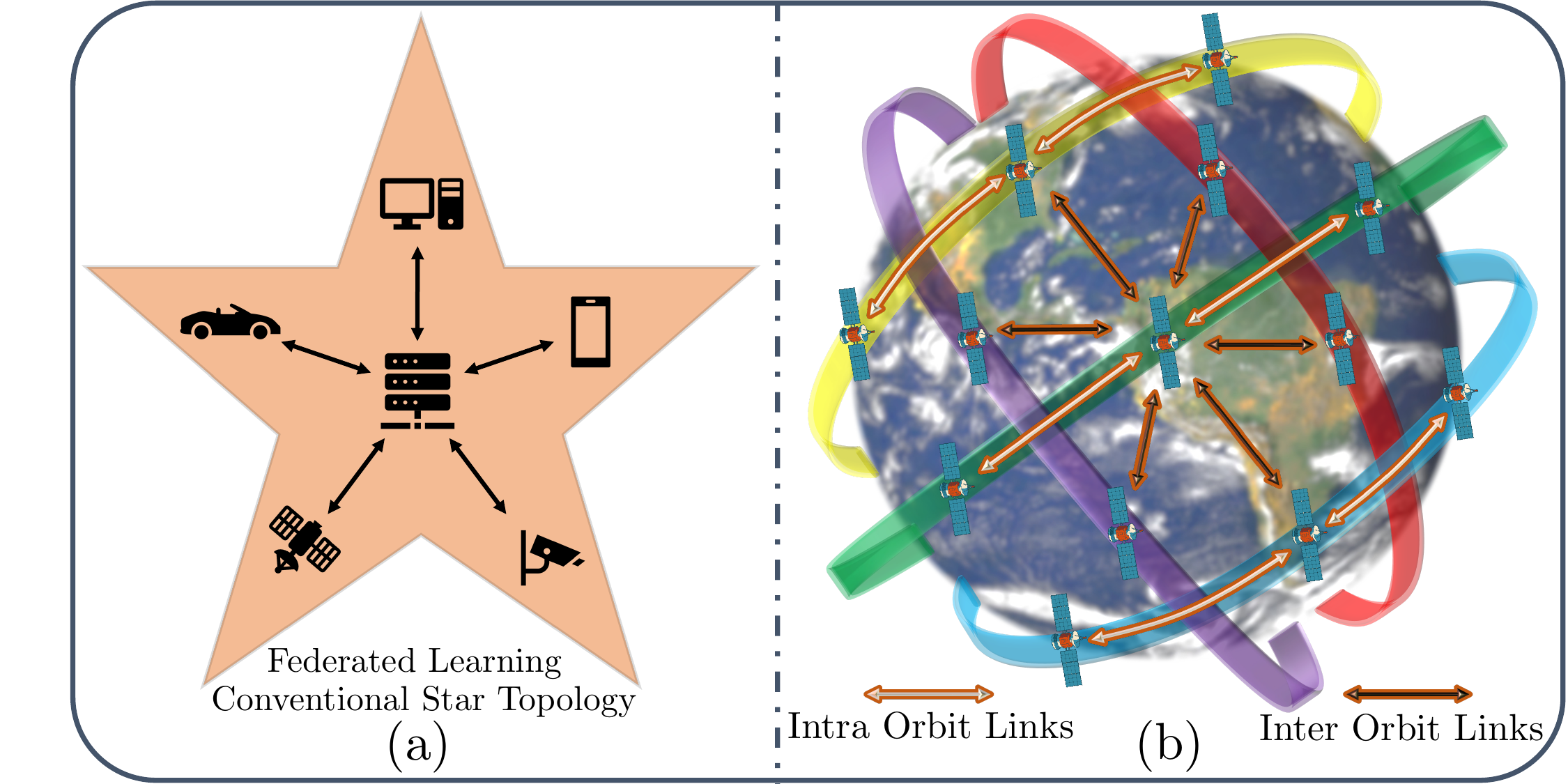}
\vspace{-1mm}
\caption{\textbf{(a)} Star topology of FedL. \textbf{(b)} Satellites with different orbits (denoted by various colors) around the Earth, using inter- and intra-orbit links.}
\label{fig:simpleFL}
\vspace{-0.1mm}
\end{figure}
Addressing this question starts with recognizing that the state-of-the-art LEO satellite constellations, such as Starlink, OneWeb, and Telesat~\cite{zong2021design}, are equipped with optical communication technology featuring multiple laser terminals for \textit{inter-satellite laser links (ISLLs)}\cite{9393372}. ISLLs offer high-data-rate, long-distance, and multi-directional communications, reducing latency and energy consumption while enabling more flexible network construction.
Despite these advantages, integrating ISLLs into Fed-LS architectures entails implementation challenges. In particular, in LEO constellations (Fig.~\ref{fig:simpleFL}(b)), satellites traverse multiple orbits: satellites within the same orbit share similar trajectories, while those in different orbits follow distinct paths. Consequently, \textit{intra-orbit links} (within the same orbit) are often considered to be stable, but \textit{inter-orbit links} (between different orbits) are considered to be available only during brief connection windows between the satellites.
Such a transient nature of inter-orbit links, coupled with the Doppler effect \cite{10043628}, result in unstable inter-satellite links that demand \textit{fine-grained, short time-scale orchestration and control} for reliable integration into Fed-LS.
Moreover, satellites employ multiple \textit{laser terminals} (e.g., four in Starlink \cite{9393372}), each dedicated to a distinct direction, which adds further complexity and necessitates the study of \textit{inter-terminal satellite network formation} (i.e.,  network topologies formed via the satellites' laser terminals). Without proper modeling and optimization, these factors hinder the integration of ISLLs into Fed-LS.

Despite these challenges, the integration of ISLLs into Fed-LS architectures has shown promising results. For instance, FedLEO \cite{elmahallawy2023optimizing} leverages intra-orbit ISLLs by designating a \textit{sink/aggregator} satellite in each orbit to collect and aggregate models. These sink satellites then transmit the aggregated orbital models to the GSs via RF links for global aggregations, reducing dependence on RF links. Similarly, FedShot \cite{elmahallawy2023one} employs a similar sink satellite framework to improve the convergence of ground-assisted Fed-LS through \textit{one-shot} FedL.
Although these works have taken the first steps towards utilizing ISLLs for intra-orbit model aggregations, they continue to depend on RF links for global model aggregations with GSs. As a result, a critical question remains: \textit{With the integration of ISLLs, is there a continued need to communicate with aerial/ground nodes for Fed-LS model aggregations?}

Addressing this question requires revisiting the traditional reliance on aerial or ground nodes for model aggregation in existing Fed-LS architectures. This reliance stems from the historical role of LEO satellites as data collectors, transmitting data to terrestrial nodes for processing. Recent advancements, however, have redefined satellites as edge computing units capable of training ML models as independent computational nodes equipped with ISLLs \cite{tang2021computation,10.1145/3528416.3530985}.
This technological advancement points to a novel approach: performing model aggregations for Fed-LS directly over satellite constellations, without the assistance of aerial or ground nodes. 
However, leveraging ISLLs for this purpose raises key questions, including \textit{when}, \textit{where}, and \textit{how} to employ these links to enable seamless distributed ML training and aggregation across satellite constellations without ground reliance.
Addressing these questions forms the central motivation behind {Fed-Span}.
\subsubsection{Dimensions of Innovations} In this work, we introduce {Fed-Span} based on four innovative design dimensions:

\textbf{(Dimension 1) Over-the-Space Aggregations:} 
{Fed-Span} introduces a spatially-aware system model that accounts for inter- and intra-orbit satellite trajectories and ISLLs. It features fine-grained temporal discretization and ISLL channel modeling, thereby unlocking the potential of ISLLs for Fed-LS.

\textbf{(Dimension 2)  Minimum Spanning Tree:}
{Fed-Span} models LEO satellite constellations as spherical networks, representing satellites as \textit{nodes} and ISLLs as \textit{edges}. It subsequently introduces a novel \textit{graph-theoretical framework} for connectivity, optimization, and network utilization in Fed-LS, marking a novel contribution to the field. This approach adopts \textit{minimum spanning trees} (MSTs) as an intuitive topology for low-overhead model aggregations in Fed-LS, with the conventional FedL star topology recoverable as a special case.
Crucially, {Fed-Span} optimizes the topology of MSTs based on the satellites' trajectories, link availabilities, ML computation capabilities, and data distributions.

\textbf{(Dimension 3)  Hierarchical Aggregations and Clustering:} 
{Fed-Span} reimagines the hierarchical FedL structure used in current Fed-LS implementations, replacing the star topology with stage-wise model aggregations over MSTs formed within satellite clusters. It further extends prior work on clustering satellite constellations~\cite{razmi2024board} by introducing a three-sided trade-off between ML performance, latency, and energy consumption to guide network formation.
Specifically, {Fed-Span} groups satellites into clusters, termed virtual constellations (VCs), and constructs \textit{multi-objective directed spanning trees} (MoDSTs) within each VC, forming a \textit{multi-objective directed spanning forest} (MoDSF) across all satellites. This approach contributes to both Fed-LS and graph theory \cite{barnes1983graph} by advancing the modeling and formulation of MoDSTs and MoDSFs.
Using these graph topologies, {Fed-Span} performs energy- and latency-efficient two-stage model aggregations. At the first stage, it executes \textit{intra-VC model aggregations} over a MoDSF to derive \textit{VC models}. Then, at the second stage, an MoDST that spans all satellites enables \textit{global model aggregations}. 

\textbf{(Dimension 4)  Satellite-Specific Constraints and Dynamics:} 
{Fed-Span} addresses two design considerations overlooked in existing Fed-LS studies~\cite{10021101, 9674028, nguyen2022federated, elmahallawy2023optimizing, elmahallawy2023one}. First, 
compared to these studies that often assume static datasets at satellites, ignoring the \textit{temporal variations of data} caused by satellites' movements and human activities on the Earth, {Fed-Span} accounts for the satellites' data dynamics. Second, compared to these studies that presume continuous local model training and aggregations for satellites, an unrealistic assumption given the satellites' battery limitations, {Fed-Span} considers optimized \textit{inactive/idle times} for satellites, making it a more practical/sustainable solution for Fed-LS.
\vspace{-1mm}
\subsection{Overview and Summary of Contributions}
Our major contributions can be summarized as follows:
\begin{itemize}[leftmargin=4.0mm]
    \item 
    We introduce {Fed-Span}, a novel framework that enables timely, over-the-space Fed-LS operations by explicitly modeling ISLLs and satellite laser terminal characteristics. Unlike conventional approaches, {Fed-Span} features a \textit{dynamic learning topology represented via MSTs that naturally adapts to satellite network configurations}.
    \item We offer a fresh perspective on the graph-theoretic concepts of \textit{trees} and \textit{forests} by modeling their topological structures through a set of continuous constraint representations (CCRs). Building on this, we formalize {Fed-Span} operations by establishing a connection between the model aggregation and dispatching phases of FedL and these CCRs. 
    \item We analytically characterize the ML model convergence of {Fed-Span}, deriving new convergence bounds for non-convex ML loss functions. These bounds, for the first time in the literature, account for (i) varying numbers of SGD iterations across the satellites, (ii) flexible frequencies/periods of model aggregations, (iii) both  intra- and inter-VC data heterogeneity of satellites, (iv) time-varying datasets on satellites, and (v) idle periods for satellites.
    \item We orchestrate {Fed-Span} through an optimization problem that jointly optimizes ML- and network-related parameters, including \textit{the topology of MoDSTs/MoDSFs}, \textit{resource allocation} across the satellites, and \textit{ML training schedules}. This formulation is one of the first to link the construction of optimal graph topologies to the underlying performance of distributed learning, measured by the ML model loss, energy consumption, and latency. We further reveal that this problem is an NP-hard Signomial programming formulation.
    \item By leveraging the inherent structure of the problem, we present a systematic approach to transform it into \textit{continuous non-convex programming} using a series of inequalities. The problem is then reformulated as a tractable \textit{geometric programming} and solved using a successive convex approximation methodology with performance guarantees. 
    \item We evaluate {Fed-Span} on real-world datasets, demonstrating its superior performance in terms of model accuracy, energy consumption, and latency compared to other methods. 
\end{itemize}

\section{System Model and Design Consideration}\label{sec:system_model}
\noindent In this section, we lay the groundwork for {Fed-Span}: we outline the properties of ISLL terminals in Sec.~\ref{sec:terminals}, their data rates in Sec. \ref{sec:transmission_data_Rate}, and the clustering of satellites in Sec.~\ref{sec:ML-sysmod}. For reference, in Appendix~\ref{app:notaions}, we provide a complete list of all notations and abbreviations used throughout the paper.




\begin{remark}[Multi-Granularity Representation of Network Operations]
As will be discussed later, {Fed-Span} unfolds through a series of local and global model aggregation rounds across the satellites. In this setting, to address the rapidly evolving structure of satellite constellations, we disentangle the notions of time and model aggregation rounds. This is because the satellites' locations may significantly change during one global/local aggregation round, making the network formation for the beginning and end of an aggregation round notably different. In particular, we introduce multiple levels of time granularity: (i) the wall-clock time instances denoted by {\small$t$} (measured in seconds, mostly used in modeling the satellite constellation structure), (ii) global aggregation rounds denoted by {\small$k$}, and (iii) local aggregation rounds denoted by {\small$\ell$}.
\end{remark}
\vspace{-5mm}
\subsection{Modeling of ISLL Terminals}\label{sec:terminals}
We consider an LEO constellation with {\small$N$} satellites, represented by the set {\small$\mathcal{N} = \{1, \dots, N\}$}, where each satellite {\small$n \in \mathcal{N}$} is equipped with {\small$M_n$} laser terminals/ports, denoted by the set {\small$\mathcal{M}_n$}. For simplicity, we assume that all satellites have the same number of terminals, {\small$M = |\mathcal{M}_n|$,~$\forall n$}.\footnote{For instance, SpaceX satellites feature four laser terminals~\cite{9393372}.} Each terminal {\small$m \in \mathcal{M}_n$} can transmit/receive data to/from terminals of other satellites.
To model ISLL formation between two satellites {\small$n, n' \in \mathcal{N}$}, we define a binary decision variable {\small$\psi_{m,m'}(t)$}, where {\small$\psi_{m,m'}(t) = 1$} indicates that terminal {\small$m \in \mathcal{M}_n$} establishes an ISLL to terminal {\small$m' \in \mathcal{M}_{n'}$} at time {\small$t$}; otherwise, {\small$\psi_{m,m'}(t) = 0$}. Subsequently, the connectivity matrix {\small$\Psi_{n,n'}(t) = [\psi_{m,m'}(t)]_{m \in \mathcal{M}_n, m' \in \mathcal{M}_{n'}}$} represents ISLLs between two satellites {\small$n$} and {$n'$}. Across all satellites, we define the \textit{connectivity block matrix} (CBM) as {\small$\bm{\Psi}(\mathcal{N},t) = [\Psi_{n,n'}(t)]_{n, n' \in \mathcal{N}}$}, an {\small$N \times N$} block matrix, where each block is an {\small$M \times M$} matrix.
Since intra-satellite connections are infeasible, the diagonal blocks of the CBM are  {\small$M \times M$} zero matrices, i.e., {\small$\psi_{m,m'}(t) = 0,~\forall m, m' \in \mathcal{M}_n,~\forall n$}. To incorporate this constraint in our later formulations, we express it as a continuous constraint representation (CCR):
\begin{equation}\label{eq:hollow}
  \sum_{n \in \mathcal{N}} \sum_{m \in \mathcal{M}_n}\sum_{m' \in \mathcal{M}_n} \psi_{m,m'}(t) = 0,
\end{equation}
which characterizes CBM as a hollow block-diagonal matrix. \textit{Such CCR representations will also be used throughout the paper, mainly to make the formulations tractable for later integration into our main formulation in Sec.~\ref{sec:optimization_problem}.}


\begin{remark}[Interpretation of Continuous Constraint Representations (CCRs)]
The term ``continuous" in CCR stems from the fact that these constraints are expressed using continuous mathematical operations, such as summations and products, rather than discrete logic (e.g., \texttt{if}/\texttt{else} conditions), which in turn enables their integration into continuous optimization frameworks. More specifically, CCRs readily accommodate relaxed (i.e., continuous) versions of binary decision variables, which is essential for developing a tractable solution to our later-described optimization problem. The binary nature of the decision variables is later reinstated through the integration of a set of nonconvex constraints, whose formulation and enforcement are detailed in the Appendix~\ref{app:optTransform}  (the specific transformation is described in Appendix~\ref{subsec:RelaxBinary}).
\end{remark}

To ensure that (i) each terminal {\small$m$} can either send or receive (not both) information, and (ii) each satellite pair {\small$n$} and {\small$n'$} can only establish one ISLL at a time, we introduce two CCRs:
\begin{equation}\label{rate2}
    \hspace{-1.5mm}\sum_{n'\in\mathcal{N}} \hspace{-0.5mm}\sum_{m'\in\mathcal{M}_{n'}}\hspace{-2mm}\psi_{m,m'}(t) + \psi_{m',m}(t) {\le}1, ~~\forall m\in\mathcal{M}_n, n\in\mathcal{N},
\end{equation}
\vspace{-1mm}
\begin{equation}\label{rate3}
    \hspace{-4mm}\sum_{m\in\mathcal{M}_n} \hspace{-0.5mm} \sum_{m'\in\mathcal{M}_{n'}}\hspace{-2mm} \psi_{m,m'}(t) + \psi_{m',m}(t) {\le}1, ~~\forall n,n' \in \mathcal{N}.
\end{equation}
\vspace{-7mm}
\subsection{Transmission Data Rate Modeling for ISLLs}\label{sec:transmission_data_Rate}
\vspace{-.5mm}
 ISLLs utilize free-space optics (FSO), providing low-latency and interference-resistant communications. In this setting, upon transmitting a signal from terminal {\small$m\in \mathcal{M}_n$} of satellite {\small$n$} at time {\small$t$} with the transmit power of {\small$P^{\mathsf{Tx}}_{m}$}, the received signal power at terminal {\small$m'\in \mathcal{M}_n$} of satellite {\small$n'$} is given by:
\begin{equation}\label{eq:CG}
\hspace{-1.5mm}
\begin{aligned}
P^{\mathsf{Rx}}_{m,m'}(t) &= P^{\mathsf{Tx}}_{m} ~a^{\mathsf{Rx}}_{m'}~ a^{\mathsf{Tx}}_{m} ~G^{\mathsf{Rx}}_{m'} ~G^{\mathsf{Tx}}_{m}~ L^{\mathsf{Rx}}_{m'} ~L^{\mathsf{Tx}}_{m} ~L^{\mathsf{PL}}_{m,m'}(t),
\end{aligned}
\hspace{-2.5mm}
\end{equation}
where {\small$a^{\mathsf{Rx}}_{m'}$} and {\small$a^{\mathsf{Tx}}_{m}$} are the optics efficiency of the receiver and transmitter, {\small$G^{\mathsf{Rx}}_{m'}$} and {\small$G^{\mathsf{Tx}}_{m}$} are the receiver and transmitter gains, and {\small$ L^{\mathsf{Rx}}_{m'}$ and $L^{\mathsf{Tx}}_{m}$} are the receiver and transmitter pointing losses, respectively~\cite{9842823}. Also, {\small$L^{\mathsf{PL}}_{m, m'}(t) = \left(\frac{\lambda_{m, m'}}{4\pi \delta_{m, m'}(t)}\right)^{2}$} is the FSO path loss,  where $\lambda_{m, m'}$ is the signal wavelength at receiver $m'$, influenced by the Doppler effect (considered in Sec.~\ref{simulations}), and {\small$\delta_{m,m'}(t)$} is the distance between terminals {\small$m$} and {\small$m'$}. 

To obtain {\small$\delta_{m,m'}(t)$}, we define the spherical position of terminal {\small$m$} as {\small$\left[ \phi_m^{\mathsf{Lat}}(t), \lambda_m^{\mathsf{Long}}(t) \right]$}, where {\small$\phi_m^{\mathsf{Lat}}(t)$} and {\small$\lambda_m^{\mathsf{Long}}(t)$} are the latitude and longitude components, respectively.
Subsequently, using the Haversine formula \cite{1068eb6b-091c-3263-9522-c68d480117c4}, {\small$\delta_{m,m'}(t)$} is given by: 
\begin{equation}\label{eq:Haversine}
\delta_{m,m'}(t) = 2 (R^{\mathsf{Earth}} + h^{\mathsf{Cons}}) \atantwo \left(\sqrt{\widehat{a}}, \sqrt{1-\widehat{a}}\right),
\end{equation} 
where {\small$R^{\mathsf{Earth}}$} is the radius of the Earth, {\small$h^{\mathsf{Cons}}$} is the constellation altitude, and {\small$\widehat{a}$} is the haversine of the central angle, given by {\small$\widehat{a} = \sin^2(\frac{\Delta \phi^{\mathsf{Lat}}_{m,m'}}{2}) + \cos(\phi_m^{\mathsf{Lat}}(t)) \cos(\phi^{\mathsf{Lat}}_{m'}(t)) \sin^2(\frac{\Delta \lambda^{\mathsf{Long}}_{m,m'}}{2})$}, where {\small$\Delta \phi^{\mathsf{Lat}}_{m,m'} = \phi^{\mathsf{Lat}}_{m} - \phi^{\mathsf{Lat}}_{m'}$} and {\small$\Delta \lambda^{\mathsf{Long}}_{m,m'} = \lambda^{\mathsf{Long}}_{m} - \lambda^{\mathsf{Long}}_{m'}$}. 

Considering~\eqref{eq:CG} and~\eqref{eq:Haversine}, we obtain the data rate (in bits/sec) between terminals $m$ and $m'$ at time $t$ as follows:
\begin{equation}
    \overline{\mathfrak{R}}_{m,m'}(t) = B \log_{2}\left(1 + {P^{\mathsf{Rx}}_{m,m'}(t)}\big/{P^{\mathsf{N}}_{m'}(t)}\right),
\end{equation}
where {\small$B$} is the transmission bandwidth, and {\small$P^{\mathsf{N}}_{m'}(t)$} is the noise power at the receiving terminal {\small$m'$}. 
To ensure that all communication links established via a CBM {\small$\bm{\Psi}(\mathcal{N},t)$} meet a minimum data rate {\small$\mathfrak{R}^{\mathsf{min}}$}, we impose the following CCR:
\begin{equation}\label{eq:min_data_rate}
\hspace{-3mm}
 \big(\overline{\mathfrak{R}}_{m,m'}(t) - \mathfrak{R}^{\mathsf{min}}\big) \psi_{m,m'}(t) {\geq} 0, \forall m\hspace{-1mm}\in\hspace{-1mm}\mathcal{M}_n, \forall m'\hspace{-1.5mm}\in\hspace{-1mm}\mathcal{M}_{n'},
\hspace{-2mm}
\end{equation}
implying that {\small$\psi_{m,m'}(t)$} can only take the value of one (i.e., establishment of an ISLL) when {\small$\mathfrak{R}_{m,m'}(t) > \mathfrak{R}^{\mathsf{min}}$}; otherwise {\small$\psi_{m,m'}(t)=0$} to satisfy~\eqref{eq:min_data_rate}.
Subsequently, the data rate from satellite {\small$n$} to satellite {\small$n'$} at time {\small$t$} under CBM {\small$\bm{\Psi}(\mathcal{N},t){=}[\Psi_{n,n'}(t)]_{n,n'{\in}\mathcal{N}}$}, where {\small$\Psi_{n,n'}(t) = [\psi_{m,m'}(t)]_{m {\in}\mathcal{M}_n, m'{\in}\mathcal{M}_{n'}}$}, is given by:
\begin{equation}\label{eq:finalRate}
\hspace{-3mm}
   \mathfrak{R}_{n,n'}\big(\bm{\Psi}(\mathcal{N}, t)\big) {=} \sum_{m\in\mathcal{M}_{n}} \sum_{m'\in\mathcal{M}_{n'}}\psi_{m,m'}(t)\overline{\mathfrak{R}}_{m,m'}(t).
\hspace{-3mm}
\end{equation}
\begin{figure*}[t]
\vspace{-7mm}
\centering
\noindent\includegraphics[width=\textwidth]{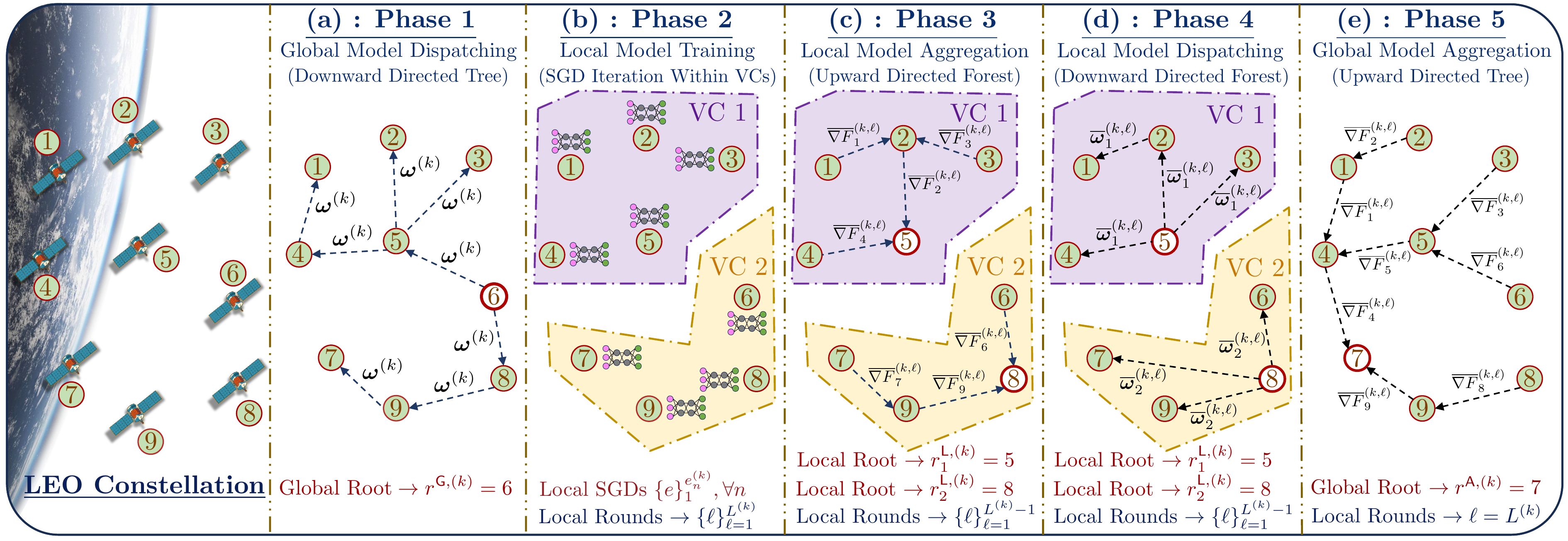} 
\vspace{-5.7mm}
\caption{A schematic of the operations that take place in each global round of {Fed-Span}. \textbf{{(a) Phase 1:}} The global model is dispatched through a root-to-leaf continuum using a downward-directed tree structure. \textbf{{(b) Phase 2:}} Upon receiving the global model, satellites conduct local model training rounds. \textbf{(c) Phase 3:} Once each local training round is completed, a local model aggregation occurs, where satellites are organized into Virtual Clusters (VCs) and transfer their models to the respective VC root node via upward-directed trees, resembling a forest topology. \textbf{(d) Phase 4:} The aggregated models of VCs are disseminated back to their satellites using downward-directed trees. \textbf{(e) Phase 5:} After the execution of the last local training round of each global round, 
satellites' models are aggregated at the global root node through model transfers along a leaf-to-root continuum via an upward-directed tree. The newly aggregated global model is then dispatched across the satellites through \textbf{Phase 1}, which marks the start of the next global round.}
\label{fig:system_model}
\vspace{-6mm}
\end{figure*} 
\vspace{-6mm}
\subsection{Virtual Constellations (VCs)}\label{sec:ML-sysmod}
{Fed-Span} operates over {\small$K$} global rounds, collected by the set {\small$\mathcal{K} = \{1, 2, \dots, K\}$}. For each global round {\small$k \in \mathcal{K}$}, 
we define a set of VCs {\small$\mathcal{C}^{(k)}$} partitioning the satellite constellation, with {\small$C^{(k)} = |\mathcal{C}^{(k)}|$} denoting the number of VCs. 
To model satellite-to-VC assignment/association, let {\small$\mathcal{N}^{(k)}_{c}=\{n|n\in\mathcal{N},\gamma^{(k)}_{c,n}=1\}$} comprise the satellites contained in a VC $c\in \mathcal{C}^{(k)}$, where {\small$\gamma^{(k)}_{c,n}\in\{0,1\}$} is a binary decision variable, capturing whether satellite $n$ is contained in VC $c$ ({\small$\gamma^{(k)}_{c,n}=1$}) or not ({\small$\gamma^{(k)}_{c,n}=0$}).
We ensure proper satellite-to-VC association via three CCRs:
\begin{align}
    &\sum_{c{\in}\mathcal{C}^{(k)}}\sum_{n\in\mathcal{N}} \gamma^{(k)}_{c,n}= N, & \forall k \in \mathcal{K},\label{cons:CSC_4}\\
    &\sum_{c{\in}\mathcal{C}^{(k)}} \gamma^{(k)}_{c,n}=1,&\forall n\in\mathcal{N},~\forall k \in \mathcal{K},\label{cons:CSC_2}\\
    &     \sum_{n\in\mathcal{N}} \gamma^{(k)}_{c,n}\ge 2,&\forall c\in\mathcal{C}^{(k)},\label{cons:CSC_1}
\end{align}
which imply that (i) VCs cover all the satellites, (ii) each satellite is contained in only one VC, and (iii) a VC exists if it is associated with at least two satellites.
\vspace{-.3mm}
\vspace{-1.5mm}
\section{Fundamentals of Machine Learning (ML) and Integration of Tree Structures in  {\large{{Fed-Span}}}}\label{sec:fundamentals}
\noindent In this section, we first outline the ML objective of {Fed-Span} in Sec.~\ref{sec:MLform}. Then, we provide an overview of the learning methodology employed by {Fed-Span} in Sec.~\ref{sec:nutshell}. Finally, we present a sketch of the tree topologies used in {Fed-Span} in Sec.~\ref{subsec:sketch}, which serve as a basis for our later formulations.

\vspace{-3.5mm}
\subsection{ML Formulation and Overarching Goal of {Fed-Span}}\label{sec:MLform}
We extend the literature on Fed-LS, which focuses on static datasets at satellites \cite{9749193,10021101, chen2022satellite ,nguyen2022federated, 10216376, 9674028, tang2022federated,10121575,10039157, elmahallawy2023optimizing, elmahallawy2023one}, to a scenario that incorporates real-time data collection at satellites. Specifically, during global round {\small$k$}, each satellite {\small$n \in \mathcal{N}$} is assumed to hold a time-varying dataset {\small$\mathcal{D}_n^{(k)}$} with the size of {\small$D_n^{(k)}=|\mathcal{D}_n^{(k)}|$}. Additionally, the cumulative dataset across all satellites is represented as {\small$\mathcal{D}^{(k)}{=}\cup_{n{\in}\mathcal{N}}\mathcal{D}_{n}^{(k)}$}, with the size of {\small$D^{(k)}{=}\sum_{n{\in}\mathcal{N}}D_{n}^{(k)}$}.
We consider the ML \textit{loss function}  {\small${f}(\bm{\omega},d)$} (e.g., cross-entropy) to quantify the  model performance for each data point {\small$d$} under model parameter {\small$\bm{\omega}\in\mathbb{R}^{M^{\mathsf{Dim}}}$}, where $M^{\mathsf{Dim}}$ denotes the model's dimension. Subsequently, letting {\small$F_{n}^{(k)}(\bm{\omega}){\triangleq} F_{n}(\bm{\omega}, \mathcal{D}_n^{(k)})=\sum_{d \in \mathcal{D}_n^{(k)}} {{f}(\bm{\omega},d)}\big/{D_n^{(k)}}$} denote the local loss at satellite {\small$n$} for ML model {\small$\bm{\omega}$},  we define the \textit{global loss} of the system at global round $k$ as follows:
\begin{equation}\label{eq:globalLoss}
 F^{(k)}(\bm{\omega})=\sum_{n\in \mathcal{N}} \frac{D_n^{(k)}}{D^{(k)}} F^{(k)}_{n}(\bm{\omega}).
\end{equation}

Due to the temporal evolution of local datasets, the optimal global model is time-varying, and thus {Fed-Span} aims to track the sequence of optimal models {\small $\big\{\bm{\omega}^{{(k)}^\star}\big\}_{k=1}^{K}$}, where:
\begin{equation}\label{eq:genForm}
\bm{\omega}^{{(k)}^\star}= \underset{\bm{\omega}\in \mathbb{R}^M}{\argmin} \;  F^{(k)}(\bm{\omega}), ~\forall k\in\mathcal{K}.
\end{equation}

\vspace{-4mm}
\subsection{{Fed-Span} in a Nutshell and Significance of Tree Structures}\label{sec:nutshell}
{Fed-Span} achieves the aforementioned goal via a sequence of operations categorized into five phases, detailed later in Sec.~\ref{sec:phasesofFEDSPAN}. An overview of these phases is presented in Fig.~\ref{fig:system_model}. In the following, we present a bird’s-eye view of these five phases before delving into their mathematical modeling. 

The model-training cycle of {Fed-Span} begins with \textbf{\textit{Phase~1}}, where the global model is dispatched across the satellites to initiate a new global round. This is followed by \textbf{\textit{Phase 2}}, during which local model training takes place. Each local training round is then followed by a local model aggregation in \textbf{\textit{Phase 3}}, which in turn is followed by a local model dispatching in \textbf{\textit{Phase 4}}.
In practice, multiple local training,  aggregation, and dispatching rounds (i.e., transitions between \textbf{\textit{Phases 2\&3\&4}}) are executed in each global round to progressively debias the satellites' models during local training. Finally, after the last local aggregation round, a global model aggregation is performed in \textbf{\textit{Phase 5}} to produce the updated global model, which is then dispatched across the satellites, again via \textbf{\textit{Phase 1}}, to initiate the next global round.

To conduct model dispatching in global and local dispatching rounds (i.e., \textbf{\textit{Phases 1\&4}}), {Fed-Span} forms a set of \textit{minimum spanning trees} (MSTs) across the satellites. As these MSTs are used to dispatch the model from an aggregation point to the rest of the satellites, they are referred to as \textit{downward-directed trees}.
Similarly, to conduct model aggregation in global and local aggregation rounds (i.e., \textbf{\textit{Phases 3\&5}}), {Fed-Span} forms another set of MSTs across the satellites. As these MSTs collect model updates from all satellites toward an aggregation point, they are referred to as \textit{upward-directed trees}.

The selection and integration of MST structures in {Fed-Span} for model aggregation and dispatching are motivated by two factors:
(i) MST structures are inherently suited for information dissemination with minimal redundancy: they possess the minimum number of edges required to maintain connectivity while avoiding cycles, thus eliminating redundant information exchanges across the network.
(ii) These structures can be carefully designed/formed to construct multi-objective directed spanning trees (MoDSTs). In particularity, in this work, we provide the first systematic approach to create a set of MoDSTs that can jointly optimize: (a) the energy consumption of model transfers, (b) the latency of model transfers, and (c) the ML performance of the model, which are the three objectives of this work reflected in our overarching problem formulation in Sec.~\ref{sec:optimization_problem}. Notably, constructing MoDSTs is known to be an NP-hard problem, and a major contribution of this work is showing that MoDST formation can be achieved through a set of non-convex optimization techniques with convergence guarantees, enabled by our careful derivation of CCRs.

Next, we present a sketch of the upward- and downward-directed trees, representing them as CCRs to make them optimizable. This provides a fresh perspective on the modeling and optimization of these trees. Henceforth, the set of nodes that can be reached from a given node in a single hop is referred to as its \textit{successors} (e.g., nodes 2, 3, and 4 are successors of node 5 in Fig.~\ref{fig:system_model}(a)). Conversely, for these successor nodes, the node above them is referred to as their \textit{predecessor}.

\vspace{-3mm}
\subsection{Sketch of Upward- and Downward-Directed Trees}\label{subsec:sketch}
We first define {\small$\bm{\Gamma}\big(\bm{\Psi}(\mathcal{N},t)\big)$} as an \textit{adjacency matrix} of a directed graph, formed by CBM {\small$\bm{\Psi}(\mathcal{N},t)$}, which was introduced in Sec.~\ref{sec:terminals}. In particular, let {\small$\bm{\Gamma}_{n,n'}(\bm{\Psi}(\mathcal{N},t))$} denote the element in the row {\small$n$} and column {\small$n'$} of this matrix. We have:
\begin{equation}\label{edgetolink}
 \hspace{-2mm} \bm{\Gamma}_{n,n'}(\bm{\Psi}(\mathcal{N},t))= \sum_{m\in\mathcal{M}_{n}} \sum_{m'\in\mathcal{M}_{n'}}\hspace{-2mm}\psi_{m,m'}(t), ~\forall n,n' \in  \mathcal{N}.\hspace{-2mm}
\end{equation}
In words, {\small$\bm{\Gamma}_{n,n'}(\bm{\Psi}(\mathcal{N},t))=1$} if an ISLL is established from one of the laser terminals of satellite {\small$n$} to satellite {\small$n'$}; otherwise {\small$\bm{\Gamma}_{n,n'}(\bm{\Psi}(\mathcal{N},t))=0$}. 
The upward- and downward-directed trees formed during the various phases of {Fed-Span} correspond to link formations determined by different CBMs {\small$\bm{\Psi}(\mathcal{N},t)$} at different times {\small$t$}. These trees will be concretized as we detail the five phases of {Fed-Span} later in Sec.~\ref{sec:phasesofFEDSPAN}. Nevertheless, a common characteristic across all these trees is the presence of a root and the minimum number of edges to connect the nodes, which are formalized below as CCRs.

\subsubsection{Presence of a Root}
In a directed tree structure, the root node lacks a higher-level predecessor and plays a pivotal role in determining the tree's directionality. To explicitly capture the root node,  extending the above-described notion of an adjacency matrix, we represent the adjacency matrix of a directed tree as {\small$\bm{\Gamma}\big(\bm{\Psi}(\mathcal{N},t), r\big)$}, where {\small$r$} denotes the root node.

We first aim to model the root node for global model aggregation and dispatching, which corresponds to two phases of {Fed-Span} (Phases 1\&5 detailed later in Sec.~\ref{subsec:ph1} and~\ref{subsec:ph5}). 
In these two phases, the models from all satellites are first gathered at the root node through an upward-directed tree. Subsequently, the same root node dispatches the aggregated model back to all satellites via a downward-directed tree (Fig.~\ref{fig:system_model}).
Let {\small$r^{\mathsf{G}, (k)} \in \mathcal{N}$} represent the root node for the global round {\small$k$}, which we formalize using the following CCR:
\begin{equation}\label{eq:rootDef}
r^{\mathsf{G}, (k)}=\sum_{n\in\mathcal{N}}  \pi_{n}^{\mathsf{G}, (k)}~ \cdot ~n,~\forall k \in\mathcal{K},
\end{equation}
where {\small$ \pi_{n}^{\mathsf{G}, (k)}\in\{0,1\}$} is a binary decision variable. In particular, {\small$\pi_{n}^{\mathsf{G}, (k)}=1$}, which results in {\small$r^{\mathsf{G}, (k)}{=}n$} according to \eqref{eq:rootDef}, implies that satellite $n$ is selected to be the root; otherwise {\small$\pi_{n}^{\mathsf{G}, (k)}=0$}.
Since only one satellite can serve as the root node during each global round, we enforce the following CCR:
\begin{equation}\label{cons:SGD_1}
    \sum_{n\in\mathcal{N}}  \pi_{n}^{\mathsf{G}, (k)} =1,~\forall k \in\mathcal{K}.
\end{equation}
~~We next focus on local model aggregation and dispatching during global round {\small $k$}, which corresponds to another two phases of {Fed-Span} (Phases 3\&4 detailed in Sec.~\ref{subsec:ph3} and~\ref{subsec:ph4}). 
In these two phases, for each VC, the models from all satellites belonging to the VC are first gathered at a root node through an upward-directed tree. Subsequently, the same root node dispatches the aggregated model back to all satellites in the VC through a downward-directed tree (Fig. \ref{fig:system_model}(d)).
Let {\small$r_{c}^{\mathsf{L}, (k)}$} denote the index of the local root node for VC {\small$c$}, which we formalize using the following CCR:
\begin{equation}\label{eq:rootDef2}
r_{c}^{\mathsf{L}, (k)}=\sum_{n{\in}\mathcal{N}} \pi_{n}^{\mathsf{L}, (k)} \cdot \gamma^{(k)}_{c,n} \cdot n,~~~ \forall c \in \mathcal{C}^{(k)},
\end{equation}
where {\small$\pi_{n}^{\mathsf{L}, (k)} \in \{0,1\}$} is a binary decision variable and {\small$\gamma^{(k)}_{c,n}$} represents the satellite-to-VC association defined in Sec.~\ref{sec:ML-sysmod}.
Specifically, {\small$\pi_{n}^{\mathsf{L}, (k)} = 1$}, which results in {\small$r_{c}^{\mathsf{L}, (k)} = n$} as per \eqref{eq:rootDef2}, signifies that satellite $n$ is selected as the local root for VC {\small $c$}; otherwise {\small$\pi_{n}^{\mathsf{L}, (k)} = 0$}.
To ensure that, during each global round {\small $k$}, only one satellite in each VC is the root, we impose:
\begin{equation}\label{cons:SCA_1}
    \sum_{n{\in}\mathcal{N}} \pi_{n}^{\mathsf{L}, (k)} \cdot \gamma^{(k)}_{c,n} =1,~  \forall c\in\mathcal{C}^{(k)}.
\end{equation}

\subsubsection{Edge Configuration}
The edge configuration of a tree must satisfy three key coupled conditions:
(i) 
The tree must contain exactly one fewer edge than its number of nodes.
(ii) The edges must form a connected structure.
(iii) The tree must be cycle-free.
To present a sketch of these conditions, we consider {\small$\bm{\Gamma}\big(\bm{\Psi}(\mathcal{N},t), r\big)$} to represent an arbitrary tree over the entire constellation {\small$\mathcal{N}$}. We can then enforce the correct number of edges in this tree using the following CCR:
\begin{equation}\label{cons:tree_1}
  \sum_{n \in \mathcal{N}} \sum_{n' \in \mathcal{N}}\sum_{m\in\mathcal{M}_{n}} \sum_{m'\in\mathcal{M}_{n'}}\psi_{m,m'}(t) = N-1.
\end{equation}
Additionally, we enforce a structure that is both connected and cycle-free by ensuring that there is only a single connected path between each node in the tree and its root. For a downward-directed tree, this condition is equivalent to the following CCR:
\begin{equation}\label{cons:tree_2}
    \sum_{q=1}^{N-1}\Big(\bm{\Gamma}\big(\bm{\Psi}(\mathcal{N},t), r\big)\Big)^{q}_{r,n}=1,~\forall n\in\mathcal{N}\setminus \{r\}.
\end{equation}
For an upward-directed tree, we can express this condition as:
\begin{equation}\label{cons:tree_3}
    \sum_{q=1}^{N-1}\Big(\bm{\Gamma}\big(\bm{\Psi}(\mathcal{N},t), r\big)\Big)^{q}_{n,r}=1,~\forall n\in\mathcal{N}\setminus \{r\}.
\end{equation}
In the expressions above, we exploited the fact that {\small$\Big(\bm{\Gamma}\big(\bm{\Psi}(\mathcal{N},t), r\big)\Big)^{q}_{r,n}$}, which is the element in row {\small $r$} and column {\small $n$} of the q-th power of {\small$\bm{\Gamma}\big(\bm{\Psi}(\mathcal{N},t), r)$}, represents the number of paths of length  {\small $q$} between the root and node {\small $n$}. Similarly, {\small$\Big(\bm{\Gamma}\big(\bm{\Psi}(\mathcal{N},t), r\big)\Big)^{q}_{n,r}$} represents the number of paths of length {\small $q$} between node {\small $n$} and the root.
Note that~\eqref{cons:tree_1},~\eqref{cons:tree_2}, and~\eqref{cons:tree_1} provide only a sketch of the models used to form our MoDSTs for global and local model aggregation and dispatching in {Fed-Span}. These modelings are tailored to their specific deployment and are detailed in Sec.~\ref {sec:phasesofFEDSPAN}.

\vspace{-2mm}
\section{{Fed-Span} Operations}\label{sec:phasesofFEDSPAN}
\noindent Building on the above-discussed modeling of upward- and downward-directed trees, we next detail the five phases of {Fed-Span} in chronological order. To simplify their explanations, each phase is divided into three parts: \textit{(i) Time Specification}, detailing the schedule of the phase.
\textit{(ii) Topological Properties and Network Operation}, outlining network formation and the processes involved.
\textit{(iii) Energy and Latency Computations}, deriving closed-form expressions for energy and latency during the phase.
These five phases form a cycle, starting with an initial model broadcast in \textbf{\textit{Phase 1}} and concluding with a global model aggregation in \textbf{\textit{Phase 5}}, collectively constituting a global round in {Fed-Span}. 

A pseudo-code of these five phases is presented in Alg.~\ref{alg:Fed-Span_alg} and their timeline is visualized in Fig.~\ref{fig:training_process} (the same color-coding is used in both the algorithm and figure to enhance readability).

\begin{figure}[t]
\vspace{-6mm}
 \centering
\noindent\includegraphics[width=0.5\textwidth]{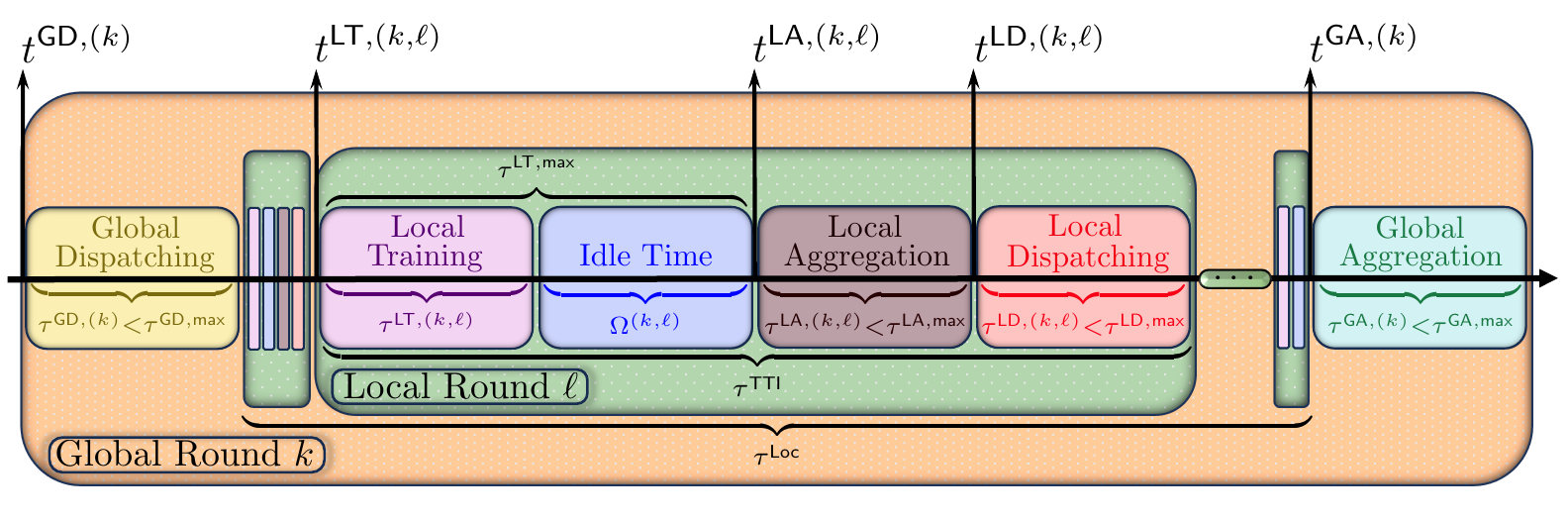} 
\vspace{-7mm}
\caption{Time instances and durations of different operations in {Fed-Span}.}
\label{fig:training_process}
\vspace{-.5mm}
\end{figure}

\vspace{-2mm}
\subsection{\underline{\textbf{Phase 1:}} Global Model Dispatching (lines \ref{Ph1:start}-\ref{Ph1:end} in Alg.~\ref{alg:Fed-Span_alg})}\label{subsec:ph1}
Each global round {\small $k$} of {Fed-Span}, starts with dispatching/broadcasting a global model {\small$\bm{\omega}^{(k)}$}.\footnote{At the initial stage of {Fed-Span}, {\footnotesize $\bm{\omega}^{(0)}$} can be randomly instantiated. In the later stages, this model corresponds to the global model obtained at the end of the previous round, as discussed later in Sec.~\ref{subsec:ph5}.}  This phase is visualized in Fig.~\ref{fig:system_model}(a) and further discussed below.

\subsubsection{Time Specification}
The start of \underline{\textbf{g}}lobal \underline{\textbf{d}}ispatching (GD) for round {\small $k$} occurs at time {\small$t^{\mathsf{GD},(k)} = T^{\mathsf{Init},(k)}$}, when a global model is instantiated at a root satellite and disseminated to all satellites via multi-hop transmissions over ISLLs. The link topology forms a downward-directed tree, as detailed next.

\subsubsection{Topological Properties and Network Operations}
We formalize the global downward-directed tree used for model dispatching at round {\small$k$}, which is formed based on a CBM {\small$\bm{\Psi}(\mathcal{N}, t^{\mathsf{GD},(k)})$}, through its adjacency matrix {\small$\bm{\Gamma}^{\mathsf{GD}, (k)}$} as follows:
\begin{equation}\label{eq:GM_dispatching_adjacney_brevity}
    \bm{\Gamma}^{\mathsf{GD}, (k)}=\bm{\Gamma}\big(\bm{\Psi}(\mathcal{N},t^{\mathsf{GD},(k)}), r^{\mathsf{G}, (k)}\big),
\end{equation}
where {\small$\bm{\Gamma}(.,.)$} was defined in Sec.~\ref{subsec:sketch}. Also, using {\small$ \pi_{n}^{\mathsf{G}, (k)}\in\{0,1\}$} to capture whether satellite {\small$n{\in}\mathcal{N}$} is the root, {\small$r^{\mathsf{G}, (k)}$} and {\small$ \pi_{n}^{\mathsf{G}, (k)}\in\{0,1\}$} hold in~\eqref{eq:rootDef} and~\eqref{cons:SGD_1}.  
Henceforth, we let {\small$\bm{\Gamma}^{\mathsf{GD}, (k)}_{n,n'}$} denote the element at row {\small $n$} and column {\small $n'$} of matrix {\small $\bm{\Gamma}^{\mathsf{GD}, (k)}$}. The compact representation in~\eqref{eq:GM_dispatching_adjacney_brevity} will later enable us to obtain closed-form expressions for energy and latency. 

To ensure that {\small$\bm{\Gamma}^{\mathsf{GD}, (k)}$} complies with the tree prerequisites outlined in Sec.~\ref{subsec:sketch}, we note that imposing \eqref{cons:tree_2} (the cycle-removal constraint) for the downward-directed tree requires matrix multiplications up to the {\small $(N-1)^{\text{th}}$} order, which is computationally intensive. Additionally, without~\eqref{cons:tree_2}, solely applying~\eqref{cons:tree_1} (the edge-limiting constraint) can result in disconnected (cycled) graphs.
To address these challenges, we take two steps:
(i) Tailoring~\eqref{cons:tree_1} specifically for downward-directed trees to reduce the solution space.
(ii) Integrating~\eqref{cons:tree_2} directly into our optimization solution.
Specifically, we modify~\eqref{cons:tree_1} for downward-directed trees, resulting in the following CCR:
\begin{equation}\label{cons:SGD_2}
\hspace{-2mm}~ \overbrace{\sum_{n\in\mathcal{N}}\bm{\Gamma}^{\mathsf{GD}, (k)}_{n,n'}}^{(a)} + \overbrace{\pi_{n'}^{\mathsf{G}, (k)}}^{(b)}{=}1,~\forall n'\in\mathcal{N}.
\end{equation}
This CCR ensures that each satellite {\small $n'$} in the tree either receives the dispatched model from exactly one satellite (i.e., it has one predecessor node, making term (a) equal to one) or serves as the root (making term (b) equal to one), but not both. This is enforced by setting the sum of (a) and (b) to one.
Further, integration of \eqref{cons:tree_2} into our optimization turns out to be a natural phenomenon detailed in the following remark.
\vspace{-2mm}
\begin{remark}[Natural Cycle Removal with Delay and Energy Consideration]\label{remark:Cycle}
    The constraint \eqref{cons:tree_2}, combined with \eqref{cons:tree_1}, serves two purposes: (i) ensuring the tree is connected, and (ii) making it cycle-free. Henceforth, we will derive the energy and latency expressions for model dispatching and aggregations over the downward- and upward-directed trees. These expressions are later integrated into the optimization problem in Sec.~\ref{sec:optimization_problem}, which jointly minimizes the overall energy, latency, and ML model performance metrics. As cycles in the directed graph can lead to redundant model exchanges that in turn cause excessive delay/energy, minimizing delay/energy metrics naturally discourages the formation of cycles, and combined with \eqref{cons:tree_1} or its equivalent \eqref{cons:SGD_2} ensures a proper tree structure.
\end{remark}
\vspace{-2mm}
\subsubsection{Energy and Latency Computations}
Let {\small$\alpha^{\mathsf{Bit}} M^{\mathsf{Dim}}$} denote the total size of the model vector {\small$\bm{\omega}^{(k)}$}, where {\small$M^{\mathsf{Dim}}$} is the model dimension and {\small$\alpha^{\mathsf{Bit}}$} is the number of bits required to represent one parameter of {\small$\bm{\omega}^{(k)}$}.
During the global model dispatching phase, we obtain the latency of transmitting the global model from satellite {\small$n$} to its successor satellite {\small$n'$} as:
\begin{equation}\label{eq:GM_dispatching_transmission_latency}
    \tau_{n,n'}^{\mathsf{GD}, (k)} = {\bm{\Gamma}^{\mathsf{GD}, (k)}_{n,n'}\alpha^{\mathsf{Bit}} M^{\mathsf{Dim}}}\big/{\mathfrak{R}^{\mathsf{GD}, (k)}_{n,n'}},
\end{equation}
where the binary value $\bm{\Gamma}^{\mathsf{GD}, (k)}_{n,n'}$ is added as a multiplicative term to avoid the computation of the delay of unused links and {\small$\mathfrak{R}^{\mathsf{GD}, (k)}_{n,n'} = \mathfrak{R}_{n,n'}\big(\bm{\Psi}(\mathcal{N},t^{\mathsf{GD},(k)})\big)$} is computed using \eqref{eq:finalRate}.
 Subsequently, we obtain the overall delay of model arrival to satellite $n'$ via the following recursive expression:
\begin{equation} \label{eq:GM_dispatching_latency}
    \tau_{n'}^{\mathsf{GD}, (k)}{=}\hspace{-2mm}\sum_{n\in\mathcal{N}\setminus\{n'\}}\hspace{-1mm}\Big(
    \overbrace{
    \bm{\Gamma}^{\mathsf{GD}, (k)}_{n, n'} \tau_{n}^{\mathsf{GD}, (k)}}^{(a)} + \overbrace{\tau_{n, n'}^{\mathsf{GD}, (k)}}^{(b)}\Big),
\end{equation}
where term $(a)$ represents the model arrival delay at the predecessor of satellite {\small$n'$}, while term $(b)$ accounts for the transmission delay from this predecessor node to {\small$n'$}. 
The overall latency of global dispatching is then given by
{\small$\tau^{\mathsf{GD},(k)}=\max_{n'\in\mathcal{N}} \big\{ \tau_{n'}^{\mathsf{GD}, (k)}\big\}$}.
We assume that global dispatching must be completed within a specified time interval as follows:
\begin{equation}\label{cons:GM_dispatch_latency_1}
    \tau^{\mathsf{GD},(k)} \leq \tau^{\mathsf{GD},\mathsf{max}},~~~\forall k \in\mathcal{K}.
\end{equation}
\begin{algorithm}[t]
\caption{Fed-Span operations (to improve readability, the same color coating as in Fig.~\ref{fig:training_process} is used for different phases)}
\label{alg:Fed-Span_alg}
{\footnotesize
\begin{algorithmic}[1]
 \STATE \textbf{Input:} Solving problem $\bm{\mathcal{P}}$ in Sec.~\ref{sec:optimization_problem} and obtaining the link formations {\scriptsize $[\psi_{m,m'}(t)]_{m \in \mathcal{M}_n, m' \in \mathcal{M}_{n'}}$} across time $t$, which leads to the creation of network topology and CBMs {\scriptsize$\bm{\Psi}(\mathcal{N}, t)=[\Psi_{n,n'}(t)]_{n, n' \in \mathcal{N}}$}, where  {\scriptsize$\Psi_{n,n'}(t) = [\psi_{m,m'}(t)]_{m \in \mathcal{M}_n, m' \in \mathcal{M}_{n'}}$} across $t$.
 \\
\STATE Obtaining the adjacency matrices {\scriptsize$\bm{\Gamma}^{\mathsf{GD}, (k)}$}, {\scriptsize$\bm{\Gamma}^{\mathsf{LA}, (k,\ell)}$}, {\scriptsize$\bm{\Gamma}^{\mathsf{LD}, (k,\ell)}$}, {\scriptsize$\bm{\Gamma}^{\mathsf{GA}, (k)}$} based on the CBMs according to~\eqref{eq:GM_dispatching_adjacney_brevity},~\eqref{eq:CM_aggregation_adjacency_brevity},~\eqref{eq:CM_dispatching_adjacency_brevity},~\eqref{eq:GM_aggregation_adjacney_brevity}
\FOR{$k = 0$ to $K-1$}
    \STATE {\rshadeGD{{\textbf{Phase 1}} \textbf{(Global Model Dispatching --- Sec.~\ref{subsec:ph1})}}} \label{Ph1:start}
    \\
    \STATE {Transmission of {\scriptsize$\bm{\omega}^{(k)}$} from  {\scriptsize$r^{\mathsf{G},(k)}$} to all satellites through  {\scriptsize$\bm{\Gamma}^{\mathsf{GD}, (k)}$}}\label{Ph1:end}
    \\
    \FOR{$\ell = 1$ to $L^{(k)}$}
        \FOR{$c = 1$ to $C^{(k)}$}
           \FOR{$n\in \mathcal{N}^{(k)}_c$}
            \STATE {\rshadeLT{{\textbf{Phase 2}} \textbf{(Local Model Training --- Sec.~\ref{subsec:ph2})}}} \label{Ph2:start}
            \\
            \STATE {Executing {\scriptsize$e^{(k)}_{n}$} SGDs to obtain  {\scriptsize$\bm{\omega}_{n}^{(k,\ell), e^{(k)}_{n}}$} as in \eqref{eq:SGD_iterations}}
            \STATE {Computing the cumulative gradient {\scriptsize$\widetilde{\nabla F}_{n}^{(k,\ell)}$} as in \eqref{eq:cumulateGrad}} \label{Ph2:end}
            \STATE {\rshadeIT{{\textbf{Idle Time}}}}\label{ID:start}
            
            \STATE {Remaining idle for the period of {\scriptsize$\Omega^{(k,\ell)}$} given by~\eqref{eq:idletime}}\label{ID:end}
            
            \IF{$\ell \neq L^{(k)}$}
            
            \STATE {\rshadeLA{{\textbf{Phase 3}} \textbf{(Local Model Aggregation --- Sec.~\ref{subsec:ph3})}}} \label{Ph3:start}
            \\ \STATE {Transmitting {\scriptsize$\widetilde{\nabla F}_{n}^{(k,\ell)}$} to the predecessor  using $\bm{\Gamma}^{\mathsf{LA}, (k,\ell)}$}
             \\
            \STATE {If serving as a predecessor, forming  {\scriptsize$\overline{\nabla F}_{n}^{(k,\ell)}$} as in \eqref{clusteraggregation} and transmitting it using $\bm{\Gamma}^{\mathsf{LA}, (k,\ell)}$}
            \\
              \STATE {If serving as the root {\scriptsize$r^{\mathsf{L},(k)}_{c}$}, aggregating the received models to obtain the VC model  {\scriptsize$\overline{\bm{\omega}}^{(k,\ell+1)}_{c}$} as in~\eqref{localagggregationformula}} \label{Ph3:end}
            
            \STATE {\rshadeLD{{\textbf{Phase 4}} \textbf{(Local Model Dispatching --- Sec.~\ref{subsec:ph4})}}}\label{Ph4:start}
            \\ \STATE {Engaging in the dispatching of the aggregated model {\scriptsize$\overline{\bm{\omega}}^{(k,\ell+1)}_{c}$} according to  $\bm{\Gamma}^{\mathsf{LD}, (k,\ell)}$}\label{Ph4:end}
            
            \ELSE
            
            \STATE {\rshadeGA{{\textbf{Phase 5}} \textbf{(Global Model Aggregation --- Sec.~\ref{subsec:ph5})}}}\label{Ph5:start}
            \\
            \STATE {Computing the cumulative gradient for the entire global round as {\scriptsize$\widehat{\nabla F}_{n}^{(k,{\ell})} {=} (\bm{\omega}^{(k)}{-}\bm{\omega}_{n}^{(k,{\ell}),e^{(k)}_{n}})\big/\eta_{_k}$} with {\scriptsize$\ell=L^{(k)}$} and transmitting it to the predecessor using $\bm{\Gamma}^{\mathsf{GA}, (k)}$}
            \\
             \STATE {If serving as a predecessor, forming  {\scriptsize$\overline{\nabla F}_{n}^{(k,\ell)}$}  with {\scriptsize$\ell=L^{(k)}$}  as in \eqref{globalaggregation} and transmitting it according to $\bm{\Gamma}^{\mathsf{GA}, (k)}$}
             \\\STATE {If serving as the global root $r^{\mathsf{A},(k)}$, obtaining the next global model {\scriptsize$\bm{\omega}^{(k{+}1)}$} as in \eqref{finalmodelaggregation}}\label{Ph5:end}
            \ENDIF
          \ENDFOR
        \ENDFOR
    \ENDFOR
\ENDFOR
\end{algorithmic}
}
\end{algorithm}
~~To ensure uninterrupted global model dispatching, the links of the downward-directed tree must maintain connectivity throughout the dispatching phase (i.e., during {\small$[t^{\mathsf{GD},(k)}, t^{\mathsf{GD},(k)} + \tau^{\mathsf{GD},(k)}]$}), which we impose through the following CCR:
\begin{equation}\label{cons:SGD_4}
  \sum_{t=1}^{\tau^{\mathsf{GD},(k)}}  \hspace{-0.5mm}\Big(\bm{\Gamma}^{\mathsf{GD}, (k)}{-}\bm{\Gamma}\big(\bm{\Psi}(\mathcal{N},t^{\mathsf{GD},(k)}+ t), r^{\mathsf{G}, (k)}\big)\Big)= \bf{0},
\end{equation}
where $\bf{0}$ denotes an $N\times N$ all-zero matrix. 
Finally, the total energy consumption of performing global model dispatching can be calculated as:
\begin{equation}\label{GDenergy}
     E^{\mathsf{GD}, (k)} = \sum_{n\in\mathcal{N}}\sum_{n'\in\mathcal{N}} \tau_{n,n'}^{\mathsf{GD}, (k)} P_{n},
\end{equation}
where {\small$P_{n}{=}P^{\mathsf{Tx}}_{m}$}, $m\in\mathcal{M}_n$, is the transmit power of satellite {\small$n$}.

\vspace{-1mm}
\subsection{\underline{\textbf{Phase 2:}} Local Model Training (lines \ref{Ph2:start}-\ref{Ph2:end} in Alg.~\ref{alg:Fed-Span_alg})}\label{subsec:ph2}
After \textbf{\textit{Phase 1}} concludes, the global model received by satellites (i.e., {\small$\bm{\omega}^{(k)}$)} is used to synchronize their models. The satellites then train their models and perform local model aggregations within their VCs to debias their trained models before the final global aggregation. Below, we detail the processes involved in the local model training.

\subsubsection{Time Specification}
We denote the set of local rounds performed during global round {\small $k$} by {\small$\mathcal{L}^{(k)}$} with each round indexed by {\small$\ell \in \mathcal{L}^{(k)}$} comprising a local model training and a local model aggregation. The number of local rounds, represented by {\small$L^{(k)}=|\mathcal{L}^{(k)}|$}, is a decision variable.
The final local round, {\small${\ell} = L^{(k)}$}, involves only local training, as it directly precedes the global aggregation.
To capture the number and timing of local rounds, we divide the interval of local rounds during global round {\small$k$}, denoted by {\small$\tau^{\mathsf{Loc}}$}, into equal periods called \textit{transmission time intervals} (TTIs). Each TTI represents a time window to either (i) perform a local round or (ii) skip it for a later TTI. Introducing TTIs enables a tractable model for scheduling local model training and aggregations, which is crucial given the dynamic nature of satellite networks.

To model TTIs, we define {\small$t^{(k)}_x \in [T^{\mathsf{Init},(k)} + \tau^{\mathsf{GD},\mathsf{max}}, T^{\mathsf{Init},(k)} + \tau^{\mathsf{GD},\mathsf{max}} + \tau^{\mathsf{Loc}}]$} as the starting time of TTI {\small$x$}, where {\small$x \in \mathcal{X}$} and {\small$\mathcal{X} = \{1, \dots, X\}$} represents the set of TTIs.
Equal intervals across TTIs are enforced via {\small$t^{(k)}_{x+1} - t^{(k)}_x = \tau^{\mathsf{TTI}}$} for all {\small$x \in \mathcal{X} \setminus \{X\}$}, where {\small$\tau^{\mathsf{TTI}}$} is the constant duration of each TTI. Consequently, the number of TTIs is given by {\small$X = \lfloor \tau^{\mathsf{Loc}} / \tau^{\mathsf{TTI}} \rfloor$}.

Let {\small$\lambda^{(k,\ell)}_{x}$} be a binary decision variable, where {\small$\lambda^{(k,\ell)}_{x} = 1$} represents that the {\small$\ell^{\text{th}}$} local round  takes place at the start of TTI {\small$x$} (i.e., at {\small$t^{(k)}_x$}); otherwise {\small$\lambda^{(k,\ell)}_{x} = 0$}.
To ensure the proper execution of local rounds, we impose the following CCRs:
\begin{align}
    &\sum_{x\in\mathcal{X}}\lambda^{(k,\ell)}_{x} = 1,~~~~\forall \ell{\in}\mathcal{L}^{(k)}\setminus\{L^{(k)}\},~\forall k\in\mathcal{K}, \label{l-ch1}
   \\
   &\sum_{\ell \in\mathcal{L}^{(k)}\setminus\{L^{(k)}\}} \sum_{x\in\mathcal{X}}\lambda^{(k,\ell)}_{x} = L^{(k)}-1, ~~\forall k\in\mathcal{K}, \label{l-ch2}
\end{align}
implying that (i) each local round {\small$\ell$}, except the final one corresponding to the global aggregation, is executed at a TTI, and (ii) the total number of TTIs used for local aggregations is {\small$L^{(k)} - 1$}, as the final local round is followed by a global aggregation. We subsequently obtain the \underline{\textbf{l}}ocal \underline{\textbf{t}}raining (LT) initialization time {\small$t^{\mathsf{LT},(k,\ell)}$} as follows:
\begin{equation}\label{eq:timeselection}
t^{\mathsf{LT},(k,\ell)}=\sum_{x\in\mathcal{X}^{(k)}}\lambda^{(k,\ell)}_{x} t_{x}^{(k)},~\forall \ell{\in}\mathcal{L}^{(k)}\setminus\{L^{(k)}\},~\forall k\in\mathcal{K}.
\end{equation}

\subsubsection{Network Operations}
To detail the network operations during local model training, we first need to identify the model on which local training begins across the satellites. To this end, we note that following each global dispatching step (i.e., after \textbf{\textit{Phase 1}} and when $\ell=0$), when all satellite models are unified, we have
{\small$\overline{\bm{\omega}}^{(k,0)}_{c} = \bm{\omega}^{(k)}$},
where {\small$\overline{\bm{\omega}}^{(k,0)}_{c}$} denotes the \textit{initial VC model} for VC {\small$c$}, a unified model across the satellites in the VC on which model training is performed.
For subsequent training rounds (i.e., {\small$\ell \geq 1$}), the model
{\small$\overline{\bm{\omega}}^{(k,\ell)}_{c}$} corresponds to the latest dispatched model across all satellites within the VC (detailed in \textbf{\textit{Phase 4}}), again yielding a unified model across the VC on which the local training round is carried out.


Let {\small$e^{(k)}_{n}$} denote the number of SGD iterations performed by satellite $n$  during each local round of global round {\small$k$}.
At each iteration {\small$e$} of local training round {\small$\ell$} in global round {\small$k$}, satellite {\small$n$} randomly selects a mini-batch of data points {\small$\mathcal{B}^{(k,\ell),e}_{n} \subseteq \mathcal{D}^{(k)}_{n}$}, where {\small$e \in \{1, \dots, e^{(k)}_{n}\}$}. The mini-batch size is fixed across iterations and denoted as {\small$B^{(k,\ell)}_{n} \triangleq |\mathcal{B}^{(k,\ell),e}_{n}| = \varsigma^{(k,\ell)}_{n} D^{(k)}_{n}$}, where {\small$\varsigma^{(k,\ell)}_{n} \in (0,1]$} is a decision variable representing the fraction of the local dataset used in each mini-batch.
The evolution of the local model of satellite {\small$n \in \mathcal{N}^{(k)}_{c}$} for $\ell\in\mathcal{L}^{(k)}$ is given by:
\begin{equation}\label{eq:SGD_iterations}
  \hspace{-4mm} \resizebox{0.46\textwidth}{!}{
    $ \bm{\omega}_{n}^{(k,\ell), e} = \bm{\omega}_{n}^{(k,\ell),e-1} - \eta_{_k} \widetilde{\nabla} F_{n}(\bm{\omega}^{(k,\ell), e-1}),~ 1\leq e \leq e^{(k)}_{n}\hspace{-0.5mm} ,$} \hspace{-3mm} 
\end{equation}
\noindent where {\small$\eta_{k}$} denotes the SGD step-size during global round {\small$k \in \mathcal{K}$}.  The iterations in~\eqref{eq:SGD_iterations} start with
{\small$\bm{\omega}^{(k,\ell),0}_{n}{=}\overline{\bm{\omega}}^{(k,\ell)}_{c}$}, which represents the VC model dispatched to satellites in VC {\small$\mathcal{N}^{(k)}_{c}$} at local round {\small$\ell$}. Also, {\small$\widetilde{\nabla} F_{n}(\bm{\omega}^{(k,\ell), e-1})$} is the stochastic gradient of the ML loss, defined as:
{\small $
    \widetilde{\nabla} F_{n}(\bm{\omega}^{(k,\ell), e-1}) =  \sum_{d\in \mathcal{B}^{(k,\ell),e}_{n}}\nabla f(\bm{\omega}^{(k,\ell),e-1}_{n},d)\big/B^{(k,\ell)}_{n}
$}.

At the end of local round {\small$\ell$}, each satellite {\small$n$} computes its cumulative gradient, which is given by:
\begin{equation}\label{eq:cumulateGrad}
    \widetilde{\nabla F}_{n}^{(k,\ell)} {=} (\overline{\bm{\omega}}^{(k,\ell)}_{c}{-}\bm{\omega}_{n}^{(k,\ell),e^{(k)}_{n}})\big/\eta_{_k}
\end{equation}
and transmits it for model aggregations as discussed in \textbf{\textit{Phase 3\&5}}.
To unify operations within each VC {\small$c$}, we assume that all satellites in the VC perform the same number of SGD iterations, defined as {\small$e_c^{(k)} = e_n^{(k)}$}, $\forall n \in \mathcal{N}_c^{(k)}$. Notably, {\small$e_c^{(k)}$} can vary across VCs, allowing for different levels of model refinements and satellite computation capabilities.

\subsubsection{Energy and Latency Computations}
Let {\small $a_{n}$} denote the number of CPU cycles required to process one data point by satellite {\small $n \in \mathcal{N}$}. During global round {\small$k \in \mathcal{K}$} and local round {\small$\ell \in \mathcal{L}^{(k)}$}, the \textit{training latency} for satellite {\small$n$} to perform {\small$e^{(k)}_{n}$} local SGD iterations as per~\eqref{eq:SGD_iterations} is given by:
\begin{equation}\label{eq:training_latency}
    \begin{aligned}
     \tau_{n}^{\mathsf{LT},{(k,\ell)}}= e^{(k)}_{n} a_{n} {B}^{(k,\ell)}_{n} {/}f^{(k,\ell)}_{n},
    \end{aligned}
\end{equation}
where {\small $f^{(k,\ell)}_{n}{\le} f^{\mathsf{max}}_{n}$}
is a decision variable representing the CPU frequency of satellite {\small $n$}. Subsequently, {\small$\tau^{\mathsf{LT},{(k,\ell)}} = \max_{n \in \mathcal{N}} \{\tau_{n}^{\mathsf{LT},{(k,\ell)}}\}$} denotes the total training latency of satellites during local round {\small$\ell$} of global round {\small$k$}, which we ensure its timely conclusion by imposing the following bound: 
\begin{equation}\label{eq:T^f}
    \tau^{\mathsf{LT},{(k,\ell)}} \leq \tau^{\mathsf{LT}, \mathsf{max}}.
\end{equation}

Subsequently, the total energy consumed by satellite {\small$n$} for local training during global round {\small$k$} and local round {\small$\ell$} is:
\begin{equation}\label{eq:EN_LC}
    E_{n}^{\mathsf{LT},{(k,\ell)}}= \frac{\alpha^{\mathsf{chip}}_{n}}{2}(f^{(k)}_{n})^{3}\tau_{n}^{\mathsf{LT},{(k,\ell)}},
\end{equation}
where {\small${\alpha^{\mathsf{chip}}_{n}}$} denotes the effective chipset capacitance~\cite{8737464}. Considering this, the total computation energy consumption  during each local round {\small$\ell$} is given by:
\begin{equation}\label{eq:EN_LC_sum}
    E^{\mathsf{LT},(k,\ell)}= \sum_{n \in \mathcal{N}} E_{n}^{\mathsf{LT},{(k,\ell)}}.
\end{equation}

\vspace{-4mm}
\subsection{\hspace{-0.5mm}\underline{\textbf{Phase\hspace{-0.5mm} 3:}} Spanning \hspace{-0.15mm}Local Aggregations\hspace{-0.5mm} (lines\hspace{-0.3mm} \ref{Ph3:start}-\ref{Ph3:end}\hspace{-0.5mm} in\hspace{-0.1mm} Alg.~\ref{alg:Fed-Span_alg})}\label{subsec:ph3}
After completing local model training at each local round {\small$\ell \in \mathcal{L}^{(k)} \setminus \{L^{(k)}\}$}, {Fed-Span} conducts a local model aggregation phase, transferring local models through the \textit{leaves-to-roots continuum} as visualized in Fig.~\ref{fig:system_model}(c) and discussed below.

\subsubsection{Time Specification}
During \underline{\textbf{l}}ocal \underline{\textbf{a}}ggregation (LA) {\small$\ell \in \mathcal{L}^{(k)} \setminus \{L^{(k)}\}$}, occurring at time {\small$t^{\mathsf{LA},(k,\ell)}$}, satellites within each VC transfer their models via an upward-directed tree to their respective root node for local model aggregation. We can obtain the time {\small$t^{\mathsf{LA},(k,\ell)}$} as follows:
\begin{equation}\label{eq:localaggtime}
t^{\mathsf{LA},(k,\ell)}=t^{\mathsf{LT},(k,\ell)} + \tau^{\mathsf{LT}, \mathsf{max}}.
\end{equation}
\subsubsection{Topological Properties and Network Operations}
As shown in Fig.~\ref{fig:system_model}(c), the upward-directed trees used for model aggregations within VCs collectively form a \textit{forest}, a graph topology comprising disjoint trees that span all satellites. 
Let {\small$\bm{\Psi}(\mathcal{N}^{(k)}_{c},t^{\mathsf{LA},(k,\ell)})$} denote the CBM of the local upward-directed aggregation tree in VC {\small$c$}. The forest used for model aggregation at local round {\small$\ell$} is then formalized by combining the CBMs across VCs as {\small$\bm{\Psi}(\mathcal{N},t^{\mathsf{LA},(k,\ell)}) = \sum_{c \in \mathcal{C}^{(k)}} \bm{\Psi}(\mathcal{N}^{(k)}_{c},t^{\mathsf{LA},(k,\ell)})$}. The adjacency matrix of this forest is further represented in the following compact form:
\begin{equation}\label{eq:CM_aggregation_adjacency_brevity}
 \bm{\Gamma}^{\mathsf{LA},(k,\ell)}=\sum_{c \in \mathcal{C}^{(k)}}\bm{\Gamma}\big(\bm{\Psi}(\mathcal{N}^{(k)}_{c},t^{\mathsf{LA},(k,\ell)}),r_{c}^{\mathsf{L}, (k)}\big),
\end{equation}
where {\small$r_{c}^{\mathsf{L}, (k)}$} represents the root of the upward-directed tree of VC {\small$c$}. Using the methodology introduced in \textbf{\textit{Phase 1}} (Sec.~\ref{subsec:ph1}) and adapting \eqref{cons:SGD_2} to this forest, we derive the following CCR:
\begin{equation}\label{cons:SCA_2}
\hspace{-2mm}~\sum_{n'\in\mathcal{N}}\bm{\Gamma}^{\mathsf{LA}, (k,\ell)}_{n,n'} + \pi_{n}^{\mathsf{L}, (k)}{=}1,~\forall n\in\mathcal{N},
\end{equation}
where {\small$\pi_{n}^{\mathsf{L}, (k)} \in \{0,1\}$} indicates whether satellite {\small $n \in \mathcal{N}$} serves as the root of one of the trees in the forest. The cycle-free property of this forest, which could be enforced by a condition similar to~\eqref{cons:tree_3}, is inherently satisfied in Sec.~\ref{sec:optimization_problem} as per Remark~\ref{remark:Cycle}.

To conduct a local model aggregation at each local round {\small$\ell \in \mathcal{L}^{(k)} \setminus \{L^{(k)}\}$}, after computing it as in~\eqref{eq:cumulateGrad}, each satellite sends its cumulative gradient to its immediate successor in the forest. Subsequently, each satellite {\small $n \in \mathcal{N}^{(k)}_{c}$} in VC {\small $c$} that serves as a successor receives cumulative gradients from its predecessor satellites and aggregates them into a single gradient vector.
This vector is expressed recursively based on the satellite’s cumulative gradient and the aggregated gradients from its predecessors as follows:
\begin{equation}\label{clusteraggregation}
     \overline{\nabla F}_{n}^{(k,\ell)}= {{D}^{(k)}_n}\widetilde{\nabla F}_{n}^{(k,\ell)}+\hspace{-3mm}\sum_{n'\in \mathcal{N}\setminus\{n\}}\hspace{-3mm} \bm{\Gamma}^{\mathsf{LA}, (k,\ell)}_{n,n'} \overline{\nabla F}_{n'}^{(k,\ell)}.
\end{equation}
Each satellite {\small $n$} then forwards its aggregated gradient vector to its successor, and this process continues until reaching the root node of each VC {\small $c$}, which obtains the VC model as follows:
\begin{equation}\label{localagggregationformula}
\overline{\bm{\omega}}^{(k,\ell+1)}_{c}=\overline{\bm{\omega}}^{(k,\ell)}_{c}-\eta_{_k}{\overline{\nabla F}_{n}^{(k,\ell)}}\big/{D^{(k)}_c},~n=r_{c}^{\mathsf{L}, (k)},
\end{equation}
where {\small $D_c^{(k)}= \sum_{n \in \mathcal{N}_{c}^{(k)}} D^{(k)}_n$}. VC model 
$\overline{\bm{\omega}}^{(k,\ell+1)}_{c}$ is then dispatched to all satellites in VC {\small $c$} and used to synchronize their local models, as detailed in \textbf{\textit{Phase 4}}.
\subsubsection{Energy and Latency Computations}
Using~\eqref{eq:CM_aggregation_adjacency_brevity}, the latency for transmitting a local model from satellite {\small$n$} to its predecessor satellite {\small$n'$} can be derived similarly to~\eqref{eq:GM_dispatching_transmission_latency} as:
\begin{equation}\label{eq:CM_aggregation_transmission_latency}
    \tau_{n,n'}^{\mathsf{LA},(k,\ell)} = {\bm{\Gamma}^{\mathsf{LA}, (k,\ell)}_{n,n'} \alpha^{\mathsf{Bit}} M^{\mathsf{Dim}}}\big/{\mathfrak{R}^{\mathsf{LA}, (k,\ell)}_{n,n'}},
\end{equation}
where {\small$\mathfrak{R}^{\mathsf{LA}, (k,\ell)}_{n,n'} = \mathfrak{R}_{n,n'}\big(\bm{\Psi}(\mathcal{N}^{(k)},t^{\mathsf{LA},(k,\ell)})\big)$} is the data rate between the respective satellites.
The total latency of model arrival at satellite {\small$n'$}, denoted by {\small$\tau_{n'}^{\mathsf{LA}, (k,\ell)}$}, can  then be obtained similar to~\eqref{eq:GM_dispatching_latency} through the following recursive relation:
\begin{equation}\label{eq:CM_aggregation_latency}
    \tau_{n'}^{\mathsf{LA}, (k,\ell)}{=}\hspace{-2mm}\max_{n\in\mathcal{N}\setminus\{n'\}}\hspace{-1mm}\Big\{ \bm{\Gamma}^{\mathsf{LA}, (k,\ell)}_{n, n'}\tau_{n}^{\mathsf{LA}, (k,\ell)} + \tau_{n,n'}^{\mathsf{LA}, (k,\ell)} \Big\}.
\end{equation}
In~\eqref{eq:CM_aggregation_latency}, the {\small$\max\{\cdot\}$} function is used because satellite {\small $n'$} may need to wait for local models from multiple successors.
The overall latency of the model aggregation over the forest is then determined by the most time-consuming aggregation across the roots of its trees: {\small $\tau^{\mathsf{LA}, (k,\ell)}=\max_{n'\in\{r_{c}^{\mathsf{L}, (k)}\}_{c\in\mathcal{C}^{(k)}}} \{\tau_{n'}^{\mathsf{LA}, (k,\ell)}\}$}, the timeliness of which is ensured through  
the following bound:
\begin{equation}\label{cons:CM_aggregation_latency_1}
    \tau^{\mathsf{LA}, (k,\ell)} \leq  \tau^{\mathsf{LA},\mathsf{max}}, ~\forall \ell{\in}\mathcal{L}^{(k)}\setminus\{L^{(k)}\},~k\in\mathcal{K}.
\end{equation}

Similar to~\eqref{cons:SGD_4}, we ensure proper local aggregation without link interruptions across the forest via the following CCR:
\begin{equation}\label{cons:SCA_4}
  \hspace{-3mm}  \sum_{t=1}^{\tau^{\mathsf{LA}, (k,\ell)}}\hspace{-3mm}  \Big(\bm{\Gamma}^{\mathsf{LA}, (k,\ell)}{-}\hspace{-2.5mm} \sum_{c \in \mathcal{C}^{(k)}}\hspace{-2mm} \bm{\Gamma}\big(\bm{\Psi}(\mathcal{N}^{(k)}_{c}\hspace{-1mm},t^{\mathsf{LA},(k,\ell)}{+}t), r_{c}^{\mathsf{L}, (k)}\big)\Big) \hspace{-.5mm}=\hspace{-.5mm} \bf{0}.\hspace{-3mm} 
\end{equation}
Subsequently, the total energy consumption for performing model aggregations across the forest can be calculated as:
\begin{equation}\label{forestenergyAG} 
    E^{\mathsf{LA},(k,\ell)} = \sum_{n\in\mathcal{N}}\sum_{n'\in\mathcal{N}} \tau_{n,n'}^{\mathsf{LA},(k,\ell)} P_{n}.
\end{equation}

\subsection{\hspace{-0.5mm}\underline{\textbf{Phase\hspace{-0.2mm} 4:}}\hspace{-0.2mm} Spanning\hspace{-0.2mm} Local\hspace{-0.2mm} Dispatching\hspace{-0.2mm} (lines\hspace{-0.2mm} \ref{Ph4:start}-\ref{Ph4:end} in Alg.~\ref{alg:Fed-Span_alg})}\label{subsec:ph4}

Upon completing \textbf{\textit{Phase 3}} for each local round {\small$\ell \in \mathcal{L}^{(k)} \setminus \{L^{(k)}\}$}, {Fed-Span} proceeds with a local model dispatching phase. This involves forming a downward-directed tree within each VC, collectively creating a \textit{forest}, where aggregated models are dispatched through the \textit{roots-to-leaves continuum} as illustrated in Fig.~\ref{fig:system_model}(d) (this forest is different from the one constructed in \textbf{\textit{Phase 3}}, where models were transmitted along the \textit{leaves-to-roots continuum} for local aggregations). We detail the characteristics of this phase below.


\subsubsection{Time Specification}
During the \underline{\textbf{l}}ocal \underline{\textbf{d}}ispatching (LD) phase of local round {\small$\ell \in \mathcal{L}^{(k)} \setminus \{L^{(k)}\}$}, the VC models obtained at the roots of all VCs are dispatched to the satellites within each VC at time {\small$t^{\mathsf{LD},(k,\ell)}$}, which is given by:
\begin{equation}
    t^{\mathsf{LD},(k,\ell)}{=}t^{\mathsf{LA},(k,\ell)}{+}\tau^{\mathsf{LA},\mathsf{max}}.
\end{equation}

\subsubsection{Topological Properties and Network Operations}
The formulation of the downward-directed trees and the resulting forest follows a similar logic to the upward-directed trees and their forest in \textbf{\textit{Phase 3}} (Sec.~\ref{subsec:ph3}). Specifically, letting {\small$\bm{\Psi}(\mathcal{N}^{(k)}_{c},t^{\mathsf{LD},(k,\ell)})$} represent the CBM of the downward-directed tree for each VC {\small$c$}, we formalize the forest used for model dispatching by combining the CBMs across VCs as follows:
\begin{equation}\label{eq:CM_dispatching_adjacency_brevity}
 \bm{\Gamma}^{\mathsf{LD},(k,\ell)}=\sum_{c \in \mathcal{C}^{(k)}}\bm{\Gamma}\big(\bm{\Psi}(\mathcal{N}^{(k)}_{c},t^{\mathsf{LD},(k,\ell)}),r_{c}^{\mathsf{L}, (k)}\big),
\end{equation}
where the root {\small$r_{c}^{\mathsf{L}, (k)}$} is the same as in the upward-directed tree from \textbf{\textit{Phase 3}}. Similar to~\eqref{cons:SCA_2}, the forest formed in this phase must satisfy the following CCR:
\begin{equation}\label{cons:SCD_1}
\hspace{-2mm}~\sum_{n\in\mathcal{N}}\bm{\Gamma}^{\mathsf{LD}, (k,\ell)}_{n,n'} + \pi_{n'}^{\mathsf{L}, (k)}{=}1,~\forall n'\in\mathcal{N},
\end{equation}
where {\small$\pi_{n'}^{\mathsf{L}, (k)} \in \{0,1\}$} indicates whether satellite {\small$n' \in \mathcal{N}$} serves as a root of one of the trees in this forest.

\subsubsection{Energy and Latency Computations}
Similar to~\eqref{eq:CM_aggregation_transmission_latency}, using the adjacency matrix of~\eqref{eq:CM_dispatching_adjacency_brevity}, the latency for transmitting the model from satellite {\small$n$} to its predecessor satellite {\small$n'$} is:
\begin{equation}\label{eq:CM_dispatching_transmission_latency}
    \tau_{n,n'}^{\mathsf{LD}, (k,\ell)} = {\bm{\Gamma}^{\mathsf{LD}, (k,\ell)}_{n,n'}\alpha^{\mathsf{Bit}} M^{\mathsf{Dim}}}\big/{\mathfrak{R}^{\mathsf{LD}, (k,\ell)}_{n,n'}},
\end{equation}
where {\small$\mathfrak{R}^{\mathsf{LD}, (k,\ell)}_{n,n'} = \mathfrak{R}_{n,n'}\big(\bm{\Psi}(\mathcal{N},t^{\mathsf{LD},(k,\ell)})\big)$}. Following the similar logic as in~\eqref{eq:GM_dispatching_latency}, the total latency of model arrival at satellite {\small $n'$} can be obtained using the following recursive expression:
\begin{equation}\label{eq:CM_dispatching_latency}
    \tau_{n'}^{\mathsf{LD}, (k,\ell)}{=}\hspace{-2mm}\sum_{n\in\mathcal{N}\setminus\{n'\}}\hspace{-1mm} \left(\bm{\Gamma}^{\mathsf{LD}, (k,\ell)}_{n, n'} \tau_{n}^{\mathsf{LD}, (k,\ell)} + \tau_{n, n'}^{\mathsf{LD}, (k,\ell)} \right).
\end{equation}
Using~\eqref{eq:CM_dispatching_latency}, the overall latency of the model dispatching forest is determined by the most time-consuming transmission path, given by {\small$\tau^{\mathsf{LD},(k,\ell)} = \max_{n'\in\mathcal{N}} \{\tau_{n'}^{\mathsf{LD}, (k,\ell)}\}$},
 the timeliness of which is imposed by the following bound:
\begin{equation}\label{cons:CM_dispatching_latency_1}
    \tau^{\mathsf{LD}, (k,\ell)} \leq \tau^{\mathsf{LD},\mathsf{max}},~\forall \ell{\in}\mathcal{L}^{(k)}\setminus\{L^{(k)}\},~ k \in\mathcal{K}.
\end{equation}

Similar to~\eqref{cons:SCA_4}, to ensure proper local model dispatching without link interruptions, we impose the following CCR:
\begin{equation}\label{cons:SCD_3}
     \hspace{-3mm}  \sum_{t=1}^{\tau^{\mathsf{LD}, (k,\ell)}}\hspace{-3mm} \Big(\hspace{-0.2mm}\bm{\Gamma}^{\mathsf{LD}, (k,\ell)}{-}\hspace{-2.5mm}\sum_{c \in \mathcal{C}^{(k)}}\hspace{-2mm} \bm{\Gamma}\big(\bm{\Psi}(\mathcal{N}^{(k)}_{c}\hspace{-1mm},t^{\mathsf{LD},(k,\ell)}{+}t), \hspace{-0.2mm} r_{c}^{\mathsf{L}, (k)}\big)\hspace{-0.8mm}\Big) = \bf{0}.
\hspace{-3mm}
\end{equation}
Finally, we obtain the total energy consumption for performing local dispatching as follows:
\begin{equation}\label{LDenergy}
     E^{\mathsf{LD},(k,\ell)} = \sum_{n\in\mathcal{N}}\sum_{n'\in\mathcal{N}} \tau_{n,n'}^{\mathsf{LD}, (k,\ell)} P_{n}.
\end{equation}
\subsection{\underline{\textbf{Phase 5:}} Global Model Aggregation (lines \ref{Ph5:start}-\ref{Ph5:end} in Alg.~\ref{alg:Fed-Span_alg})}\label{subsec:ph5}
Upon completing \textbf{\textit{Phase 3}} for the last local round {\small$\ell = L^{(k)}$}, {Fed-Span} concludes with a global aggregation conducted via an upward-directed tree to form the new global model {\small$\bm{\omega}^{(k+1)}$}. This phase is illustrated in Fig.~\ref{fig:system_model}(e) and detailed below.
\subsubsection{Time Specification}
After completing the final round of local training {\small$\ell=L^{(k)}$}, the resulting local models are transmitted to the global root {\small$r^{\mathsf{G}, (k+1)}$} for \underline{\textbf{g}}lobal \underline{\textbf{a}}ggregation (GA) in round {\small$k$}, occurring at time {\small$t^{\mathsf{GA},(k)}$}, which is given by:
\begin{equation}
    t^{\mathsf{GA},(k)} = T^{\mathsf{Init},(k)}{+}\tau^{\mathsf{GD},\mathsf{max}}+\tau^{\mathsf{Loc}}.
\end{equation}
\subsubsection{Topological Properties and Network Operations}
We next formalize the global upward-directed tree rooted at {\small$r^{\mathsf{A}, (k)}$}, used for model aggregation in round {\small$k$}. This tree is derived from the CBM {\small$\bm{\Psi}(\mathcal{N},t^{\mathsf{GA},(k)})$} and compactly represented by the adjacency matrix {\small$\bm{\Gamma}^{\mathsf{GD}, (k)}$}, defined as:
\begin{equation}\label{eq:GM_aggregation_adjacney_brevity}
    \bm{\Gamma}^{\mathsf{GA}, (k)}=\bm{\Gamma}\big(\bm{\Psi}(\mathcal{N},t^{\mathsf{GA},(k)}),r^{\mathsf{A}, (k)}\big).
\end{equation}
Let {\small$\pi_{n}^{\mathsf{A}, (k)} \in \{0,1\}$} indicate whether satellite {\small$n$} serves as the root of this upward-directed tree that is identical to the root of the downward-directed tree formed in \textbf{\textit{Phase 1}} of the next global round {\small$k+1$} (i.e., {\small$r^{\mathsf{A}, (k)} = r^{\mathsf{G}, (k+1)}$} and {\small$\pi_{n}^{\mathsf{A}, (k)} = \pi_{n}^{\mathsf{G}, (k+1)}$}).\footnote{This common root is responsible for forming the global model and subsequently dispatching it across the satellites for the next global round.} Subsequently, considering~\eqref{cons:SGD_2}, the tree in this phase must satisfy the following CCR:
\begin{equation}\label{cons:SGA_1}
\hspace{-2mm}~\sum_{n'\in\mathcal{N}}\bm{\Gamma}^{\mathsf{GA}, (k)}_{n,n'} + \pi_{n}^{\mathsf{A}, (k)}{=}1,~\forall n\in\mathcal{N}.
\end{equation}


To conduct a global aggregation after the last local aggregation {\small$\ell=L^{(k)}$}, each satellite {\small$n$} in VC {\small$c$}, computes its cumulative gradient for the entire global round as {\small$\widehat{\nabla F}_{n}^{(k,{\ell})} {=} (\bm{\omega}^{(k)}{-}\bm{\omega}_{n}^{(k,{\ell}),e^{(k)}_{n}})\big/\eta_{_k}$} with {\small$\ell=L^{(k)}$}, similar to~\eqref{eq:cumulateGrad}, and transmits it to its successor satellite.
Subsequently, each successor satellite {\small $n \in \mathcal{N}$} receives cumulative gradients from its predecessors and aggregates them into a single gradient vector. This vector, at satellite {\small $n$}, can be expressed recursively in terms of its own cumulative gradient and the aggregated gradients of its predecessors as:
\begin{equation}\label{globalaggregation}
\hspace{-2mm}
\resizebox{0.94\linewidth}{!}{$
     \overline{\nabla F}_{n}^{(k,{\ell})}{=} \frac{{D}^{(k)}_n}{e^{(k)}_{n} L^{(k)} }\widehat{\nabla F}_{n}^{(k,{\ell})}\hspace{-1.5mm}+\hspace{-1mm}\sum_{n'\in \mathcal{N}\setminus\{n\}}\hspace{-1mm} \bm{\Gamma}^{\mathsf{GA}, (k)}_{n,n'} \overline{\nabla F}_{n'}^{(k,{\ell})}\hspace{-0.5mm},\hspace{-0.79mm}~\ell {=} L^{(k)}\hspace{-0.75mm}.
     $}
\hspace{-1.8mm}
\end{equation}
The normalization by the number of SGD iterations {\small$e^{(k)}_n$} in~\eqref{globalaggregation} reduces the risk of bias toward satellites performing more SGD iterations during local rounds~\cite{DBLP:journals/corr/abs-2202-02947}.

Based on~\eqref{globalaggregation}, the next global model is computed at the root satellite {\small$r^{\mathsf{A}, (k)}$} as follows:
\begin{equation}\label{finalmodelaggregation}
\hspace{-2mm}
\resizebox{0.94\linewidth}{!}{$
    \bm{\omega}^{(k+1)}{=}\bm{\omega}^{(k)}{-}\eta_{_k}\Xi^{(k)}{\overline{\nabla F}_{n}^{(k,{\ell})}}\big/{D^{(k)}},~\ell = L^{(k)},~n{=}r^{\mathsf{A}, (k)},
      $}
\hspace{-1.8mm}
\end{equation}
where {\small $D^{(k)}= \sum_{n \in \mathcal{N}} D^{(k)}_n$} and $\Xi^{(k)}=\sum_{n{\in}\mathcal{N}}\frac{{D}^{(k)}_n L^{(k)}e^{(k)}_{n}}{{D}^{(k)}}$ is a boosting coefficient, controlling the convergence speed. 
Following this, {\small$\bm{\omega}^{(k+1)}$} is dispatched as in \textbf{\textit{Phase 1}} (Sec.~\ref{subsec:ph1}), initiating the next sequence of phases for global round {\small$k+1$}.

\subsubsection{Energy and Latency Computations}
Similar to~\eqref{eq:GM_dispatching_transmission_latency}, using the adjacency matrix from~\eqref{eq:GM_aggregation_adjacney_brevity}, the latency of model transmission from satellite {\small$n$} to its predecessor satellite {\small$n'$} is:
\begin{equation}\label{eq:GM_aggregation_transmission_latency}
    \tau_{n,n'}^{\mathsf{GA}, (k)} = {\bm{\Gamma}^{\mathsf{GA}, (k)}_{n,n'}\alpha^{\mathsf{Bit}} M^{\mathsf{Dim}}}\big/{\mathfrak{R}_{n,n'}^{\mathsf{GA}, (k)}},
\end{equation}
where {\small$\mathfrak{R}_{n,n'}^{\mathsf{GA}, (k)} = \mathfrak{R}_{n,n'}\big(\bm{\Psi}(\mathcal{N},t^{\mathsf{GA},(k)})\big)$}.
Similar to~\eqref{eq:CM_aggregation_latency}, the latency of model arrival at satellite {\small$n'$} during the global aggregation can be obtained as follows:
\begin{equation}\label{eq:GM_aggregation_latency}
    \tau_{n'}^{\mathsf{GA}, (k)}{=}\hspace{-2mm}\max_{n\in\mathcal{N}\setminus\{n'\}}\hspace{-1mm}\Big\{\bm{\Gamma}^{\mathsf{GA}, (k)}_{n, n'} \left(\tau_{n,n'}^{\mathsf{GA}, (k)} + \tau_{n}^{\mathsf{GA}, (k)}\right) \Big\}.
\end{equation}
Consequently, the overall latency of global aggregation is determined by the latency of model arrival at the root satellite, given by {\small$\tau^{\mathsf{GA},(k)} = \tau_{n'}^{\mathsf{GA}, (k)}$}, where {\small$n' = r^{\mathsf{A}, (k)}$}. We assume a bounded global aggregation window as follows:
\begin{equation}\label{cons:GM_aggregation_latency_1}
    \tau^{\mathsf{GA},(k)} \leq \tau^{\mathsf{GA},\mathsf{max}},~~~\forall k \in\mathcal{K}.
\end{equation}

Similar to~\eqref{cons:SGD_4}, proper model aggregation without link interruptions is ensured through the following CCR:
\begin{equation}\label{cons:SGA_3}
    \hspace{-3mm}  \sum_{t=1}^{\tau^{\mathsf{GA}, (k)}}\hspace{-2mm} \Big(\bm{\Gamma}^{\mathsf{GA}, (k)}{-}\bm{\Gamma}\big(\bm{\Psi}(\mathcal{N},t^{\mathsf{GA},(k)}{+}t), r^{\mathsf{A}, (k)}\big)\Big) = \bf{0}.
\end{equation}
 Finally, the total energy consumption of performing global aggregation {\small $k$} is given by:
\begin{equation}\label{GAenergy}
    E^{\mathsf{GA}, (k)} = \sum_{n\in\mathcal{N}}\sum_{n'\in\mathcal{N}} \tau_{n,n'}^{\mathsf{GA}, (k)} P_{n}.
\end{equation}
\vspace{-1mm}

\subsection{Notion of Idle Times in {Fed-Span} (lines \ref{ID:start}-\ref{ID:end} in Alg.~\ref{alg:Fed-Span_alg})}
Existing works on Fed-LS \cite{9749193,10021101, chen2022satellite ,nguyen2022federated, 10216376, 9674028, tang2022federated,10121575,10039157, elmahallawy2023optimizing, elmahallawy2023one} often assume continuous back-to-back model training on satellites, which is unrealistic due to satellites' resource limitations, such as their limited battery capacities.
To address this, 
in {Fed-Span}, a grace period, called \textit{idle time}, follows each local training round {\small$\ell{\in}\mathcal{L}^{(k)}$}, during which satellites pause operations to conserve resources.
Let {\small$\Omega^{(k,\ell)}$} denote the duration of this idle time (visualized in Fig.~\ref{fig:training_process}), which is considered as a degree of freedom and later optimized in Sec.~\ref{sec:optimization_problem}; mathematically:
\begin{equation}\label{eq:idletime}
\hspace{-3mm}
    \Omega^{(k,\ell)} {=} \tau^{\mathsf{LT}, \mathsf{max}} - \tau^{\mathsf{LT},{(k,\ell)}},~\ell{\in}\mathcal{L}^{(k)}, k \in\mathcal{K}.
    \hspace{-3mm}
\end{equation}
We further define the total idle time for global round {\small$k$} as {\small$\Omega^{(k)} = \sum_{\ell \in \mathcal{L}^{(k)}} \Omega^{(k,\ell)}$}. 

\begin{table*}[t]
\vspace{-9mm}
\begin{minipage}{1\textwidth}
{\footnotesize
\begin{equation}\label{eq:gen_conv}
\hspace{-3mm}
\resizebox{0.97\linewidth}{!}{$
 {\small
\begin{aligned} 
    &\frac{1}{K} \sum_{k=0}^{K-1}\mathbb E\left \Vert\nabla F^{(k)}(\bm{\omega}^{(k)})\right\Vert^2 \leq \underbrace{\frac{F^{(0)}(\bm{\omega}^{(0)}) - F^{(K)^{\ast}}}{K \Phi_{\mathsf{min}}\left(1-\Lambda_{\mathsf{max}}\right)}}_{(a)} + \frac{1}{K \left(1-\Lambda_{\mathsf{max}}\right)} \sum_{k=0}^{K-1} \vast[ \underbrace{\frac{\Omega^{(k)} \Delta^{(k)}}{\Phi_{\mathsf{min}}}}_{(b)} + \underbrace{8 \beta \Theta^2 \Phi_{\mathsf{max}} \hspace{-2mm}\sum_{c\in \mathcal{C}^{(k)}} \hspace{-2mm} \left(\frac{1}{D^{(k)}L^{(k)}} \right)^2 \hspace{-1mm} \frac{1}{e^{(k)}_{c}} \sum_{n\in \mathcal{N}} {D}^{(k)}_{n} \sum_{\ell=1}^{L^{(k)}}\left(1- \varsigma^{(k,\ell)}_n\right)\frac{\gamma^{(k)}_{c,n}\left(\sigma^{(k)}_n\right)^2}{ \varsigma^{(k,\ell)}_n}}_{(c)}
    \\&  + \underbrace{16 \eta_k^2 \beta^2 \left(L^{(k)}  \left(L^{(k)}-1\right) \left(e^{(k)}_{\mathsf{max}}\right)^{2} +  e^{(k)}_{\mathsf{max}}\left(e^{(k)}_{\mathsf{max}}-1\right)\right) \zeta^{\mathsf{Glob}}_{2}\hat{\zeta}^{\mathsf{Loc}}_{1}}_{(d)}
     + \underbrace{16 \eta_k^2 \beta^2 \Theta^2 \hspace{-2mm} \sum_{c\in \mathcal{C}^{(k)}} \hspace{-1mm} \frac{1}{D^{(k)}L^{(k)}} \left(\left(L^{(k)}-1\right) e_{c}^{(k)} + \left(e_{c}^{(k)}-1\right)\right)  \sum_{n\in \mathcal{N}}  \sum_{\ell=1}^{L^{(k)}} \left(1- \varsigma^{(k,\ell)}_n\right)\frac{\gamma^{(k)}_{c,n}\left(\sigma^{(k)}_n\right)^2}{\varsigma^{(k,\ell)}_n}}_{(e)}\vast] 
\end{aligned}
 }
$}
\hspace{-2.5mm}
\end{equation}
}
\hrule
\end{minipage}
\vspace{-4mm}
\end{table*}
\section{ML Convergence Analysis of {\large{Fed-Span}}}\label{sec:conv}
\noindent We next analyze the ML convergence of {Fed-Span}.  To do the analysis, we first introduce various definitions and assumptions. Henceforth, the notation {\small$\Vert.\Vert$} refers to the Euclidean norm-2.
\begin{definition}[Local Data Variability]\label{Assump:DataVariabilit}
    The local data variability at each satellite $n$ is measured via $\Theta_n{\geq} 0$, which $\forall \bm{\omega} , k$ satisfies
    \vspace{-4mm}
    \begin{equation}
   \Vert \nabla f(\bm{\omega},d) - \nabla f\hspace{-.4mm}(\bm{\omega},d')\Vert   \leq  \Theta_n \Vert d\hspace{-.4mm}-d' \Vert,~\forall d,d'\in\mathcal{D}^{(k)}_n.
    \end{equation}
    We also define $\Theta = \max_{n \in \mathcal{N}} \{\Theta_n\}$. We further let $\sigma^{(k)}_{n}$ denote the variance of the features of datapoints in dataset $\mathcal{D}_n^{(k)}$.
\end{definition}
\begin{assumption}[Smoothness of the Loss Functions]\label{Assup:lossFun}
    The local loss function $F^{(k)}_{n}$ of each satellite $n$ is  $\beta$-smooth,~$\forall n,k$:
     \vspace{-.5mm}
    \begin{equation*}
       \hspace{-1.5mm} \Vert F^{(k)}_{n}(\bm{\omega})\hspace{-.4mm}-\hspace{-.4mm} F^{(k)}_{n}(\bm{\omega}') \Vert \hspace{-.6mm}\leq \hspace{-.6mm} \beta \Vert \bm{\omega}-\bm{\omega}' \Vert,~ \hspace{-.8mm}\forall \bm{\omega},\bm{\omega}'\hspace{-.8mm}\in\hspace{-.4mm}\mathbb{R}^M\hspace{-.7mm}, \hspace{-2.3mm} 
        \vspace{-.1mm}
    \end{equation*}
     \noindent which implies the {\small $\beta$}-smoothness of {\small$F^{(k)}$} defined in~\eqref{eq:globalLoss}.
\end{assumption}

Next, inspired by~\cite{10106478}, we define \textit{model drift} to capture the dynamics of satellites' datasets. 
\begin{definition} [Model Drift]\label{def:cons} Let {\small $\mathcal{D}_n(t)$}, with size {\small${D}_n(t)=|\mathcal{D}_n(t)|$}, capture the instantaneous dataset of satellite $n$ at wall-clock time $t$ and {\small$D(t)=\sum_{n\in \mathcal{N}} D_n(t)$}. Also, let 
{\small$ F_n  \left(\bm{\omega}\big|\mathcal{D}_n(t) \right)=\sum_{d \in \mathcal{D}_n(t)} {{f}(\bm{\omega},d)}\big/{D_n(t)}$} refer to the instantaneous local loss of satellite $n$ at wall-clock time $t$.
For each satellite {\small$n\in\mathcal{N}$}, we measure the online model drift for two consecutive wall-clock time instances $t-1$ and $t$ during which the satellite does not conduct ML model training by {\small$\Delta_n(t)\in\mathbb{R}$}, which captures the variation of the local loss for any model parameter, {\small$ \forall \bm{\omega}\in\mathbb{R}^M$}:
\begin{equation}\label{eq:conceptDrift}
\hspace{-3mm}
\resizebox{.90\linewidth}{!}{$
 \frac{D_n(t)}{D(t)}  F_n  \left(\bm{\omega}\big|\mathcal{D}_n(t) \right) -  \frac{D_n(t-1)}{D (t-1)} F_n   \left( \bm{\omega}\big|\mathcal{D}_n(t-1)\right) \leq  \Delta_n (t).
 $}
 \hspace{-3mm}
\end{equation}
\end{definition}
Greater values of the model drift $\Delta_n(t)\gg 0$ imply larger local loss variations under real-time data collection at satellite $n$, which in turn makes tracking the optimal ML model parameters in \eqref{eq:genForm} (discussed in Sec.~\ref{sec:MLform}) more challenging.


The next two assumptions, built upon \cite{DBLP:journals/corr/abs-2202-02947}, quantify data heterogeneity across the satellites within a VC and across different VCs in the constellation.
 \begin{assumption}[Dissimilarity of Loss Functions inside each VC]\label{Assup:IntraClusterDissimilarity}
 For each VC {\small$c$}, there exist two finite constants {\small$\zeta^{\mathsf{Loc}}_{c,1} \geq 1$} and {\small$\zeta^{\mathsf{Loc}}_{c,2} \geq 0$} such that {\small$\forall k,\bm{\omega}$} hold in:
    \vspace{-1mm}
   \begin{equation*}
   \hspace{-.5mm}
   \resizebox{.99\linewidth}{!}{$
        \sum_{n\in \mathcal{N}^{(k)}_{c}} {a_n} \left\Vert \nabla F^{({k})}_n (\bm{\omega}) \right\Vert^2 \leq \zeta^{\mathsf{Loc}}_{c,1}  \left\Vert  \sum_{n\in \mathcal{N}^{(k)}_{c}}a_n \nabla F^{({k})}_n (\bm{\omega}) \right\Vert^2 {+}\zeta^{\mathsf{Loc}}_{c,2}
        $}
   \end{equation*}
    for any set of coefficients {\small$\{a_n\geq 0\}$}, where {\small$\sum_{n\in \mathcal{N}^{(k)}_{c}} a_n=1$}.
\end{assumption}
\begin{assumption}[Dissimilarity of Loss Functions across VCs]\label{Assup:InterClusterDissimilarity}
There exist two finite constants {\small$\zeta^{\mathsf{Glob}}_1 \geq 1$} and {\small$\zeta^{\mathsf{Glob}}_2 \geq 0$} such that  {\small$\forall k,\bm{\omega}$} hold in:
\begin{equation*}
       \hspace{-1mm}
       \resizebox{.99\linewidth}{!}{$\sum_{c\in \mathcal{C}^{(k)}} b_c  \left\Vert  \nabla F^{({k})}_c (\bm{\omega}) \right\Vert^2 \hspace{-1.5mm}{\leq}  \zeta^{\mathsf{Glob}}_1  \left\Vert  \sum_{c\in \mathcal{C}^{(k)}}\hspace{-0.2mm}b_c \nabla F^{({k})}_c (\bm{\omega})    \right\Vert^2\hspace{-1mm} {+}\zeta^{\mathsf{Glob}}_2$}
\end{equation*}
 for any set of coefficients {\small$\{b_c\geq 0\}$}, where {\small$\sum_{c\in \mathcal{C}^{(k)}} b_c=1$}.
\end{assumption}
The four parameters {\small$\zeta^{\mathsf{Loc}}_{c,1}$}, {\small$\zeta^{\mathsf{Loc}}_{c,2}$}, {\small$\zeta^{\mathsf{Glob}}_1$}, {\small$\zeta^{\mathsf{Glob}}_2$} described above capture the heterogeneity of data (i.e., non-iid-ness of data) inside each VC and across the VCs.
To guide satellite clustering (i.e., forming VCs), we need to obtain a relationship between the clustering and the data heterogeneity of satellites across VCs. This requires obtaining a nuanced result in the literature of both hierarchical FedL and Fed-LS, which we present next.

\begin{proposition}[Bounding the Dissimilarity of Loss Functions across VCs]\label{th:clus}
Using Assumptions \ref{Assup:IntraClusterDissimilarity} and \ref{Assup:InterClusterDissimilarity} while presuming the feasible values of {\small$
\zeta^{\mathsf{Loc}}_{c,2} = 0, \forall c \in \mathcal{C}^{(k)}$}, and {\small$
\zeta^{\mathsf{Glob}}_1 = 1$}, let {\small$\zeta^{\mathsf{Loc},\min}_{c,1}$} denote the minimum value of {\small$\zeta^{\mathsf{Loc}}_{c,1}$} under which Assumption \ref{Assup:IntraClusterDissimilarity} holds. By defining  {\small$\hat{\zeta}^{\mathsf{Loc}}_{1} \triangleq \max_{c \in \mathcal{C}^{(k)}} \{\zeta^{\mathsf{Loc},\min}_{c,1}\}$}, the global dissimilarity of loss functions {\small$\zeta^{\mathsf{Glob}}_{2}$} in Assumption~\ref{Assup:InterClusterDissimilarity} satisfies:
\begin{equation*}
\resizebox{.99\linewidth}{!}{$
\frac{ \displaystyle \sum_{c \in \mathcal{C}^{(k)}} \hspace{-1mm} {b_c} \sum_{n \in \mathcal{N}^{(k)}_c}\hspace{-1mm} a_n \left\Vert \nabla F^{(k)}_n (\bm{\omega}) \right\Vert^2 
    - \hat{\zeta}^{\mathsf{Loc}}_{1} \left\Vert \sum_{c \in \mathcal{C}^{(k)}} \hspace{-1mm}b_c \sum_{n \in \mathcal{N}^{(k)}_c} \hspace{-1mm}a_n \nabla F^{(k)}_n (\bm{\omega}) \right\Vert^2}{\displaystyle\hat{\zeta}^{\mathsf{Loc}}_{1}} 
    {\leq} 
       \zeta^{\mathsf{Glob}}_2.
       $}
\end{equation*}

\begin{proof}
Refer to Appendix~\ref{app:th:clus} for the detailed proof.
\end{proof}
\end{proposition}
\begin{remark}[Interpretation of Proposition~\ref{th:clus}]\label{remark4}
Proposition~\ref{th:clus} unveils that the dissimilarity of loss functions across VCs is ultimately constrained by the VC with the highest internal dissimilarity (reflected by {\small$\hat{\zeta}^{\mathsf{Loc}}_{1}$}). Importantly,  across all combinatorial configurations of VCs, increasing the number of VCs tends to result in sparser and more heterogeneous data distributions within each VC, which raises {\small$\hat{\zeta}^{\mathsf{Loc}}_{1}$} and consequently increases {\small$\zeta^{\mathsf{Glob}}_{2}$}. We empirically validate this phenomenon in Fig.~\ref{Fig:lowerBound}, which visualizes the distribution of {\small$\zeta^{\mathsf{Glob}}_{2}$} values across multiple (20,000 samples) VC configurations.
\end{remark}

We next present the general convergence behavior of {Fed-Span} as one of our main results.

\begin{theorem}[General Convergence Behavior of {Fed-Span}]\label{th:main}
Assume that the step-size of SGD used in~\eqref{eq:SGD_iterations}, satisfies:
\vspace{-0.05mm}
\begin{equation*}
\resizebox{.99\linewidth}{!}{$
    \eta_k \hspace{-0.25mm}\leq\hspace{-0.25mm} \min \Bigg\{\hspace{-0.7mm}\frac{1}{\beta}\hspace{-0.5mm} \sqrt{ \frac{\Lambda^{(k)}}{8\left(\hspace{-0.25mm}\zeta^{\mathsf{Glob}}_1\hat{\zeta}^{\mathsf{Loc}}_{1}+\Lambda^{(k)}\hspace{-0.45mm}\right)\left(\hspace{-0.65mm}L^{(k)}  \left(L^{(k)}-1\right) \left(\hspace{-0.25mm}e^{(k)}_{\mathsf{max}}\hspace{-0.25mm}\right)^{2} \hspace{-0.45mm}{+}  e^{(k)}_{\mathsf{max}}\left(\hspace{-0.25mm}e^{(k)}_{\mathsf{max}}-1\right)\hspace{-0.25mm}\hspace{-0.55mm}\right)}}
     , \Big(\hspace{-0.65mm}2\beta\hspace{-1.85mm} \displaystyle\sum_{c'\in \mathcal{C}^{(k)}}\hspace{-1.45mm}\frac{{D}^{(k)}_{c'}L^{(k)}e^{(k)}_{c'}}{D^{(k)}}\hspace{-0.65mm}\Big)^{-1}\hspace{-0.75mm}\Bigg\},
    $}
    \vspace{-0.2mm}
\end{equation*}
where {\small$\Lambda^{(k)}<1$} is a constant bounded as {\small$\max_{(k)} \left\{\Lambda^{(k)}\right\} \leq \Lambda_{\mathsf{max}}< 1$} and {\small$e^{(k)}_{\mathsf{max}}$} holds in {\small$\max_{c \in \mathcal{C}^{(k)}} \{e^{(k)}_{c}\} \leq e^{(k)}_{\mathsf{max}}$}.
We define {\small$\Phi^{(k)} = \frac{\eta_k}{2} \sum_{c\in \mathcal{C}^{(k)}}{{D}^{(k)}_{c} L^{(k)} e^{(k)}_{c}}/{D^{(k)}}$} bounded as {\small$\Phi_{\mathsf{min}} \leq \Phi_{(k)}\leq \Phi_{\mathsf{max}}$} for finite constants {\small$\Phi_{\mathsf{min}}$} and {\small$\Phi_{\mathsf{max}}$}, and also consider {\small$\Delta^{(k)} = \sum_{\ell \in\mathcal{L}^{(k)}} \sum_{n\in \mathcal{N}} \Delta^{(k,\ell)}_{n}$} where {\small$\Delta^{(k,\ell)}_{n} = \max_{t \in (t^{\mathsf{LT},{(k,\ell)}}+\tau^{\mathsf{LT},{(k,\ell)}},t^{\mathsf{LT},{(k,\ell)}}+\tau^{\mathsf{LT},{\mathsf{max}}}]}\{\Delta_{n}(t)\}$}. Then, the cumulative average of the gradient of the global loss over the training period of {Fed-Span} satisfies the upper bound in~\eqref{eq:gen_conv}. 
\end{theorem}
\begin{proof}
Refer to Appendix~\ref{app:th:main} for the detailed proof.
\end{proof}

\begin{table*}[t!]
\vspace{-6mm}
\begin{minipage}{0.99\textwidth}
{\footnotesize
\begin{equation}\label{eq:corr2_conv}
\hspace{-7mm}
\resizebox{.75\linewidth}{!}{$
{\small
\begin{aligned}
    &\frac{1}{K} \sum_{k=0}^{K-1}\mathbb E\left \Vert\nabla F^{(k)}(\bm{\omega}^{(k)})\right\Vert^2 \leq \frac{1}{\sqrt{K} \left(1-\Lambda_{\mathsf{max}}\right)} \Bigg[2 \sqrt{\ell_{\mathsf{max}} \widehat{e}_{\mathsf{max}}}\frac{F^{(0)}(\bm{\omega}^{(0)}) - F^{(K)^{\ast}}}{\ell_{\mathsf{min}}\overline{e}_{\mathsf{min}} \alpha \sqrt{N}} + \frac{2 \sqrt{\ell_{\mathsf{max}} \widehat{e}_{\mathsf{max}}} \chi}{\ell_{\mathsf{min}}\overline{e}_{\mathsf{min}} \alpha \sqrt{N}} + \frac{4 \ell_{\mathsf{max}}\overline{e}_{\mathsf{max}} \Theta^2 \alpha \beta \sqrt{N}}{\sqrt{ K \ell_{\mathsf{min}} \widehat{e}_{\mathsf{min}}}} \sigma_{\mathsf{max}}
    \\& + \frac{16 \alpha^2 \beta^2 N}{\sqrt{K} \ell_{\mathsf{min}} \widehat{e}_{\mathsf{min}}} \left(\ell_{\mathsf{max}} \left(\ell_{\mathsf{max}}-1\right) \left(e_{\mathsf{max}}\right)^{2} +  e_{\mathsf{max}}\left(e_{\mathsf{max}}-1\right)\right) \zeta^{\mathsf{Glob}}_{2}\hat{\zeta}^{\mathsf{Loc}}_{1} +\frac{16 \alpha^2 \beta^2 \Theta^2 N}{\sqrt{K} \ell_{\mathsf{min}} \widehat{e}_{\mathsf{min}}} \left(\left(\ell_{\mathsf{max}}-1\right) e_{\mathsf{max}} + \left(e_{\mathsf{max}}-1\right)\right) \sigma_{\mathsf{max}}\Bigg]
\end{aligned}
}
$}
\hspace{-5mm}
\end{equation}
}
\hrule
\end{minipage}
\vspace{-5mm}
\end{table*}
\begin{figure}[t]
\centering
\includegraphics[width=0.49\textwidth, trim= 1 1 1 1, clip]{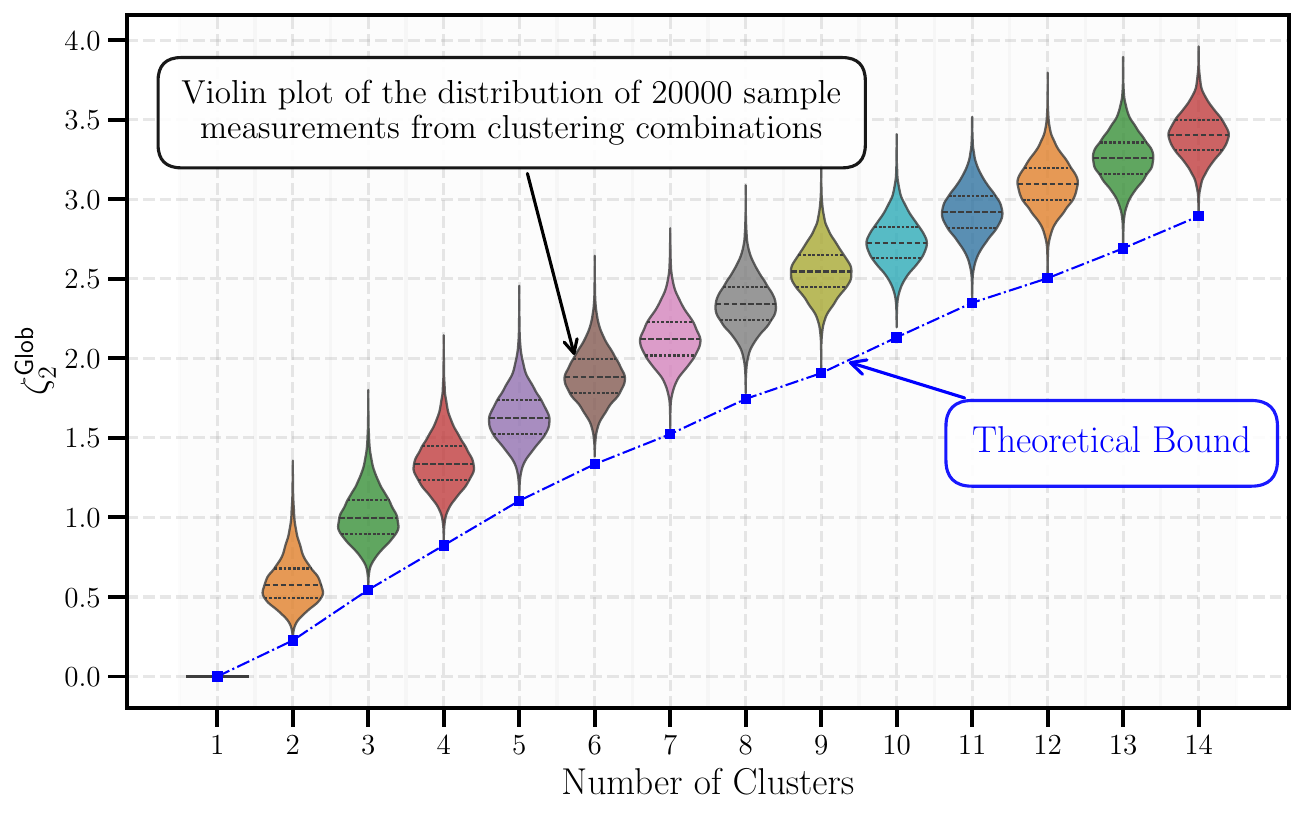}
\vspace{-6mm}
\caption{The lower-bound obtained in Proposition~\ref{th:clus} (i.e., the left hand side of the bound) vs the actual value of the right hand side (i.e., {\small$\zeta^{\mathsf{Glob}}_2$}).\label{Fig:lowerBound}}
\end{figure}

\begin{remark}[Interpretation of Theorem~\ref{th:main}]
The bound in \eqref{eq:gen_conv} captures the influence of both ML- and network-related parameters on the model convergence of {Fed-Span}. 
Term \textit{(a)} highlights the impact of the initial loss value on the convergence.
Term \textit{(b)} reflects that the multiplication of model drift {\small$\Delta^{(k)}$} and idle times {\small$\Omega^{(k)}$} has a linear impact on the convergence. Term \textit{(c)} shows the impact of mini-batch sizes used by satellites {\small$\varsigma^{(k,\ell)}_n$}.
Term \textit{(d)} accounts for the effect of loss dissimilarity: by invoking Proposition~\ref{th:clus}, it shows that the bound deteriorates linearly with the data heterogeneity of the VC exhibiting the highest intra-cluster heterogeneity among all VCs, as quantified by {\small$\hat{\zeta}^{\mathsf{Loc}}_{1}$}.
Also, the impacts of {\small$\hat{\zeta}^{\mathsf{Loc}}_{1}$} and {\small$\zeta^{\mathsf{Glob}}_1$} are reflected in the choice of the step size {\small$\eta_k$} in the statement of the theorem, where larger values of these quantities induce smaller step sizes to avoid the bias of satellites' local models. 
Term \textit{(e)} captures the influence of SGD iterations {\small$e_c^{(k)}$} and the number of local training rounds {\small$L^{(k)}$} on the convergence. Further, terms \textit{(c)} and \textit{(e)} capture the impact of various configurations of the satellite-to-VC association on the convergence through {\small$\gamma_{c,n}^{(k)}$}.
The bound in \eqref{eq:gen_conv} further signals the specific conditions under which a stronger convergence of {\small$\mathcal{O}(1/\sqrt{K})$} can be achieved. These conditions are obtained in the following corollary.
\end{remark}

 \begin{corollary}[Convergence under Proper Choice of ML- and Network-Related Parameters]\label{cor:1}
In addition to the conditions mentioned in Theorem~\ref{th:main}, further assume that (i) {\small$\Phi_{\mathsf{min}} = \eta \ell_{\mathsf{min}} \overline{e}_{\mathsf{min}} / 2$} and {\small$\Phi_{\mathsf{max}} = \eta \ell_{\mathsf{max}} \overline{e}_{\mathsf{max}} / 2$}, (ii) bounded number of local aggregation rounds {\small$\ell_{\mathsf{min}} \leq L^{(k)} \leq \ell_{\mathsf{max}}$} for finite positive constants {\small$\ell_{\mathsf{min}}$ and $\ell_{\mathsf{max}}$}, (iii) bounded average SGD {\small$\overline{e}_{\mathsf{min}}\leq e_{\mathsf{avg}}^{(k)} \triangleq \sum_{c\in \mathcal{C}^{(k)}}\frac{{D}^{(k)}_{c} e^{(k)}_{c}}{D^{(k)}} \leq \overline{e}_{\mathsf{max}}$} for finite positive constants {\small$\overline{e}_{\mathsf{min}}$} and {\small$\overline{e}_{\mathsf{max}}$}, (iv) step size choice of {\small$\eta = \alpha \big /{\sqrt{ \ell_{\mathsf{max}} \widehat{e}_{\mathsf{max}} K \big /N}}$} with a finite positive constant {\small$\alpha$} to satisfy the condition on {\small$\eta_k$} in Theorem~\ref{th:main} {\small$\forall k$}, (v) bounded total SGD for each local round {\small$\widehat{e}_{\mathsf{min}} \leq e^{(k)}_{\mathsf{sum}}\leq  \widehat{e}_{\mathsf{max}}$} for finite positive constants {\small$\widehat{e}_{\mathsf{min}}$} and {\small$\widehat{e}_{\mathsf{max}}$}, where  {\small$e^{(k)}_{\mathsf{sum}}=\sum_{c\in \mathcal{C}^{(k)}} e^{(k)}_{c} N^{(k)}_{c}$}.
We further define the bounds on (i) the SGD sampling noise as {\small$\max\limits_{k, \ell, n} \left\{\left(1- \varsigma^{(k,\ell)}_n\right)\frac{\left(\sigma^{(k)}_n\right)^2}{D^{(k)}_n \varsigma^{(k,\ell)}_n} \right\} \leq \sigma_{\mathsf{max}}$}, (ii) the local iterations as {\small$\max_{k} \{e_{\mathsf{max}}^{(k)}\} \leq e_{\mathsf{max}}$}, (iii) the idle period as {\small$\Delta^{(k)} \leq \left[\frac{\chi}{K\Omega^{(k)}}\right]^{+}$} for a finite non-negative constant {\small$\chi$}. Then, the cumulative average of the global loss function gradient for {Fed-Span} satisfies~\eqref{eq:corr2_conv}, which in turn guarantees that {\small$\frac{1}{K} \sum_{k=0}^{K-1}\mathbb E\left \Vert\nabla F^{(k)}(\bm{\omega}^{(k)})\right\Vert^2 \leq \mathcal{O}(1/\sqrt{K})$}. 
 \end{corollary} 
 \vspace{-3mm}
 \begin{proof}
 Refer to Appendix~\ref{app:cor:1} for the detailed proof.
 \end{proof}

\section{Optimized  Network Formation and Operation for {\large {Fed-Span}}}\label{sec:optimization_problem}
\noindent 
We next integrate the convergence analysis from Sec.~\ref{sec:conv} with the modeling of {Fed-Span} operations from Sec.~\ref{sec:phasesofFEDSPAN} to provide a unified framework for optimizing {Fed-Span}. This is achieved through an overarching optimization problem {\small$\bm{\mathcal{P}}$} presented below that jointly optimizes: (i) ML-related control decisions,
(ii) network topology for model aggregation and dispatching, and (iii) resource allocation across satellites. To our knowledge, this represents one of the first formulations in the literature to synergize graph theory and Fed-LS. 

\begin{align}
    \nonumber &\footnotesize\hspace{-0mm} \resizebox{.67\linewidth}{!}{$(\bm{\mathcal{P}}):~~~\hspace{-1mm}\min~ \Bigg[\underbrace{\alpha_1 \hspace{-.1mm}\frac{1}{K}\sum_{k=0}^{K-1}\mathbb E\left[ \big\Vert \nabla F^{({k})}(\bm{\omega}^{({k})})\big\Vert^2\right]}_{(a) \text{~$\rightarrow$ Obtained via \eqref{eq:gen_conv}}}
    $}
    \\ \nonumber
    &\resizebox{.98\linewidth}{!}{$\small\hspace{-2mm} + \alpha_{2} \underbrace{\bigg( \sum_{k\in\mathcal{K}}\Big[ \tau^{\mathsf{GD},(k)} +  \tau^{\mathsf{GA},(k)} + \sum_{\ell \in \mathcal{L}^{(k)}} \tau^{\mathsf{LD},(k,\ell)} +  \tau^{\mathsf{LA},(k,\ell)} + \tau^{\mathsf{LT},(k,\ell)} \Big]\bigg)}_{(b)} $}\\ \nonumber
    &\resizebox{.98\linewidth}{!}{$\small\hspace{-2mm} + \alpha_{3}\underbrace{\bigg(  \sum_{k\in\mathcal{K}}\bigg[  E^{\mathsf{GD}, (k)} + E^{\mathsf{GA}, (k)} +  \sum_{\ell \in \mathcal{L}^{(k)}}  E^{\mathsf{LD},(k,\ell)} + E^{\mathsf{LA},(k,\ell)} + E^{\mathsf{LT},(k,\ell)} \bigg]\bigg)}_{(c)} \Bigg]
    $}
    \\
    &\hspace{-1mm}\textrm{{\textbf{s.t.}}}\nonumber\\[-0.2em]
    &\hspace{-1mm}\textrm{\underline{{\textbf{Constraints}}}:}\nonumber\\[-0.2em]
    &\nonumber \sbullet[0.7]\text{\small {CBM Feasibility and Usage:} } \eqref{eq:hollow}, \eqref{rate2}, \eqref{rate3}, \eqref{edgetolink}\\
    &\nonumber \sbullet[0.7]\text{\small {Proper link and VC  Formation:} } \eqref{eq:min_data_rate}, \eqref{eq:finalRate}, \eqref{cons:CSC_4}, \eqref{cons:CSC_2}, \eqref{cons:CSC_1}\\
    &\nonumber \sbullet[0.7]\text{\small {Root Selection of Trees and Forests:} } \eqref{eq:rootDef}, \eqref{cons:SGD_1}, \eqref{eq:rootDef2}, \eqref{cons:SCA_1} \\
    &\nonumber \sbullet[0.7]\text{\small {Global Dispatching Properties:} } \eqref{cons:SGD_2}, \eqref{eq:GM_dispatching_transmission_latency} ,\eqref{eq:GM_dispatching_latency}, \eqref{cons:GM_dispatch_latency_1}, \eqref{cons:SGD_4}, \eqref{GDenergy}\\
    &\nonumber \sbullet[0.7]\text{\small {Local Training Properties:} } \eqref{l-ch1} ,\eqref{l-ch2}, \eqref{eq:timeselection}, \eqref{eq:training_latency}, \eqref{eq:T^f}, \eqref{eq:EN_LC}, \eqref{eq:EN_LC_sum} \\ 
    &\nonumber \sbullet[0.7]\text{\small {Local Aggregation Properties:} } \eqref{cons:SCA_2}, \eqref{eq:CM_aggregation_transmission_latency}, \eqref{eq:CM_aggregation_latency} ,\eqref{cons:CM_aggregation_latency_1}, \eqref{cons:SCA_4},\eqref{forestenergyAG} \\
    &\nonumber \sbullet[0.7]\text{\small {Local Dispatching Properties:} } \eqref{cons:SCD_1}, \eqref{eq:CM_dispatching_transmission_latency}, \eqref{eq:CM_dispatching_latency}, \eqref{cons:CM_dispatching_latency_1}, \eqref{cons:SCD_3}, \eqref{LDenergy}\\
    &\nonumber \sbullet[0.7]\text{\small {Global Aggregation Properties:} } \eqref{cons:SGA_1}, \eqref{eq:GM_aggregation_transmission_latency} ,\eqref{eq:GM_aggregation_latency}, \eqref{cons:GM_aggregation_latency_1},\eqref{cons:SGA_3}, \eqref{GAenergy} \\
    &\nonumber \sbullet[0.7]\text{\small {Idle Times:} } \eqref{eq:idletime}\\
    &\nonumber \textrm{\underline{\textbf{{Variables}:} }} \underbrace{\{ \psi_{m,m'}(t)\}_{m \in\mathcal{M}_n, m' \in\mathcal{M}_{n'}}}_{\text{Link Establishment and Network Topology}}, \underbrace{\{\gamma^{(k)}_{c,n}\}_{k\in\mathcal{K}, c \in \mathcal{C}, n \in \mathcal{N}}}_{\text{Satellite-to-VC Association}}\\
    & \nonumber, \underbrace{\{\pi_{n}^{\mathsf{G}, (k)}, \pi_{n}^{\mathsf{L}, (k)}\}_{k\in\mathcal{K}, n \in \mathcal{N}}}_{\text{Global/Local Root Selection}}, \underbrace{\{\lambda^{(k,\ell)}_{x}\}_{k\in\mathcal{K}, \ell \in \mathcal{L}, x \in \mathcal{X}}}_{\text{Local Round Tuning}}, \underbrace{\{e^{(k)}_{n}\}_{k\in\mathcal{K}, n \in \mathcal{N}}}_{\text{Satellite SGD Count}}\\
    &\nonumber \underbrace{\{\varsigma^{(k,\ell)}_n\}_{k\in\mathcal{K}, \ell \in \mathcal{L}, n \in \mathcal{N}}}_{\text{Mini-Batch Fraction}}, \underbrace{\{f^{(k,\ell)}_{n}\}_{k\in\mathcal{K}, \ell \in \mathcal{L}, n \in \mathcal{N}}}_{\text{CPU Frequency}}, \underbrace{\{\Omega^{(k,\ell)}\}_{k\in\mathcal{K}, \ell \in \mathcal{L}}}_{\text{Idle Time}}
\end{align}
\renewcommand{\theequation}{\arabic{equation}}

\vspace{-0.0mm}
\begin{figure*}[t]
\vspace{-3.5mm}
    \centering
        \begin{minipage}[t]{1\linewidth}
    \begin{tcolorbox}[colback=white,
        colbacktitle=white, coltitle=black,
        width=1\textwidth, left=0pt, right=0pt, top=0pt, bottom=0pt]
        \centering
        \includegraphics[width=1\linewidth, trim= 135 50 140 50, clip]{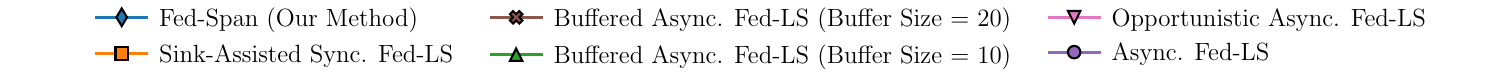}
    \end{tcolorbox}
     \end{minipage}
    \vspace{-3.55mm}
 
   \begin{minipage}[t]{0.32\linewidth}
    \begin{tcolorbox}[colback=redMod,
    colframe=mygreen,
    colbacktitle=mygreen,
    coltitle=white,
    left=0pt, right=0pt, top=2pt, bottom=-1pt,
    title=Fashion-MNIST (Latency/Time),
    halign title=flush center,
    boxrule=1.5pt]
    \hspace{-2.7mm}
    \begin{tabularx}{1.065\textwidth}{X}
            \centering
            \includegraphics[width=\linewidth, trim= 8 9 8 8, clip]{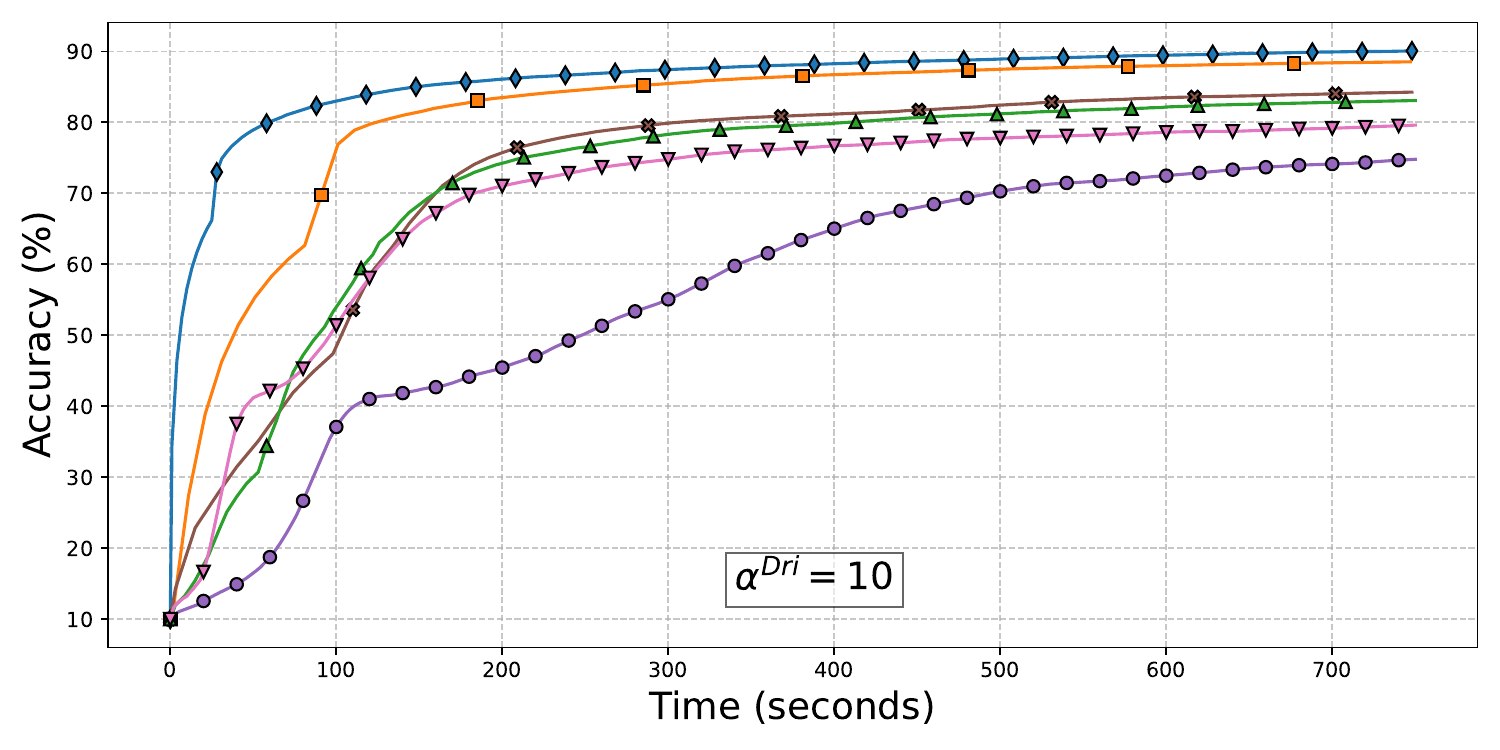}
        \\[0.17em]
            \centering
            \includegraphics[width=\linewidth, trim= 8 9 8 8, clip]{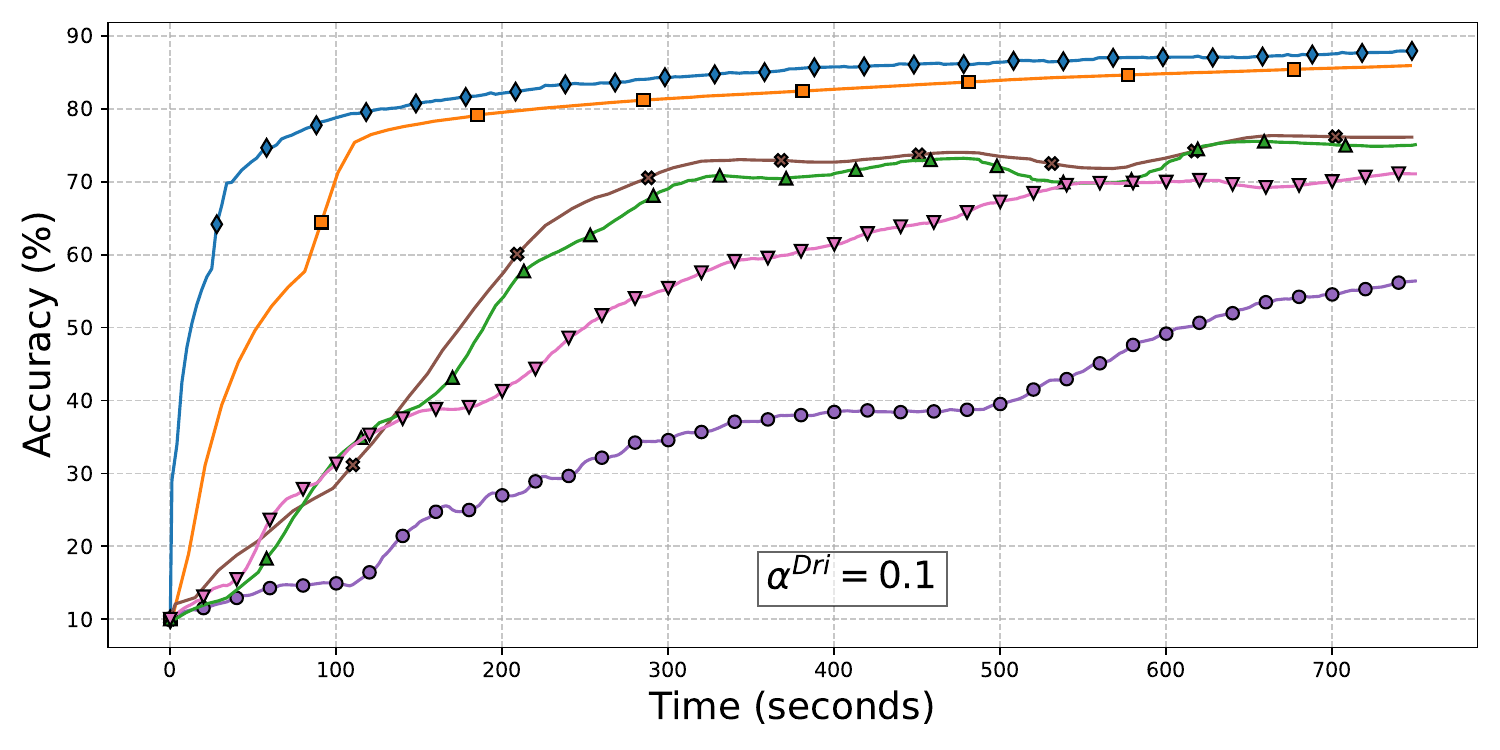}
    \end{tabularx}
\end{tcolorbox}
\vspace{-4mm}
 \end{minipage}
    \hfill
        \begin{minipage}[t]{0.32\linewidth}
        \begin{tcolorbox}[colback=redMod,
    colframe=mygreen,
    colbacktitle=mygreen,
    coltitle=white,
    left=0pt, right=0pt, top=2pt, bottom=-1pt,
    title=CIFAR-10 (Latency/Time), halign title=flush center, boxrule=1.5pt]
    \hspace{-2.7mm}
    \begin{tabularx}{1.065\textwidth}{X}
            \centering
            \includegraphics[width=\linewidth, trim= 8 9 8 8, clip]{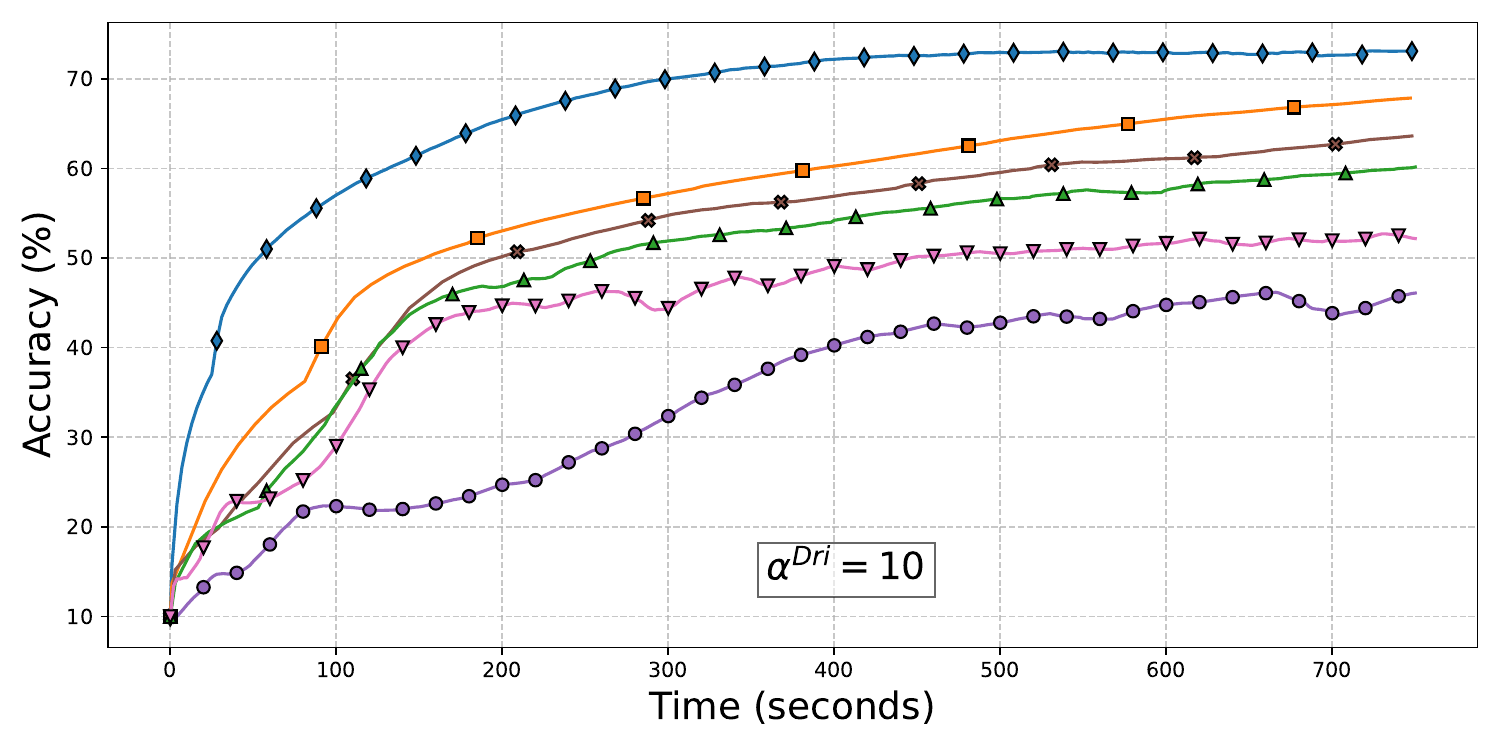}
        \\[0.17em]
            \centering
            \includegraphics[width=\linewidth, trim= 8 9 8 8, clip]{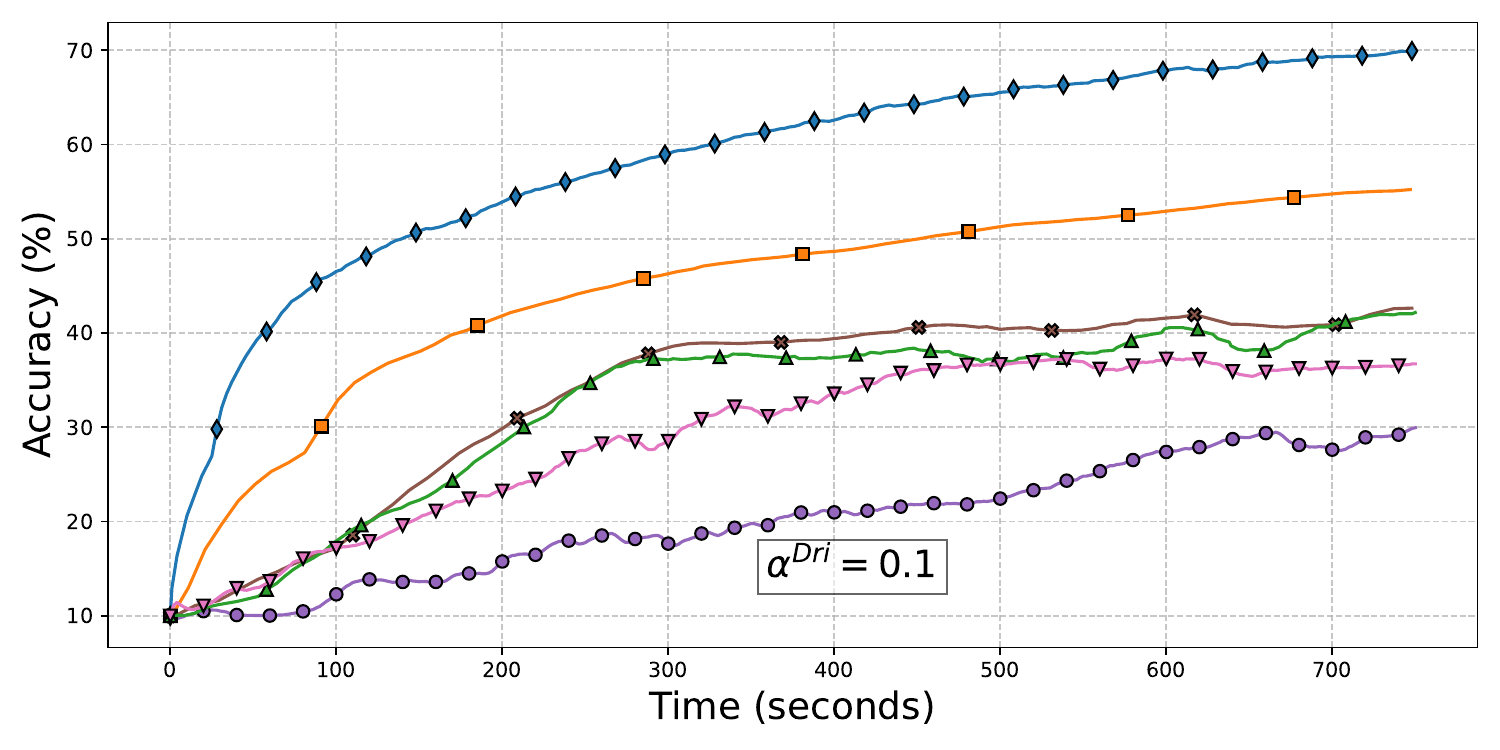}
    \end{tabularx}
    \end{tcolorbox}
    \vspace{-4mm}
    \end{minipage}
    \hfill
  \begin{minipage}[t]{0.32\linewidth}
    \begin{tcolorbox}[colback=redMod,
    colframe=mygreen,
    colbacktitle=mygreen,
    coltitle=white,
    left=0pt, right=0pt, top=2pt, bottom=-1pt,
    title=FMoW (Latency/Time), halign title=flush center, boxrule=1.5pt]
    \hspace{-2.7mm}
    \begin{tabularx}{1.065\textwidth}{X}
            \centering
            \includegraphics[width=\linewidth, trim= 8 9 8 8, clip]{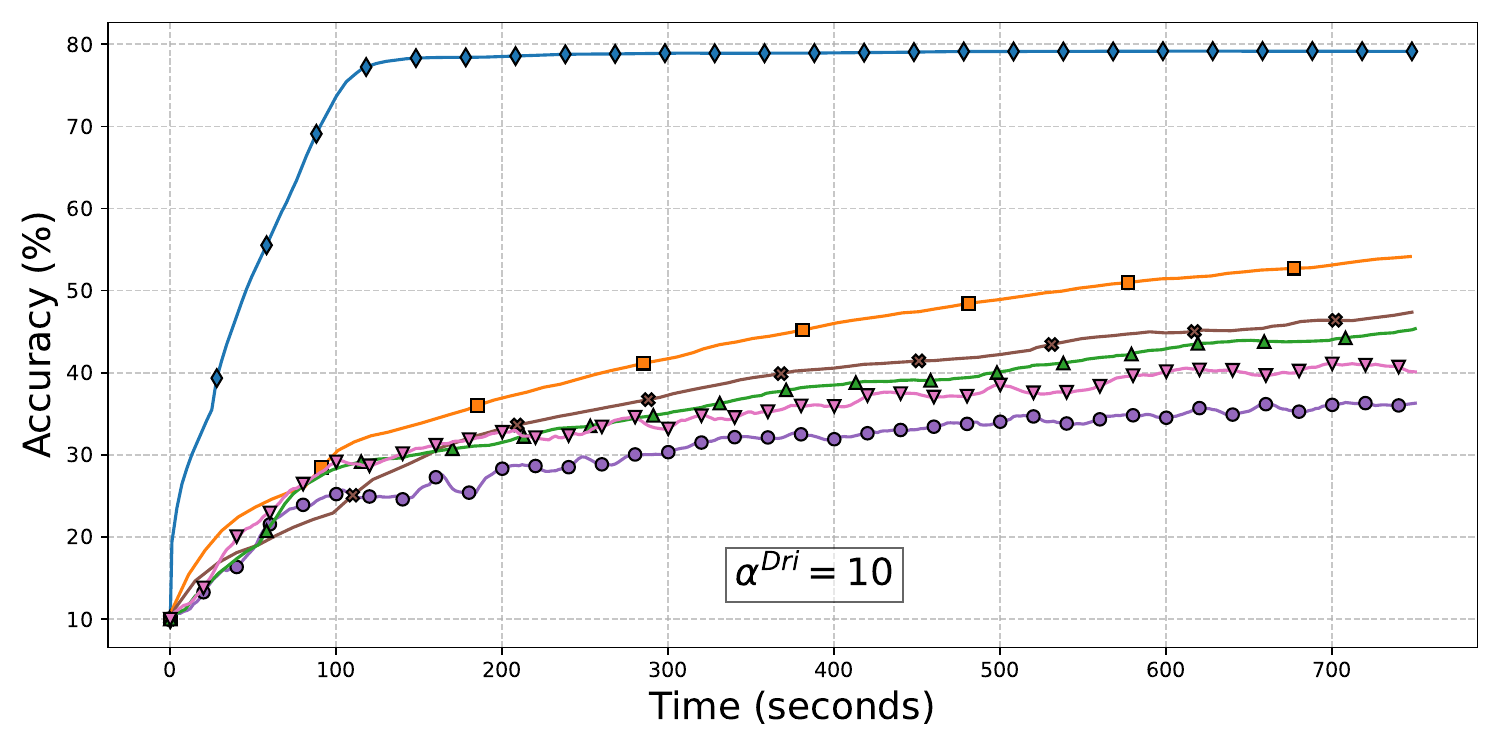}
        \\[0.17em]
            \centering
            \includegraphics[width=\linewidth, trim= 8 8 9 8, clip]{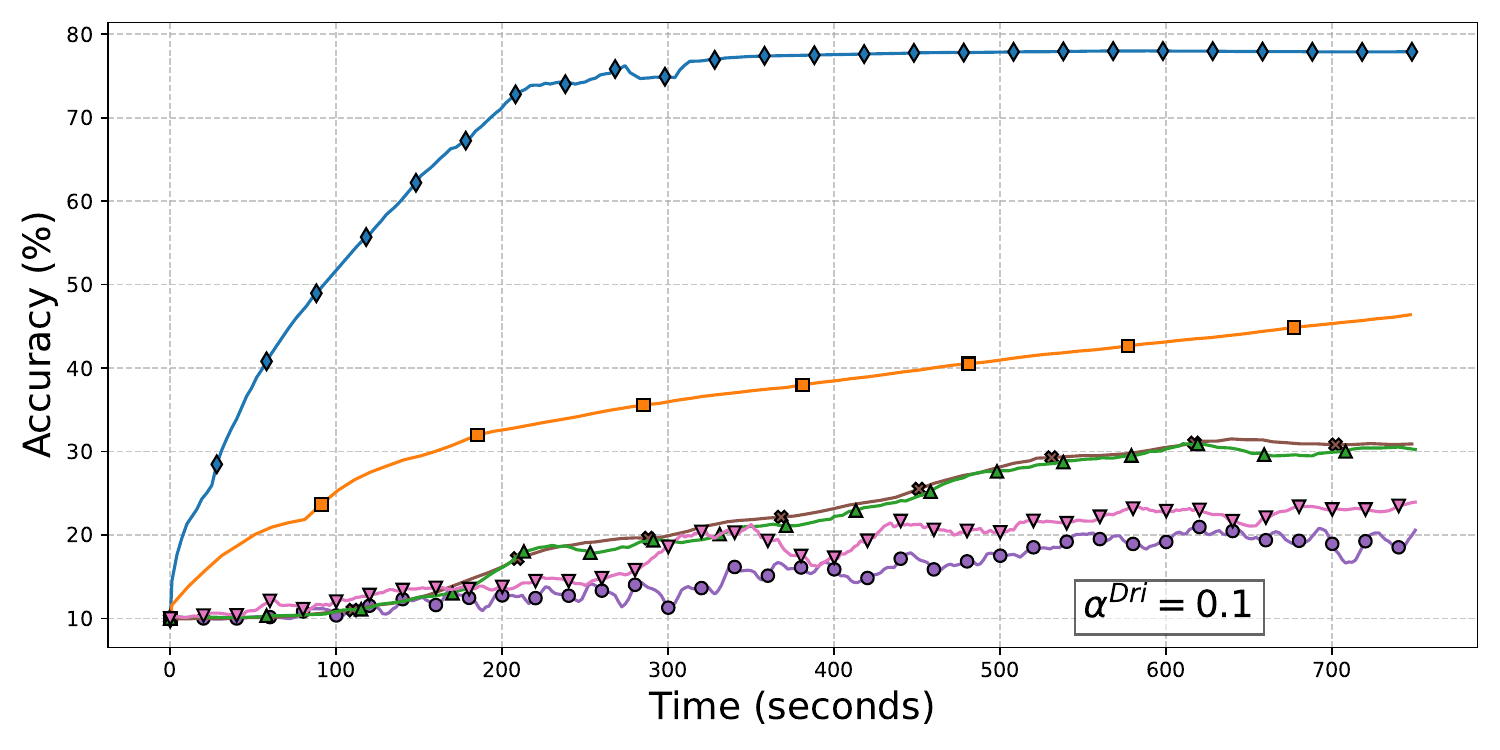}
    \end{tabularx}
\end{tcolorbox}
\vspace{-1.3mm}
 \end{minipage}

  \begin{minipage}[t]{0.32\linewidth}
    \begin{tcolorbox}[colback=redMod,
    colframe=mygreen,
    colbacktitle=mygreen,
    coltitle=white,
    left=0pt, right=0pt, top=2pt, bottom=-1pt,
    title=Fashion-MNIST (Energy Usage), halign title=flush center, boxrule=1.5pt]
    \hspace{-2.7mm}
    \begin{tabularx}{1.065\textwidth}{X}
            \centering
            \includegraphics[width=\linewidth, trim= 8 9 8 8, clip]{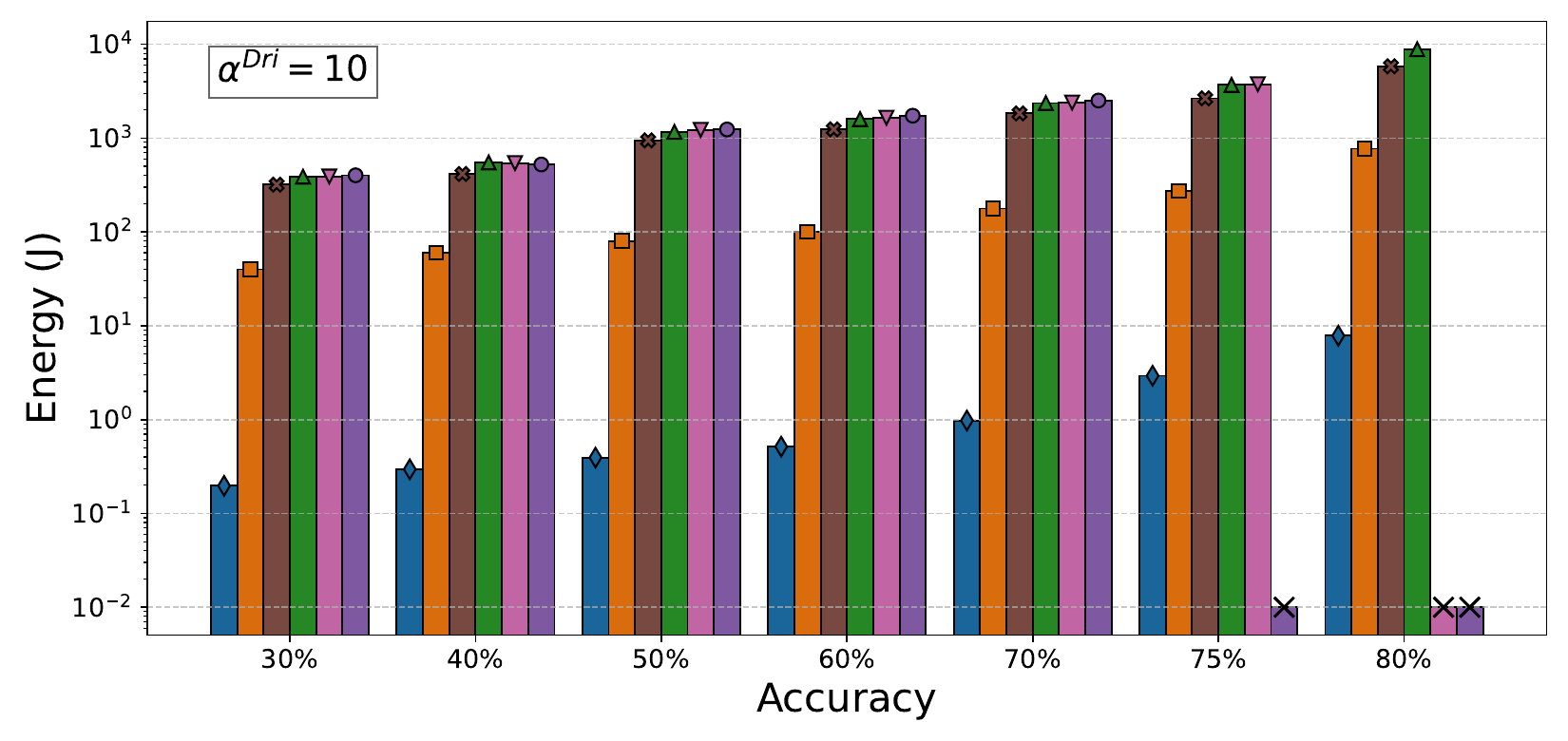}
        \\[0.17em]
            \centering
            \includegraphics[width=\linewidth, trim= 8 9 8 8, clip]{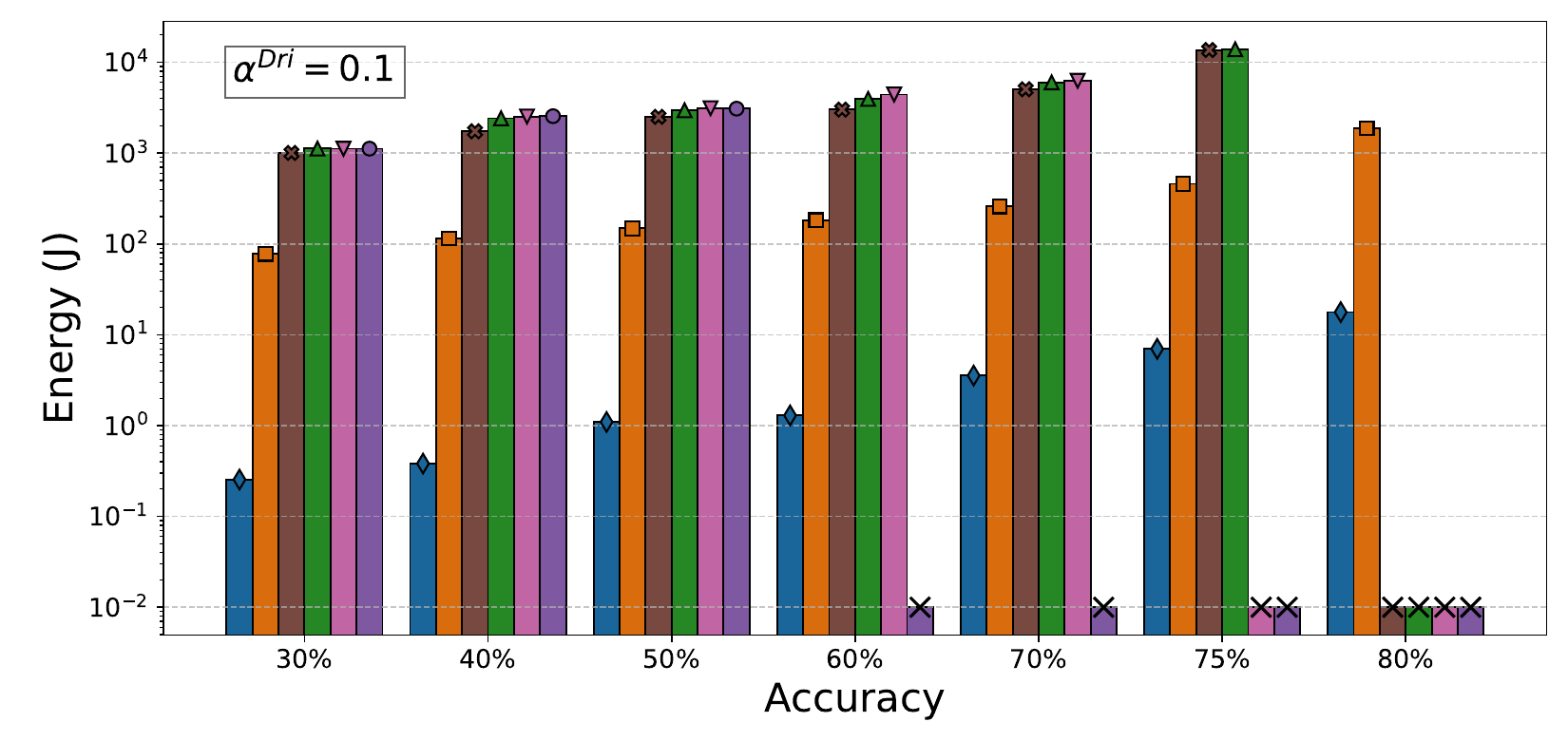}
    \end{tabularx}
\end{tcolorbox}
\vspace{-4mm}
 \end{minipage}
    \hfill 
    \hspace{0.9mm}
        \begin{minipage}[t]{0.32\linewidth}
        \begin{tcolorbox}[colback=redMod,
    colframe=mygreen,
    colbacktitle=mygreen,
    coltitle=white,
    left=0pt, right=0pt, top=2pt, bottom=-1pt,
    title=CIFAR-10 (Energy Usage), halign title=flush center, boxrule=1.5pt]
    \hspace{-2.7mm}
    \begin{tabularx}{1.065\textwidth}{X}
            \centering
            \includegraphics[width=\linewidth, trim= 8 9 8 8, clip]{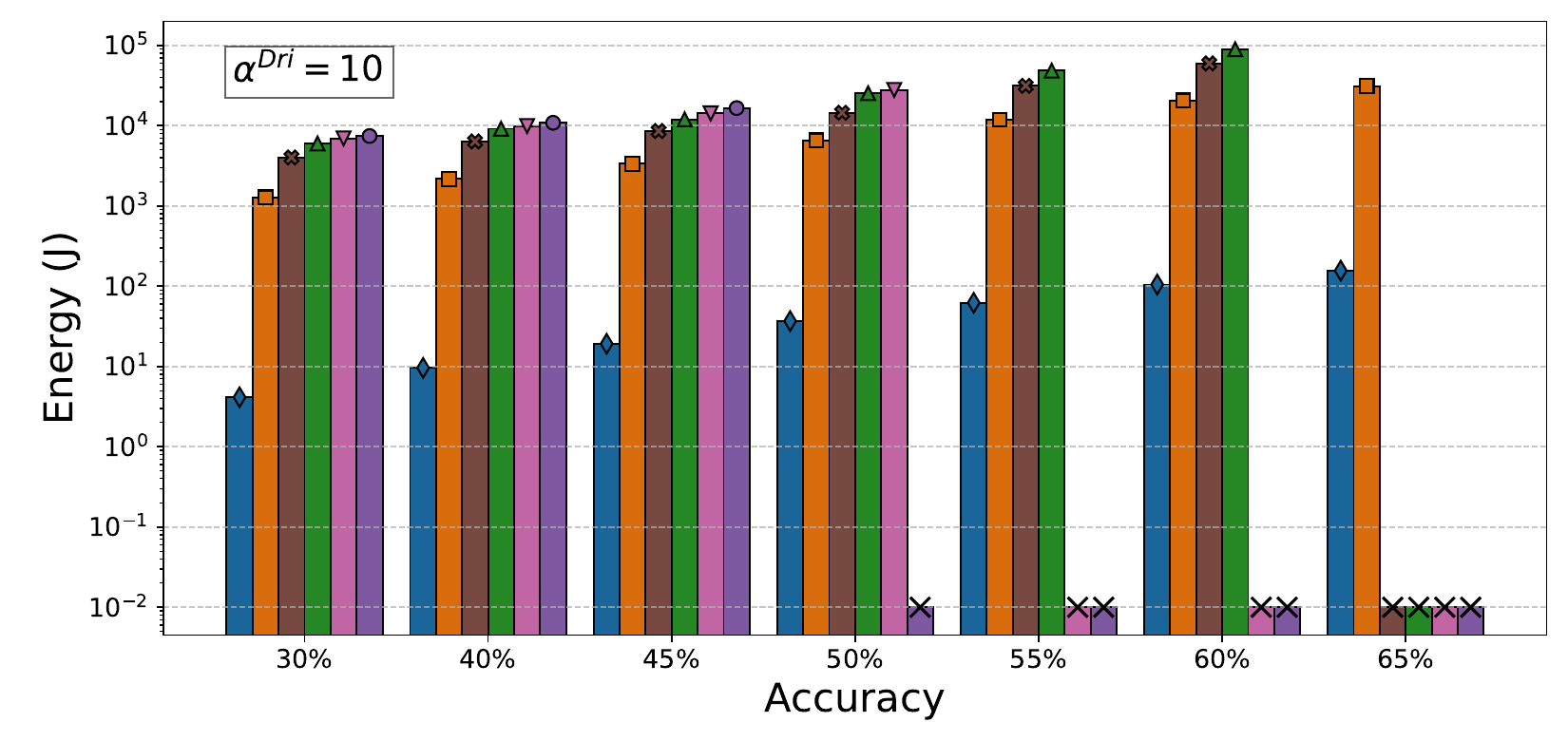}
        \\[0.17em]
            \centering
            \includegraphics[width=\linewidth, trim= 8 9 8 8, clip]{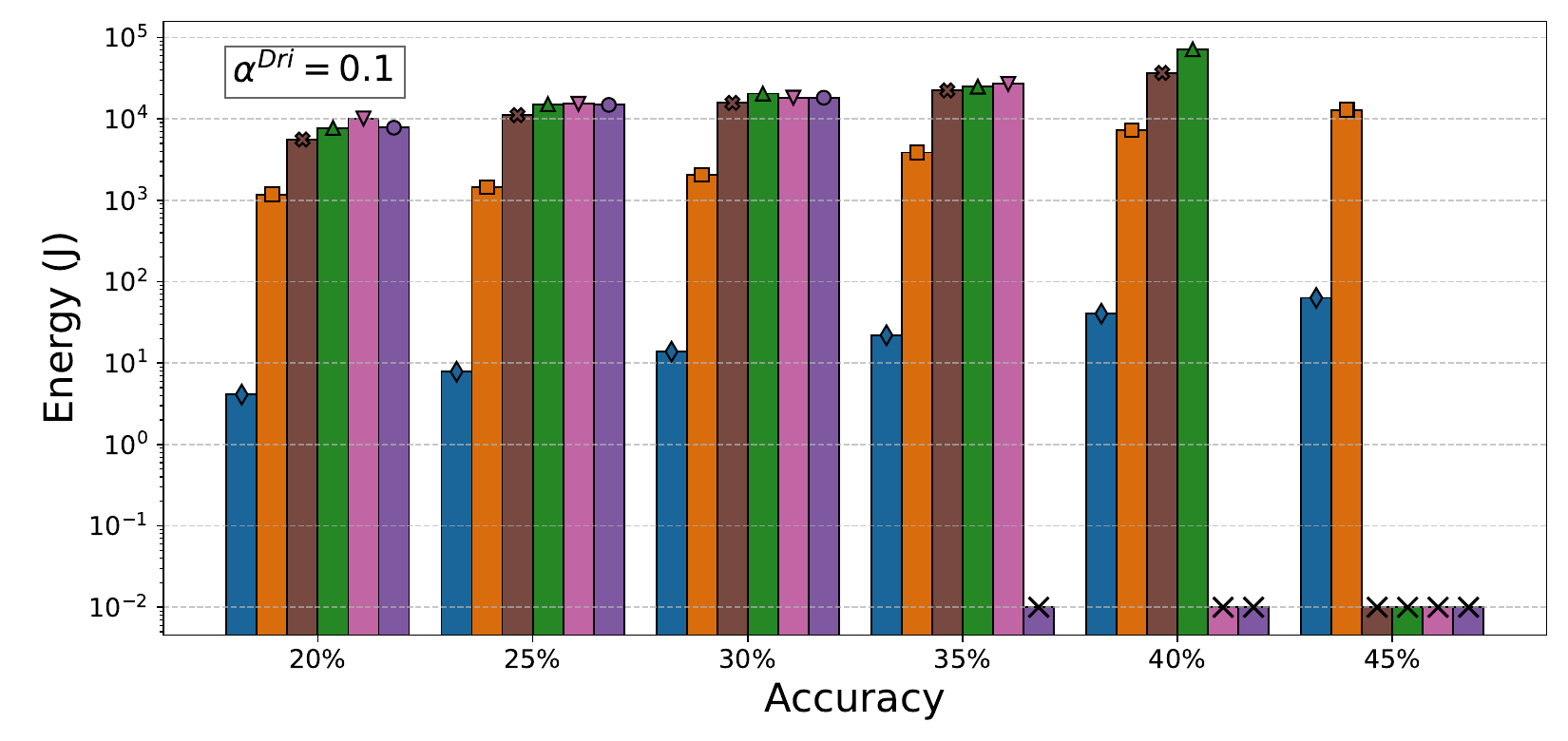}
    \end{tabularx}
    \end{tcolorbox}
    \vspace{-4mm}
    \end{minipage}
    \hfill
     \hspace{1.5mm}
  \begin{minipage}[t]{0.32\linewidth}
    \begin{tcolorbox}[colback=redMod,
    colframe=mygreen,
    colbacktitle=mygreen,
    coltitle=white,
    left=0pt, right=0pt, top=2pt, bottom=-1pt,
    title=FMoW (Energy Usage), halign title=flush center, boxrule=1.5pt]
    \hspace{-2.7mm}
    \begin{tabularx}{1.065\textwidth}{X}
            \centering
            \includegraphics[width=\linewidth, trim= 8 9 8 8, clip]{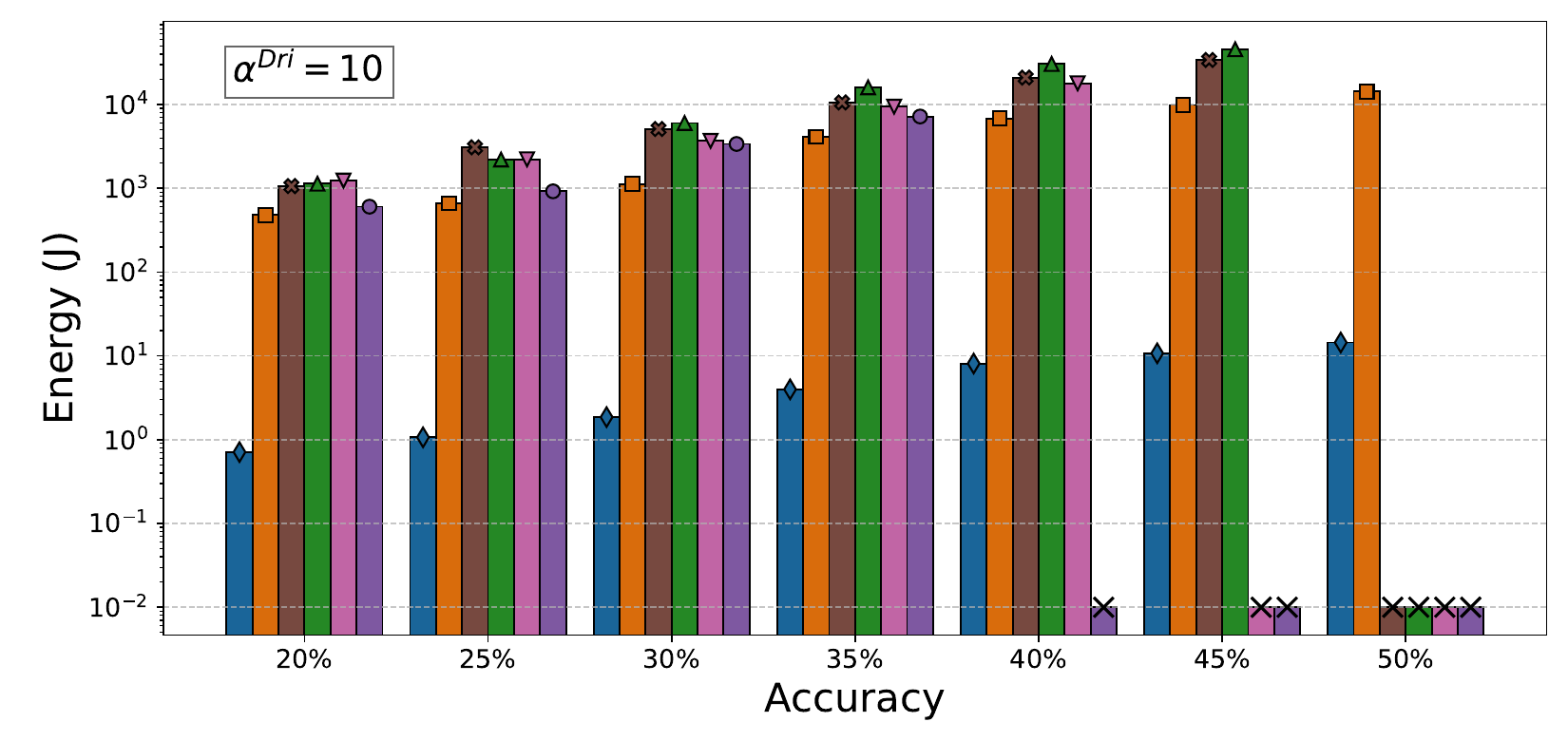}
        \\[0.17em]
            \centering
            \includegraphics[width=\linewidth, trim= 8 8 9 8, clip]{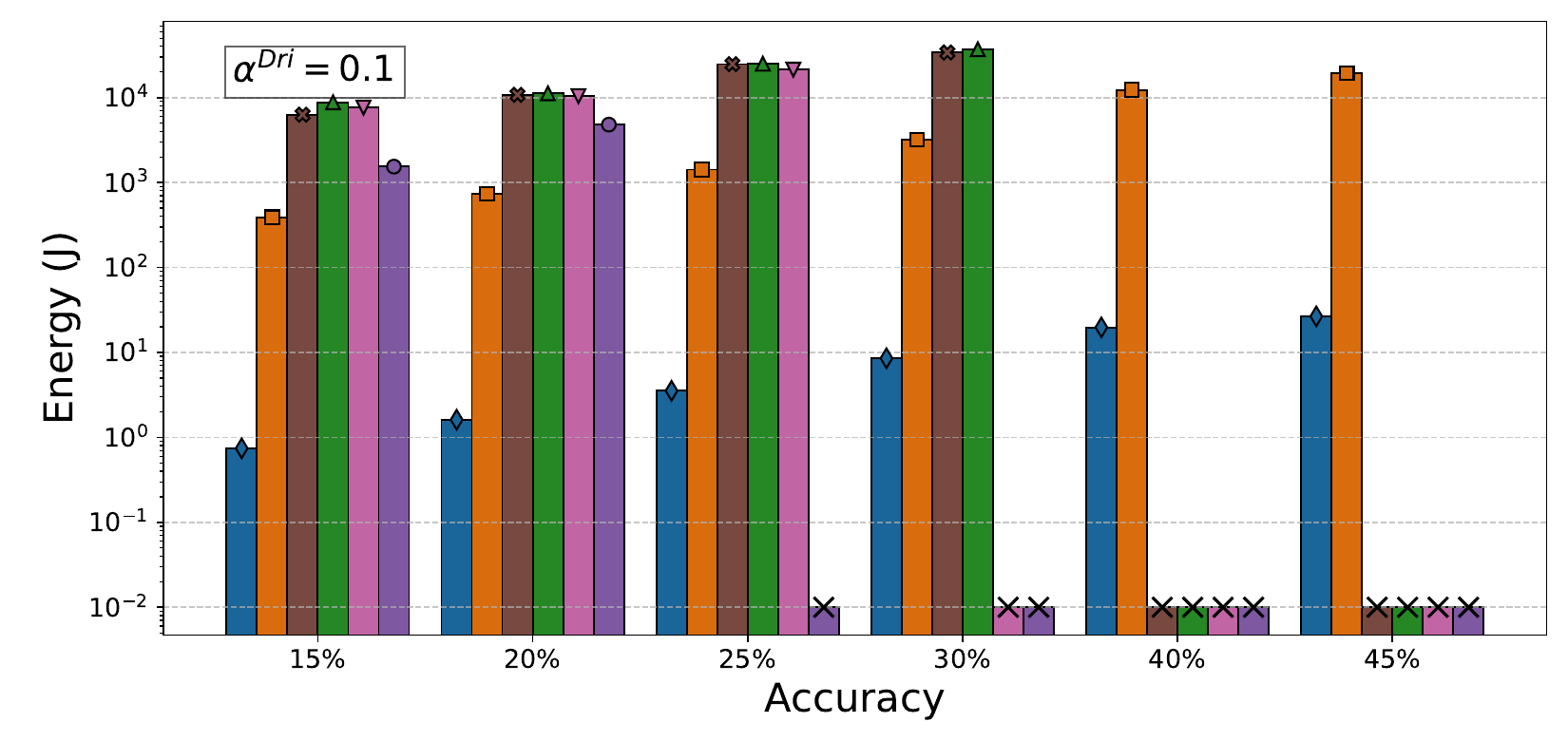}
    \end{tabularx}
\end{tcolorbox}
\vspace{-4mm}
 \end{minipage}
 \vspace{-5.5mm}
    \caption{Performance comparisons between our method and baselines in terms of (i) model test accuracy vs latency/time (top row) and (ii) energy usage required to reach various model test accuracy (bottom row) across three datasets (Fashion-MNIST, CIFAR-10, and FMoW) under various non-iid data configurations determined by Dirichlet parameter {\small$\alpha^{\mathsf{Dri}}$}. The same color-coding is used in all the plots (both line and bar plots) as described in the top legend. Lower {\small$\alpha^{\mathsf{Dri}}$} means more heterogeneous data. In the bar plots, sign {\small$\times$} on top of a bar implies that the respective algorithm could not reach the desired accuracy. Also, the y-axis of the bar plots are presented in the logarithmic format due to the large gap between our method and the rest of the baselines, caused by the use of fast optical links between the satellites, whose data rate can reach 30 Gbps.    
    Results show that our method consistently achieves superior performance compared to the baselines, with the performance gap widening as the datasets become more complex (i.e., CIFAR-10 and FMoW) and the data distribution becomes more non-iid. This is because for challenging datasets, the choice of training strategy and resource allocations plays a more prominent role in model training accuracy.}
    \label{fig:Latency_Energy_all_plots}
    \vspace{-5mm}
\end{figure*}
\textbf{Objective Function of $\bm{\mathcal{P}}$:} 
The objective function of {\small$\bm{\mathcal{P}}$} captures a three-sided trade-off between 
ML model performance (term (a)),
total latency of {Fed-Span} operations (term (b)), and 
total energy consumption of satellites (term (c)).
Specifically, for term (a), we utilize the generic bound from Theorem~\ref{th:main}, incorporating the step size condition from Corollary~\ref{cor:1}, while terms (b) and (c) use the derived results in Sec.~\ref{sec:phasesofFEDSPAN}. These three (potentially competing) objectives are weighted by non-negative coefficients {\small$\alpha_{1}$}, {\small$\alpha_{2}$}, and {\small$\alpha_{3}$}.
Additionally, {\small$\bm{\mathcal{P}}$} includes a set of constraints, organized into multiple categories using bullet points for improved readability.

\begin{figure*}[t]
\vspace{-6mm}
\centering
\includegraphics[width=0.9\textwidth, trim= 3 3 6.5 5, clip]{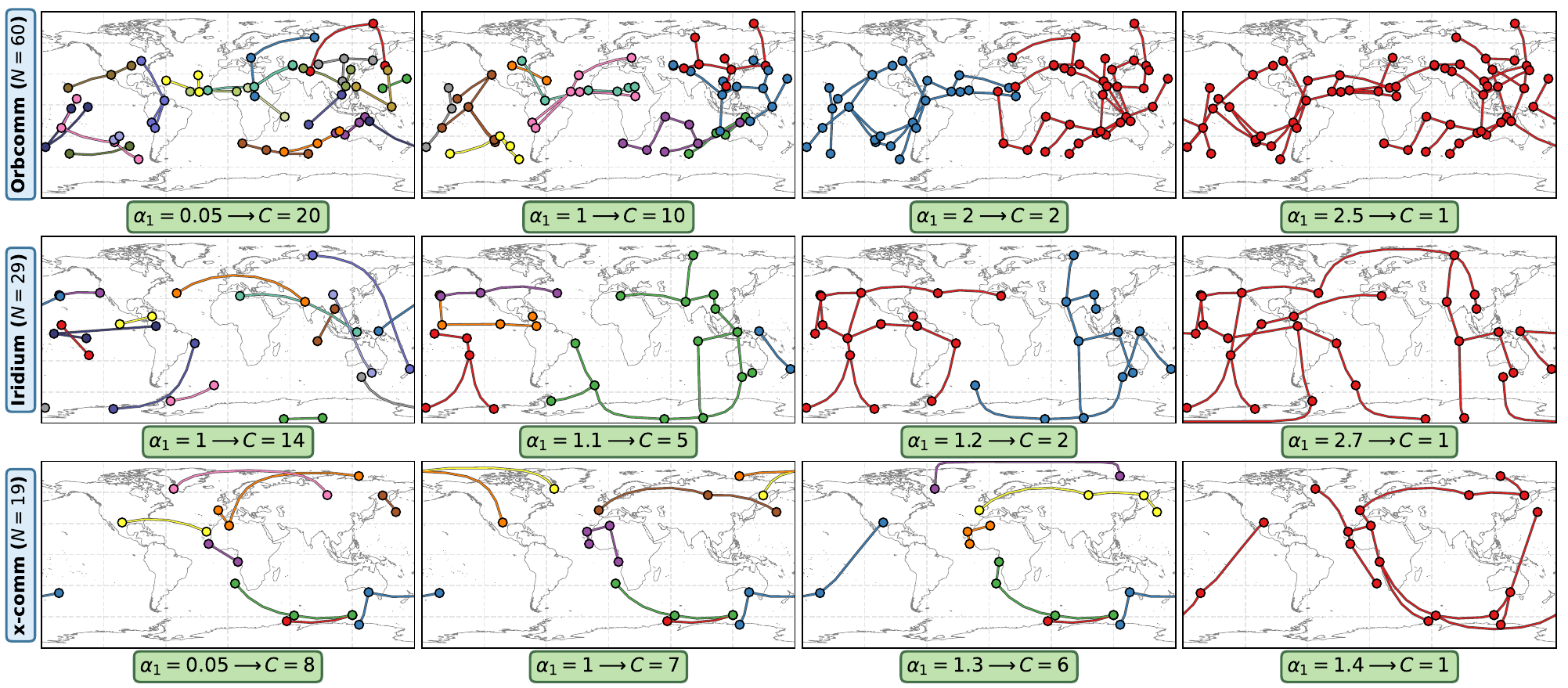}
\vspace{-1mm}
\caption{Visualization of the clustering of satellites and the number of clusters {\small$C=C^{(k)}$} for {\small$k=1$} over the Equirectangular projection of the Earth for various satellite constellations (i.e., Orbcomm, Iridium, X-comm) under different choices of the importance of ML model convergence on the optimization objective function (i.e. {\small$\alpha_1$} in {\small$\bm{\mathcal{P}}$}).\label{fig:WorldViewCondensed}}
\vspace{-3.0mm}
\end{figure*}

\textbf{Challenges Faced in Solving $\bm{\mathcal{P}}$:} The objective function and constraints of {\small$\bm{\mathcal{P}}$} are highly non-convex and complex to address due to several factors:
(i) terms involving the multiplications of the optimization variables in the objective function and constraints (e.g., in the convergence bound \eqref{eq:gen_conv} and energy consumption calculations such as \eqref{eq:EN_LC}),
(ii) terms with negative signs in the convergence bound \eqref{eq:gen_conv},
(iii) recursive functions for node delay calculations (e.g., \eqref{eq:GM_dispatching_latency}, \eqref{eq:CM_dispatching_latency}), (iv) non-convex recursive {\small$\max$} functions for node delay calculations (e.g., \eqref{eq:CM_aggregation_latency}, \eqref{eq:GM_aggregation_latency}), and (v) the binary/integer nature of some of the optimization variables (e.g., {\small $\pi_{n}^{\mathsf{G}, (k)}$}, {\small$\pi_{n}^{\mathsf{L}}$}, and {\small$\lambda^{(k,\ell)}_{x}$}). These factors collectively make solving the problem {\small$\bm{\mathcal{P}}$} highly non-trivial.

\textbf{Inherent Characteristics and Solution Design:}
We devise a systematic approach to solve the problem {\small$\bm{\mathcal{P}}$}. 
 Our method begins by relaxing the binary variables to continuous ones using a set of non-convex constraints, which ensure that these variables ultimately retain binary values in the final solution. Furthermore, we demonstrate that, through a sequence of algebraic manipulations and operations, {\small$\bm{\mathcal{P}}$} can be reformulated as a \textit{signomial programming} (SP) problem.
\textit{Given the generality of {\small$\bm{\mathcal{P}}$}, the steps in our methodology are adaptable and can be extended to a broader range of optimization problems, particularly within the domains of FedL and graph-theoretic network formations.}
To maintain clarity and readability, we omit the detailed mathematical derivations in this section. Interested readers are referred to Appendix~\ref{app:optTransform} for an in-depth explanation, and an overview of the steps taken in that appendix is discussed in the following. In Appendix~\ref{app:optTransform}, we first discuss SP problems, which are inherently non-convex and generally NP-hard, and highlight their similarities to {\small$\bm{\mathcal{P}}$}. We then outline the steps required to transform {\small$\bm{\mathcal{P}}$} into an SP problem and explore the unique characteristics of this transformed SP. Building on these insights, we propose a method to further reformulate the resulting SP into a \textit{geometric programming} (GP) problem~\cite{boyd2007tutorial,chiang2005geometric}. This is achieved using algebraic techniques applied to \textit{monomials} and \textit{posynomials}, such as leveraging the arithmetic-geometric mean inequality to bound \textit{posynomials} in terms of \textit{monomials}~\cite{duffin1972reversed}.
Finally, we use the fact that GP problems can be converted into convex optimization problems via a logarithmic change of variables, enabling efficient solutions using  CVXPY~\cite{diamond2016cvxpy}.
Additionally, we demonstrate that our GP solver (i.e. Algorithm~\ref{alg:cent} in Appendix~\ref{app:cons:sudo}) generates a series of solutions that converge to Karush–Kuhn–Tucker (KKT) conditions of {\small$\bm{\mathcal{P}}$} (see Proposition~\ref{propo:KKT} in Appendix~\ref{app:cons:sudo}). Complexity analysis and feasibility of solving our optimization problem {\small$\bm{\mathcal{P}}$} using CVXPY are also presented in Appendix~\ref{subsec:Complexity}.

\section{Simulation Results}\label{simulations}
 \noindent \textbf{Simulation Setup:} We next numerically evaluate the performance of {Fed-Span}. Our experiments utilize three datasets: two common datasets (CIFAR-10 and Fashion-MNIST) and one satellite-specific dataset, Functional Map of the World (FMoW), which is designed for land-use classification based on satellite images.\footnote{CIFAR-10: \url{https://www.cs.toronto.edu/~kriz/cifar.html}, ~~
Fashion-MNIST: \url{https://github.com/zalandoresearch/fashion-mnist}, ~~
FMoW: \url{https://github.com/fMoW/dataset}
} To emulate non-i.i.d. data across satellites, the datasets are partitioned using a Dirichlet distribution \cite{nguyen2022federated} with concentration parameter {\small$\alpha^{\mathsf{Dir}}\in\{0.1,10\}$} (default value: {\small$\alpha^{\mathsf{Dir}} = 0.1$}). We incorporate real-world orbital trajectory of satellites obtained from CelesTrak\footnote{CelesTrak: \url{https://celestrak.org/}} to enhance the realism of our simulations. Specifically, we model the movement of LEO satellites for three real-world constellations: X‑Comm ({\small$N=19$}), Iridium ({\small$N=29$}), and Orbcomm ({\small$N=60$}). 
As discussed in~Sec.~\ref{sec:transmission_data_Rate}, we also incorporate the Doppler effect in data rate calculations, following the approach in \cite{khalid2024characterization}.  The parameter values used in our simulations are summarized in Table~\ref{tab:sim_optical_params} in Appendix~\ref{parapp}, and an illustrative example of the {Fed-Span} network formation is provided in Appendix~\ref{App:Illust}.


 \textbf{Baselines:}
We evaluate the performance of {Fed-Span} in comparison to the following baseline methods: 

\begin{enumerate}[leftmargin=4mm]
    \item \textit{Asynchronous Fed-LS \cite{10021101}}: Global aggregations are performed whenever a ground gateway receives a local model from a satellite. This occurs when the satellite completes its local training and enters the observable range of a  gateway.    
    \item \textit{Buffered Asynchronous Fed-LS \cite{nguyen2022federated}}: Global aggregations are triggered when a predefined number of satellite models, reflecting the \textit{buffer size}, are received by the ground gateway. 
    \item \textit{Opportunistic Asynchronous Fed-LS}: We propose this method as a flexible variant of the above baseline, where an aggregation occurs with any number of satellite models available at the ground gateway, eliminating the need to wait for a fixed number of satellites. While this approach reduces waiting time, it may result in less stable training due to the varying number of contributing satellites per aggregation.    
    \item \textit{Sink-assisted Synchronous Fed-LS \cite{elmahallawy2024secure, elmahallawy2023optimizing}}: Satellites within the same orbit transmit their models sequentially to a designated sink satellite via ISLLs, following the shortest path. The sink satellite then relays all collected models to the ground gateway. Aggregation is performed only after all models have been gathered across the ground gateways.
\end{enumerate}

\begin{figure*}[t]
\vspace{-8mm}
    \centering
    \includegraphics[width=0.32\textwidth,trim=0.5cm 0cm 0cm 1.6cm,clip]{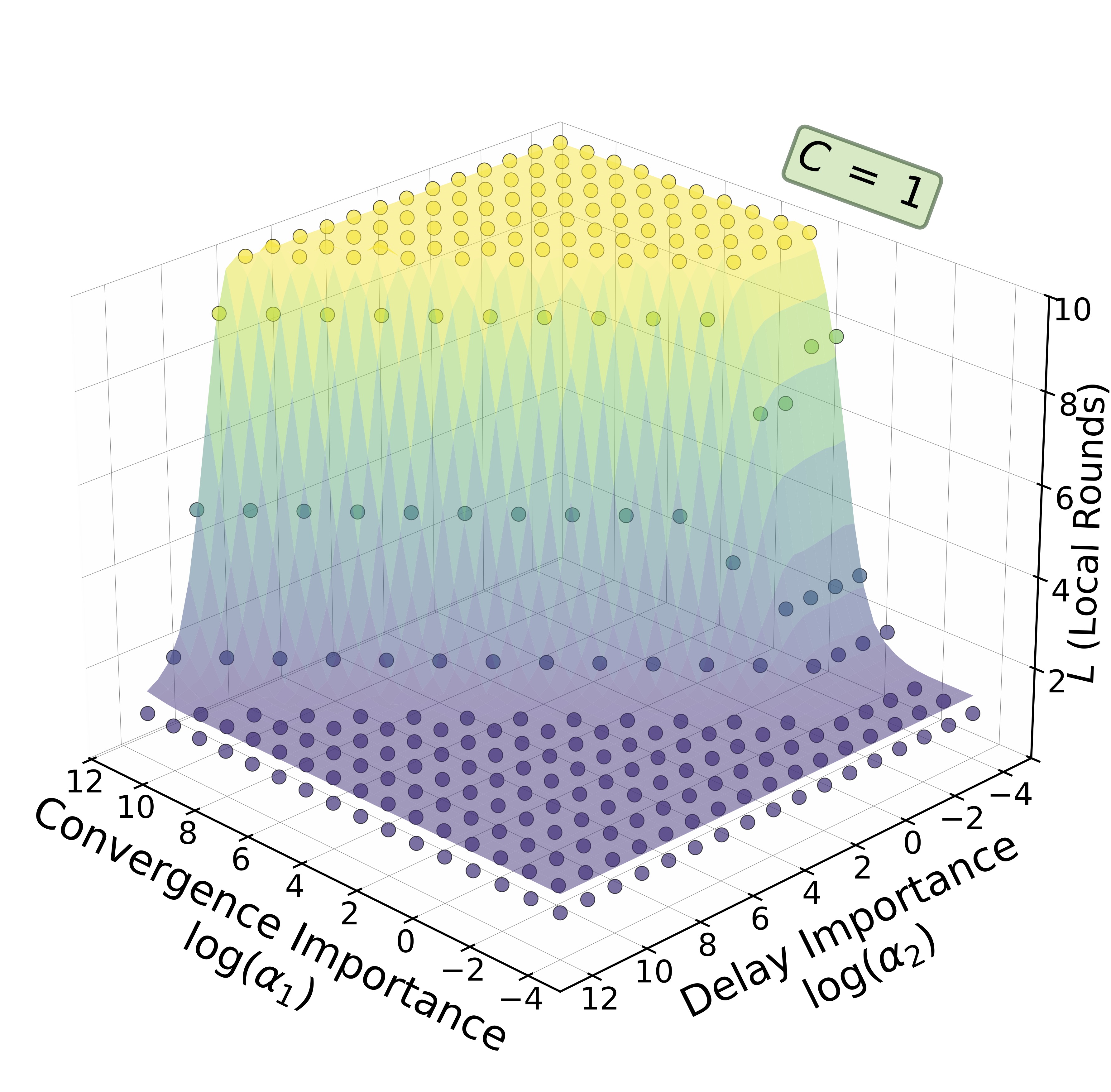}\hspace{0.01\textwidth}%
    \includegraphics[width=0.32\textwidth,trim=0.5cm 0cm 0cm 1.6cm,clip]{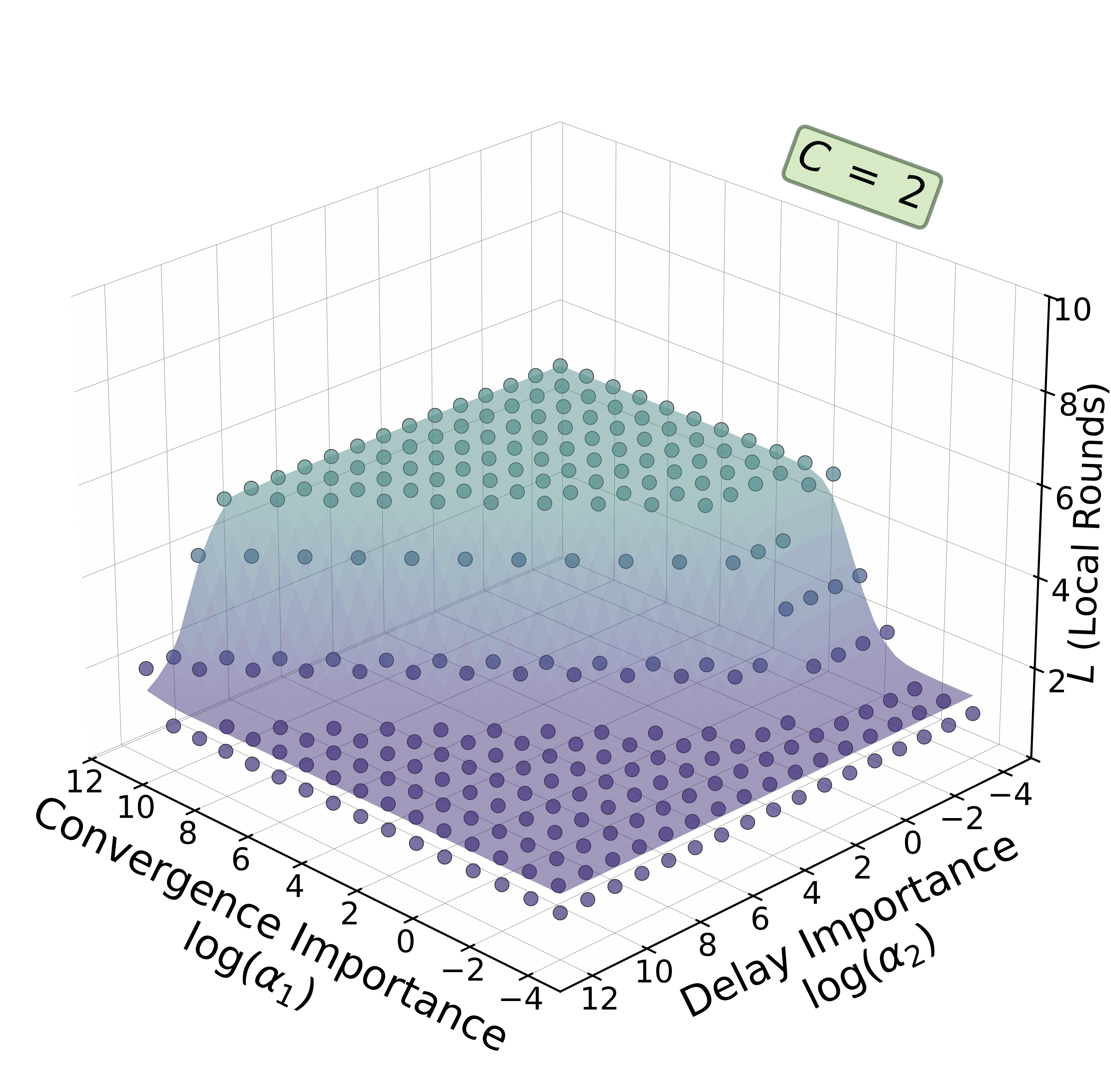}\hspace{0.01\textwidth}%
    \includegraphics[width=0.32\textwidth,trim=0.5cm 0cm 0cm 1.6cm,clip]{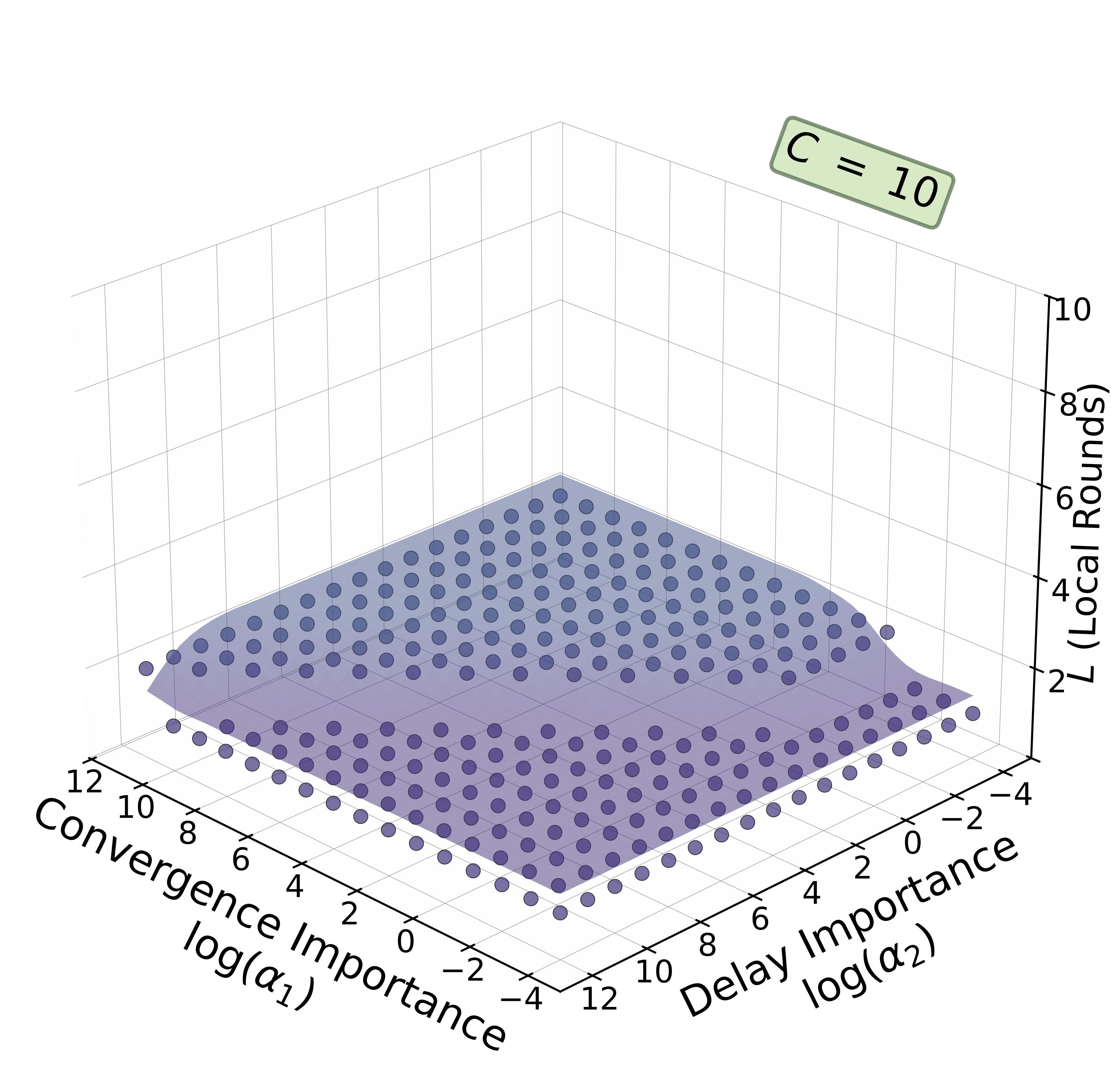}
    
    \caption{Impact of the importance of ML convergence and latency (i.e. {\small$\alpha_1$} and {\small$\alpha_2$} in {\small$\bm{\mathcal{P}}$}) on the number of local model training rounds {\small$L= L^{(k)}$, $k=1$}, under various number of satellite VCs/clusters {\small$C= C^{(k)}$}, {\small$k=1$}.  \label{fig:3DVisual}}
    \vspace{-4mm}
\end{figure*}

 For baselines that require satellite-to-ground communications, we consider 15 terrestrial gateway nodes uniformly distributed across the Earth's continents. The maximum effective communication range {\small$\mathfrak{R}^{\mathsf{min}}$} is set to 5000 km for ISLLs and 2300 km for satellite-to-ground RF links.

\textbf{Main Results:}  We present our main results in Fig.~\ref{fig:Latency_Energy_all_plots}, where we focus on the Orbcomm constellation (i.e., the largest constellation with {\small$N=60$})\footnote{The results for x‑comm ({\small$N=19$}) and Iridium ({\small$N=29$}) are qualitatively similar and thus omitted for brevity.} and present the global ML model performance vs. time/latency and energy consumption. The results demonstrate that our method achieves superior convergence speed (i.e., it attains higher model accuracy under lower latency and energy usage) compared to baseline methods across various levels of data heterogeneity (i.e.,  {\small$\alpha^{\mathsf{Dir}} \in \{0.1, 10\}$}). In particular, the results show the performance gap between our method and baselines widens as the datasets increase in complexity (i.e., CIFAR-10 and FMoW) and as the data distribution becomes more non-iid (i.e., {\small$\alpha^{\mathsf{Dir}} =0.1$}). This trend highlights that, for more challenging datasets, the choice of training strategy and resource allocation becomes increasingly critical to achieving high model accuracy. We note that the improvements obtained by our method are primarily attributed to: (i) the efficient use of local model aggregations, which help debias the local models of the satellites, and (ii) the use of high-speed ISLLs, which eliminate the reliance on satellite-to-ground communications and the need for satellite visibility from ground gateways during model aggregations.

\begin{figure}[b]
\centering
\includegraphics[width=0.49\textwidth, trim= 3 3 6.5 5, clip]{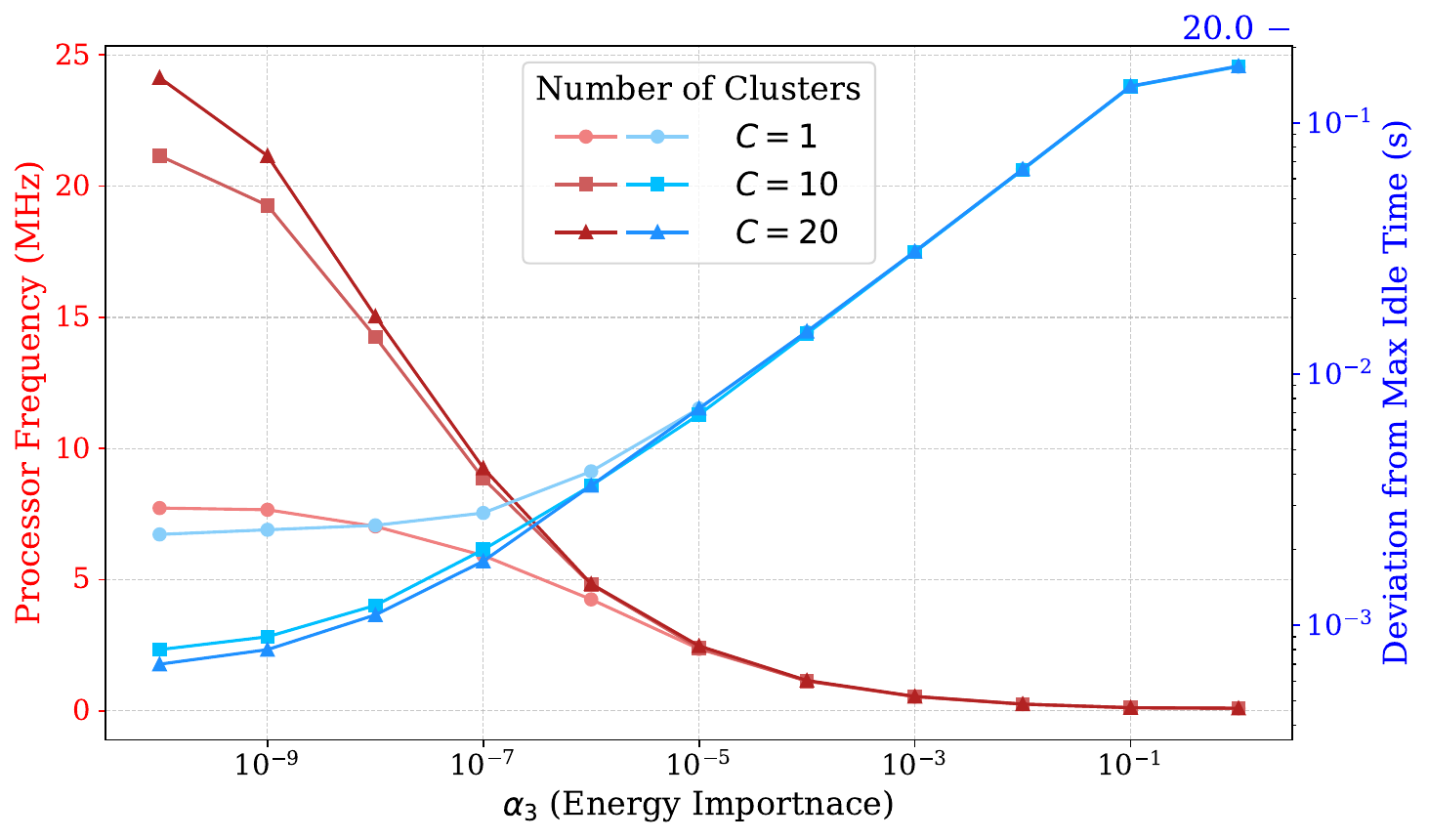}
\vspace{-5mm}
\caption{Resource allocation and satellite orchestration in terms of the CPU frequency speed (left axis) and idle times (right axis) vs various choices of the importance of the energy usage on the optimization objective (i.e., {\small$\alpha_3$} in {\small$\bm{\mathcal{P}}$}) and different numbers of satellite VCs/clusters (i.e., {\small$C=C^{(k)},~k=1$}). \label{fig:doubleAxis}}
\vspace{-0.1mm}
\end{figure}

\textbf{Satellite Clustering Topology:}
As this work is the first in the Fed-LS literature to obtain optimal clustering of satellites, it is insightful to examine the actual topologies of the formed clusters overlaid on a global map. Subsequently, in Fig.~\ref{fig:WorldViewCondensed}, we visualize the resulting satellite clusters/VCs for all three constellations.
In particular, the results reveal, for the first time in the Fed-LS literature, a direct connection between satellite cluster/VC topologies and the optimization objective's emphasis on ML model performance (i.e., the weight {\small$\alpha_1$} in {\small$\bm{\mathcal{P}}$}). As {\small$\alpha_1$} increases, {Fed-Span} prioritizes model convergence by reducing the number of VCs (as observed when moving from the left to the right subplots), resulting in deeper tree networks with higher aggregation delays and more energy-intensive communication paths (this leads to improved model convergence due to the higher collective data quality within each VC). Conversely, when {\small$\alpha_1$} decreases, the algorithm prioritizes resource efficiency, minimizing energy consumption and latency by forming more (smaller) clusters/VCs. 

\textbf{Resource Allocation and Satellite Orchestration:} We next investigate how  formulation~{\small$\bm{\mathcal{P}}$} adjusts resource allocation and satellite orchestration. Specifically, in Fig.~\ref{fig:doubleAxis}, we illustrate the impact of the energy usage weight (i.e., {\small$\alpha_3$} in {\small$\bm{\mathcal{P}}$}) on two key system behaviors: (i) CPU frequency, representing resource allocation decisions, and (ii) satellite idle times, representing orchestration decisions. The results are shown for varying numbers of clusters in the Orbcomm constellation ({\small$C\in \{1, 10, 20\}$} sampled from the first row of Fig.~\ref{fig:WorldViewCondensed}).
The results indicate that increasing {\small$\alpha_3$} leads to a reduction in satellite idle times (right y-axis and blue curves), reflected by a higher deviation from the maximum idle time (20 seconds). This suggests that satellites increasingly refrain from performing model training to conserve energy. Additionally, a higher {\small$\alpha_3$} results in a sharp decrease in CPU frequency (left y-axis and red curves), implying that satellites process data slower, thereby sacrificing latency in favor of reduced energy consumption.

\textbf{Local Model Training Configuration:} 
To supplement the above results, we further examine the joint impact of (i) clustering (with cluster/VC topologies for {\small$C \in \{1, 2, 10\}$} are sampled from the first row of Fig.~\ref{fig:WorldViewCondensed}) and (ii) the relative importance of ML convergence and latency (i.e., {\small$\alpha_1$} and {\small$\alpha_2$} in {\small$\bm{\mathcal{P}}$}) on the local model training, specifically, the number of local training rounds {\small$L = L^{(k)}$}. The results are visualized in Fig.~\ref{fig:3DVisual}.
The results indicate that when the weight on latency (i.e., {\small$\alpha_2$}) is reduced and the emphasis on convergence (i.e., {\small$\alpha_1$}) is increased, the number of local aggregations rises. This enables more frequent debiasing of satellites' models, which in turn accelerates convergence. This phenomenon is amplified as the number of clusters/VCs decreases (e.g., for $C=1$), since a higher number of local aggregations is only effective when each satellite cluster/VC contains a sufficiently diverse dataset. 

\vspace{-1mm}
\section{Conclusion and Future Work}
\noindent In this work, we presented a graph-theoretic formulations for Fed-LS with ISLLs. Our approach involved a systematic modeling framework that integrates graph theory, non-convex approximations, signomial programming, and geometric programming. Through simulations, we showed the superior performance of our method compared to Fed-LS baselines and uncovered a set of trade-offs that had previously received limited attention in the Fed-LS literature. Promising future works include: (i) extending our framework to quantum Fed-LS (incorporating either quantum computation or communication), and (ii) exploring multi-modal/multi-task Fed-LS with ISLLs.

\bibliographystyle{IEEEtran}
\bibliography{ref}

\vspace{-10mm}
\begin{IEEEbiography}[\vspace{-14mm}{\includegraphics[width=1in,height=0.9in,clip,keepaspectratio]{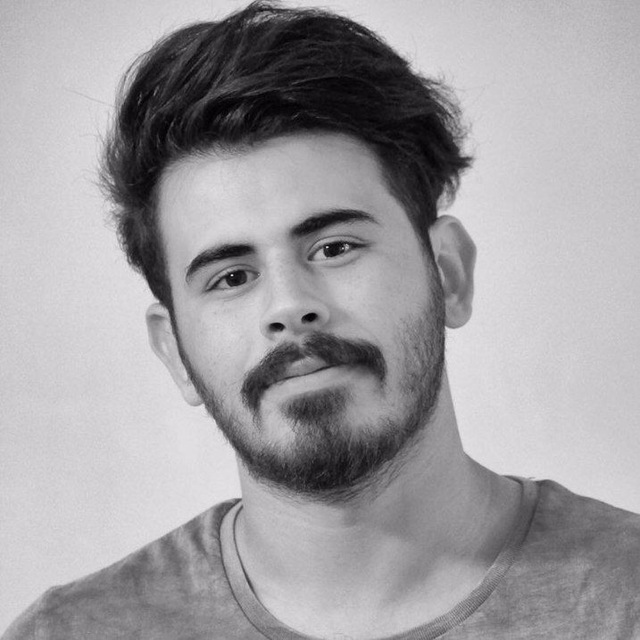}}]{Fardis Nadimi}
[S] is currently a PhD student in the Department of Electrical Engineering at the University at Buffalo–SUNY, NY, USA. 
He received his M.Sc. degree from the Iran University of Science and Technology, Tehran, Iran.
\end{IEEEbiography}
\vspace{-18.15mm}
\begin{IEEEbiography}[\vspace{-14mm}
{\includegraphics[width=1in,height=0.9in, clip,keepaspectratio]{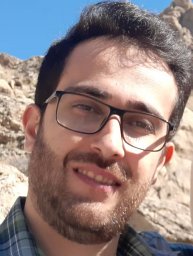}}]{Payam Abdisarabshali}[S] is currently a PhD student at the Electrical Engineering Department at the University at Buffalo--SUNY, NY, USA. He received his M.Sc. degree from Razi University, Kermanshah, Iran, with top-rank recognition.
\end{IEEEbiography}
\vspace{-18.15mm}
\begin{IEEEbiography}[
\vspace{-14mm}
{\includegraphics[width=1in,height=0.9in, clip,keepaspectratio]{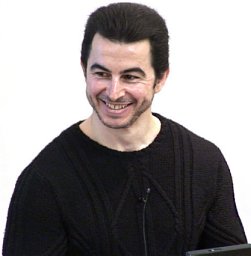}}]{Jacob Chakareski}[SM] is currently an Associate Professor with the Informatics Department at the New Jersey Institute of Technology, NJ, USA. He received his Ph.D. degree from Stanford, CA, USA, and Rice University, TX, USA.
\end{IEEEbiography}
\vspace{-18.15mm}
\begin{IEEEbiography}[\vspace{-14mm}
{\includegraphics[width=1in,height=0.9in, trim= 0 35 0 30, clip,keepaspectratio]{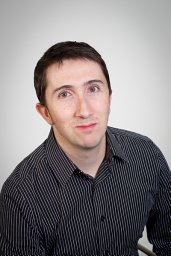}}]{Nicholas Mastronarde}[SM] is currently a Professor with the Electrical Engineering Department at the University at Buffalo--SUNY, NY, USA. He received his Ph.D. degree from the University of California--Los Angeles (UCLA), CA, USA.
\end{IEEEbiography}
\vspace{-18.15mm}
\begin{IEEEbiography}[\vspace{-14mm}{\includegraphics[width=1in,height=0.9in, trim= 300 0 300 50, clip,keepaspectratio]{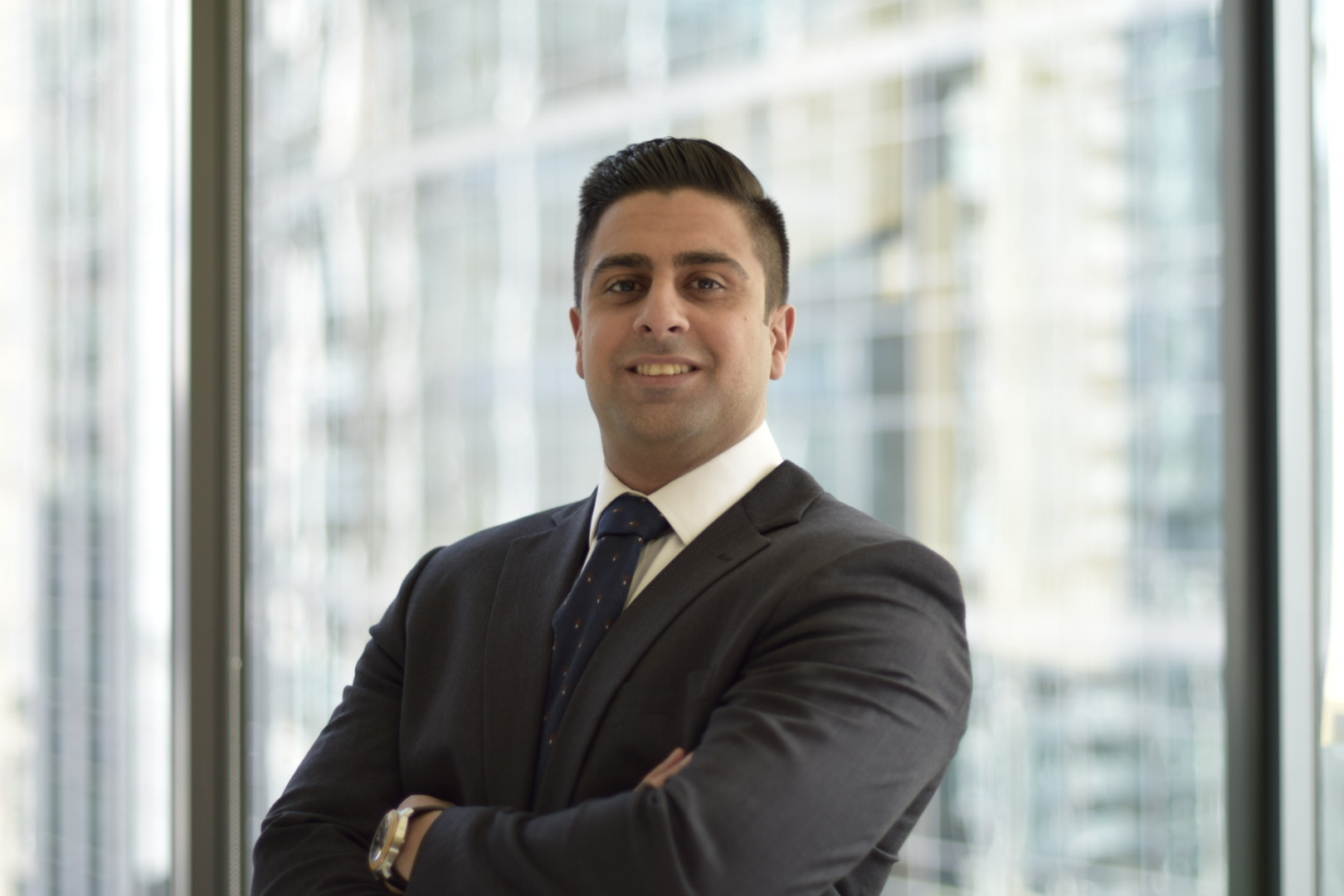}}]{Seyyedali Hosseinalipour}[SM] is currently an Assistant Professor with the Electrical Engineering Department at the University at Buffalo--SUNY, NY, USA. He received his Ph.D. degree from NC State University, NC, USA.
\end{IEEEbiography}

\vfill
\begingroup
\onecolumn

\appendices
\section{Notations and Acronyms Table}\label{app:notaions}

\setlength{\tabcolsep}{2pt}
{\footnotesize
\vspace{-3mm}
\begin{longtable}{|c|p{4.1cm}||c|p{4.1cm}||c|p{4.1cm}|}
\caption{List of notations used in the paper.}
\vspace{-2mm}
\label{table:notations_comprehensive_3col} \\
\hline
\rowcolor[HTML]{e3e2a5}\multicolumn{6}{|c|}{\textbf{System Model, Physical Layer, ML Setup, Tree Structures (Sec. \ref{sec:system_model} and Sec.\ref{sec:fundamentals})}} \\
\hline
\rowcolor[HTML]{e8cbc2}
\textbf{Notation} & \centering{\textbf{Description}} & \textbf{Notation} & \centering{\textbf{Description}} & \textbf{Notation} & {\hspace{12mm}\centering{\textbf{Description}}} \\
\hline
$\mathcal{N}$ & Set of satellites & $\mathcal{M}_n$ & Set of laser terminals on satellite $n$ & $T^{\mathsf{Init},(k)}$ & Wall-clock start of a global round \\ \hline
\rowcolor[HTML]{ededed}$\psi_{m,m'}$ & Binary variable for ISLLs & $\bm{\Psi}$ & Connectivity block matrix & $\Psi_{n,n'}$ & Connectivity matrix ($n, n'$) \\ \hline
$\overline{\mathfrak{R}}_{m,m'}$ & Data rate of link between $m$ and $m'$ & $P^{\mathsf{Tx}}_{m}$ & Transmit power of terminal $m$ & $P_{n}$ & Total tx power of satellite $n$ \\ \hline
\rowcolor[HTML]{ededed}$P^{\mathsf{Rx}}_{m,m'}$ & Received power at $m'$ & $B$ & Transmission bandwidth & $P^{\mathsf{N}}_{m'}$ & Noise power at receiver $m'$ \\ \hline
$a^{\mathsf{Tx}}/a^{\mathsf{Rx}}$ & Tx/Rx optics efficiency & $G^{\mathsf{Tx}} / G^{\mathsf{Rx}}$ & Tx/Rx antenna gains & $L^{\mathsf{Tx}} / L^{\mathsf{Rx}}$ & Tx/Rx Pointing Losses \\ \hline
\rowcolor[HTML]{ededed}$L^{\mathsf{PL}}_{m,m'}$ & Free-space optical path loss & $\lambda_{m, m'}$ & Signal wavelength at rx $m'$ & $\delta_{m,m'}$ & Distance between $m$ and $m'$ \\ \hline
$R^{\mathsf{Earth}}$ & Radius of the Earth & $h^{\mathsf{Cons}}$ & Altitude of constellation & $\phi_m^{\mathsf{Lat}}, \lambda_m^{\mathsf{Long}}$ & Latitude/Longitude Coordinates \\ \hline
\rowcolor[HTML]{ededed}$\Delta \phi_m^{\mathsf{Lat}}, \Delta \lambda_m^{\mathsf{Long}}$ & Difference in Latitude/Longitude & $\widehat{a}$ & Haversine of central angle & $\mathfrak{R}^{\mathsf{min}}$ & Min ISLL data rate threshold \\ \hline
$\mathfrak{R}_{n,n'}$ & Data rate of link between $n, n'$ & $\mathcal{K}$ & Set of global aggregation rounds & $\mathcal{D}_n^{(k)}$ & Dataset of satellite $n$ (size $D_n^{(k)}$) \\ \hline
\rowcolor[HTML]{ededed}$D^{(k)}$ & Total dataset size & $F^{(k)}_{n}$ & Local loss function of satellite $n$ & $F^{(k)}$ & Global loss function \\ \hline
$\bm{\omega}^{(k)}$ & Global ML model parameters & $\bm{\omega}^{{(k)}^\star}$ & Optimal global model parameters & $M^{\mathsf{Dim}}$ & Dimension of of ML model vector \\ \hline
\rowcolor[HTML]{ededed}$\mathcal{C}^{(k)}$ & Set of VCs (size $C^{(k)}$) & $\gamma^{(k)}_{c,n}$ & Binary satellite-to-VC association & $\mathcal{N}^{(k)}_c$ & Set of satellites in VC $c$ \\ \hline
$D_c^{(k)}$ & Dataset size of VC $c$ & $\pi_{n}^{\mathsf{L}, (k)}$ & Binary variable for local roots & $r^{\mathsf{G}, (k)}$ & Global dispatching root node \\ \hline
\rowcolor[HTML]{ededed}$\pi_{n}^{\mathsf{G}, (k)} / \pi_{n}^{\mathsf{A}, (k)}$ & Binary variables for global root & $r_{c}^{\mathsf{L}, (k)}$ & Root node for VC $c$ & & \\ \hline
\addlinespace[3pt]
\hline
\rowcolor[HTML]{e3e2a5}\multicolumn{6}{|c|}{\textbf{Phase 1: Global Model Dispatching (Sec. \ref{subsec:ph1})}} \\ \hline
\rowcolor[HTML]{e8cbc2}
\textbf{Notation} &  \centering{\textbf{Description}} & \textbf{Notation} &  \centering{\textbf{Description}} & \textbf{Notation} & {\hspace{12mm}\centering{\textbf{Description}}} \\
\hline
\rowcolor[HTML]{ededed}$\bm{\Gamma}^{\mathsf{GD},(k)}$ & Adjacency matrix global dispatching & $t^{\mathsf{GD},(k)}$ & Global dispatching start time & $\tau^{\mathsf{GD}, (k)}$ & Global dispatching latency\\ \hline
$\tau_{n,n'}^{\mathsf{GD},(k)}$ & Latency of global dispatching link & $\tau_{n'}^{\mathsf{GD},(k)}$ & Model arrival latency at $n'$ & $\tau^{\mathsf{GD},\mathsf{max}}$ & Max global dispatching time \\ \hline
\rowcolor[HTML]{ededed}$E^{\mathsf{GD}, (k)}$ & Global dispatching energy cons. & $\alpha^{\mathsf{Bit}}$ & Bits per model parameter & $\mathfrak{R}^{\mathsf{GD},(k)}_{n,n'}$ & Global dispatching link data rate \\ \hline
\addlinespace[3pt]
\hline
\rowcolor[HTML]{e3e2a5}\multicolumn{6}{|c|}{\textbf{Phase 2: Local Model Training (Sec. \ref{subsec:ph2})}} \\ \hline
\rowcolor[HTML]{e8cbc2}
\textbf{Notation} &  \centering{\textbf{Description}} & \textbf{Notation} &  \centering{\textbf{Description}} & \textbf{Notation} & {\hspace{12mm}\centering{\textbf{Description}}} \\
\hline
$\mathcal{L}^{(k)}$ & Set of local rounds & $L^{(k)}$ & Number of local rounds & $\tau^{\mathsf{Loc}}$ & Total time for local rounds \\ \hline
\rowcolor[HTML]{ededed} $\tau^{\mathsf{TTI}}$ & Duration of a single TTI & $t^{(k)}_x$ & Start time of TTI $x$ & $\mathcal{X}$ & Subset of TTIs $X = |\mathcal{X}|$ \\ \hline
$\lambda^{(k,\ell)}_{x}$ & Binary variable of local training time & $t^{\mathsf{LT}}$ & Initial time for local training & $\overline{\bm{\omega}}^{(k,\ell)}_{c}$ & Unified VC model \\ \hline
\rowcolor[HTML]{ededed} $\overline{\bm{\omega}}^{(k,0)}_{c} $& Initial VC model & $\eta_{k}$ & SGD step size & $e^{(k)}_{n}$ & SGD iterations at satellite $n$ \\ \hline
$e^{(k)}_{c}$ & Unified SGD iterations of VC $c$ & $\bm{\omega}_{n}^{(k,\ell),e}$ & Iteration $e$ on satellite $n$ local model & $\widetilde{\nabla F}^{(k)}_{n}$ & Cumulative gradient of satellite $n$ \\ \hline
\rowcolor[HTML]{ededed} $\mathcal{B}^{(k,\ell)}_{n}$ & Mini-batch samples & $\varsigma^{(k,\ell)}_{n}$ & Decision variable of mini-batch & $a_n$ & CPU cycles per data \\ \hline
$f^{(k,\ell)}_{n}$ & CPU frequency of satellite $n$ & $f^{\mathsf{max}}_{n}$ & Max available CPU frequency & $\alpha^{\mathsf{Chip}}_n$ & Chipset capacitance \\ \hline
\rowcolor[HTML]{ededed} $\tau^{\mathsf{LT},{(k,\ell)}}$ & Local training latency & $\tau^{\mathsf{LT}, \mathsf{max}}$ & Max local training time & $E^{\mathsf{LT},{(k,\ell)}}$ & Local training energy cons. \\ \hline
\addlinespace[3pt]
\hline
\rowcolor[HTML]{e3e2a5}\multicolumn{6}{|c|}{\textbf{Phase 3: Local Model Aggregation (Sec. \ref{subsec:ph3})}} \\ \hline
\rowcolor[HTML]{e8cbc2}
\textbf{Notation} &  \centering{\textbf{Description}} & \textbf{Notation} &  \centering{\textbf{Description}} & \textbf{Notation} & {\hspace{12mm}\centering{\textbf{Description}}} \\
\hline
\rowcolor[HTML]{ededed}$\Omega^{(k,\ell)}$ & Idle time duration & $t^{\mathsf{LA},(k,\ell)}$ & Local aggregation start time & $\bm{\Gamma}^{\mathsf{LA},{(k,\ell)}}$ & Adjacency matrix local aggregation \\ \hline
$\overline{\nabla F}_{n}$ & Aggregated gradient & $\tau^{\mathsf{LA},{(k,\ell)}}$ & Local aggregation latency & $\tau_{n,n'}^{\mathsf{LA},{(k,\ell)}}$ & Latency of local aggregation link \\ \hline
\rowcolor[HTML]{ededed}$\tau_{n'}^{\mathsf{LA},{(k,\ell)}}$ & Model arrival latency at $n'$ & $\tau^{\mathsf{LA},\mathsf{max}}$ & Max local aggregation time & $E^{\mathsf{LA},{(k,\ell)}}$ & Local aggregation energy cons. \\ \hline
$\Omega^{(k)}$ & Total idle time & $\mathfrak{R}^{\mathsf{LA},{(k,\ell)}}_{n,n'}$ & Local aggregation link data rate & & \\ \hline
\addlinespace[3pt]
\hline
\rowcolor[HTML]{e3e2a5}\multicolumn{6}{|c|}{\textbf{Phase 4: Local Model Dispatching (Sec. \ref{subsec:ph4})}} \\ \hline
\rowcolor[HTML]{e8cbc2}
\textbf{Notation} &  \centering{\textbf{Description}} & \textbf{Notation} &  \centering{\textbf{Description}} & \textbf{Notation} & {\hspace{12mm}\centering{\textbf{Description}}}\\
\hline
\rowcolor[HTML]{ededed}$\bm{\Gamma}^{\mathsf{LD},{(k,\ell)}}$ & Adjacency matrix local dispatching & $t^{\mathsf{LD},{(k,\ell)}}$ & Local dispatching start time & $\tau^{\mathsf{LD},{(k,\ell)}}$ & Local dispatching latency \\ \hline
$\tau_{n,n'}^{\mathsf{LD},{(k,\ell)}}$ & Latency of local dispatching link & $\tau_{n'}^{\mathsf{LD},{(k,\ell)}}$ & Model arrival latency at $n'$ & $\tau^{\mathsf{LD},\mathsf{max}}$ & Max local dispatching time \\ \hline
\rowcolor[HTML]{ededed}$E^{\mathsf{LD},{(k,\ell)}}$ & Local dispatching energy cons. & $\mathfrak{R}^{\mathsf{LD},{(k,\ell)}}_{n,n'}$ & Local dispatching link data rate & & \\ \hline
\addlinespace[3pt]
\hline
\rowcolor[HTML]{e3e2a5}\multicolumn{6}{|c|}{\textbf{Phase 5: Global Model Aggregation (Sec. \ref{subsec:ph5})}} \\ \hline
\rowcolor[HTML]{e8cbc2}
\textbf{Notation} & \centering{\textbf{Description}} & \textbf{Notation} &  \centering{\textbf{Description}} & \textbf{Notation} & {\hspace{12mm}\centering{\textbf{Description}}} \\
\hline
$\bm{\Gamma}^{\mathsf{GA},(k)}$ & Adjacency matrix global aggregation & $t^{\mathsf{GA},(k)}$ & Global aggregation start time & $r^{\mathsf{A}, (k)}$ & Global aggregation root node\\ \hline
\rowcolor[HTML]{ededed}$\widehat{\nabla F}^{(k)}_{n}$ & Cumulative global gradient & $\Xi^{(k)}$ & Boosting coefficient & $\tau^{\mathsf{GA},(k)}$ & Global aggregation latency\\ \hline
$\tau_{n,n'}^{\mathsf{GA},(k)}$ & Latency of global aggregation link & $\tau_{n'}^{\mathsf{GA},(k)}$ & Model arrival latency at $n'$ & $\tau^{\mathsf{GA},\mathsf{max}}$ & Max global aggregation time\\ \hline
\rowcolor[HTML]{ededed}$E^{\mathsf{GA}, (k)}$ & Global aggregation energy cons. & $\mathfrak{R}^{\mathsf{GA},(k)}_{n,n'}$ & Global aggregation link data rate & & \\ \hline
\addlinespace[3pt]
\hline
\rowcolor[HTML]{e3e2a5} \multicolumn{6}{|c|}{\textbf{Convergence Analysis \& Optimization (Sec. \ref{sec:conv} \& \ref{sec:optimization_problem})}} \\ \hline
\rowcolor[HTML]{e8cbc2}
\textbf{Notation} & \centering{\textbf{Description}} & \textbf{Notation} & \centering{\textbf{Description}} & \textbf{Notation} & {\hspace{12mm}\centering{\textbf{Description}}} \\ \hline
$\beta$ & Smoothness constant & $\zeta^{\mathsf{Loc}}_{c,1}/ \zeta^{\mathsf{Loc}}_{c,2}$ & Intra-VC heterogeneity constants & $\Delta^{(k)}$ & Total model drift \\ \hline
\rowcolor[HTML]{ededed} 
$\Delta^{(k,\ell)}_{n}$ & Model drift & $\Theta_n/ \Theta$ & Local data variability constants & $\zeta^{\mathsf{Loc},\min}_{c,1} /\hat{\zeta}^{\mathsf{Loc}}_{1}$ & Min/Max intra-VC heterogeneity \\ \hline
$\zeta^{\mathsf{Glob}}_{1}/ \zeta^{\mathsf{Glob}}_{2}$ & Inter-VC heterogeneity constants & $\alpha^{\mathsf{Dri}}$ & Dirichlet heterogeneity parameter & $\sigma^{(k)}_{n}$ & Variance feature distribution \\ \hline
\rowcolor[HTML]{ededed} 
$\sigma_{\mathsf{max}}$ & Max SGD sampling noise & $\Lambda^{(k)}$ & Convergence constant & $\Phi^{(k)}$ & Scaled step size factor \\ \hline
$e^{(k)}_{\mathsf{max}}$ & Max SGD iterations across VCs & $\ell_{\min} / \ell_{\max}$ & Min/Max local aggregation rounds & $\overline{e}_{\min} / \overline{e}_{\max}$ & Min/Max average SGD iterations \\ \hline
\rowcolor[HTML]{ededed} 
$\chi$ & Bound for idle period & $e^{(k)}_{\mathsf{avg}} / e^{(k)}_{\mathsf{sum}}$ & Average/Total SGD iterations & $\widehat{e}_{\min} / \widehat{e}_{\max}$ & Min/Max total SGD iterations \\ \hline
\end{longtable}
}

\begin{table*}[h]
\caption{List of acronyms used in the paper.}
\label{table:acronyms}
\centering
{\begin{tabular}{|c|l||c|l|}
\hline
\rowcolor[HTML]{e8cbc2} 
\textbf{Acronym} & {\hspace{16mm}\centering{\textbf{Description}}}  & \textbf{Acronym} & {\hspace{16mm}\centering{\textbf{Description}}}  \\ 
\hline
\hline
\rowcolor[HTML]{ededed} CBM & Connectivity block matrix & CCRs & Continuous constraint representations \\ \hline
CIFAR & Canadian Inst. for Advanced Research & CPU & Central processing unit \\ \hline
\rowcolor[HTML]{ededed} Fed-LS & FedL over LEO satellites & FedL & Federated learning \\ \hline
Fed-Span & FedL with spanning aggregation over LEO constellations & FMoW & Functional Map of the World \\ \hline
\rowcolor[HTML]{ededed} FSO & Free-space optics & GA & Global aggregation \\ \hline
GD & Global dispatching & GSs & Ground stations \\ \hline
\rowcolor[HTML]{ededed} IID & Independent and identically distributed & ISLLs & Inter-satellite laser links \\ \hline
LA & Local aggregation & LD & Local dispatchinging \\ \hline
\rowcolor[HTML]{ededed} LEO & Low Earth orbit & LT & Local training \\ \hline
ML & Machine learning & MNIST & Modified NIST database \\ \hline
\rowcolor[HTML]{ededed} MoDSF & Multi-objective directed spanning forest & MoDSTs & Multi-objective directed spanning trees \\ \hline
MSF & Minimum spanning forest & MST & Minimum spanning tree \\ \hline
\rowcolor[HTML]{ededed} RF & Radio-frequency & SGD & Stochastic gradient descent \\ \hline
 TTI & Transmission time intervals & VCs & Virtual constellations \\ \hline
\end{tabular}%
}
\end{table*}
\pagebreak

\section{Parameters for the {Fed-Span} Framework}\label{parapp}
\begin{table*}[h]
\vspace{1.5mm}
\caption{System parameters used in the {Fed-Span} simulations. Abbreviations: FAL $\rightarrow$ Full Angle Divergence of the optical beam; AD $\rightarrow$ Receiver aperture diameter. The use of $\pm$ denotes the range of the random parameter around its center/mean.}

\vspace{-1.5mm}
\centering
\label{tab:sim_optical_params} 
\begin{tabular}{|c c|c c|c c|c c|} 
\hline
\rowcolor[HTML]{e8cbc2} \textbf{Parameter} & \textbf{Value} & \textbf{Parameter} & \textbf{Value} & \textbf{Parameter} & \textbf{Value} & \textbf{Parameter} & \textbf{Value} \\
\hline
$L^{\mathsf{Tx}}_{m}$ & $e^{(-G^{\mathsf{Tx}}_{m} \cdot 10^{-12})}$ & $L^{\mathsf{Rx}}_{m'}$ & $e^{(-G^{\mathsf{Rx}}_{m'} \cdot 10^{-12})}$ & $a^{\mathsf{Rx}}_{m'} = a^{\mathsf{Tx}}_{m}$ & $0.8$ & $M$ & $4$ \\

$G^{\mathsf{Tx}}_{m}$ & $\tfrac{16}{(\text{FAL})^2}$ & $G^{\mathsf{Rx}}_{m'}$ & $\left(\tfrac{\pi \times AD}{\lambda}\right)^2$ & $\lambda_{m,m'}$ & $1550$ nm & $P^{\mathsf{Tx}}_{m}$ & $1$ W  \\

FAL & $15 \times 10^{-6}$ & AD & $80 \times 10^{-3}$ & Satellite Speed & $7800$ km/s & $h^{\mathsf{Cons}}$ & $500$ km \\

$R^{\mathsf{Earth}}$ & $6371$ km & $B$ & $2.2 \times 10^{9}$ Hz & $P_{\mathsf{RF}}$ & $10$ W & $\mathfrak{R}^{\mathsf{min}}$ & $7$ Gbps \\

$\mathfrak{B}_{n}^{\mathsf{cap}}$ & 500 \text{W} & $D_n^{(k)}$ & $1000 \pm 125$ & $\sigma^{(k)}_{n}$ & $3 \pm 2$ & $\alpha^{\mathsf{Bit}} M^{\mathsf{Dim}}$ & $6.5,\,18,\,35$ MB \\

$a_{n}$ & $1360$ & $\alpha^{\mathsf{chip}}_{n}$ & $2 \times 10^{-27}$ & $f^{\mathsf{max}}_{n}$ & $2.3 \times 10^9~\text{Hz}$ & $\beta$ & $10$ \\

$\Lambda$ & $0.9$ & $\Delta$ & $0.1$ & $\Theta$ & $3$ & $\tau^{\mathsf{GD},\mathsf{max}}$ & $1$ s \\
$\tau^{\mathsf{LT},\mathsf{max}}$ & $1$ s & $\tau^{\mathsf{LD},\mathsf{max}}$ & $1$ s & $\tau^{\mathsf{LA},\mathsf{max}}$ & $1$ s & $\tau^{\mathsf{GA},\mathsf{max}}$ & $1$ s \\
\hline
\end{tabular}
\vspace{-2mm}
\end{table*}

\newpage
\section{Illustrative Explanation of {Fed-Span} Framework}\label{App:Illust}
The idea behind {Fed-Span} is rather intuitive: instead of dealing with the constant connectivity complications with ground stations, the goal is to establish efficient satellite constellation networks that enable Fed-LS to operate seamlessly.
To address this, {Fed-Span} introduces tree-based network topologies that ensure connectivity and localized model aggregations while selecting links that lead to overall energy and delay efficiency as well as ML convergence improvements. 

To help with the comprehension, we illustrate the concept behind {Fed-Span} in Fig. \ref{threefig} using the Orbcomm constellation with $N=60$ satellites. To determine a viable network based on the satellites' locations (the left plot in Fig. \ref{threefig}), it is necessary to identify which ISSLs are usable. For this purpose, we impose a minimum ISSL data rate threshold of $\mathfrak{R}^{\mathsf{min}} = 7\text{Gbps}$. Furthermore, conforming to ISSL terminal constraints, such as the limited number of terminals per satellite, shapes the feasible connectivity across the satellites.
Incorporating these constraints yields the connectivity block matrix (CBM) $\bm{\Psi}(N,t)$, specifying the subset of ISSL terminals available for network formation at time $t$ (the middle plot in Fig. \ref{threefig}). The {Fed-Span} framework then leverages this CBM and solves problem {\small$\bm{\mathcal{P}}$} to construct tree-shaped networks for local and global aggregation rounds, represented as an adjacency matrix $\bm{\Gamma}(\bm{\Psi}(N,t))$ (the right plot in Fig. \ref{threefig}). Specifically, the right plot of Fig.~\ref{threefig} illustrates an adjacency matrix for a two-cluster/VC case ($C^{(k)}=2$), where nodes and links of each cluster are distinguished by different colors, and the root nodes of each cluster are highlighted in green.


\begin{figure}[H]
\centering
\includegraphics[width=1\textwidth, trim= 22 7 22 22, clip]{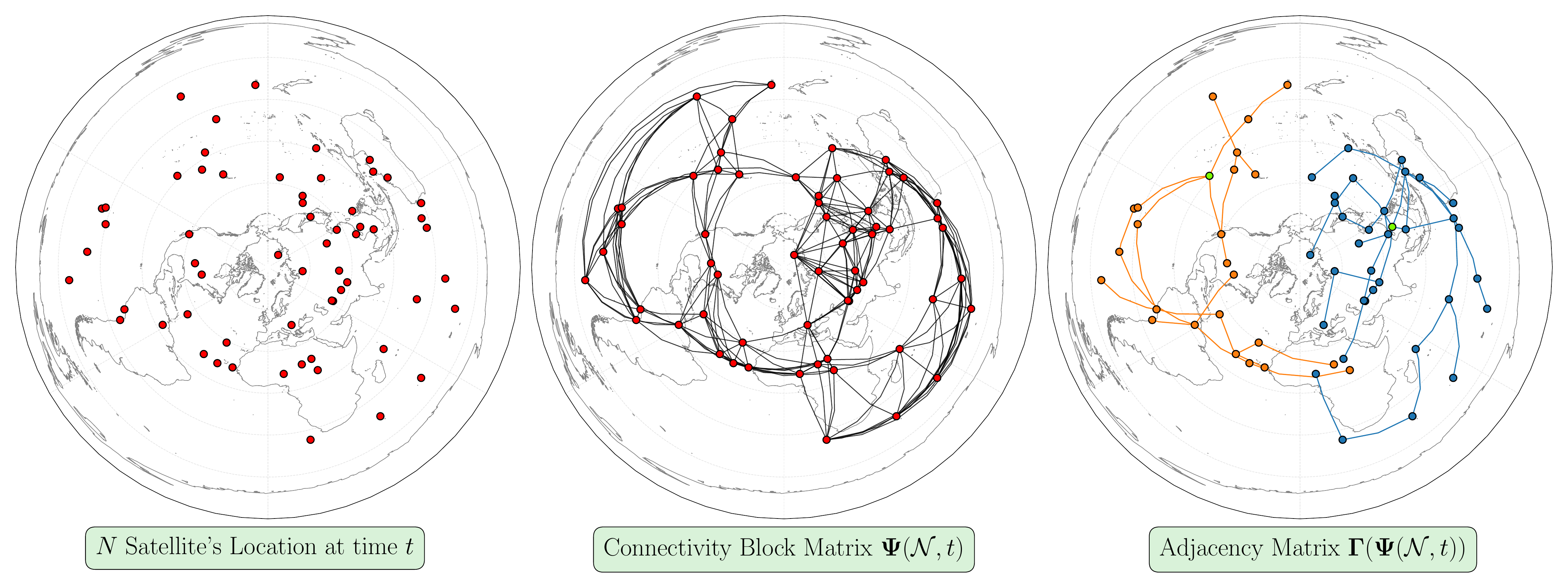}
\vspace{-1mm}
\caption{Network formation in {Fed-Span}: starting with acquiring satellite locations at any given time (the left subplot), then constructing the CBM of available links (the middle subplot), and finally determining the network topology (i.e., the established links) via solving {\small$\bm{\mathcal{P}}$}, balancing a trade-off among ML performance, training/transmission energy consumption, and training/transmission delay (the right subplot).}\label{threefig}

\vspace{-0.1mm}
\end{figure}

\newpage

\section{Proof of Theorem~\ref{th:main}}\label{app:th:main}

\begin{lemma}[\textbf{Mini-batch SGD noise characterization}]\label{Lemma:SGD}
        The variance of stochastic gradient during mini-batch gradient descent iterations can be bounded based on the result of \cite{lohr2021sampling} (Chapter 3, Eq. (3.5)) as follows:
         \begin{equation}
            \mathbb{E}_k\left[\left\Vert \sum_{d\in \mathcal{B}^{(k,\ell),e}_{n}} \frac{\nabla f(\bm{\omega}^{(k,\ell),e-1}_{n},d)}{B^{(k, \ell)}_{n}} - \nabla F_{n}^{(k)}(\bm{\omega}_{n}^{(k,\ell),e-1})\right\Vert^2\right] = \left(1- \varsigma^{(k,\ell)}_n\right)\frac{\left(\sigma^{(k)}_n\right)^2}{\varsigma^{(k,\ell)}_n D^{(k)}_n} 2 \Theta^2, \forall e,
         \end{equation}
         where $(\sigma^{(k)}_n)^2$ denotes the sampled data variance at the respective satellite.
\end{lemma}
\subsection{Unfolding of {Fed-Span} Aggregations}
Recall that, based on \eqref{finalmodelaggregation}, the global aggregation of {Fed-Span}  follows the following rule at the global root node:
\begin{equation}\label{eq:proofAgg1}
    \bm{\omega}^{(k+1)}{=}\bm{\omega}^{(k)}{-}\eta_{_k}\Xi^{(k)}\overline{\nabla F}_{n}^{(k,{\ell})}\big/D^{(k)},~\ell = L^{(k)},~n{=}r^{\mathsf{A}, (k)},
\end{equation}
where
\begin{equation}\label{eq:proofAgg2}
     \overline{\nabla F}_{n}^{(k,{\ell})}{=} \frac{{D}^{(k)}_n}{e^{(k)}_{n} L^{(k)} }\widehat{\nabla F}_{n}^{(k,{\ell})}\hspace{-1.5mm}+\hspace{-1mm}\sum_{n'\in \mathcal{N}\setminus\{n\}}\hspace{-1mm} \bm{\Gamma}^{\mathsf{GA}, (k)}_{n,n'} \overline{\nabla F}_{n'}^{(k,{\ell})}\hspace{-0.5mm},\hspace{-0.79mm}~~~~~~\ell {=} L^{(k)}\hspace{-0.75mm},~n{=}r^{\mathsf{A}, (k)}.
\end{equation}
Inspecting~\eqref{eq:proofAgg1} and focusing on the term $\Xi^{(k)}= \sum_{n{\in}\mathcal{N}}\frac{{D}^{(k)}_n L^{(k)}e^{(k)}_{n}}{{D}^{(k)}}$, we have:
\begin{equation}
\Xi^{(k)}=\sum_{c\in \mathcal{C}^{(k)}}\sum_{n\in \mathcal{N}_c^{(k)}}\frac{{D}^{(k)}_n L^{(k)}e^{(k)}_{c}}{{D}^{(k)}}=\sum_{c\in\mathcal{C}^{(k)}}e^{(k)}_{c}\sum_{n\in \mathcal{N}_c^{(k)}}\frac{{D}^{(k)}_n L^{(k)}}{{D}^{(k)}}=\sum_{c\in\mathcal{C}^{(k)}}\frac{{D}^{(k)}_c L^{(k)} e^{(k)}_{c}}{{D}^{(k)}},
\end{equation}
where $e^{(k)}_c$ is the number of SGD iterations uniformly performed at the satellites in VC $c$. Thus,~\eqref{eq:proofAgg1} can be written as follows:
\begin{equation}\label{eq:proofAgg3}
    \bm{\omega}^{(k+1)}{=}\bm{\omega}^{(k)}{-}\eta_{_k} \sum_{c\in\mathcal{C}^{(k)}}\frac{{D}^{(k)}_c L^{(k)} e^{(k)}_{c}}{{D}^{(k)}}\overline{\nabla F}_{n}^{(k,{\ell})}\big/D^{(k)},~\ell = L^{(k)},~n{=}r^{\mathsf{A}, (k)}.
\end{equation}
Next, we focus on the term $\overline{\nabla F}_{n}^{(k,{\ell})}$ in~\eqref{eq:proofAgg3}, using~\eqref{eq:proofAgg2}, and the fact that the tree formed via $\bm{\Gamma}^{\mathsf{GA}, (k)}$ covers all the satellites and satellites have a unique cycle-free path to the root node (i.e., the model of each satellite is only forwarded through one path to the root node), we have:
\begin{align}
     &\overline{\nabla F}_{n}^{(k,{\ell})}{=} \frac{{D}^{(k)}_n}{e^{(k)}_{n} L^{(k)} }\widehat{\nabla F}_{n}^{(k,{\ell})}\hspace{-1.5mm}+\hspace{-1mm}\sum_{n'\in \mathcal{N}\setminus\{n\}}\hspace{-1mm} \bm{\Gamma}^{\mathsf{GA}, (k)}_{n,n'} \overline{\nabla F}_{n'}^{(k,{\ell})}\hspace{-0.5mm},\hspace{-0.79mm}~\ell {=} L^{(k)},~n{=}r^{\mathsf{A}, (k)}
     \\
    & \Rightarrow  
 \overline{\nabla F}_{n}^{(k,{\ell})}{=} \frac{{D}^{(k)}_n}{ e^{(k)}_{n} L^{(k)} }\widehat{\nabla F}_{n}^{(k,{\ell})}\hspace{-1.5mm}+\hspace{-1mm}\sum_{n'\in \mathcal{N}\setminus\{n\}}\hspace{-1mm} \frac{{D}^{(k)}_{n'}}{ e^{(k)}_{n'} L^{(k)} }\widehat{\nabla F}_{n'}^{(k,{\ell})}\hspace{-0.5mm},\hspace{-0.79mm}~\ell {=} L^{(k)}, ~n{=}r^{\mathsf{A}, (k)}\\
 &\Rightarrow \overline{\nabla F}_{n}^{(k,{\ell})}{=}\sum_{n'\in \mathcal{N}} \frac{{D}^{(k)}_{n'}}{e^{(k)}_{n'} L^{(k)} }\widehat{\nabla F}_{n'}^{(k,{\ell})}\hspace{-0.5mm},\hspace{-0.79mm}~\ell {=} L^{(k)}, ~n{=}r^{\mathsf{A}, (k)}.
\end{align}
Replacing the above result in~\eqref{eq:proofAgg3}, we have:
\begin{equation*}\label{eq:proofAgg51}
    \bm{\omega}^{(k+1)}{=}\bm{\omega}^{(k)}{-}\eta_{_k} \sum_{c'\in\mathcal{C}^{(k)}}\frac{{D}^{(k)}_{c'} L^{(k)} e^{(k)}_{c'}}{{D}^{(k)}}\sum_{n\in \mathcal{N}}\hspace{-1mm} \frac{{D}^{(k)}_{n}}{{D}^{(k)} e^{(k)}_{n} L^{(k)} } \widehat{\nabla F}_{n}^{(k,{\ell})},~~\ell = L^{(k)}
\end{equation*}
\begin{equation*}\label{eq:proofAgg52}
   \Rightarrow \bm{\omega}^{(k+1)}{=}\bm{\omega}^{(k)}{-}\eta_{_k} \sum_{c'\in\mathcal{C}^{(k)}}\frac{{D}^{(k)}_{c'} L^{(k)} e^{(k)}_{c'}}{{D}^{(k)}}\sum_{n\in \mathcal{N}}\hspace{-1mm} \frac{{D}^{(k)}_{n}}{{D}^{(k)} e^{(k)}_{n} L^{(k)} } \frac{{D}^{(k)}_{c}}{{D}^{(k)}_{c}}\widehat{\nabla F}_{n}^{(k,{\ell})} ,~~\ell = L^{(k)}
\end{equation*}
\begin{equation*}\label{eq:proofAgg53}
   \Rightarrow \bm{\omega}^{(k+1)}{=}\bm{\omega}^{(k)}{-}\eta_{_k} \sum_{c'\in\mathcal{C}^{(k)}}\frac{{D}^{(k)}_{c'} L^{(k)} e^{(k)}_{c'}}{{D}^{(k)}}\sum_{n\in \mathcal{N}}\hspace{-1mm} \frac{{D}^{(k)}_{c}}{{D}^{(k)} e^{(k)}_{n} L^{(k)} } \frac{{D}^{(k)}_{n}}{{D}^{(k)}_{c}}\widehat{\nabla F}_{n}^{(k,{\ell})} ,~~\ell = L^{(k)}
\end{equation*}
\begin{equation*}\label{eq:proofAgg54}
  \Rightarrow  \bm{\omega}^{(k+1)}{=}\bm{\omega}^{(k)}{-}\eta_{_k} \sum_{c'\in\mathcal{C}^{(k)}}\frac{{D}^{(k)}_{c'} L^{(k)} e^{(k)}_{c'}}{{D}^{(k)}}\sum_{c\in \mathcal{C}^{(k)}}\sum_{n\in \mathcal{N}^{(k)}_c}\hspace{-1mm} \frac{{D}^{(k)}_{c}}{{D}^{(k)} e^{(k)}_{c} L^{(k)} } \frac{{D}^{(k)}_{n}}{{D}^{(k)}_{c}}\widehat{\nabla F}_{n}^{(k,{\ell})} ,~~\ell = L^{(k)}
\end{equation*}
\begin{equation*}\label{eq:proofAgg555}
  \Rightarrow  \bm{\omega}^{(k+1)}{=}\bm{\omega}^{(k)}{-}\eta_{_k} \sum_{c'\in\mathcal{C}^{(k)}}\frac{{D}^{(k)}_{c'} L^{(k)} e^{(k)}_{c'}}{{D}^{(k)}}\sum_{c\in \mathcal{C}^{(k)}} \frac{{D}^{(k)}_{c}}{{D}^{(k)} e^{(k)}_{c} L^{(k)} }\sum_{n\in \mathcal{N}^{(k)}_c} \frac{{D}^{(k)}_{n}}{{D}^{(k)}_{c}}\widehat{\nabla F}_{n}^{(k,{\ell})} ,~~\ell = L^{(k)}
\end{equation*}
\begin{equation*}\label{eq:proofAgg55}
  \xRightarrow{(a)}  \bm{\omega}^{(k+1)}{=}\bm{\omega}^{(k)}{-}\eta_{_k} \sum_{c'\in\mathcal{C}^{(k)}}\frac{{D}^{(k)}_{c'} L^{(k)} e^{(k)}_{c'}}{{D}^{(k)}}\sum_{c\in \mathcal{C}^{(k)}} \frac{{D}^{(k)}_{c}}{{D}^{(k)} e^{(k)}_{c} L^{(k)} }\sum_{n\in \mathcal{N}^{(k)}_c} \frac{{D}^{(k)}_{n}}{{D}^{(k)}_{c}}\left(\left[\sum_{\ell=1}^{L^{(k)}-1} \sum_{n' \in \mathcal{N}^{(k)}_c}\frac{D_{n'}^{(k)}}{D_c^{(k)}}\widetilde{\nabla F}_{n'}^{(k,{\ell})} \right] + \widetilde{\nabla F}_{n}^{(k,{L^{(k)}})}  \right)
\end{equation*}
\begin{equation*}\label{eq:proofAgg566}
  \Rightarrow  \bm{\omega}^{(k+1)}{=}\bm{\omega}^{(k)}{-}\eta_{_k} \sum_{c'\in\mathcal{C}^{(k)}}\frac{{D}^{(k)}_{c'} L^{(k)} e^{(k)}_{c'}}{{D}^{(k)}}\sum_{c\in \mathcal{C}^{(k)}} \frac{{D}^{(k)}_{c}}{{D}^{(k)} e^{(k)}_{c} L^{(k)} }\left(\left[\sum_{\ell=1}^{L^{(k)}-1} \sum_{n' \in \mathcal{N}^{(k)}_c}\frac{D_{n'}^{(k)}}{D_c^{(k)}}\widetilde{\nabla F}_{n'}^{(k,{\ell})} \right]+\sum_{n\in \mathcal{N}^{(k)}_c} \frac{{D}^{(k)}_{n}}{{D}^{(k)}_{c}} \widetilde{\nabla F}_{n}^{(k,{L^{(k)}})}  \right)
\end{equation*}
\begin{equation}\label{eq:proofAgg56}
 \Rightarrow   \bm{\omega}^{(k+1)}{=}\bm{\omega}^{(k)}{-}\eta_{_k} \sum_{c'\in\mathcal{C}^{(k)}}\frac{{D}^{(k)}_{c'} L^{(k)} e^{(k)}_{c'}}{{D}^{(k)}}\sum_{c\in \mathcal{C}^{(k)}} \frac{{D}^{(k)}_{c}}{{D}^{(k)} e^{(k)}_{c} L^{(k)} }\sum_{\ell=1}^{L^{(k)}} \sum_{n \in \mathcal{N}^{(k)}_c}\frac{D_{n}^{(k)}}{D_c^{(k)}}\widetilde{\nabla F}_{n}^{(k,{\ell})},
\end{equation}
where in $\xRightarrow{(a)}$ in the above derivations, we have used the fact that:
\begin{align}
    & \widehat{\nabla F}_{n}^{(k,{\ell})},~~\ell = L^{(k)}~~~ \equiv~~~  \widehat{\nabla F}_{n}^{(k,L^{(k)})}
   = ({\bm{\omega}}^{(k)}{-}\bm{\omega}_{n}^{(k,L^{(k)}),e^{(k)}_{n}})\big/\eta_{_k}
\\ & = ({\bm{\omega}}^{(k)}{-} \overline{\bm{\omega}}^{(k,1)}_c + \overline{\bm{\omega}}^{(k,1)}_c - \overline{\bm{\omega}}^{(k,2)}_c+ \overline{\bm{\omega}}^{(k,2)}_c- \cdots -\overline{\bm{\omega}}^{(k,L^{(k)}-1)}_c + \overline{\bm{\omega}}^{(k,L^{(k)}-1)}_c -\bm{\omega}_{n}^{(k,L^{(k)}-1),e^{(k)}_{n}})\big/\eta_{_k}
\\ & \overset{(a)}{=} \left[\sum_{\ell=1}^{L^{(k)}-1} \frac{\overline{\nabla F}_{n'}^{(k,{\ell})}}{D_c^{(k)}} \right] + \widetilde{\nabla F}_{n}^{(k,{L^{(k)}})},~~ n'=r_{c}^{\mathsf{L}, (k)}
\\ & \overset{(b)}{=}  \left[\sum_{\ell=1}^{L^{(k)}-1} \sum_{n' \in \mathcal{N}^{(k)}_c}\frac{D_{n'}^{(k)}}{D_c^{(k)}}\widetilde{\nabla F}_{n'}^{(k,{\ell})} \right] + \widetilde{\nabla F}_{n}^{(k,{L^{(k)}})},
\end{align}
where in $\overset{(a)}{=}$ in the derivations above, we have used the fact that based on~\eqref{localagggregationformula}, we have (note that $\overline{\bm{\omega}}^{(k,0)}_{c}=\overline{\bm{\omega}}^{(k)}$):
\begin{align}
& {\overline{\nabla F}_{n}^{(k,\ell)}}\big/{D^{(k)}_c} = (\overline{\bm{\omega}}^{(k,\ell)}_{c} - \overline{\bm{\omega}}^{(k,\ell+1)}_{c})\big/\eta_{_k} ,~n=r_{c}^{\mathsf{L}, (k)}.
\end{align}
Also, $\overset{(b)}{=}$ in the derivations above is  obtained based on~\eqref{clusteraggregation} and the fact that the tree formed via $\bm{\Gamma}^{\mathsf{LA}, (k,\ell)}$ covers all the satellites in each VC and satellites have a unique cycle-free path to the root node (i.e., the model of each satellite is only forwarded through one path to the root node), where we have: 
\begin{align}
     &\overline{\nabla F}_{n}^{(k,\ell)} = {{D}^{(k)}_n}\widetilde{\nabla F}_{n}^{(k,\ell)}+ \sum_{n'\in \mathcal{N}\setminus\{n\}}\hspace{-3mm} \bm{\Gamma}^{\mathsf{LA}, (k,\ell)}_{n,n'} \overline{\nabla F}_{n'}^{(k,\ell)}, \forall \ell \in L^{(k)}-1,~n{=}r^{\mathsf{L}, (k)}
     \\& \Rightarrow \overline{\nabla F}_{n}^{(k,\ell)} = {{D}^{(k)}_n}\widetilde{\nabla F}_{n}^{(k,\ell)} + \sum_{n' \in \mathcal{N}^{(k)}_c} {{D}^{(k)}_{n'}}\widetilde{\nabla F}_{n'}^{(k,\ell)}, \forall \ell \in L^{(k)}-1,~n{=}r^{\mathsf{L}, (k)}
    \\&\Rightarrow \overline{\nabla F}_{n}^{(k,\ell)} = \sum_{n' \in \mathcal{N}^{(k)}_c} {{D}^{(k)}_{n'}}\widetilde{\nabla F}_{n'}^{(k,\ell)}, \forall \ell \in L^{(k)}-1,~n{=}r^{\mathsf{L}, (k)}.
\end{align}

Inspecting the result in~\eqref{eq:proofAgg56}, we have the following fundamental relationship between the two global models obtained in {Fed-Span}, which we will use in our subsequent proofs:
\begin{equation}\label{eq:proofAgg60}
 \bm{\omega}^{(k+1)}-\bm{\omega}^{(k)}{=}{-}\eta_{_k} \sum_{c'\in\mathcal{C}^{(k)}}\frac{{D}^{(k)}_{c'} L^{(k)} e^{(k)}_{c'}}{{D}^{(k)}}\sum_{c\in \mathcal{C}^{(k)}} \frac{{D}^{(k)}_{c}}{{D}^{(k)} e^{(k)}_{c} L^{(k)} }\sum_{\ell=1}^{L^{(k)}} \sum_{n \in \mathcal{N}^{(k)}_c}\frac{D_{n}^{(k)}}{D_c^{(k)}}\widetilde{\nabla F}_{n}^{(k,{\ell})}.
\end{equation}

\subsection{Main Convergence Analysis}
From the $\beta$-smoothness of the global loss function (Assumption~\ref{Assup:lossFun}), we have:
\begin{equation}
F^{(k)}(\bm{\omega}^{(k+1)}) \leq F^{(k)}(\bm{\omega}^{(k)}) + \nabla F^{(k)}(\bm{\omega}^{(k)})^\top \left(\bm{\omega}^{(k+1)} - \bm{\omega}^{(k)}\right) + \frac{\beta}{2}\left\Vert \bm{\omega}^{(k+1)} - \bm{\omega}^{(k)} \right\Vert^2.
\end{equation}
Replacing the updating rule for $\bm{\omega}^{(k+1)}$ and taking the conditional expectation from both hand sides yields to:
\begin{equation}\label{betaupdaterule}
\begin{aligned}
 &\mathbb{E}_k \left[F^{(k)}(\bm{\omega}^{(k+1)})\right] \leq F^{(k)}(\bm{\omega}^{(k)}) - \nabla F^{(k)}(\bm{\omega}^{(k)})^\top \mathbb{E}_k\Bigg[ \eta_{_k} \sum_{c'\in \mathcal{C}^{(k)}}\frac{{D}^{(k)}_{c'}L^{(k)}e^{(k)}_{c'}}{D^{(k)}} \sum_{c\in \mathcal{C}^{(k)}}\frac{{D}^{(k)}_{c}}{D^{(k)}L^{(k)}e^{(k)}_{c}}  \sum_{\ell=1}^{L^{(k)}} \sum_{n\in \mathcal{N}^{(k)}_c}\frac{{D}^{(k)}_{n}}{D^{(k)}_c} \widetilde{\nabla F}_{n}^{(k, \ell)}\Bigg] \\&  + \frac{\beta}{2} \mathbb{E}_k\Bigg[\Bigg\Vert \eta_{_k} \sum_{c'\in \mathcal{C}^{(k)}}\frac{{D}^{(k)}_{c'}L^{(k)}e^{(k)}_{c'}}{D^{(k)}} \sum_{c\in \mathcal{C}^{(k)}}\frac{{D}^{(k)}_{c}}{D^{(k)}L^{(k)}e^{(k)}_{c}}  \sum_{\ell=1}^{L^{(k)}} \sum_{n\in \mathcal{N}^{(k)}_c}\frac{{D}^{(k)}_{n}}{D^{(k)}_c} \widetilde{\nabla F}_{n}^{(k, \ell)}\Bigg\Vert^2\Bigg].
    \end{aligned}
\end{equation}
Since $\widetilde{\nabla F}_{n}^{(k, \ell)} = -\left( \bm{\omega}_{n}^{(k,\ell-1), e_{c}^{(k)}} - \overline{\bm{\omega}}_{c}^{(k,\ell-1)}\right)\Big/ {\eta_{_k}}$, we have:
\begin{equation}\label{collectiveSGD}
    \widetilde{\nabla F}_{n}^{(k, \ell)} = \frac{1}{B^{(k, \ell)}_{n}} \sum_{e=1}^{e_{c}^{(k)}} \sum_{d\in \mathcal{B}^{(k,\ell),e}_{n}}\nabla f(\bm{\omega}^{(k, \ell-1),e-1}_{n},d).
\end{equation}
Since the mini-batch size is fixed during local SGD iterations at each G-round and SGD is unbiased, taking the expectation from both hand sides of \eqref{collectiveSGD} leads to:
\begin{equation}\label{expcollectiveSGD}
   \mathbb{E}_k \left[\widetilde{\nabla F}_{n}^{(k, \ell)}\right] = \sum_{e=1}^{e_{c}^{(k)}} \nabla F_{n}^{(k, \ell)}(\bm{\omega}^{(k, \ell-1),e-1}_{n}).
\end{equation}
Replacing \eqref{expcollectiveSGD} in \eqref{betaupdaterule}, we get:
\begin{equation}
\begin{aligned}
 &\mathbb{E}_k \left[F^{(k)}(\bm{\omega}^{(k+1)})\right] \leq F^{(k)}(\bm{\omega}^{(k)})
 \\& - \nabla F^{(k)}(\bm{\omega}^{(k)})^\top \mathbb{E}_k\Bigg[ \eta_{_k} \sum_{c'\in \mathcal{C}^{(k)}}\frac{{D}^{(k)}_{c'}L^{(k)}e^{(k)}_{c'}}{D^{(k)}} \sum_{c\in \mathcal{C}^{(k)}}\frac{{D}^{(k)}_{c}}{D^{(k)}L^{(k)}e^{(k)}_{c}}   \sum_{n\in \mathcal{N}^{(k)}_c} \frac{{D}^{(k)}_{n}}{D^{(k)}_c} \sum_{\ell=1}^{L^{(k)}}\sum_{e=1}^{e_{c}^{(k)}} \nabla F_{n}^{(k, \ell)}(\bm{\omega}^{(k, \ell-1),e-1}_{n}) \Bigg] \\&  + \frac{\beta}{2} \mathbb{E}_k\Bigg[\Bigg\Vert \eta_{_k} \sum_{c'\in \mathcal{C}^{(k)}}\frac{{D}^{(k)}_{c'}L^{(k)}e^{(k)}_{c'}}{D^{(k)}} \sum_{c\in \mathcal{C}^{(k)}}\frac{{D}^{(k)}_{c}}{D^{(k)}L^{(k)}e^{(k)}_{c}}  \sum_{\ell=1}^{L^{(k)}}  \sum_{n\in \mathcal{N}^{(k)}_c} \frac{{D}^{(k)}_{n}}{D^{(k)}_c} \widetilde{\nabla F}_{n}^{(k, \ell)}\Bigg\Vert^2\Bigg].
    \end{aligned}
\end{equation}
By using the rule $2\bm{a}^\top\bm{b}=\Vert \bm{a}\Vert^2+\Vert \bm{b}\Vert^2-\Vert \bm{a}-\bm{b}\Vert^2$ for any two real valued vectors $\bm{a}$ and $\bm{b}$ with the same length, we get:
\begin{equation}\label{aftervector}
\begin{aligned}
 &\mathbb{E}_k \left[F^{(k)}(\bm{\omega}^{(k+1)})\right] \leq F^{(k)}(\bm{\omega}^{(k)}) - \frac{\eta_k}{2} \hspace{-2mm}\sum_{c'\in \mathcal{C}^{(k)}}\hspace{-2mm}\frac{{D}^{(k)}_{c'}L^{(k)}e^{(k)}_{c'}}{D^{(k)}} \mathbb{E}_k\Bigg[ \left \Vert \nabla F^{(k)}(\bm{\omega}^{(k)})\right\Vert^2
 \\& + \left \Vert \sum_{c\in \mathcal{C}^{(k)}}\frac{{D}^{(k)}_{c}}{D^{(k)}L^{(k)}e^{(k)}_{c}}  \sum_{\ell=1}^{L^{(k)}} \hspace{-1mm} \sum_{n\in \mathcal{N}^{(k)}_c} \frac{{D}^{(k)}_{n}}{D^{(k)}_c} \sum_{e=1}^{e_{c}^{(k)}}  \nabla F_{n}^{(k, \ell)}(\bm{\omega}^{(k, \ell-1),e-1}_{n}) \right \Vert^2 
  \\& - \left \Vert \nabla F^{(k)}(\bm{\omega}^{(k)}) - \sum_{c\in \mathcal{C}^{(k)}}\frac{{D}^{(k)}_{c}}{D^{(k)}L^{(k)}e^{(k)}_{c}}   \sum_{n\in \mathcal{N}^{(k)}_c} \frac{{D}^{(k)}_{n}}{D^{(k)}_c} \sum_{\ell=1}^{L^{(k)}} \sum_{e=1}^{e_{c}^{(k)}}  \nabla F_{n}^{(k, \ell)}(\bm{\omega}^{(k, \ell-1),e-1}_{n}) \right \Vert^2 \Bigg]
 \\& + \frac{\beta}{2} \mathbb{E}_k\Bigg[\Bigg\Vert \eta_{_k} \sum_{c'\in \mathcal{C}^{(k)}}\frac{{D}^{(k)}_{c'}L^{(k)}e^{(k)}_{c'}}{D^{(k)}} \sum_{c\in \mathcal{C}^{(k)}}\frac{{D}^{(k)}_{c}}{D^{(k)}L^{(k)}e^{(k)}_{c}}  \sum_{\ell=1}^{L^{(k)}} \sum_{n\in \mathcal{N}^{(k)}_c} \frac{{D}^{(k)}_{n}}{D^{(k)}_c} \widetilde{\nabla F}_{n}^{(k, \ell)}\Bigg\Vert^2\Bigg].
    \end{aligned}
\end{equation}
We bound the last term on the right hand side of the above inequality as follows:
\begin{equation}
\hspace{-2.5mm}
\resizebox{.99\linewidth}{!}{$
    \begin{aligned}
    & \frac{\eta_{k}^2 \beta }{2} \mathbb{E}_k\left[\left\Vert \sum_{c'\in \mathcal{C}^{(k)}}\frac{{D}^{(k)}_{c'}L^{(k)}e^{(k)}_{c'}}{D^{(k)}} \sum_{c\in \mathcal{C}^{(k)}}\frac{{D}^{(k)}_{c}}{D^{(k)}L^{(k)}e^{(k)}_{c}}  \sum_{\ell=1}^{L^{(k)}}  \sum_{n\in \mathcal{N}^{(k)}_c} \frac{{D}^{(k)}_{n}}{D^{(k)}_c} \widetilde{\nabla F}_{n}^{(k, \ell)}\right\Vert^2\right]
    \\& =\frac{ \eta_{k}^2 \beta }{2} \mathbb{E}_k\left[\left\Vert \sum_{c'\in \mathcal{C}^{(k)}}\frac{{D}^{(k)}_{c'}L^{(k)}e^{(k)}_{c'}}{D^{(k)}} \sum_{c\in \mathcal{C}^{(k)}}\frac{{D}^{(k)}_{c}}{D^{(k)}L^{(k)}e^{(k)}_{c}}   \sum_{n\in \mathcal{N}^{(k)}_c} \frac{{D}^{(k)}_{n}}{D^{(k)}_c} \sum_{\ell=1}^{L^{(k)}}\frac{1}{B^{(k, \ell)}_{n}} \sum_{e=1}^{e_{c}^{(k)}} \sum_{d\in \mathcal{B}^{(k,\ell),e}_{n}}\nabla f(\bm{\omega}^{(k, \ell-1),e-1}_{n},d)\right\Vert^2\right]
    \\& =\frac{\eta_{k}^2 \beta }{2} \mathbb{E}_k\Bigg[\Bigg\Vert \sum_{c'\in \mathcal{C}^{(k)}}\frac{{D}^{(k)}_{c'}L^{(k)}e^{(k)}_{c'}}{D^{(k)}} \sum_{c\in \mathcal{C}^{(k)}}\frac{{D}^{(k)}_{c}}{D^{(k)}L^{(k)}e^{(k)}_{c}}   \sum_{n\in \mathcal{N}^{(k)}_c} \frac{{D}^{(k)}_{n}}{D^{(k)}_c} \sum_{\ell=1}^{L^{(k)}}\frac{1}{B^{(k, \ell)}_{n}} \sum_{e=1}^{e_{c}^{(k)}} \sum_{d\in \mathcal{B}^{(k,\ell),e}_{n}}\nabla f(\bm{\omega}^{(k, \ell-1),e-1}_{n},d)
    \\& + \sum_{c'\in \mathcal{C}^{(k)}}\frac{{D}^{(k)}_{c'}L^{(k)}e^{(k)}_{c'}}{D^{(k)}} \sum_{c\in \mathcal{C}^{(k)}}\frac{{D}^{(k)}_{c}}{D^{(k)}L^{(k)}e^{(k)}_{c}}   \sum_{n\in \mathcal{N}^{(k)}_c} \frac{{D}^{(k)}_{n}}{D^{(k)}_c} \sum_{\ell=1}^{L^{(k)}}\sum_{e=1}^{e_{c}^{(k)}} \left(\nabla F_{n}^{(k, \ell)}(\bm{\omega}^{(k, \ell-1),e-1}_{n}) - \nabla F_{n}^{(k, \ell)}(\bm{\omega}^{(k, \ell-1),e-1}_{n}) \right)\Bigg\Vert^2\Bigg]
    \\& \overset{(i)}{\leq} \eta_k^2 \beta \hspace{-2mm}\sum_{c'\in \mathcal{C}^{(k)}} \hspace{-1mm} \left(\frac{{D}^{(k)}_{c'}L^{(k)}e^{(k)}_{c'}}{D^{(k)}}\right)^2 \hspace{-1mm}\underbrace{ \mathbb{E}_k\left[\left\Vert \sum_{c\in \mathcal{C}^{(k)}}\hspace{-1mm}\frac{{D}^{(k)}_{c}}{D^{(k)}L^{(k)}e^{(k)}_{c}}\sum_{n\in \mathcal{N}^{(k)}_c} \hspace{-1mm}\frac{{D}^{(k)}_{n}}{D^{(k)}_c} \sum_{\ell=1}^{L^{(k)}} \sum_{e=1}^{e^{(k)}_{c}} \left( \sum_{d\in \mathcal{B}^{(k,\ell),e}_{n}} \hspace{-3mm} \frac{\nabla f(\bm{\omega}^{(k, \ell-1),e-1}_{n},d)}{B^{(k,\ell)}_{n}} - \nabla F_{n}^{(k, \ell)}(\bm{\omega}^{(k, \ell-1),e-1}_{n})\right)\right\Vert^2\right]}_{(a)}
    \\& + \eta_k^2 \beta \sum_{c'\in \mathcal{C}^{(k)}} \hspace{-1mm} \left(\frac{{D}^{(k)}_{c'}L^{(k)}e^{(k)}_{c'}}{D^{(k)}}\right)^2 \mathbb{E}_k\left[\left\Vert  \sum_{c\in \mathcal{C}^{(k)}}\hspace{-1mm}\frac{{D}^{(k)}_{c}}{D^{(k)}L^{(k)}e^{(k)}_{c}}\sum_{n\in \mathcal{N}^{(k)}_c} \hspace{-1mm}\frac{{D}^{(k)}_{n}}{D^{(k)}_c} \sum_{\ell=1}^{L^{(k)}} \sum_{e=1}^{e^{(k)}_{c}}\nabla F_{n}^{(k, \ell)}(\bm{\omega}^{(k, \ell-1),e-1}_{n})\right\Vert^2\right],
    \end{aligned}
    $}
    \hspace{-5.5mm}
\end{equation}
where in inequality $(i)$ we have used the Cauchy-Schwarz inequality $\Vert \mathbf{a}+\mathbf{b} \Vert^2\leq 2 \Vert \mathbf{a} \Vert^2+2\Vert \mathbf{b} \Vert^2$. We bound term $(a)$ by using the fact that each local gradient estimation is unbiased, together with the assumption that the noise of gradient estimation is independent across the nodes. Therefore, we get:
\begin{equation}
    \begin{aligned}
     & (a) \hspace{-1mm} \overset{(i)} = \sum_{c\in \mathcal{C}^{(k)}} \hspace{-2mm} \left(\frac{{D}^{(k)}_{c}}{D^{(k)}L^{(k)}e^{(k)}_{c}} \right)^2 \hspace{-3mm} \sum_{n\in \mathcal{N}^{(k)}_c} \hspace{-2mm} \left(\frac{{D}^{(k)}_{n}}{D^{(k)}_c} \right)^2 \hspace{-1mm} \mathbb{E}_k\hspace{-1mm}\left[\left\Vert \sum_{\ell=1}^{L^{(k)}} \sum_{e=1}^{e^{(k)}_{c}} \hspace{-0.5mm} \underbrace{\left( \sum_{d\in \mathcal{B}^{(k,\ell),e}_{n}} \hspace{-3.5mm} \frac{\nabla f(\bm{\omega}^{(k, \ell-1),e-1}_{n},d)}{B^{(k,\ell)}_{n}} \hspace{-0.5mm} - \hspace{-0.5mm}\nabla F_{n}^{(k, \ell)}(\bm{\omega}^{(k, \ell-1),e-1}_{n})\right)}_{b}\right\Vert^2\right]
     \\&  \overset{(ii)} {=} \sum_{c\in \mathcal{C}^{(k)}} \hspace{-2mm} \left(\frac{{D}^{(k)}_{c}}{D^{(k)}L^{(k)}e^{(k)}_{c}} \right)^2 \hspace{-3mm} \sum_{n\in \mathcal{N}^{(k)}_c} \hspace{-2mm} \left(\frac{{D}^{(k)}_{n}}{D^{(k)}_c} \right)^2 \sum_{\ell=1}^{L^{(k)}} \sum_{e=1}^{e^{(k)}_{c}} \mathbb{E}_k\left[\left\Vert \sum_{d\in \mathcal{B}^{(k,\ell),e}_{n}} \hspace{-3mm} \frac{\nabla f(\bm{\omega}^{(k, \ell-1),e-1}_{n},d)}{B^{(k,\ell)}_{n}} \hspace{-0.5mm} - \hspace{-0.5mm} \nabla F_{n}^{(k, \ell)}(\bm{\omega}^{(k, \ell-1),e-1}_{n})\right\Vert^2\right]
     \\& \overset{(iii)}  = \sum_{c\in \mathcal{C}^{(k)}} \hspace{-1mm} \left(\frac{{D}^{(k)}_{c}}{D^{(k)}L^{(k)}e^{(k)}_{c}} \right)^2 \sum_{\ell=1}^{L^{(k)}} \sum_{n\in \mathcal{N}^{(k)}_c} \hspace{-1mm} \left(\frac{{D}^{(k)}_{n}}{D^{(k)}_c} \right)^2 \sum_{e=1}^{e^{(k)}_{c}} 2 \left(1- \varsigma^{(k,\ell)}_n\right)\frac{\left(\sigma^{(k)}_n\right)^2}{\varsigma^{(k,\ell)}_n D^{(k)}_n} \Theta^2,
    \end{aligned}
\end{equation}
where to obtain $(i)$, we expanded term $(a)$ and used the fact that the noise of gradient estimation across the mini-batches is independent and zero mean, and to obtain $(ii)$, we used the fact that each individual term in term $(b)$ is zero mean. To further bound this term, Lemma \ref{Lemma:SGD} (Mini-batch SGD noise characterization) is exploited to obtain $(iii)$. Replacing the above result back in \eqref{aftervector} results in:
\begin{equation}
\begin{aligned}
 & \mathbb{E}_k \left[F^{(k)}(\bm{\omega}^{(k+1)})\right] \leq F^{(k)}(\bm{\omega}^{(k)}) - \frac{\eta_k}{2} \hspace{-2mm}\sum_{c'\in \mathcal{C}^{(k)}}\hspace{-2mm}\frac{{D}^{(k)}_{c'}L^{(k)}e^{(k)}_{c'}}{D^{(k)}} \mathbb{E}_k\Bigg[ \left \Vert\nabla F^{(k)}(\bm{\omega}^{(k)})\right\Vert^2 
 \\& + \left \Vert \sum_{c\in \mathcal{C}^{(k)}}\frac{{D}^{(k)}_{c}}{D^{(k)}L^{(k)}e^{(k)}_{c}}  \sum_{\ell=1}^{L^{(k)}} \hspace{-1mm} \sum_{n\in \mathcal{N}^{(k)}_c} \frac{{D}^{(k)}_{n}}{D^{(k)}_c} \sum_{e=1}^{e_{c}^{(k)}}  \nabla F_{n}^{(k, \ell)}(\bm{\omega}^{(k, \ell-1),e-1}_{n}) \right \Vert^2 
 \\& - \left \Vert \nabla F^{(k)}(\bm{\omega}^{(k)}) - \sum_{c\in \mathcal{C}^{(k)}}\frac{{D}^{(k)}_{c}}{D^{(k)}L^{(k)}e^{(k)}_{c}}   \sum_{n\in \mathcal{N}^{(k)}_c} \frac{{D}^{(k)}_{n}}{D^{(k)}_c} \sum_{\ell=1}^{L^{(k)}} \sum_{e=1}^{e_{c}^{(k)}}  \nabla F_{n}^{(k, \ell)}(\bm{\omega}^{(k, \ell-1),e-1}_{n}) \right \Vert^2 \Bigg]
 \\& + 2 \eta_k^2 \beta \Theta^2 \sum_{c'\in \mathcal{C}^{(k)}} \left(\frac{{D}^{(k)}_{c'}L^{(k)}e^{(k)}_{c'}}{D^{(k)}}\right)^2 \sum_{c\in \mathcal{C}^{(k)}} \left(\frac{{D}^{(k)}_{c}}{D^{(k)}L^{(k)}e^{(k)}_{c}} \right)^2 \sum_{\ell=1}^{L^{(k)}} \sum_{n\in \mathcal{N}^{(k)}_c} \left(\frac{{D}^{(k)}_{n}}{D^{(k)}_c} \right)^2 \sum_{e=1}^{e^{(k)}_{c}} \left(1- \varsigma^{(k,\ell)}_n\right)\frac{\left(\sigma^{(k)}_n\right)^2}{\varsigma^{(k,\ell)}_n D^{(k)}_n}
 \\& + \eta_k^2 \beta \sum_{c'\in \mathcal{C}^{(k)}} \hspace{-1mm} \left(\frac{{D}^{(k)}_{c'}L^{(k)}e^{(k)}_{c'}}{D^{(k)}}\right)^2 \mathbb{E}_k\left[\left\Vert  \sum_{c\in \mathcal{C}^{(k)}}\hspace{-1mm}\frac{{D}^{(k)}_{c}}{D^{(k)}L^{(k)}e^{(k)}_{c}}\sum_{\ell=1}^{L^{(k)}}\sum_{n\in \mathcal{N}^{(k)}_c} \hspace{-1mm}\frac{{D}^{(k)}_{n}}{D^{(k)}_c} \sum_{e=1}^{e^{(k)}_{c}}\nabla F_{n}^{(k, \ell)}(\bm{\omega}^{(k, \ell-1),e-1}_{n})\right\Vert^2\right].
    \end{aligned}
\end{equation}
Conducting algebraic manipulation on the above inequality and gathering terms results in:
\begin{equation}\label{mainbeforediss}
\hspace{-3mm}
\resizebox{.99\linewidth}{!}{$
\begin{aligned}
 &\mathbb{E}_k \left[F^{(k)}(\bm{\omega}^{(k+1)})\right] \leq F^{(k)}(\bm{\omega}^{(k)}) - \frac{\eta_k}{2} \sum_{c'\in \mathcal{C}^{(k)}}\frac{{D}^{(k)}_{c'}L^{(k)}e^{(k)}_{c'}}{D^{(k)}}  \left \Vert\nabla F^{(k)}(\bm{\omega}^{(k)})\right\Vert^2 
 \\& + \underbrace{\left(\eta_k^2 \beta \sum_{c'\in \mathcal{C}^{(k)}} \left(\frac{{D}^{(k)}_{c'}L^{(k)}e^{(k)}_{c'}}{D^{(k)}}\right)^2 - \frac{\eta_k}{2} \sum_{c'\in \mathcal{C}^{(k)}}\frac{{D}^{(k)}_{c'}L^{(k)}e^{(k)}_{c'}}{D^{(k)}} \right) \mathbb{E}_k \left \Vert \sum_{c\in \mathcal{C}^{(k)}}\hspace{-1mm}\frac{{D}^{(k)}_{c}}{D^{(k)}L^{(k)}e^{(k)}_{c}}\sum_{\ell=1}^{L^{(k)}}\sum_{n\in \mathcal{N}^{(k)}_c} \hspace{-1mm}\frac{{D}^{(k)}_{n}}{D^{(k)}_c} \sum_{e=1}^{e^{(k)}_{c}}\nabla F_{n}^{(k, \ell)}(\bm{\omega}^{(k, \ell-1),e-1}_{n})\right \Vert^2}_{(c)}
 \\& + \frac{\eta_k}{2} \sum_{c'\in \mathcal{C}^{(k)}}\frac{{D}^{(k)}_{c'}L^{(k)}e^{(k)}_{c'}}{D^{(k)}}  \underbrace{\mathbb{E}_k \left \Vert \nabla F^{(k)}(\bm{\omega}^{(k)}) - \sum_{c\in \mathcal{C}^{(k)}}\frac{{D}^{(k)}_{c}}{D^{(k)}L^{(k)}e^{(k)}_{c}}   \sum_{n\in \mathcal{N}^{(k)}_c} \frac{{D}^{(k)}_{n}}{D^{(k)}_c} \sum_{\ell=1}^{L^{(k)}} \sum_{e=1}^{e_{c}^{(k)}}  \nabla F_{n}^{(k, \ell)}(\bm{\omega}^{(k, \ell-1),e-1}_{n}) \right \Vert^2}_{(d)}
 \\& + 2 \eta_k^2 \beta \Theta^2 \sum_{c'\in \mathcal{C}^{(k)}} \left(\frac{{D}^{(k)}_{c'}L^{(k)}e^{(k)}_{c'}}{D^{(k)}}\right)^2 \sum_{c\in \mathcal{C}^{(k)}} \left(\frac{{D}^{(k)}_{c}}{D^{(k)}L^{(k)}e^{(k)}_{c}} \right)^2 \sum_{\ell=1}^{L^{(k)}} \sum_{n\in \mathcal{N}^{(k)}_c} \left(\frac{{D}^{(k)}_{n}}{D^{(k)}_c} \right)^2 \sum_{e=1}^{e^{(k)}_{c}} \left(1- \varsigma^{(k,\ell)}_n\right)\frac{\left(\sigma^{(k)}_n\right)^2}{\varsigma^{(k,\ell)}_n D^{(k)}_n}.
    \end{aligned}
    $}\hspace{-7mm}
\end{equation}
Assuming $\eta_{t} \leq \left(2\beta \sum_{c'\in \mathcal{C}^{(k)}}\frac{{D}^{(k)}_{c'}L^{(k)}e^{(k)}_{c'}}{D^{(k)}}\right)^{-1}$, term $(c)$ would be negative and can be removed from the upper bound. Then, we bound the $(d)$ as follows:
\begin{equation} \label{weightdiff}
\begin{aligned}
    (d) & \overset{(i)}{\leq} \sum_{c\in \mathcal{C}^{(k)}}\frac{{D}^{(k)}_{c}}{D^{(k)}}\mathbb{E}_k \left \Vert \sum_{n\in \mathcal{N}^{(k)}_c} \frac{{D}^{(k)}_{n}}{D^{(k)}_c} \nabla F^{(k)}_{n}(\bm{\omega}^{(k)}) - \frac{1}{L^{(k)} e^{(k)}_{c}} \sum_{n\in \mathcal{N}^{(k)}_c} \frac{{D}^{(k)}_{n}}{D^{(k)}_c} \sum_{\ell=1}^{L^{(k)}} \sum_{e=1}^{e_{c}^{(k)}}  \nabla F_{n}^{(k, \ell)}(\bm{\omega}^{(k, \ell-1),e-1}_{n}) \right \Vert^2
    \\&
    \overset{(ii)}{\leq} \sum_{c\in \mathcal{C}^{(k)}} \frac{{D}^{(k)}_{c}}{D^{(k)}L^{(k)}e^{(k)}_{c}}   \sum_{n\in \mathcal{N}^{(k)}_c} \frac{{D}^{(k)}_{n}}{D^{(k)}_c} \sum_{\ell=1}^{L^{(k)}} \sum_{e=1}^{e_{c}^{(k)}}\mathbb{E}_k \left[\left \Vert \nabla F^{(k)}_{n}(\bm{\omega}^{(k)}) -  \nabla F_{n}^{(k, \ell)}(\bm{\omega}^{(k, \ell-1),e-1}_{n}) \right \Vert^2 \right] 
    \\& \leq  \beta^2\sum_{c\in \mathcal{C}^{(k)}} \frac{{D}^{(k)}_{c}}{D^{(k)}L^{(k)}e^{(k)}_{c}}   \sum_{n\in \mathcal{N}^{(k)}_c} \frac{{D}^{(k)}_{n}}{D^{(k)}_c} \sum_{\ell=1}^{L^{(k)}} \sum_{e=1}^{e_{c}^{(k)}} \mathbb{E}_k \left[\left \Vert  \bm{\omega}^{(k)} - \bm{\omega}^{(k, \ell-1),e-1}_{n} \right \Vert^2 \right],
\end{aligned}
\end{equation}
where for acquiring inequalities $(i)$ and $(ii)$, we used Jensen's inequality and $\beta$-smoothness of the loss function. To bound $\mathbb{E}_k \left[\left \Vert \bm{\omega}^{(k)} - \bm{\omega}^{(k, \ell-1),e-1}_{n} \right \Vert^2 \right]$, we take the following steps:
\begin{equation}
\hspace{-10mm}
\resizebox{.97\linewidth}{!}{$
  \begin{aligned}
    &\mathbb{E}_k \left[\left \Vert \bm{\omega}^{(k)} - \bm{\omega}^{(k, \ell-1),e-1}_{n} \right \Vert^2 \right] = \mathbb{E}_k \left[\left \Vert \bm{\omega}^{(k)} - \overline{\bm{\omega}}_{c}^{(k,1)} + \overline{\bm{\omega}}_{c}^{(k,1)} - ...- \overline{\bm{\omega}}_{c}^{(k,\ell-1)} + \overline{\bm{\omega}}_{c}^{(k,\ell-1)} - \bm{\omega}^{(k, \ell-1),e-1}_{n} \right \Vert^2 \right]
    \\& {=}\eta_k^2  \mathbb{E}_k \left[\left \Vert \sum_{n'\in \mathcal{N}^{(k)}_c} \frac{{D}^{(k)}_{n'}}{D^{(k)}_c} \sum_{\ell'=1}^{\ell-1} \sum_{e'=1}^{e} \sum_{d\in \mathcal{B}^{(k,\ell'),e'}_{n'}} \frac{\nabla f(\bm{\omega}^{(k, \ell'-1),e'-1}_{n'},d)}{B^{(k,\ell')}_{n'}} +  \sum_{e'=1}^{e-1} \sum_{d\in \mathcal{B}^{(k, \ell),e'}_{n}} \frac{\nabla f(\bm{\omega}^{(k, \ell-1),e'-1}_{n},d)}{B^{(k,\ell)}_{n}} \right \Vert^2 \right]
    \\& \overset{(i)} \leq 2 \eta_k^2  \mathbb{E}_k \left[\left \Vert \sum_{n'\in \mathcal{N}^{(k)}_c} \frac{{D}^{(k)}_{n'}}{D^{(k)}_c} \sum_{\ell'=1}^{\ell-1} \sum_{e'=1}^{e} \sum_{d\in \mathcal{B}^{(k, \ell'),e'}_{n'}} \frac{\nabla f(\bm{\omega}^{(k, \ell'-1),e'-1}_{n'},d)}{B^{(k,\ell')}_{n'}} \right \Vert^2 \right] + 2 \eta_k^2  \mathbb{E}_k \left[\left \Vert \sum_{e'=1}^{e-1} \sum_{d\in \mathcal{B}^{(k, \ell),e'}_{n}} \frac{\nabla f(\bm{\omega}^{(k, \ell-1),e'-1}_{n},d)}{B^{(k,\ell)}_{n}} \right \Vert^2 \right]
    \\& \overset{(ii)} \leq \underbrace{2 \eta_k^2 \sum_{n'\in \mathcal{N}^{(k)}_c} \frac{{D}^{(k)}_{n'}}{D^{(k)}_c} \mathbb{E}_k \left[\left \Vert \sum_{\ell'=1}^{\ell-1} \sum_{e'=1}^{e} \sum_{d\in \mathcal{B}^{(k, \ell'),e'}_{n'}} \hspace{-3mm} \frac{\nabla f(\bm{\omega}^{(k, \ell'-1),e'-1}_{n'},d)}{B^{(k,\ell')}_{n'}} \right \Vert^2 \right]}_{(e)} \hspace{-1mm} + \underbrace{2 \eta_k^2  \mathbb{E}_k \left[\left \Vert \sum_{e'=1}^{e-1} \sum_{d\in \mathcal{B}^{(k, \ell),e'}_{n}} \frac{\nabla f(\bm{\omega}^{(k, \ell-1),e'-1}_{n},d)}{B^{(k,\ell)}_{n}} \right \Vert^2 \right]}_{(f)},
    \end{aligned}
    $}
    \hspace{-6mm}
    \end{equation}
where we used Cauchy-Schwarz inequality for inequality $(i)$ and Jensen's inequality for inequality $(ii)$. Next, we bound term $(e)$ as follows:
\begin{equation}
\hspace{-12mm}
\resizebox{.97\linewidth}{!}{$
  \begin{aligned}
    \\& (e) = 2 \eta_k^2 \sum_{n'\in \mathcal{N}^{(k)}_c} \frac{{D}^{(k)}_{n'}}{D^{(k)}_c} \mathbb{E}_k\left[\left\Vert
   \sum_{\ell'=1}^{\ell-1} \sum_{e'=1}^{e} \left(\sum_{d\in \mathcal{B}^{(k, \ell'),e'}_{n'}} \hspace{-3mm} \frac{\nabla f(\bm{\omega}^{(k, \ell'-1),e'-1}_{n},d)}{B^{(k,\ell')}_{n'}} - \nabla F_{n'}^{(k)}(\bm{\omega}^{(k, \ell'-1),e'-1}_{n'}) + \nabla F_{n'}^{(k)}(\bm{\omega}^{(k, \ell'-1),e'-1}_{n'})
    \right)\right\Vert^2 \right]
    \\& \overset{(i)} \leq 4 \eta_k^2 \sum_{n'\in \mathcal{N}^{(k)}_c} \frac{{D}^{(k)}_{n'}}{D^{(k)}_c}  \mathbb{E}_k\left[\left\Vert
   \sum_{\ell'=1}^{\ell-1} \sum_{e'=1}^{e} \left(\sum_{d\in \mathcal{B}^{(k, \ell'),e'}_{n'}} \hspace{-3mm} \frac{\nabla f(\bm{\omega}^{(k, \ell'-1),e'-1}_{n'},d)}{B^{(k,\ell')}_{n'}} - \nabla F_{n'}^{(k)}(\bm{\omega}^{(k, \ell'-1),e'-1}_{n'})
    \right)\right\Vert^2 \right] 
    \\& + 4 \eta_k^2 \sum_{n'\in \mathcal{N}^{(k)}_c} \frac{{D}^{(k)}_{n'}}{D^{(k)}_c}  \mathbb{E}_k\left[\left\Vert \sum_{\ell'=1}^{\ell-1}  \sum_{e'=1}^{e} \nabla F_{n'}^{(k)}(\bm{\omega}^{(k, \ell'-1),e'-1}_{n'})\right\Vert^2 \right]
    \\& \overset{(ii)} = 4 \eta_k^2 \sum_{n'\in \mathcal{N}^{(k)}_c} \frac{{D}^{(k)}_{n'}}{D^{(k)}_c}  \sum_{\ell'=1}^{\ell-1} \sum_{e'=1}^{e}  \mathbb{E}_k\left[\left\Vert \sum_{d\in \mathcal{B}^{(k, \ell'),e'}_{n'}} \hspace{-3mm} \frac{\nabla f(\bm{\omega}^{(k, \ell'-1),e'-1}_{n'},d)}{B^{(k,\ell')}_{n'}} - \nabla F_{n'}^{(k)}(\bm{\omega}^{(k, \ell'-1),e'-1}_{n'})
    \right\Vert^2 \right] 
    \\& + 4 \eta_k^2 \sum_{n'\in \mathcal{N}^{(k)}_c} \frac{{D}^{(k)}_{n'}}{D^{(k)}_c}  \mathbb{E}_k\left[\left\Vert \sum_{\ell'=1}^{\ell-1}  \sum_{e'=1}^{e} \nabla F_{n'}^{(k)}(\bm{\omega}^{(k, \ell'-1),e'-1}_{n'})\right\Vert^2 \right]
    \\& \overset{(iii)} = 8 \eta_k^2 \sum_{n'\in \mathcal{N}^{(k)}_c} \frac{{D}^{(k)}_{n'}}{D^{(k)}_c}  \sum_{\ell'=1}^{\ell-1} \sum_{e'=1}^{e} \left(1- \varsigma^{(k,\ell'')}_{n'}\right) \frac{\left(\sigma^{(k)}_{n'}\right)^2}{\varsigma^{(k,\ell'')}_{n'} D^{(k)}_{n'}} \Theta^2 + 4 \eta_k^2 \sum_{n'\in \mathcal{N}^{(k)}_c} \frac{{D}^{(k)}_{n'}}{D^{(k)}_c}  \mathbb{E}_k\left[\left\Vert \sum_{\ell'=1}^{\ell-1}  \sum_{e'=1}^{e} \nabla F_{n'}^{(k)}(\bm{\omega}^{(k, \ell'-1),e'-1}_{n'})\right\Vert^2 \right],
    \end{aligned}
    $}
        \hspace{-19mm}
\end{equation}
 where we used Cauchy–Schwarz inequality for the inequality $(i)$, the fact that each local gradient estimation is unbiased for the inequality $(ii)$, and the Lemma \ref{Lemma:SGD} for the inequality $(iii)$. Likewise, term $(f)$ can be bounded as follows:
\begin{equation}
\hspace{-5mm}
\resizebox{.97\linewidth}{!}{$
  \begin{aligned}
    \\& (f) = 2 \eta_k^2  \mathbb{E}_k\left[\left\Vert
   \sum_{e'=1}^{e-1} \left(\sum_{d\in \mathcal{B}^{(k, \ell),e'}_{n}} \frac{\nabla f(\bm{\omega}^{(k, \ell-1),e'-1}_{n},d)}{B^{(k,\ell)}_{n}} - \nabla F_{n}^{(k)}(\bm{\omega}^{(k, \ell-1),e'-1}_{n}) + \nabla F_{n}^{(k)}(\bm{\omega}^{(k, \ell-1),e'-1}_{n})
    \right)\right\Vert^2 \right]
    \\& \overset{(i)} \leq 4 \eta_k^2  \mathbb{E}_k\left[\left\Vert
   \sum_{e'=1}^{e-1} \left(\sum_{d\in \mathcal{B}^{(k, \ell),e'}_{n}} \frac{\nabla f(\bm{\omega}^{(k, \ell-1),e'-1}_{n},d)}{B^{(k,\ell)}_{n}} - \nabla F_{n}^{(k)}(\bm{\omega}^{(k, \ell-1),e'-1}_{n})
    \right)\right\Vert^2 \right] + 4 \eta_t^2  \mathbb{E}_k\left[\left\Vert \sum_{e'=1}^{e-1} \nabla F_{n}^{(k)}(\bm{\omega}^{(k, \ell-1),e'-1}_{n})\right\Vert^2 \right]
    \\&\overset{(ii)} = 4 \eta_k^2 \sum_{e'=1}^{e-1} \mathbb{E}_k\left[\Bigg\Vert
   \sum_{d\in \mathcal{B}^{(k, \ell),e'}_{n}} \frac{\nabla f(\bm{\omega}^{(k, \ell-1),e'-1}_{n},d)}{B^{(k,\ell)}_{n}} - \nabla F_{n}^{(k)}(\bm{\omega}^{(k, \ell-1),e'-1}_{n})\Bigg\Vert^2\right]+4 \eta_t^2  \mathbb{E}_k\left[\left\Vert \sum_{e'=1}^{e-1} \nabla F_{n}^{(k)}(\bm{\omega}^{(k, \ell-1),e'-1}_{n})\right\Vert^2 \right]
    \\& \overset{(iii)} \leq 8 \eta_k^2 \sum_{e'=1}^{e-1} \left(1- \varsigma^{(k,\ell)}_n\right)\frac{\left(\sigma^{(k)}_n\right)^2}{\varsigma^{(k,\ell)}_n D^{(k)}_n} \Theta^2 + 4 \eta_t^2  \mathbb{E}_k\left[\left\Vert \sum_{e'=1}^{e-1} \nabla F_{n}^{(k)}(\bm{\omega}^{(k, \ell-1),e'-1}_{n})\right\Vert^2 \right],
    \end{aligned}
    $}
    \hspace{-5mm}
    \end{equation}
 where we used Cauchy–Schwarz inequality $(i)$, the fact that each local gradient estimation is unbiased, and the Lemma \ref{Lemma:SGD} for the inequalities $(i)$, $(ii)$, and $(iii)$ respectively. Thus, we have:
\begin{equation}
\hspace{-4mm}
\resizebox{.97\linewidth}{!}{$
  \begin{aligned}
    & \mathbb{E}_k \left[\left \Vert \bm{\omega}^{(k)} - \bm{\omega}^{(k, \ell-1),e-1}_{n} \right \Vert^2 \right] \leq  8 \eta_k^2 \sum_{e'=1}^{e-1} \left(1- \varsigma^{(k,\ell)}_n\right)\frac{\left(\sigma^{(k)}_n\right)^2}{\varsigma^{(k,\ell)}_n D^{(k)}_n} \Theta^2 + 8 \eta_k^2 \sum_{n'\in \mathcal{N}^{(k)}_c} \frac{{D}^{(k)}_{n'}}{D^{(k)}_c}  \sum_{\ell'=1}^{\ell-1} \sum_{e'=1}^{e} \left(1- \varsigma^{(k,\ell'')}_{n'}\right) \frac{\left(\sigma^{(k)}_{n'}\right)^2}{\varsigma^{(k,\ell'')}_{n'} D^{(k)}_{n'}} \Theta^2
    \\& + \underbrace{4 \eta_t^2  \mathbb{E}_k\left[\left\Vert \sum_{e'=1}^{e-1} \nabla F_{n}^{(k)}(\bm{\omega}^{(k, \ell-1),e'-1}_{n})\right\Vert^2 \right] + 4 \eta_k^2 \sum_{n'\in \mathcal{N}^{(k)}_c} \frac{{D}^{(k)}_{n'}}{D^{(k)}_c}  \mathbb{E}_k\left[\left\Vert \sum_{\ell'=1}^{\ell-1}  \sum_{e'=1}^{e} \nabla F_{n'}^{(k)}(\bm{\omega}^{(k, \ell'-1),e'-1}_{n'})\right\Vert^2 \right]}_{(g)}.
    \end{aligned}
    $}
    \hspace{-4mm}
    \end{equation}
Also, for term $(g)$, we have:
\begin{equation}
  \begin{aligned}
& (g) \overset{(i)} \leq 4 \eta_k^2 \left(e-1\right) \sum_{e'=1}^{e-1}\mathbb{E}_k\left[\left\Vert \nabla F_{n}^{(k)}(\bm{\omega}^{(k, \ell-1),e'-1}_{n}) - \nabla F_{n}^{(k)}(\bm{\omega}^{(k)}) + \nabla F_{n}^{(k)}(\bm{\omega}^{(k)})\right\Vert^2 \right]
\\& + 4 \eta_k^2 \sum_{n'\in \mathcal{N}^{(k)}_c} \frac{{D}^{(k)}_{n'}}{D^{(k)}_c} \left(\ell-1\right) \sum_{\ell'=1}^{\ell-1} e \sum_{e'=1}^{e}\mathbb{E}_k\left[\left\Vert \nabla F_{n'}^{(k)}(\bm{\omega}^{(k, \ell'-1),e'-1}_{n'}) - \nabla F_{n'}^{(k)}(\bm{\omega}^{(k)}) + \nabla F_{n'}^{(k)}(\bm{\omega}^{(k)})\right\Vert^2 \right]
 \\& \overset{(ii)}\leq 8 \eta_k^2 \left(e-1\right)  \sum_{e'=1}^{e-1}\mathbb{E}_k\left[\left\Vert \nabla F_{n}^{(k)}(\bm{\omega}^{(k, \ell-1),e'-1}_{n}) - \nabla F_{n}^{(k)}(\bm{\omega}^{(k)}) \right\Vert^2 \right] + 8 \eta_k^2 \left(e-1\right)  \sum_{e'=1}^{e-1}\left\Vert\nabla F_{n}^{(k)}(\bm{\omega}^{(k)})\right\Vert^2
 \\& + 8 \eta_k^2 \sum_{n'\in \mathcal{N}^{(k)}_c} \frac{{D}^{(k)}_{n'}}{D^{(k)}_c} \left(\ell-1\right) \sum_{\ell'=1}^{\ell-1} e \sum_{e'=1}^{e}\mathbb{E}_k\left[\left\Vert \nabla F_{n'}^{(k)}(\bm{\omega}^{(k, \ell'-1),e'-1}_{n'}) - \nabla F_{n'}^{(k)}(\bm{\omega}^{(k)})\right\Vert^2 \right]
  \\& + 8 \eta_k^2 \sum_{n'\in \mathcal{N}^{(k)}_c} \frac{{D}^{(k)}_{n'}}{D^{(k)}_c} \left(\ell-1\right) \sum_{\ell'=1}^{\ell-1} e \sum_{e'=1}^{e} \left\Vert \nabla F_{n'}^{(k)}(\bm{\omega}^{(k)})\right\Vert^2
 \\& \leq  8 \eta_k^2\beta^2 \left(e-1\right) \sum_{e'=1}^{e-1}\mathbb{E}_k\left[\left\Vert \bm{\omega}^{(k, \ell-1),e'-1}_{n} - \bm{\omega}^{(k)} \right\Vert^2 \right] + 8 \eta_k^2 \left(e-1\right) \sum_{e'=1}^{e-1}\left\Vert\nabla F_{n}^{(k)}(\bm{\omega}^{(k)})\right\Vert^2
 \\& + 8 \eta_k^2 \beta^2 \sum_{n'\in \mathcal{N}^{(k)}_c} \frac{{D}^{(k)}_{n'}}{D^{(k)}_c} \left(\ell-1\right) \sum_{\ell'=1}^{\ell-1} e \sum_{e'=1}^{e}\mathbb{E}_k\left[\left\Vert \bm{\omega}^{(k, \ell'-1),e'-1}_{n'} - \bm{\omega}^{(k)}\right\Vert^2 \right]
 \\& + 8 \eta_k^2 \sum_{n'\in \mathcal{N}^{(k)}_c} \frac{{D}^{(k)}_{n'}}{D^{(k)}_c} \left(\ell-1\right) \sum_{\ell'=1}^{\ell-1} e \sum_{e'=1}^{e} \left\Vert \nabla F_{n'}^{(k)} (\bm{\omega}^{(k)})\right\Vert^2,
\end{aligned}
\end{equation}
which implies:
\begin{equation}
\hspace{-4mm}
\resizebox{.97\linewidth}{!}{$
\begin{aligned}
 & \sum_{n\in \mathcal{N}^{(k)}_c} \frac{{D}^{(k)}_{n}}{D^{(k)}_c} \sum_{\ell=1}^{L^{(k)}}\sum_{e=1}^{e_{c}^{(k)}} \mathbb{E}_k\left[\left\Vert \bm{\omega}^{(k)} - \bm{\omega}^{(k, \ell-1),e-1}_{n}  \right\Vert^2 \right]\leq 8 \eta_k^2 \sum_{n\in \mathcal{N}^{(k)}_c} \frac{{D}^{(k)}_{n}}{D^{(k)}_c} \sum_{\ell=1}^{L^{(k)}} \sum_{e=1}^{e_{c}^{(k)}} \sum_{e'=1}^{e-1} \left(1- \varsigma^{(k,\ell)}_n\right)\frac{\left(\sigma^{(k)}_n\right)^2}{\varsigma^{(k,\ell)}_n D^{(k)}_n} \Theta^2 
 \\& + 8 \eta_k^2 \sum_{n\in \mathcal{N}^{(k)}_c} \frac{{D}^{(k)}_{n}}{D^{(k)}_c} \sum_{n'\in \mathcal{N}^{(k)}_c} \frac{{D}^{(k)}_{n'}}{D^{(k)}_c}  \sum_{\ell=1}^{L^{(k)}} \sum_{\ell'=1}^{\ell-1} \sum_{e=1}^{e_{c}^{(k)}} \sum_{e'=1}^{e} \left(1- \varsigma^{(k,\ell'')}_{n'}\right) \frac{\left(\sigma^{(k)}_{n'}\right)^2}{\varsigma^{(k,\ell'')}_{n'} D^{(k)}_{n'}} \Theta^2
  \\& + 8 \eta_k^2\beta^2 \sum_{n\in \mathcal{N}^{(k)}_c} \frac{{D}^{(k)}_{n}}{D^{(k)}_c} \sum_{\ell=1}^{L^{(k)}} \sum_{e=1}^{e_{c}^{(k)}} \left(e-1\right) \sum_{e'=1}^{e-1}\mathbb{E}_k\left[\left\Vert \bm{\omega}^{(k, \ell-1),e'-1}_{n} - \bm{\omega}^{(k)} \right\Vert^2 \right] + 8 \eta_k^2 \sum_{n\in \mathcal{N}^{(k)}_c} \frac{{D}^{(k)}_{n}}{D^{(k)}_c} \sum_{\ell=1}^{L^{(k)}} \sum_{e=1}^{e_{c}^{(k)}} \left(e-1\right) \sum_{e'=1}^{e-1}\left\Vert\nabla F_{n}^{(k)}(\bm{\omega}^{(k)})\right\Vert^2
 \\& + 8 \eta_k^2 \beta^2 \sum_{n\in \mathcal{N}^{(k)}_c} \frac{{D}^{(k)}_{n}}{D^{(k)}_c} \sum_{n'\in \mathcal{N}^{(k)}_c} \frac{{D}^{(k)}_{n'}}{D^{(k)}_c}\sum_{\ell=1}^{L^{(k)}} \left(\ell-1\right) \sum_{\ell'=1}^{\ell-1} \sum_{e=1}^{e_{c}^{(k)}} e \sum_{e'=1}^{e}\mathbb{E}_k\left[\left\Vert \bm{\omega}^{(k, \ell'-1),e'-1}_{n'} - \bm{\omega}^{(k)}\right\Vert^2 \right]
 \\& + 8 \eta_k^2 \sum_{n\in \mathcal{N}^{(k)}_c} \frac{{D}^{(k)}_{n}}{D^{(k)}_c} \sum_{n'\in \mathcal{N}^{(k)}_c} \frac{{D}^{(k)}_{n'}}{D^{(k)}_c} \sum_{\ell=1}^{L^{(k)}} \left(\ell-1\right) \sum_{\ell'=1}^{\ell-1} \sum_{e=1}^{e_{c}^{(k)}} e \sum_{e'=1}^{e} \left\Vert \nabla F_{n'}^{(k)}(\bm{\omega}^{(k)})\right\Vert^2
 \\& \leq 8 \eta_k^2 \left(e_{c}^{(k)}\right) \left(e_{c}^{(k)}-1\right) \hspace{-2mm}\sum_{n\in \mathcal{N}^{(k)}_c} \hspace{-1mm}\frac{{D}^{(k)}_{n}}{D^{(k)}_c} \sum_{\ell=1}^{L^{(k)}}  \hspace{-1mm} \left(1- \varsigma^{(k,\ell)}_n\right)\frac{\left(\sigma^{(k)}_n\right)^2}{\varsigma^{(k,\ell)}_n D^{(k)}_n} \Theta^2 + 8 \eta_k^2 \left(L^{(k)}-1\right)\left(e_{c}^{(k)}\right)^{2} \hspace{-2mm} \sum_{n\in \mathcal{N}^{(k)}_c}\hspace{-1mm} \frac{{D}^{(k)}_{n}}{D^{(k)}_c} \sum_{\ell=1}^{L^{(k)}} \hspace{-1mm} \left(1- \varsigma^{(k,\ell)}_n\right)\frac{\left(\sigma^{(k)}_n\right)^2}{\varsigma^{(k,\ell)}_n D^{(k)}_n} \Theta^2
  \\& + 8 \eta_k^2\beta^2 \left(e_{c}^{(k)}\right) \left(e_{c}^{(k)}-1\right)\hspace{-2mm} \sum_{n\in \mathcal{N}^{(k)}_c} \hspace{-1mm} \frac{{D}^{(k)}_{n}}{D^{(k)}_c}   \sum_{\ell=1}^{L^{(k)}}\sum_{e=1}^{e_{c}^{(k)}} \mathbb{E}_k\left[\left\Vert \bm{\omega}^{(k, \ell-1),e-1}_{n} - \bm{\omega}^{(k)} \right\Vert^2 \right] \hspace{-1mm} + 8 \eta_k^2 \left(e_{c}^{(k)}\right) \left(e_{c}^{(k)}-1\right)\hspace{-2mm} \sum_{n\in \mathcal{N}^{(k)}_c} \hspace{-1mm} \frac{{D}^{(k)}_{n}}{D^{(k)}_c} \sum_{\ell=1}^{L^{(k)}} \sum_{e=1}^{e_{c}^{(k)}} \left\Vert\nabla F_{n}^{(k)}(\bm{\omega}^{(k)})\right\Vert^2
 \\& + 8 \eta_k^2 \beta^2 \left(L^{(k)}\right)\left(L^{(k)}-1\right) \left(e_{c}^{(k)}\right)^{2} \sum_{n\in \mathcal{N}^{(k)}_c} \frac{{D}^{(k)}_{n}}{D^{(k)}_c} \sum_{\ell=1}^{L^{(k)}} \sum_{e=1}^{e_{c}^{(k)}} \mathbb{E}_k\left[\left\Vert \bm{\omega}^{(k, \ell-1),e-1}_{n} - \bm{\omega}^{(k)}\right\Vert^2 \right] 
 \\& + 8 \eta_k^2 \left(L^{(k)}\right)\left(L^{(k)}-1\right) \left(e_{c}^{(k)}\right)^{2} \sum_{n\in \mathcal{N}^{(k)}_c} \frac{{D}^{(k)}_{n}}{D^{(k)}_c} \sum_{\ell=1}^{L^{(k)}} \sum_{e=1}^{e_{c}^{(k)}}\left\Vert \nabla F_{n}^{(k)}(\bm{\omega}^{(k)})\right\Vert^2.
\end{aligned}
$}
\hspace{-4mm}
\end{equation}
By assuming $\eta_k \leq \left(\beta \sqrt{8 \left(\left(L^{(k)}\right)\left(L^{(k)}-1\right) \left(e_{c}^{(k)}\right)^{2} + \left(e_{c}^{(k)}\right) \left(e_{c}^{(k)}-1\right)\right)}\right)^{-1},\forall c$, we get:
\begin{equation}
\begin{aligned}
&  \sum_{n\in \mathcal{N}^{(k)}_c} \frac{{D}^{(k)}_{n}}{D^{(k)}_c} \sum_{\ell=1}^{L^{(k)}}\sum_{e=1}^{e_{c}^{(k)}} \mathbb{E}_k\left[\left\Vert \bm{\omega}^{(k)} -  \bm{\omega}^{(k, \ell-1),e-1}_{n} \right\Vert^2 \right]
\\& \leq
\frac{8 \eta_k^2 \left(\left(L^{(k)}-1\right) \left(e_{c}^{(k)}\right)^{2} + \left(e_{c}^{(k)}\right) \left(e_{c}^{(k)}-1\right) \right) \sum_{n\in \mathcal{N}^{(k)}_c} \frac{{D}^{(k)}_{n}}{D^{(k)}_c} \sum_{\ell=1}^{L^{(k)}}  \left(1- \varsigma^{(k,\ell)}_n\right)\frac{\left(\sigma^{(k)}_n\right)^2}{\varsigma^{(k,\ell)}_n D^{(k)}_n} \Theta^2}{1 - 8 \eta_k^2 \beta^2 \left(\left(L^{(k)}\right)\left(L^{(k)}-1\right) \left(e_{c}^{(k)}\right)^{2} + \left(e_{c}^{(k)}\right) \left(e_{c}^{(k)}-1\right) \right)}
\\& + \frac{8 \eta_k^2 \left(\left(L^{(k)}\right)^{2}  \left(L^{(k)}-1\right) \left(e_{c}^{(k)}\right)^{3} + \left(L^{(k)}\right) \left(e_{c}^{(k)}\right)^{2}\left(e_{c}^{(k)}-1\right)\right) \sum_{n\in \mathcal{N}^{(k)}_c} \frac{{D}^{(k)}_{n}}{D^{(k)}_c} \left\Vert\nabla F_{n}^{(k)}(\bm{\omega}^{(k)})\right\Vert^2}{1 - 8 \eta_k^2 \beta^2 \left(\left(L^{(k)}\right)\left(L^{(k)}-1\right) \left(e_{c}^{(k)}\right)^{2} + \left(e_{c}^{(k)}\right) \left(e_{c}^{(k)}-1\right) \right)}.
\end{aligned}
\end{equation}
Subsequently, by replacing the above bound in \eqref{weightdiff}, we get:
\begin{equation}
\hspace{-4mm}
\resizebox{.97\linewidth}{!}{$
\begin{aligned}
& (d) \leq \sum_{c\in \mathcal{C}^{(k)}} \frac{{D}^{(k)}_{c}}{D^{(k)}L^{(k)}e^{(k)}_{c}}
\frac{8 \eta_k^2 \beta^2 \left(\left(L^{(k)}-1\right) \left(e_{c}^{(k)}\right)^{2} + \left(e_{c}^{(k)}\right) \left(e_{c}^{(k)}-1\right) \right) \sum_{n\in \mathcal{N}^{(k)}_c} \frac{{D}^{(k)}_{n}}{D^{(k)}_c} \sum_{\ell=1}^{L^{(k)}}  \left(1- \varsigma^{(k,\ell)}_n\right)\frac{\left(\sigma^{(k)}_n\right)^2}{\varsigma^{(k,\ell)}_n D^{(k)}_n} \Theta^2}{1 - 8 \eta_k^2 \beta^2 \left(\left(L^{(k)}\right)\left(L^{(k)}-1\right) \left(e_{c}^{(k)}\right)^{2} + \left(e_{c}^{(k)}\right) \left(e_{c}^{(k)}-1\right) \right)}
\\& + \sum_{c\in \mathcal{C}^{(k)}} \frac{{D}^{(k)}_{c}}{D^{(k)}L^{(k)}e^{(k)}_{c}} \frac{8 \eta_k^2 \beta^2 \left(\left(L^{(k)}\right)^{2}  \left(L^{(k)}-1\right) \left(e_{c}^{(k)}\right)^{3} + \left(L^{(k)}\right) \left(e_{c}^{(k)}\right)^{2}\left(e_{c}^{(k)}-1\right)\right) \sum_{n\in \mathcal{N}^{(k)}_c} \frac{{D}^{(k)}_{n}}{D^{(k)}_c} \left\Vert\nabla F_{n}^{(k)}(\bm{\omega}^{(k)})\right\Vert^2}{1 - 8 \eta_k^2 \beta^2 \left(\left(L^{(k)}\right)\left(L^{(k)}-1\right) \left(e_{c}^{(k)}\right)^{2} + \left(e_{c}^{(k)}\right) \left(e_{c}^{(k)}-1\right) \right)}.
\end{aligned}
$}\hspace{-4mm}
\end{equation}
We exploit the intra-cluster and inter-cluster bounded dissimilarity assumptions (Assumptions \ref{Assup:IntraClusterDissimilarity} and \ref{Assup:InterClusterDissimilarity}) for the second term as:
\begin{equation}
\hspace{-4mm}
\resizebox{.97\linewidth}{!}{$
\begin{aligned}
    & \sum_{c\in \mathcal{C}^{(k)}} \frac{{D}^{(k)}_{c}}{D^{(k)}L^{(k)}e^{(k)}_{c}} \frac{8 \eta_k^2 \beta^2 \left(\left(L^{(k)}\right)^{2}  \left(L^{(k)}-1\right) \left(e_{c}^{(k)}\right)^{3} + \left(L^{(k)}\right) \left(e_{c}^{(k)}\right)^{2}\left(e_{c}^{(k)}-1\right)\right) \sum_{n\in \mathcal{N}^{(k)}_c} \frac{{D}^{(k)}_{n}}{D^{(k)}_c} \left\Vert\nabla F_{n}^{(k)}(\bm{\omega}^{(k)})\right\Vert^2}{1 - 8 \eta_k^2 \beta^2 \left(\left(L^{(k)}\right)\left(L^{(k)}-1\right) \left(e_{c}^{(k)}\right)^{2} + \left(e_{c}^{(k)}\right) \left(e_{c}^{(k)}-1\right) \right)}
    \\& \leq \sum_{c\in \mathcal{C}^{(k)}} \frac{{D}^{(k)}_{c}}{D^{(k)}L^{(k)}} \frac{8 \eta_k^2 \beta^2\left(\left(L^{(k)}\right)^{2}  \left(L^{(k)}-1\right) \left(e_{c}^{(k)}\right)^{2} + \left(L^{(k)}\right) \left(e_{c}^{(k)}\right)\left(e_{c}^{(k)}-1\right)\right)}{1 - 8 \eta_k^2 \beta^2 \left(\left(L^{(k)}\right)\left(L^{(k)}-1\right) \left(e_{c}^{(k)}\right)^{2} + \left(e_{c}^{(k)}\right) \left(e_{c}^{(k)}-1\right) \right)} \left( \zeta^{\mathsf{Loc}}_{c,1} \left\Vert \sum_{n\in \mathcal{N}^{(k)}_c}\frac{{D}^{(k)}_{n}}{D^{(k)}_c} \nabla F_{n}^{(k)}(\bm{\omega}^{(k)})\right\Vert^2 +  \zeta^{\mathsf{Loc}}_{c,2} \right).
    \end{aligned}
    $}
    \hspace{-4mm}
\end{equation}
Following Proposition~\ref{th:clus}, we exploit the degrees of freedom in the data heterogeneity parameters by setting $\zeta^{\mathsf{Loc}}_{c,2} = 0$, which yields the minimal admissible value of $\zeta^{\mathsf{Loc}}_{c,1}$ that Assumption \ref{Assup:IntraClusterDissimilarity} holds, denoted as $\zeta^{\mathsf{Loc},\min}_{c,1}$. This bounds the last term on the right hand side of the above bound as:
\begin{equation}
    \leq \sum_{c\in \mathcal{C}^{(k)}} \frac{{D}^{(k)}_{c}}{D^{(k)}L^{(k)}} \frac{8 \eta_k^2 \beta^2
    \left(\left(L^{(k)}\right)^{2}  \left(L^{(k)}-1\right) \left(e_{c}^{(k)}\right)^{2} + \left(L^{(k)}\right) \left(e_{c}^{(k)}\right)\left(e_{c}^{(k)}-1\right)\right)}{1 - 8 \eta_k^2 \beta^2 \left(\left(L^{(k)}\right)\left(L^{(k)}-1\right) \left(e_{c}^{(k)}\right)^{2} + \left(e_{c}^{(k)}\right) \left(e_{c}^{(k)}-1\right) \right)} \zeta^{\mathsf{Loc},\min}_{c,1} \left\Vert \nabla F_{c}^{(k)}(\bm{\omega}^{(k)})\right\Vert^2.
\end{equation}
We further define
$\hat{\zeta}^{\mathsf{Loc}}_{1} \triangleq \max_{c \in \mathcal{C}^{(k)}} \big\{\zeta^{\mathsf{Loc},\min}_{c,1}\big\}$ and assume $\max_{c \in \mathcal{C}^{(k)}} \{e^{(k)}_{c}\} \leq e^{(k)}_{\mathsf{max}}$. Thus, we obtain the following bound on the above expression:
\begin{equation}
\hspace{-4mm}
\resizebox{.97\linewidth}{!}{$
\begin{aligned}
    & \leq \frac{8 \eta_k^2 \beta^2
    \left(\left(L^{(k)}\right)  \left(L^{(k)}-1\right) \left(e^{(k)}_{\mathsf{max}}\right)^{2} +  \left(e^{(k)}_{\mathsf{max}}\right)\left(e^{(k)}_{\mathsf{max}}-1\right)\right)}{1 - \left(8 \eta_k^2 \beta^2 \left(\left(L^{(k)}\right)\left(L^{(k)}-1\right) \left(e^{(k)}_{\mathsf{max}}\right)^{2} + \left(e^{(k)}_{\mathsf{max}}\right) \left(e^{(k)}_{\mathsf{max}}-1\right)\right) \right)} \left(\zeta^{\mathsf{Glob}}_{1}\hat{\zeta}^{\mathsf{Loc}}_{1}  \left\Vert \sum_{c\in \mathcal{C}^{(k)}} \frac{{D}^{(k)}_{c}}{D^{(k)}} \nabla F_{c}^{(k)}(\bm{\omega}^{(k)})\right\Vert^2 + \zeta^{\mathsf{Glob}}_{2}\hat{\zeta}^{\mathsf{Loc}}_{1} \right),
    \\& = \frac{8 \eta_k^2 \beta^2\left(\left(L^{(k)}\right)  \left(L^{(k)}-1\right) \left(e^{(k)}_{\mathsf{max}}\right)^{2} +  \left(e^{(k)}_{\mathsf{max}}\right)\left(e^{(k)}_{\mathsf{max}}-1\right)\right)}{1 - \left(8 \eta_k^2 \beta^2 \left(\left(L^{(k)}\right)\left(L^{(k)}-1\right) \left(e^{(k)}_{\mathsf{max}}\right)^{2} + \left(e^{(k)}_{\mathsf{max}}\right) \left(e^{(k)}_{\mathsf{max}}-1\right)\right) \right)} \left(\zeta^{\mathsf{Glob}}_{1}\hat{\zeta}^{\mathsf{Loc}}_{1} \left\Vert \nabla F^{(k)}(\bm{\omega}^{(k)})\right\Vert^2 + \zeta^{\mathsf{Glob}}_{2}\hat{\zeta}^{\mathsf{Loc}}_{1}\right). 
    \end{aligned}
    $}
    \hspace{-4mm}
\end{equation} 
 Replacing the above result back in \eqref{mainbeforediss}, and by gathering terms, we have:
\begin{equation}\label{eq:convmid}
\hspace{-4mm}
\resizebox{.98\linewidth}{!}{$
\begin{aligned}
    &\mathbb{E}_k \left[F^{(k)}(\bm{\omega}^{(k+1)})\right] \leq F^{(k)}(\bm{\omega}^{(k)}) + \frac{\eta_k}{2} \hspace{-2mm}\sum_{c'\in \mathcal{C}^{(k)}}\hspace{-1mm}\frac{{D}^{(k)}_{c'}L^{(k)}e^{(k)}_{c'}}{D^{(k)}} \underbrace{\left(\frac{8 \eta_k^2 \beta^2 \left(\left(L^{(k)}\right)  \left(L^{(k)}-1\right) \left(e^{(k)}_{\mathsf{max}}\right)^{2} +  \left(e^{(k)}_{\mathsf{max}}\right)\left(e^{(k)}_{\mathsf{max}}-1\right)\right)}{1 - 8 \eta_k^2 \beta^2 \left(\left(L^{(k)}\right)\left(L^{(k)}-1\right) \left(e^{(k)}_{\mathsf{max}}\right)^{2} + \left(e^{(k)}_{\mathsf{max}}\right) \left(e^{(k)}_{\mathsf{max}}-1\right) \right)} \zeta^{\mathsf{Glob}}_1\hat{\zeta}^{\mathsf{Loc}}_{1} - 1\right)}_{(h)} \left \Vert\nabla F^{(k)}(\bm{\omega}^{(k)})\right\Vert^2
    \\& + \frac{\eta_k}{2}  \vast(\sum_{c'\in \mathcal{C}^{(k)}}\frac{{D}^{(k)}_{c'}L^{(k)}e^{(k)}_{c'}}{D^{(k)}} \sum_{c\in \mathcal{C}^{(k)}} \frac{{D}^{(k)}_{c}}{D^{(k)}L^{(k)}e^{(k)}_{c}} \frac{8 \eta_k^2 \beta^2 \Theta^2 \left(  \left(L^{(k)}-1\right) \left(e_{c}^{(k)}\right)^{2} + \left(e_{c}^{(k)}\right)\left(e_{c}^{(k)}-1\right)\right)}{1 - 8 \eta_k^2 \beta^2 \left(\left(L^{(k)}\right)\left(L^{(k)}-1\right) \left(e^{(k)}_{\mathsf{max}}\right)^{2} + \left(e^{(k)}_{\mathsf{max}}\right) \left(e^{(k)}_{\mathsf{max}}-1\right) \right)} \sum_{n\in \mathcal{N}^{(k)}_c} \frac{{D}^{(k)}_{n}}{D^{(k)}_c} \sum_{\ell=1}^{L^{(k)}} \left(1- \varsigma^{(k,\ell)}_n\right)\frac{\left(\sigma^{(k)}_n\right)^2}{\varsigma^{(k,\ell)}_n D^{(k)}_n}\vast)
    \\& + \frac{\eta_k}{2} \left(\sum_{c'\in \mathcal{C}^{(k)}}\frac{{D}^{(k)}_{c'}L^{(k)}e^{(k)}_{c'}}{D^{(k)}}\frac{8 \eta_k^2 \beta^2\left(\left(L^{(k)}\right)  \left(L^{(k)}-1\right) \left(e^{(k)}_{\mathsf{max}}\right)^{2} +  \left(e^{(k)}_{\mathsf{max}}\right)\left(e^{(k)}_{\mathsf{max}}-1\right)\right)}{1 - 8 \eta_k^2 \beta^2 \left(\left(L^{(k)}\right)\left(L^{(k)}-1\right) \left(e^{(k)}_{\mathsf{max}}\right)^{2} + \left(e^{(k)}_{\mathsf{max}}\right) \left(e^{(k)}_{\mathsf{max}}-1\right)\right)}   \zeta^{\mathsf{Glob}}_{2}\hat{\zeta}^{\mathsf{Loc}}_{1} \right)
    \\&  + 2 \eta_k^2 \beta \Theta^2 \sum_{c'\in \mathcal{C}^{(k)}} \left(\frac{{D}^{(k)}_{c'}L^{(k)}e^{(k)}_{c'}}{D^{(k)}}\right)^2 \sum_{c\in \mathcal{C}^{(k)}} \left(\frac{{D}^{(k)}_{c}}{D^{(k)}L^{(k)}\sqrt{e^{(k)}_{c}}} \right)^2 \sum_{n\in \mathcal{N}^{(k)}_c} \left(\frac{{D}^{(k)}_{n}}{D^{(k)}_c} \right)^2 \sum_{\ell=1}^{L^{(k)}}\left(1- \varsigma^{(k,\ell)}_n\right)\frac{\left(\sigma^{(k)}_n\right)^2}{\varsigma^{(k,\ell)}_n D^{(k)}_n}.
    \end{aligned}
    $}
    \hspace{-4mm}
\end{equation}
We aim to make term $(h)$ negative with a proper choice of step size, for which we first impose a similar condition to our prior one (i.e. $\eta_k \leq \left(\beta \sqrt{8 \left(\left(L^{(k)}\right)\left(L^{(k)}-1\right) \left(e_{c}^{(k)}\right)^{2} + \left(e_{c}^{(k)}\right) \left(e_{c}^{(k)}-1\right)\right)}\right)^{-1},\forall c$) on the denominator of the fraction in $(h)$:
 \begin{align}
     & 1 - 8 \eta_k^2 \beta^2 \left(\left(L^{(k)}\right)\left(L^{(k)}-1\right) \left(e^{(k)}_{\mathsf{max}}\right)^{2} + \left(e^{(k)}_{\mathsf{max}}\right) \left(e^{(k)}_{\mathsf{max}}-1\right)\right) \geq 0
     \\& \Rightarrow \eta_k \leq \frac{1}{\beta \sqrt{8\left(\left(L^{(k)}\right)\left(L^{(k)}-1\right) \left(e^{(k)}_{\mathsf{max}}\right)^{2} + \left(e^{(k)}_{\mathsf{max}}\right) \left(e^{(k)}_{\mathsf{max}}-1\right)\right)}}.
 \end{align}
We further assume that there exist a set of constants $\Lambda^{(k)}$, $\forall k$, where:
\begin{equation}
\begin{aligned}
&\frac{8 \eta_k^2 \beta^2 \left(\left(L^{(k)}\right)  \left(L^{(k)}-1\right) \left(e^{(k)}_{\mathsf{max}}\right)^{2} +  \left(e^{(k)}_{\mathsf{max}}\right)\left(e^{(k)}_{\mathsf{max}}-1\right)\right)}{1 - 8 \eta_k^2 \beta^2 \left(\left(L^{(k)}\right)\left(L^{(k)}-1\right) \left(e^{(k)}_{\mathsf{max}}\right)^{2} + \left(e^{(k)}_{\mathsf{max}}\right) \left(e^{(k)}_{\mathsf{max}}-1\right)\right) }\zeta^{\mathsf{Glob}}_1\hat{\zeta}^{\mathsf{Loc}}_{1} \leq  \Lambda^{(k)}< 1
\\& \frac{1}{1 - 8 \eta_k^2 \beta^2 \left(\left(L^{(k)}\right)\left(L^{(k)}-1\right) \left(e^{(k)}_{\mathsf{max}}\right)^{2} + \left(e^{(k)}_{\mathsf{max}}\right) \left(e^{(k)}_{\mathsf{max}}-1\right)\right)}\leq \frac{ \zeta^{\mathsf{Glob}}_1\hat{\zeta}^{\mathsf{Loc}}_{1}+\Lambda^{(k)}}{ \zeta^{\mathsf{Glob}}_1\hat{\zeta}^{\mathsf{Loc}}_{1}},
\end{aligned}
\end{equation}
which can be obtained with the following condition on the step size:
\begin{align}
    \eta_k \leq \frac{1}{\beta} \sqrt{ \frac{\Lambda^{(k)}}{8\left(\zeta^{\mathsf{Glob}}_1\hat{\zeta}^{\mathsf{Loc}}_{1}+\Lambda^{(k)}\right)\left(\left(L^{(k)}\right)  \left(L^{(k)}-1\right) \left(e^{(k)}_{\mathsf{max}}\right)^{2} +  \left(e^{(k)}_{\mathsf{max}}\right)\left(e^{(k)}_{\mathsf{max}}-1\right)\right)}},
\end{align}
which is a tighter condition on the step size. By defining $\Phi^{(k)} = \dfrac{\eta_k}{2} \sum_{c'\in \mathcal{C}^{(k)}}\frac{{D}^{(k)}_{c'}L^{(k)}e^{(k)}_{c'}}{D^{(k)}}$, we can simplify \eqref{eq:convmid} as follows:
\begin{equation}
\hspace{-4mm}
\resizebox{.98\linewidth}{!}{$
 \begin{aligned}
    &\left \Vert\nabla F^{(k)}(\bm{\omega}^{(k)})\right\Vert^2 \leq \frac{F^{(k)}(\bm{\omega}^{(k)}) - \mathbb E_{k} \left[F^{(k)} (\bm{\omega}^{(k+1)})\right]}{\Phi^{(k)}\left(1-\Lambda^{(k)}\right)}
    \\& + \frac{8 \eta_k^2 \beta^2 \Theta^2}{\left(1-\Lambda^{(k)}\right)} \hspace{-1mm} \left( \sum_{c\in \mathcal{C}^{(k)}} \hspace{-1mm} \frac{{D}^{(k)}_{c}}{D^{(k)}L^{(k)}} \left(  \left(L^{(k)}-1\right) \left(e_{c}^{(k)}\right) + \left(e_{c}^{(k)}-1\right)\right) \hspace{-1mm} \left(\frac{ \zeta^{\mathsf{Glob}}_1\hat{\zeta}^{\mathsf{Loc}}_{1}+\Lambda^{(k)}}{ \zeta^{\mathsf{Glob}}_1\hat{\zeta}^{\mathsf{Loc}}_{1}}\right) \hspace{-2mm}\sum_{n\in \mathcal{N}^{(k)}_c} \hspace{-1mm} \frac{{D}^{(k)}_{n}}{D^{(k)}_c} \sum_{\ell=1}^{L^{(k)}} \left(1- \varsigma^{(k,\ell)}_n\right)\frac{\left(\sigma^{(k)}_n\right)^2}{\varsigma^{(k,\ell)}_n D^{(k)}_n}\right)
     \\& + \frac{8 \eta_k^2 \beta^2}{\left(1-\Lambda^{(k)}\right)} \hspace{-1mm}\left(\left(\left(L^{(k)}\right)  \left(L^{(k)}-1\right) \left(e^{(k)}_{\mathsf{max}}\right)^{2} +  \left(e^{(k)}_{\mathsf{max}}\right)\left(e^{(k)}_{\mathsf{max}}-1\right)\right) \left(\frac{ \zeta^{\mathsf{Glob}}_1\hat{\zeta}^{\mathsf{Loc}}_{1}+\Lambda^{(k)}}{ \zeta^{\mathsf{Glob}}_1\hat{\zeta}^{\mathsf{Loc}}_{1}}\right) \zeta^{\mathsf{Glob}}_{2}\hat{\zeta}^{\mathsf{Loc}}_{1} \right)
    \\& + \frac{8 \Theta^2 \beta \Phi^{(k)}}{\left(1-\Lambda^{(k)}\right)} \sum_{c\in \mathcal{C}^{(k)}} \left(\frac{{D}^{(k)}_{c}}{D^{(k)}L^{(k)}} \right)^2 \frac{1}{e^{(k)}_{c}} \sum_{n\in \mathcal{N}^{(k)}_c} \left(\frac{{D}^{(k)}_{n}}{D^{(k)}_c} \right)^2 \sum_{\ell=1}^{L^{(k)}}\left(1- \varsigma^{(k,\ell)}_n\right)\frac{\left(\sigma^{(k)}_n\right)^2}{\varsigma^{(k,\ell)}_n D^{(k)}_n}.
 \end{aligned}
 $}
 \hspace{-4mm}
 \end{equation}
Taking total expectation and averaging across global aggregations, and knowing that $\frac{\zeta^{\mathsf{Glob}}_1\hat{\zeta}^{\mathsf{Loc}}_{1}+\Lambda^{(k)}}{\zeta^{\mathsf{Glob}}_1\hat{\zeta}^{\mathsf{Loc}}_{1}} \leq 2$, we have:
\begin{equation}\label{befoermodeldrift}
 \begin{aligned}
    &\frac{1}{K} \sum_{k=0}^{K-1}\mathbb E\left \Vert\nabla F^{(k)}(\bm{\omega}^{(k)})\right\Vert^2 \leq \frac{1}{K} \sum_{k=0}^{K-1}\left[\frac{ \mathbb E_{k} \left[F^{(k)}(\bm{\omega}^{(k)})\right] - \mathbb E_{k} \left[F^{(k)} (\bm{\omega}^{(k+1)})\right]}{\Phi^{(k)}\left(1-\Lambda^{(k)}\right)} \right]
    \\& + \frac{1}{K} \sum_{k=0}^{K-1} \left[\frac{16 \eta_k^2 \beta^2 \Theta^2}{\left(1-\Lambda^{(k)}\right)} \sum_{c\in \mathcal{C}^{(k)}} \frac{{D}^{(k)}_{c}}{D^{(k)}L^{(k)}} \left(\left(L^{(k)}-1\right) \left(e_{c}^{(k)}\right) + \left(e_{c}^{(k)}-1\right)\right) \sum_{n\in \mathcal{N}^{(k)}_c} \frac{{D}^{(k)}_{n}}{D^{(k)}_c} \sum_{\ell=1}^{L^{(k)}} \left(1- \varsigma^{(k,\ell)}_n\right)\frac{\left(\sigma^{(k)}_n\right)^2}{\varsigma^{(k,\ell)}_n D^{(k)}_n}\right]
    \\& +\frac{1}{K} \sum_{k=0}^{K-1} \left[\frac{16 \eta_k^2 \beta^2}{\left(1-\Lambda^{(k)}\right)} \left(\left(L^{(k)}\right)  \left(L^{(k)}-1\right) \left(e^{(k)}_{\mathsf{max}}\right)^{2} +  \left(e^{(k)}_{\mathsf{max}}\right)\left(e^{(k)}_{\mathsf{max}}-1\right)\right) \zeta^{\mathsf{Glob}}_{2}\hat{\zeta}^{\mathsf{Loc}}_{1} \right]
    \\& + \frac{1}{K} \sum_{k=0}^{K-1} \left[\frac{8 \Theta^2 \beta \Phi^{(k)}}{\left(1-\Lambda^{(k)}\right)} \sum_{c\in \mathcal{C}^{(k)}} \left(\frac{{D}^{(k)}_{c}}{D^{(k)}L^{(k)}} \right)^2 \frac{1}{e^{(k)}_{c}} \sum_{n\in \mathcal{N}^{(k)}_c} \left(\frac{{D}^{(k)}_{n}}{D^{(k)}_c} \right)^2 \sum_{\ell=1}^{L^{(k)}}\left(1- \varsigma^{(k,\ell)}_n\right)\frac{\left(\sigma^{(k)}_n\right)^2}{\varsigma^{(k,\ell)}_n D^{(k)}_n}\right].
 \end{aligned}
 \end{equation}
Using the definition of model drift in Definition \ref{def:cons}, we assume that the time instances $T^{\mathsf{Init}, (k)}$ and $T^{\mathsf{Init}, (k+1)}$ denote the wall-clock time of the commencement and end of the global aggregation round $k$, and $t^{\mathsf{LT},{(k,\ell)}}$ and $t^{\mathsf{LT},{(k,\ell)}} + \tau^{\mathsf{LT},{\mathsf{max}}}$ denote commencement and ending of the local aggregation round $\ell$, respectively.
We let $ \tau^{\mathsf{LT},{(k,\ell)}}$ to denote the interval that satellites are engaged with model training at each local round $\ell$. Therefore, the idle time between round $\ell$ and $\ell+1$ where satellites are not conducting model training is given by $\Omega^{(k, \ell)}{=}\tau^{\mathsf{LT},{\mathsf{max}}}-\tau^{\mathsf{LT},{(k,\ell)}}$.  Noting that $F^{(k)}(\bm{\omega}^{(k)})$ is the loss function under which global aggregation round $k + 1$ starts from, we bound it in terms of $F^{(k-1)}(\bm{\omega}^{(k)})$, i.e., the loss under which global aggregation round $k$ concludes, as follows:
\begin{equation}
\begin{aligned}
    &  F^{(k)}(\bm{\omega}^{(k)}) = \sum_{n\in \mathcal{N}} \Bigg[\frac{|\mathcal{D}_n(T^{\mathsf{Init}, (k+1)})|}{|\mathcal{D}(T^{\mathsf{Init}, (k+1)})|} F_n(\bm{\omega}^{(k)} | \mathcal{\mathcal{D}}(T^{\mathsf{Init}, (k+1)})) - \frac{|\mathcal{D}_n(T^{\mathsf{Init}, (k+1)}-1)|}{|\mathcal{D}(T^{\mathsf{Init}, (k+1)}-1)|} F_n(\bm{\omega}^{(k)} | \mathcal{\mathcal{D}}(T^{\mathsf{Init}, (k+1)}-1)) 
    \\& + \frac{|\mathcal{D}_n(T^{\mathsf{Init}, (k+1)}-2)|}{|\mathcal{D}(T^{\mathsf{Init}, (k+1)}-2)|} F_n(\bm{\omega}^{(k)} | \mathcal{\mathcal{D}}(T^{\mathsf{Init}, (k+1)}-2)) - ... + \frac{|\mathcal{D}_n(T^{\mathsf{Init}, (k)})|}{|\mathcal{D}(T^{\mathsf{Init}, (k)})|} F_n(\bm{\omega}^{(k)} | \mathcal{\mathcal{D}}(T^{\mathsf{Init}, (k)})) \Bigg]  
    \\& \leq\sum_{\ell \in\mathcal{L}^{(k)}} \left(\tau^{\mathsf{LT},{\mathsf{max}}} - \tau^{\mathsf{LT},{(k,\ell)}}\right)  \sum_{n\in \mathcal{N}} \Delta^{(k,\ell)}_{n} + \sum_{n\in \mathcal{N}} \frac{|\mathcal{D}_n(T^{\mathsf{Init}, (k)})|}{|\mathcal{D}(T^{\mathsf{Init}, (k)})|} F_n(\bm{\omega}^{(k)} | \mathcal{\mathcal{D}}(T^{\mathsf{Init}, (k)}))
    \\& = \sum_{\ell \in\mathcal{L}^{(k)}} \Omega^{(k, \ell)}  \sum_{n\in \mathcal{N}} \Delta^{(k,\ell)}_{n} + \sum_{n\in \mathcal{N}} F_n(\bm{\omega}^{(k)} | \mathcal{D}_n^{(k-1)}) = \Omega^{(k)} \Delta^{(k)} + F^{(k-1)}(\bm{\omega}^{(k)}),  \forall k\geq 1,
   \end{aligned}
\end{equation}
where for $k=0$, we have $F^{(-1)}(\bm{\omega}^{(0)})$ denoting the loss of the initial global model, and $\Delta^{(k,\ell)}_{n} = \max_{t \in (t^{\mathsf{LT},{(k,\ell)}}+\tau^{\mathsf{LT},{(k,\ell)}},t^{\mathsf{LT},{(k,\ell)}}+\tau^{\mathsf{LT},{\mathsf{max}}}]}\Delta^{(t)}_{n}$, and $\Delta^{(k)} = \sum_{\ell \in\mathcal{L}^{(k)}} \sum_{n\in \mathcal{N}} \Delta^{(k,\ell)}_{n}$. By including model drift and idle-times effects in \eqref{befoermodeldrift}, we have:
\begin{equation}\label{app1result}
\hspace{-4mm}
\resizebox{.99\linewidth}{!}{$
 \begin{aligned}
    &\frac{1}{K} \sum_{k=0}^{K-1}\mathbb E\left \Vert\nabla F^{(k)}(\bm{\omega}^{(k)})\right\Vert^2 \leq \frac{1}{K} \left[\frac{ \mathbb E_{0} \left[F^{(0)}(\bm{\omega}^{(0)})\right] - \mathbb E_{0} \left[F^{(0)} (\bm{\omega}^{(1)})\right]}{\Phi^{(0)}\left(1-\Lambda^{(0)}\right)}+\sum_{k=1}^{K-1}\left[\frac{ \mathbb E_{k} \left[F^{(k-1)}(\bm{\omega}^{(k)})\right] - \mathbb E_{k} \left[F^{(k)} (\bm{\omega}^{(k+1)})\right]}{\Phi^{(k)}\left(1-\Lambda^{(k)}\right)} \right]  \right] + \frac{1}{K} \sum_{k=0}^{K-1} \left[\frac{\Omega^{(k)} \Delta^{(k)}}{\Phi^{(k)}\left(1-\Lambda^{(k)}\right)}\right]
    \\& + \frac{1}{K} \sum_{k=0}^{K-1} \left[\frac{16 \eta_k^2 \beta^2 \Theta^2}{\left(1-\Lambda^{(k)}\right)} \sum_{c\in \mathcal{C}^{(k)}} \frac{{D}^{(k)}_{c}}{D^{(k)}L^{(k)}} \left(\left(L^{(k)}-1\right) \left(e_{c}^{(k)}\right) + \left(e_{c}^{(k)}-1\right)\right) \sum_{n\in \mathcal{N}^{(k)}_c} \frac{{D}^{(k)}_{n}}{D^{(k)}_c} \sum_{\ell=1}^{L^{(k)}} \left(1- \varsigma^{(k,\ell)}_n\right)\frac{\left(\sigma^{(k)}_n\right)^2}{\varsigma^{(k,\ell)}_n D^{(k)}_n}\right]
    \\& +\frac{1}{K} \sum_{k=0}^{K-1} \left[\frac{16 \eta_k^2 \beta^2}{\left(1-\Lambda^{(k)}\right)} \left(\left(L^{(k)}\right)  \left(L^{(k)}-1\right) \left(e^{(k)}_{\mathsf{max}}\right)^{2} +  \left(e^{(k)}_{\mathsf{max}}\right)\left(e^{(k)}_{\mathsf{max}}-1\right)\right) \zeta^{\mathsf{Glob}}_{2}\hat{\zeta}^{\mathsf{Loc}}_{1} \right]
    \\& + \frac{1}{K} \sum_{k=0}^{K-1} \left[\frac{8 \Theta^2 \beta \Phi^{(k)}}{\left(1-\Lambda^{(k)}\right)} \sum_{c\in \mathcal{C}^{(k)}} \left(\frac{{D}^{(k)}_{c}}{D^{(k)}L^{(k)}} \right)^2 \frac{1}{e^{(k)}_{c}} \sum_{n\in \mathcal{N}^{(k)}_c} \left(\frac{{D}^{(k)}_{n}}{D^{(k)}_c} \right)^2 \sum_{\ell=1}^{L^{(k)}}\left(1- \varsigma^{(k,\ell)}_n\right)\frac{\left(\sigma^{(k)}_n\right)^2}{\varsigma^{(k,\ell)}_n D^{(k)}_n}\right].
 \end{aligned}
 $}
 \hspace{-4mm}
 \end{equation}
 Note that since the choice of initial model global aggregation is deterministic, we have $\mathbb E_{0} \left[F^{(0)}(\bm{\omega}^{(0)})\right]=F^{(0)}(\bm{\omega}^{(0)}) $. Also, at the minimizer, we have $F^{{(K)}^*}\leq F^{{(K)}}(\bm{\omega}^{(K)})$.
Ultimately, assuming $\max_{(k)} \left\{\Lambda^{(k)}\right\} \leq \Lambda_{\mathsf{max}}< 1$ and $\Phi_{\mathsf{min}} \leq \Phi_{(k)}\leq \Phi_{\mathsf{max}}$ for finite constants $\Phi_{\mathsf{min}}$ and $\Phi_{\mathsf{max}}$, and knowing $\sum_{n\in \mathcal{N}^{(k)}_c} f(n) = \sum_{n\in \mathcal{N}}\gamma^{(k)}_{c,n} f(n)$ for any function $f(n)$, by telescopic expansion of the summation in the first line of the above bound, we obtain the last result of the proof, which is:
\begin{equation}\label{app1result2}
 \begin{aligned}
    &\frac{1}{K} \sum_{k=0}^{K-1}\mathbb E\left \Vert\nabla F^{(k)}(\bm{\omega}^{(k)})\right\Vert^2 \leq  \frac{F^{(0)}(\bm{\omega}^{(0)}) - F^{(K)^{\ast}}}{K \Phi_{\mathsf{min}}\left(1-\Lambda_{\mathsf{max}}\right)}  + \frac{1}{K} \sum_{k=0}^{K-1} \left[\frac{\Omega^{(k)} \Delta^{(k)}}{\Phi_{\mathsf{min}}\left(1-\Lambda_{\mathsf{max}}\right)}\right]
    \\& + \frac{1}{K} \sum_{k=0}^{K-1} \left[\frac{16 \eta_k^2 \beta^2 \Theta^2}{\left(1-\Lambda_{\mathsf{max}}\right)} \sum_{c\in \mathcal{C}^{(k)}} \frac{{D}^{(k)}_{c}}{D^{(k)}L^{(k)}} \left(\left(L^{(k)}-1\right) \left(e_{c}^{(k)}\right) + \left(e_{c}^{(k)}-1\right)\right) \sum_{n\in \mathcal{N}} \frac{{D}^{(k)}_{n}}{D^{(k)}_c} \sum_{\ell=1}^{L^{(k)}} \left(1- \varsigma^{(k,\ell)}_n\right)\frac{\gamma^{(k)}_{c,n}\left(\sigma^{(k)}_n\right)^2}{\varsigma^{(k,\ell)}_n D^{(k)}_n}\right]
    \\& +\frac{1}{K} \sum_{k=0}^{K-1} \left[\frac{16 \eta_k^2 \beta^2}{\left(1-\Lambda_{\mathsf{max}}\right)} \left(\left(L^{(k)}\right)  \left(L^{(k)}-1\right) \left(e^{(k)}_{\mathsf{max}}\right)^{2} +  \left(e^{(k)}_{\mathsf{max}}\right)\left(e^{(k)}_{\mathsf{max}}-1\right)\right) \zeta^{\mathsf{Glob}}_{2}\hat{\zeta}^{\mathsf{Loc}}_{1} \right]
    \\& + \frac{1}{K} \sum_{k=0}^{K-1} \left[\frac{8 \Theta^2 \beta \Phi_{\mathsf{max}}}{\left(1-\Lambda_{\mathsf{max}}\right)} \sum_{c\in \mathcal{C}^{(k)}} \left(\frac{{D}^{(k)}_{c}}{D^{(k)}L^{(k)}} \right)^2 \frac{1}{e^{(k)}_{c}} \sum_{n\in \mathcal{N}} \left(\frac{{D}^{(k)}_{n}}{D^{(k)}_c} \right)^2 \sum_{\ell=1}^{L^{(k)}}\left(1- \varsigma^{(k,\ell)}_n\right)\frac{\gamma^{(k)}_{c,n}\left(\sigma^{(k)}_n\right)^2}{\varsigma^{(k,\ell)}_n D^{(k)}_n}\right].
 \end{aligned}
 \end{equation}
 \newpage
 \section{Proof of Corollary~\ref{cor:1}}\label{app:cor:1}
Considering \eqref{app1result2}, we can replace $\Phi_{\mathsf{min}} = \eta \ell_{\mathsf{min}} \overline{e}_{\mathsf{min}} / 2$ and $\Phi_{\mathsf{max}} = \eta \ell_{\mathsf{max}} \overline{e}_{\mathsf{max}} / 2$, where $\ell_{\mathsf{min}} \leq L^{(k)} \leq \ell_{\mathsf{max}}$ for finite positive constants $\ell_{\mathsf{min}}$ and $\ell_{\mathsf{max}}$ and $(\overline{e}_{\mathsf{max}})^{-1}\leq (e_{\mathsf{avg}}^{(k)})^{-1} \leq (\overline{e}_{\mathsf{min}})^{-1}$ for finite positive constants $\overline{e}_{\mathsf{min}}$ and $\overline{e}_{\mathsf{max}}$, in which $e_{\mathsf{avg}}^{(k)} = \sum_{c\in \mathcal{C}^{(k)}}\frac{{D}^{(k)}_{c} e^{(k)}_{c}}{D^{(k)}}$, as follows:
\begin{equation}
 \begin{aligned}
    &\frac{1}{K} \sum_{k=0}^{K-1}\mathbb E\left \Vert\nabla F^{(k)}(\bm{\omega}^{(k)})\right\Vert^2 \leq  \frac{F^{(0)}(\bm{\omega}^{(0)}) - F^{(K)^{\ast}}}{\left(K  \eta \ell_{\mathsf{min}} \overline{e}_{\mathsf{min}} / 2 \right)\left(1-\Lambda_{\mathsf{max}}\right)}  + \frac{1}{K} \sum_{k=0}^{K-1} \left[\frac{\Omega^{(k)} \Delta^{(k)}}{\left( \eta \ell_{\mathsf{min}} \overline{e}_{\mathsf{min}} / 2 \right)\left(1-\Lambda_{\mathsf{max}}\right)}\right]
    \\& + \frac{1}{K} \sum_{k=0}^{K-1} \left[\frac{16 \eta^2  \beta^2 \Theta^2}{\left(1-\Lambda_{\mathsf{max}}\right)} \sum_{c\in \mathcal{C}^{(k)}} \frac{{D}^{(k)}_{c}}{D^{(k)}L^{(k)}} \left(\left(L^{(k)}-1\right) \left(e_{c}^{(k)}\right) + \left(e_{c}^{(k)}-1\right)\right) \sum_{n\in \mathcal{N}} \frac{{D}^{(k)}_{n}}{D^{(k)}_c} \sum_{\ell=1}^{L^{(k)}} \left(1- \varsigma^{(k,\ell)}_n\right)\frac{\gamma^{(k)}_{c,n}\left(\sigma^{(k)}_n\right)^2}{\varsigma^{(k,\ell)}_n D^{(k)}_n}\right]
    \\& +\frac{1}{K} \sum_{k=0}^{K-1} \left[\frac{16 \eta^2 \beta^2}{\left(1-\Lambda_{\mathsf{max}}\right)} \left(\left(L^{(k)}\right)  \left(L^{(k)}-1\right) \left(e^{(k)}_{\mathsf{max}}\right)^{2} +  \left(e^{(k)}_{\mathsf{max}}\right)\left(e^{(k)}_{\mathsf{max}}-1\right)\right) \zeta^{\mathsf{Glob}}_{2}\hat{\zeta}^{\mathsf{Loc}}_{1} \right]
    \\& + \frac{1}{K} \sum_{k=0}^{K-1} \left[\frac{8 \Theta^2 \beta \left( \eta \ell_{\mathsf{max}} \overline{e}_{\mathsf{max}} / 2 \right)}{\left(1-\Lambda_{\mathsf{max}}\right)} \sum_{c\in \mathcal{C}^{(k)}} \left(\frac{{D}^{(k)}_{c}}{D^{(k)}L^{(k)}} \right)^2 \frac{1}{e^{(k)}_{c}} \sum_{n\in \mathcal{N}} \left(\frac{{D}^{(k)}_{n}}{D^{(k)}_c} \right)^2 \sum_{\ell=1}^{L^{(k)}}\left(1- \varsigma^{(k,\ell)}_n\right)\frac{\gamma^{(k)}_{c,n}\left(\sigma^{(k)}_n\right)^2}{\varsigma^{(k,\ell)}_n D^{(k)}_n}\right].
 \end{aligned}
 \end{equation}
Assuming the choice of step size as $\eta = \alpha \big /{\sqrt{ \ell_{\mathsf{max}} \widehat{e}_{\mathsf{max}} K \big /N}}$ with a finite positive constant $\alpha$, where $\widehat{e}_{\mathsf{min}} \leq e^{(k)}_{\mathsf{sum}}\leq  \widehat{e}_{\mathsf{max}}$ for finite positive constants $\widehat{e}_{\mathsf{min}}$ and $\widehat{e}_{\mathsf{max}}$, in which $e^{(k)}_{\mathsf{sum}}=\sum_{c\in \mathcal{C}^{(k)}} e^{(k)}_{c} N^{(k)}_{c}$ is  the total number of local iterations across the participating VC's, we get:
\begin{equation}\label{corr1result}
\resizebox{.99\linewidth}{!}{$
 \begin{aligned}
    &\frac{1}{K} \sum_{k=0}^{K-1}\mathbb E\left \Vert\nabla F^{(k)}(\bm{\omega}^{(k)})\right\Vert^2 \leq 2 \sqrt{ \ell_{\mathsf{max}} \widehat{e}_{\mathsf{max}}} \frac{F^{(0)}(\bm{\omega}^{(0)}) - F^{(K)^{\ast}}}{\ell_{\mathsf{min}} \overline{e}_{\mathsf{min}} \alpha \sqrt{N K} \left(1- \Lambda_{\mathsf{max}}\right)} + \frac{2 \sqrt{ \ell_{\mathsf{max}} \widehat{e}_{\mathsf{max}}}}{\ell_{\mathsf{min}} \overline{e}_{\mathsf{min}} \alpha \sqrt{N K}} \sum_{k=0}^{K-1} \left[\frac{\Omega^{(k)} \Delta^{(k)}}{1- \Lambda_{\mathsf{max}}}\right]
    \\& + \frac{1}{K} \sum_{k=0}^{K-1} \left[\frac{16 \alpha^2 \beta^2 \Theta^2 N}{ \ell_{\mathsf{min}} \widehat{e}_{\mathsf{min}} K \left(1- \Lambda_{\mathsf{max}}\right)} \hspace{-1mm} \sum_{c\in \mathcal{C}^{(k)}} \hspace{-1mm} \frac{{D}^{(k)}_{c}}{D^{(k)}L^{(k)}} \left(\left(L^{(k)}-1\right) \left(e_{c}^{(k)}\right) + \left(e_{c}^{(k)}-1\right)\right) \sum_{n\in \mathcal{N}} \frac{{D}^{(k)}_{n}}{D^{(k)}_c} \sum_{\ell=1}^{L^{(k)}} \left(1- \varsigma^{(k,\ell)}_n\right)\frac{\gamma^{(k)}_{c,n}\left(\sigma^{(k)}_n\right)^2}{\varsigma^{(k,\ell)}_n D^{(k)}_n}\right]
    \\& +\frac{1}{K} \sum_{k=0}^{K-1} \left[\frac{16 \alpha^2 \beta^2 N}{\ell_{\mathsf{min}} \widehat{e}_{\mathsf{min}} K \left(1- \Lambda_{\mathsf{max}} \right)}\left(\left(L^{(k)}\right)  \left(L^{(k)}-1\right) \left(e^{(k)}_{\mathsf{max}}\right)^{2} +  \left(e^{(k)}_{\mathsf{max}}\right)\left(e^{(k)}_{\mathsf{max}}-1\right)\right) \zeta^{\mathsf{Glob}}_{2}\hat{\zeta}^{\mathsf{Loc}}_{1} \right]
    \\& + \frac{1}{K} \sum_{k=0}^{K-1} \left[\frac{4 \ell_{\mathsf{max}}\overline{e}_{\mathsf{max}} \Theta^2 \alpha \beta \sqrt{N}}{ \sqrt{ \ell_{\mathsf{min}} \widehat{e}_{\mathsf{min}}K}\left(1- \Lambda_{\mathsf{max}} \right)} \sum_{c\in \mathcal{C}^{(k)}} \left(\frac{{D}^{(k)}_{c}}{D^{(k)}L^{(k)}} \right)^2 \frac{1}{e^{(k)}_{c}} \sum_{n\in \mathcal{N}} \left(\frac{{D}^{(k)}_{n}}{D^{(k)}_c} \right)^2 \sum_{\ell=1}^{L^{(k)}}\left(1- \varsigma^{(k,\ell)}_n\right)\frac{\gamma^{(k)}_{c,n}\left(\sigma^{(k)}_n\right)^2}{\varsigma^{(k,\ell)}_n D^{(k)}_n}\right].
 \end{aligned}
 $}
 \end{equation}
  Considering \eqref{corr1result}, assuming a bounded sampling noise $\max\limits_{k, \ell, n} \left\{\left(1- \varsigma^{(k,\ell)}_n\right)\frac{\left(\sigma^{(k)}_n\right)^2}{\varsigma^{(k,\ell)}_n D^{(k)}_n} \right\} \leq \sigma_{\mathsf{max}}$ and bounded local aggregations $\max_{k} \{e_{\mathsf{max}}^{(k)}\} \leq e_{\mathsf{max}}$, $\Delta^{(k)} \leq \left[\frac{\chi}{K\Omega^{(k)}}\right]^{+}$ (for a finite non-negative constant $\chi$), we get:
  \begin{equation}
 \begin{aligned}
    &\frac{1}{K} \sum_{k=0}^{K-1}\mathbb E\left \Vert\nabla F^{(k)}(\bm{\omega}^{(k)})\right\Vert^2 \leq 2 \sqrt{\ell_{\mathsf{max}} \widehat{e}_{\mathsf{max}}} \frac{F^{(0)}(\bm{\omega}^{(0)}) - F^{(K)^{\ast}}}{\ell_{\mathsf{min}}\overline{e}_{\mathsf{min}} \alpha \sqrt{N K} \left(1- \Lambda_{\mathsf{max}}\right)} + \frac{2 \sqrt{\ell_{\mathsf{max}} \widehat{e}_{\mathsf{max}}}}{\ell_{\mathsf{min}}\overline{e}_{\mathsf{min}} \alpha \sqrt{N K}} \sum_{k=0}^{K-1} \left[\frac{\chi}{K \left(1- \Lambda_{\mathsf{max}}\right)}\right]
    \\& + \frac{1}{K} \sum_{k=0}^{K-1} \Bigg[\frac{16 \alpha^2 \beta^2 N}{\ell_{\mathsf{min}} \widehat{e}_{\mathsf{min}} K \left(1- \Lambda_{\mathsf{max}} \right)}\left(\left(L^{(k)}\right)  \left(L^{(k)}-1\right) \left(e^{(k)}_{\mathsf{max}}\right)^{2} +  \left(e^{(k)}_{\mathsf{max}}\right)\left(e^{(k)}_{\mathsf{max}}-1\right)\right) \zeta^{\mathsf{Glob}}_{2}\hat{\zeta}^{\mathsf{Loc}}_{1}
    \\& + \frac{16 \alpha^2 \beta^2 \Theta^2 N}{\ell_{\mathsf{min}} \widehat{e}_{\mathsf{min}} K \left(1- \Lambda_{\mathsf{max}}\right)} \left(\left(L^{(k)}-1\right) \left(e^{(k)}_{\mathsf{max}}\right) + \left(e^{(k)}_{\mathsf{max}}-1\right)\right) \sigma_{\mathsf{max}} + \frac{4 \ell_{\mathsf{max}}\overline{e}_{\mathsf{max}} \Theta^2 \alpha \beta \sqrt{N}}{ \sqrt{ \ell_{\mathsf{min}} \widehat{e}_{\mathsf{min}}K}\left(1- \Lambda_{\mathsf{max}} \right)} \sigma_{\mathsf{max}} \Bigg].
 \end{aligned}
 \end{equation}
Ultimately, we obtain the final result of the proof as follows:
  \begin{equation}
 \begin{aligned}
    &\frac{1}{K} \sum_{k=0}^{K-1}\mathbb E\left \Vert\nabla F^{(k)}(\bm{\omega}^{(k)})\right\Vert^2 \leq \frac{1}{\sqrt{K} \left(1-\Lambda_{\mathsf{max}}\right)} \left[2 \sqrt{\ell_{\mathsf{max}} \widehat{e}_{\mathsf{max}}}\frac{F^{(0)}(\bm{\omega}^{(0)}) - F^{(K)^{\ast}}}{\ell_{\mathsf{min}}\overline{e}_{\mathsf{min}} \alpha \sqrt{N}} + \frac{2 \sqrt{\ell_{\mathsf{max}} \widehat{e}_{\mathsf{max}}} \chi}{\ell_{\mathsf{min}}\overline{e}_{\mathsf{min}} \alpha \sqrt{N}}\right]
    \\& + \frac{1}{K \left(1-\Lambda_{\mathsf{max}}\right)} \hspace{-0.5mm} \Bigg[\frac{16 \alpha^2 \beta^2 N}{\ell_{\mathsf{min}} \widehat{e}_{\mathsf{min}}} \left(\left(\ell_{\mathsf{max}}\right)  \left(\ell_{\mathsf{max}}-1\right) \left(e_{\mathsf{max}}\right)^{2} +  \left(e_{\mathsf{max}}\right)\left(e_{\mathsf{max}}-1\right)\right) \zeta^{\mathsf{Glob}}_{2}\hat{\zeta}^{\mathsf{Loc}}_{1} 
    \\& +\frac{16 \alpha^2 \beta^2 \Theta^2 N}{\ell_{\mathsf{min}} \widehat{e}_{\mathsf{min}}} \left(\left(\ell_{\mathsf{max}}-1\right) \left(e_{\mathsf{max}}\right) + \left(e_{\mathsf{max}}-1\right)\right) \sigma_{\mathsf{max}} + \frac{4 \ell_{\mathsf{max}}\overline{e}_{\mathsf{max}} \Theta^2 \alpha \beta \sqrt{N}}{ \sqrt{ \ell_{\mathsf{min}} \widehat{e}_{\mathsf{min}}}} \sigma_{\mathsf{max}} \Bigg].
 \end{aligned}
 \end{equation}

\newpage

\section{Proof of Proposition ~\ref{th:clus}}~\label{app:th:clus}
We previously introduced the intra- and inter-VC loss functions dissimilarity for satellites, expressed respectively as:
\begin{equation}\label{zet1}
    \sum_{n \in \mathcal{N}^{(k)}_c} a_n \left\Vert \nabla F^{(k)}_n (\bm{\omega}) \right\Vert^2 
    \leq 
    \zeta^{\mathsf{Loc}}_{c,1} \left\Vert \sum_{n \in \mathcal{N}^{(k)}_c} a_n \nabla F^{(k)}_n (\bm{\omega}) \right\Vert^2 
    + \zeta^{\mathsf{Loc}}_{c,2}, ~~~\forall c \in \mathcal{C}^{(k)}
\end{equation}
\begin{equation}\label{zet2}
    \sum_{c \in \mathcal{C}^{(k)}} b_c \left\Vert \nabla F^{(k)}_c (\bm{\omega}) \right\Vert^2 
    \leq 
    \zeta^{\mathsf{Glob}}_1 \left\Vert \sum_{c \in \mathcal{C}^{(k)}} b_c \nabla F^{(k)}_c (\bm{\omega}) \right\Vert^2 
    + \zeta^{\mathsf{Glob}}_2,
\end{equation}
subject to $\sum_{c \in \mathcal{C}^{(k)}} b_c = 1$ and $\sum_{n \in \mathcal{N}^{(k)}_c} a_n = 1$.
Without loss of generality, for inequality~\eqref{zet1}, we leverage the degrees of freedom in the constants 
$\zeta^{\mathsf{Loc}}_{c,1}$ and $\zeta^{\mathsf{Loc}}_{c,2}$ by first minimizing 
$\zeta^{\mathsf{Loc}}_{c,2}$. Since its minimum admissible value is zero, we set 
$\zeta^{\mathsf{Loc}}_{c,2}=0$. Under this choice, inequality~\eqref{zet1} reduces 
to a form in which the minimal admissible value of $\zeta^{\mathsf{Loc}}_{c,1}$, 
denoted by $\zeta^{\mathsf{Loc},\min}_{c,1}$, can be identified. To obtain 
a uniform bound across all VCs, we introduce  $\hat{\zeta}^{\mathsf{Loc}}_{1} \triangleq \max_{c \in \mathcal{C}^{(k)}} \{\zeta^{\mathsf{Loc},\min}_{c,1}\}$ which serves as a global constant valid for the entire collection of VCs.
In contrast, for inequality~\eqref{zet2}, we invert the argument: 
we first minimize $\zeta^{\mathsf{Glob}}_1$, whose smallest admissible value 
is $\zeta^{\mathsf{Glob}}_1=1$. This choice reduces inequality~\eqref{zet2}
to a form in which the minimum admissible value of $\zeta^{\mathsf{Glob}}_2$, 
denoted by $\zeta^{\mathsf{Glob}, \min}_2$, can be identified.
This leads to refined inequalities for intra- and inter-VC loss function dissimilarities, which can be shown as follows:
\begin{equation}\label{zet11}
    \sum_{n \in \mathcal{N}^{(k)}_c} a_n \left\Vert \nabla F^{(k)}_n (\bm{\omega}) \right\Vert^2 
    \leq 
    \hat{\zeta}^{\mathsf{Loc}}_{1} \left\Vert \sum_{n \in \mathcal{N}^{(k)}_c} a_n \nabla F^{(k)}_n (\bm{\omega}) \right\Vert^2,~~ \forall c \in \mathcal{C}^{(k)},
\end{equation}
\begin{equation}\label{zet22}
    \sum_{c \in \mathcal{C}^{(k)}} b_c \left\Vert \nabla F^{(k)}_c (\bm{\omega}) \right\Vert^2 
    \leq \left\Vert \sum_{c \in \mathcal{C}^{(k)}} b_c \nabla F^{(k)}_c (\bm{\omega}) \right\Vert^2 
    + \zeta^{\mathsf{Glob}, \min}_2.
\end{equation}
The inequality \eqref{zet22}, following the VC gradient definition, leads to:
\begin{equation}
    \sum_{c \in \mathcal{C}^{(k)}} b_c \left\Vert \sum_{n \in \mathcal{N}^{(k)}_c} a_n \nabla F^{(k)}_n (\bm{\omega}) \right\Vert^2 
    \leq 
     \left\Vert \sum_{c \in \mathcal{C}^{(k)}} b_c \sum_{n \in \mathcal{N}^{(k)}_c} a_n \nabla F^{(k)}_n (\bm{\omega}) \right\Vert^2 
    + \zeta^{\mathsf{Glob}, \min}_2,
\end{equation}
which, by utilizing \eqref{zet11} on its left-hand side, can be rewritten as follows:
\begin{equation}
    \frac{1}{\hat{\zeta}^{\mathsf{Loc}}_{1}}\sum_{c \in \mathcal{C}^{(k)}} b_c \sum_{n \in \mathcal{N}^{(k)}_c} a_n \left\Vert \nabla F^{(k)}_n (\bm{\omega}) \right\Vert^2 
    \leq 
     \left\Vert \sum_{c \in \mathcal{C}^{(k)}} b_c \sum_{n \in \mathcal{N}^{(k)}_c} a_n \nabla F^{(k)}_n (\bm{\omega}) \right\Vert^2 
    + \zeta^{\mathsf{Glob}, \min}_2,
\end{equation}
\begin{equation}
 \Longrightarrow \sum_{c \in \mathcal{C}^{(k)}} b_c \sum_{n \in \mathcal{N}^{(k)}_c} a_n \left\Vert \nabla F^{(k)}_n (\bm{\omega}) \right\Vert^2 
    - \hat{\zeta}^{\mathsf{Loc}}_{1} \left\Vert \sum_{c \in \mathcal{C}^{(k)}} b_c \sum_{n \in \mathcal{N}^{(k)}_c} a_n \nabla F^{(k)}_n (\bm{\omega}) \right\Vert^2 
    \leq 
      \hat{\zeta}^{\mathsf{Loc}}_{1} \zeta^{\mathsf{Glob}, \min}_2,
\end{equation}

\begin{equation} \label{deriveddis}
 \Longrightarrow \frac{ \sum_{c \in \mathcal{C}^{(k)}} b_c \sum_{n \in \mathcal{N}^{(k)}_c} a_n \left\Vert \nabla F^{(k)}_n (\bm{\omega}) \right\Vert^2 
    - \hat{\zeta}^{\mathsf{Loc}}_{1} \left\Vert \sum_{c \in \mathcal{C}^{(k)}} b_c \sum_{n \in \mathcal{N}^{(k)}_c} a_n \nabla F^{(k)}_n (\bm{\omega}) \right\Vert^2}{\hat{\zeta}^{\mathsf{Loc}}_{1}} 
    \leq 
       \zeta^{\mathsf{Glob}, \min}_2,
\end{equation}
\begin{equation} \label{deriveddis1}
 \Longrightarrow \frac{ \sum_{c \in \mathcal{C}^{(k)}} b_c \sum_{n \in \mathcal{N}^{(k)}_c} a_n \left\Vert \nabla F^{(k)}_n (\bm{\omega}) \right\Vert^2 
    - \hat{\zeta}^{\mathsf{Loc}}_{1} \left\Vert \sum_{c \in \mathcal{C}^{(k)}} b_c \sum_{n \in \mathcal{N}^{(k)}_c} a_n \nabla F^{(k)}_n (\bm{\omega}) \right\Vert^2}{\hat{\zeta}^{\mathsf{Loc}}_{1}} 
    \leq 
       \zeta^{\mathsf{Glob}}_2.
\end{equation}

\newpage
\label{gpstart}
 \section{Transforming the Problem into Geometric Programming (GP)}\label{app:optTransform}

\subsection{Useful Lemma}
 
\begin{lemma}[\textbf{Arithmetic-geometric mean inequality}~\cite{duffin1972reversed}]\label{Lemma:ArethmaticGeometric}
         Consider a posynomial function $g(\bm{y})=\sum_{i=1}^{i'} u_i(\bm{y})$, where $u_i(\bm{y})$ is a monomial, $\forall i$. The following inequality holds:
         \vspace{-2.5mm}
         \begin{equation}\label{eq:approxPosMonMain}
             g(\bm{y})\geq \hat{g}(\bm{y})\triangleq \prod_{i=1}^{i'}\left( {u_i(\bm{y})}/{\alpha_i(\bm{z})}\right)^{\alpha_i(\bm{z})},
             \vspace{-1.5mm}
         \end{equation}
         where $\alpha_i(\bm{z})=u_i(\bm{z})/g(\bm{z})$, $\forall i$, and $\bm{z}>0$ is a fixed point.
\end{lemma}

 \subsection{Problem Transformation Steps}
A standard GP is defined as minimizing a posynomial objective function subject to a set of equality constraints on monomials and inequality constraints on posynomials or monomials~\cite{chiang2005geometric}. This non-convex optimization problem can be characterized as:
\begin{equation}\label{eq:GPformat}
    \begin{aligned}
    &\min_{\bm{y}} f_0 (\bm{y})\nonumber\\
    &\textrm{s.t.} ~~~ f_i(\bm{y})\leq 1, \;\; i=1,\cdots,I,\nonumber\\
   &~~~~~~ h_l(\bm{y})=1, \;\; l=1,\cdots,L,
    \end{aligned}
\end{equation}
where  $f_i(\bm{y})=\sum_{m=1}^{M_i} d_{i,m} y_1^{\alpha^{(1)}_{i,m}} y_2^{\alpha^{(2)}_{i,m}} \cdots y_n ^{\alpha^{(n)}_{i,m}}$, $\forall i$, and $h_l(\bm{y})= d_l y_1^{\alpha^{(1)}_l} y_2^{\alpha^{(2)}_l} \cdots y_n ^{\alpha^{(n)}_l}$, $\forall l$. Due to the fact that the log-sum-exp function $f(\bm{y}) = \log \sum_{j=1}^n e^{y_j}$ is convex ($\log$ denotes the natural logarithm) with the following change of variables $z_i=\log(y_i)$, $b_{i,k}=\log(d_{i,k})$, $b_l=\log (d_l)$  the GP can be converted into the following convex format, which can be solved via standard convex solvers such as CVXPY: 
  \begin{equation}~\label{GPtoConvex}
    \begin{aligned}
    &\min_{\bm{z}} \;\log \sum_{m=1}^{M_0} e^{\left(\bm{\alpha}^{\top}_{0,m}\bm{z}+ b_{0,m}\right)}\nonumber\\
    &\textrm{s.t.} ~~~ \log \sum_{m=1}^{M_i} e^{\left(\bm{\alpha}^{\top}_{i,m}\bm{z}+ b_{i,m}\right)}\leq 0 \;\; i=1,\cdots,I,\nonumber\\
   &~~~~~~~ \bm{\alpha}_l^\top \bm{z}+b_l =0\;\; l=1,\cdots,L,
    \end{aligned}
\end{equation}
where $\bm{z}=[z_1,\cdots,z_n]^\top$, $\bm{\alpha}_{i,k}=\left[\alpha_{i,k}^{(1)},\alpha_{i,k}^{(2)}\cdots, \alpha_{i,k}^{(n)}\right]^\top$, $\forall i,k$, and $\bm{\alpha}_{l}=\left[\alpha_{l}^{(1)},a_{l}^{(2)}\cdots, \alpha_{l}^{(n)}\right]^\top$\hspace{-2mm}, $\forall l$.

Our overall objective is to transform the problem $\bm{{\mathcal{P}}}$ into GP format. This requires transforming its objective function and constraints in the aforementioned format, which we will address in the following. Our approach will involve conducting a set of corrections and approximations on the left-hand side (L.H.S) and right-hand side (R.H.S) of the constraints of the optimization problem.

\subsection{Binary Variable Approximation}\label{subsec:RelaxBinary}
Variables that are originally defined as binary can be relaxed into continuous variables by introducing quadratic (non-convex) constraints that enforce binarity. Specifically, for a binary variable $A$, we first relax it to a continuous variable  $A\in[0,1]$ and then impose the constraint $A(A-1) \leq 0,$ which ensures that $A$ can only take values in $\{0,1\}$. This relaxation allows binary decision variables to be incorporated within a continuous optimization framework. For instance, the variable $\psi_{m,m'}(t)$ is treated as a continuous variable $\psi_{m,m'}(t)\in[0,1]$, while its binarity is enforced through the following constraint:

\begin{tcolorbox}[ams align]
\psi_{m,m'}(t) (1-\psi_{m,m'}(t)) \leq 0.
\end{tcolorbox}

For brevity, we avoid presenting equivalent constraints for all binary variables as they follow the above form.

\subsection{GP Transformations}
\textbullet \hspace{2mm} \textbf{Constraint \eqref{eq:hollow}}: The constraint $\sum_{n \in \mathcal{N}} \sum_{m \in \mathcal{M}_n}\sum_{m' \in \mathcal{M}_n} \psi_{m,m'}(t) = 0$ is in the form of equality on a posynomial which is not admitted in GP. Meanwhile, the R.H.S of the constraint can only be equal to 1.\footnote{GP only admits inequality on posynomials in the form of $f(x)\leq 1$, where the posynomial is forced to be less than or equal to $1$.} Therefore, we rewrite this constrain, first with summing both side with 1, and then transforming it into two inequalities via introducing an auxiliary variable which transforms it into GP admitted format\footnote{We have used the fact that forcing an equality constraint in the form of $f(x)=1$ is equivalent to simultaneously satisfying $f(x)\leq1$ and $\frac{1}{f(x)}\leq 1$. Further, to ensure that the optimization solver has a reasonably large feasible region to find a solution, we can further transform joint satisfaction of $f(x)\leq1$ and $\frac{1}{f(x)}\leq1$ to $f(x)\leq 1$ and $\frac{A^{-1}}{f(x)}\leq 1$, where $A\geq 1$ is a coefficient, the value of which will be forced as $A \downarrow 1$ through penalizing the objective function when $A>1$.}:
\begin{equation}
\begin{aligned}
     &1 + \sum_{n \in \mathcal{N}} \sum_{m \in \mathcal{M}_n}\sum_{m' \in \mathcal{M}_n} \psi_{m,m'}(t) \leq 1,
     \\& \frac{(A_{1})^{-1}}{1 + \sum_{n \in \mathcal{N}} \sum_{m \in \mathcal{M}_n}\sum_{m' \in \mathcal{M}_n} \psi_{m,m'}(t)} \leq 1,
     \\& (A_{1})^{-1} \geq 1.
\end{aligned}
\end{equation}
Note that the fraction on the L.H.S of the above inequality represents the division of a monomial by a posynomial, which would not admit the GP format. Thus, we use Lemma \ref{Lemma:ArethmaticGeometric} to condense its denominator to a monomial as follows:
\begin{equation}
\begin{aligned}
    & H_1(\mathbf{x}) \triangleq 1 + \sum_{n \in \mathcal{N}} \sum_{m \in \mathcal{M}_n}\sum_{m' \in \mathcal{M}_n} \psi_{m,m'}(t) 
     \\& \Rightarrow H_1(\mathbf{x}) \geq \widehat{H}_1(\mathbf{x};l) \triangleq \left(H_{1}[\mathbf{x}]^{l-1}\right)^{\frac{1}{H_{1}[\mathbf{x}]^{l-1}}} \prod_{n \in \mathcal{N}} \prod_{m \in \mathcal{M}_n} \prod_{m' \in \mathcal{M}_n} \left(\frac{\psi_{m,m'}(t) H_1[\mathbf{x}]^{l-1}}{[\psi_{m,m'}(t)]^{l-1}}\right)^{\frac{[\psi_{m,m'}(t)]^{l-1}}{H_1[\mathbf{x}]^{l-1}}}.
\end{aligned}
\end{equation}
Therefore, the constraint \eqref{eq:hollow} can be represented in the GP admitted form as follows:
\begin{tcolorbox}[ams align]
    & 1 + \sum_{n \in \mathcal{N}} \sum_{m \in \mathcal{M}_n}\sum_{m' \in \mathcal{M}_n} \psi_{m,m'}(t) \leq 1,,
    \\& \frac{(A_{1})^{-1}}{\widehat{H}_1(\mathbf{x};l)} \leq 1,
    \\& (A_{1})^{-1} \geq 1.
\end{tcolorbox}

\textbullet \hspace{2mm} \textbf{Constraint \eqref{rate2}}: The constraint $\sum_{n'\in\mathcal{N}} \sum_{m'\in\mathcal{M}_{n'}}\psi_{m,m'}(t) + \psi_{m',m}(t) {\le}1, ~~\forall m\in\mathcal{M}_n, n\in\mathcal{N}$ is a standard GP admissible posynomial inequality.

\textbullet \hspace{2mm} \textbf{Constraint \eqref{rate3}}: The constraint $\sum_{m\in\mathcal{M}_n} \sum_{m'\in\mathcal{M}_{n'}} \psi_{m,m'}(t) + \psi_{m',m}(t) {\le}1, ~~\forall n,n' \in \mathcal{N}$ is a standard GP admissible posynomial inequality.

\textbullet \hspace{2mm} \textbf{Constraint \eqref{eq:min_data_rate}}: The constraint $\big(\overline{\mathfrak{R}}_{m,m'}(t) - \mathfrak{R}^{\mathsf{min}}\big) \psi_{m,m'}(t) {\geq} 0, \forall m\in\mathcal{M}_n, \forall m'\in\mathcal{M}_{n'}$ includes a subtraction that is not admitted in GP and can be rewritten as an inequality of posynomial as follows:
\begin{equation}
    \frac{\mathfrak{R}^{\mathsf{min}} \psi_{m,m'}(t)}{\overline{\mathfrak{R}}_{m,m'}(t)  \psi_{m,m'}(t)} {\leq} 1.
\end{equation}
This constraint is in GP admitted format; however, to avoid zero denominators, we sum both the nominator and denominator with a constant as follows:
\begin{equation}
    \frac{\mathfrak{R}^{\mathsf{min}} \psi_{m,m'}(t)+1}{\overline{\mathfrak{R}}_{m,m'}(t)  \psi_{m,m'}(t)+1} {\leq} 1.
\end{equation}
This transformation requires the denominator to be condensed; we do that by further utilizing the Lemma \ref{Lemma:ArethmaticGeometric} as follows:
\begin{equation}
 H_2(\mathbf{x}) \triangleq \overline{\mathfrak{R}}_{m,m'}(t)  \psi_{m,m'}(t)+1 \Rightarrow H_2(\mathbf{x}) \geq \widehat{H}_2(\mathbf{x};l) \triangleq  \left(\frac{\overline{\mathfrak{R}}_{m,m'}(t)  \psi_{m,m'}(t) H_2[\mathbf{x}]^{l-1}}{[\overline{\mathfrak{R}}_{m,m'}(t)  \psi_{m,m'}(t)]^{l-1}}\right)^{\frac{[\overline{\mathfrak{R}}_{m,m'}(t)  \psi_{m,m'}(t)]^{l-1}}{H_2[\mathbf{x}]^{l-1}}} \left(H_{2}[\mathbf{x}]^{l-1}\right)^{\frac{1}{H_{2}[\mathbf{x}]^{l-1}}}.
\end{equation}
Thus, the constraint \eqref{eq:min_data_rate} can be presented in the GP admitted format as follows:
\begin{tcolorbox}[ams align]
    \frac{\mathfrak{R}^{\mathsf{min}} \psi_{m,m'}(t) + 1}{\widehat{H}_2(\mathbf{x};l)} {\leq} 1.
\end{tcolorbox}

\textbullet \hspace{2mm} \textbf{Constraint \eqref{eq:finalRate}}: The constraint $\mathfrak{R}_{n,n'}\big(\bm{\Psi}(\mathcal{N}, t)\big) {=} \sum_{m\in\mathcal{M}_{n}} \sum_{m'\in\mathcal{M}_{n'}}\psi_{m,m'}(t)\overline{\mathfrak{R}}_{m,m'}(t)$ can first be rewritten as follows:
\begin{equation}
     \frac{\sum_{m\in\mathcal{M}_{n}} \sum_{m'\in\mathcal{M}_{n'}}  \psi_{m,m'}(t)\overline{\mathfrak{R}}_{m,m'}(t)}{\mathfrak{R}_{n,n'}\big(\bm{\Psi}(\mathcal{N}, t)\big)} = 1,
\end{equation}
which is in the form of equality on a posynomial, which is not admitted in GP format, and therefore, can be transformed into GP admitted format by introducing an auxiliary variable as follows:
\begin{equation}
\begin{aligned}
     & \frac{\sum_{m\in\mathcal{M}_{n}} \sum_{m'\in\mathcal{M}_{n'}}  \psi_{m,m'}(t)\overline{\mathfrak{R}}_{m,m'}(t) + 1}{\mathfrak{R}_{n,n'}\big(\bm{\Psi}(\mathcal{N}, t)\big) + 1} \leq 1,
     \\& \frac{(A_{2})^{-1} \left(\mathfrak{R}_{n,n'}\big(\bm{\Psi}(\mathcal{N}, t)\big) + 1\right)}{\sum_{m\in\mathcal{M}_{n}} \sum_{m'\in\mathcal{M}_{n'}}  \psi_{m,m'}(t)\overline{\mathfrak{R}}_{m,m'}(t) + 1}  \leq 1,
     \\& (A_{2})^{-1} \geq 1.
\end{aligned}
\end{equation}
Note that additions with a constant are included to avoid zero denominators. We further use Lemma \ref{Lemma:ArethmaticGeometric} to condense both equations' denominators as follows:

\begin{equation}
\begin{aligned}
& H_3(\mathbf{x}) \triangleq \sum_{m\in\mathcal{M}_{n}} \sum_{m'\in\mathcal{M}_{n'}}  \psi_{m,m'}(t)\overline{\mathfrak{R}}_{m,m'}(t) + 1
\\& \Rightarrow H_3(\mathbf{x}) \geq \widehat{H}_3(\mathbf{x};l) \triangleq \prod_{m\in\mathcal{M}_{n}} \prod_{m'\in\mathcal{M}_{n'}} \left(\frac{\psi_{m,m'}(t)\overline{\mathfrak{R}}_{m,m'}(t) H_3[\mathbf{x}]^{l-1}}{[\psi_{m,m'}(t)\overline{\mathfrak{R}}_{m,m'}(t)]^{l-1}}\right)^{\frac{[\psi_{m,m'}(t)\overline{\mathfrak{R}}_{m,m'}(t)]^{l-1}}{H_3[\mathbf{x}]^{l-1}}} \left(H_{3}[\mathbf{x}]^{l-1}\right)^{\frac{1}{H_{3}[\mathbf{x}]^{l-1}}},
\end{aligned}
\end{equation}
\begin{equation}
\begin{aligned}
\\& H_4(\mathbf{x}) \triangleq \mathfrak{R}_{n,n'}\big(\bm{\Psi}(\mathcal{N}, t)\big) + 1 \Rightarrow H_4(\mathbf{x}) \geq \widehat{H}_4(\mathbf{x};l) \triangleq \left(\frac{\mathfrak{R}_{n,n'}\big(\bm{\Psi}(\mathcal{N}, t)\big) H_4[\mathbf{x}]^{l-1}}{[\mathfrak{R}_{n,n'}\big(\bm{\Psi}(\mathcal{N}, t)\big)]^{l-1}}\right)^{\frac{[\mathfrak{R}_{n,n'}(\bm{\Psi}(\mathcal{N}, t))]^{l-1}}{H_4[\mathbf{x}]^{l-1}}} \left(H_{4}[\mathbf{x}]^{l-1}\right)^{\frac{1}{H_{4}[\mathbf{x}]^{l-1}}}.
\end{aligned}
\end{equation}
Therefore, the constraint \eqref{eq:finalRate} can be represented in the GP admitted form as follows:
\begin{tcolorbox}[ams align]
    & \frac{\sum_{m\in\mathcal{M}_{n}} \sum_{m'\in\mathcal{M}_{n'}}  \psi_{m,m'}(t)\overline{\mathfrak{R}}_{m,m'}(t) + 1}{\widehat{H}_3(\mathbf{x};l)} \leq 1,
     \\& \frac{(A_{2})^{-1}  \mathfrak{R}_{n,n'}\big(\bm{\Psi}(\mathcal{N}, t)\big)}{\widehat{H}_4(\mathbf{x};l)}  \leq 1,
     \\& (A_{2})^{-1} \geq 1.
\end{tcolorbox}

\textbullet \hspace{2mm} \textbf{Constraint \eqref{cons:CSC_4}}: The constraint $\sum_{c{\in}\mathcal{C}^{(k)}}\sum_{n\in\mathcal{N}} \gamma^{(k)}_{c,n}= N, \forall k \in \mathcal{K}$ is also in the form of equality on a posynomial and can be treated similarly to the last constraint as follows:
\begin{equation}
\begin{aligned}
     & \frac{\sum_{c \in \mathcal{C}^{(k)}}\sum_{n\in\mathcal{N}} \gamma^{(k)}_{c,n}}{N}\leq 1,
     \\& \frac{(A_{3})^{-1} N}{\sum_{c \in \mathcal{C}^{(k)}}\sum_{n\in\mathcal{N}}\gamma^{(k)}_{c,n}}  \leq 1,
     \\& (A_{3})^{-1} \geq 1.
\end{aligned}
\end{equation}
The denominator can be condensed via Lemma \ref{Lemma:ArethmaticGeometric} as follows:
\begin{equation}
H_5(\mathbf{x}) \triangleq \sum_{c \in \mathcal{C}^{(k)}}\sum_{n\in\mathcal{N}} \gamma^{(k)}_{c,n} \Rightarrow H_5(\mathbf{x}) \geq \widehat{H}_5(\mathbf{x};l) \triangleq \prod_{c \in \mathcal{C}^{(k)}} \prod_{n \in \mathcal{N}} \left(\frac{\gamma^{(k)}_{c,n} H_5[\mathbf{x}]^{l-1}}{[\gamma^{(k)}_{c,n}]^{l-1}}\right)^{\frac{[\gamma^{(k)}_{c,n}]^{l-1}}{H_5[\mathbf{x}]^{l-1}}}.
\end{equation}
This leads to the GP admitted form of constraint \eqref{cons:CSC_4} as follows:
\begin{tcolorbox}[ams align]
    & (N)^{-1}\sum_{c \in \mathcal{C}^{(k)}}\sum_{n\in\mathcal{N}} \gamma^{(k)}_{c,n}\leq 1,
    \\& \frac{(A_{3})^{-1}}{(N)^{-1} \widehat{H}_5(\mathbf{x};l)}\leq 1,
    \\& (A_{3})^{-1} \geq 1.
\end{tcolorbox}

\textbullet \hspace{2mm} \textbf{Constraint \eqref{cons:CSC_2}}: The constraint $\sum_{c{\in}\mathcal{C}^{(k)}} \gamma^{(k)}_{c,n}=1, \forall n\in\mathcal{N},~\forall k \in \mathcal{K}$ is in the form of equality on a posynomial which is not admitted in GP and we rewrite it via introducing an auxiliary variable as two inequalities to transform it into GP format\footnote{We have used the fact that forcing an equality constraint in the form of $f(x)=1$ is equivalent to simultaneously satisfying $f(x)\leq1$ and $\frac{1}{f(x)}\leq 1$. Further, to ensure that the optimization solver has a reasonably large feasible region to find a solution, we can further transform joint satisfaction of $f(x)\leq1$ and $\frac{1}{f(x)}\leq1$ to $f(x)\leq 1$ and $\frac{A^{-1}}{f(x)}\leq 1$, where $A\geq 1$ is a coefficient, the value of which will be forced as $A \downarrow 1$ through penalizing the objective function when $A>1$.}:
\begin{equation}
\begin{aligned}
     &\sum_{c \in \mathcal{C}^{(k)}} \gamma^{(k)}_{c,n}\leq 1,
     \\& \frac{(A_{4})^{-1}}{\sum_{c \in \mathcal{C}^{(k)}} \gamma^{(k)}_{c,n}} \leq 1,
     \\& (A_{4})^{-1} \geq 1.
\end{aligned}
\end{equation}
We further exploit Lemma \ref{Lemma:ArethmaticGeometric} to condense the denominator as:
\begin{equation}
H_6(\mathbf{x}) \triangleq \sum_{c \in \mathcal{C}^{(k)}} \gamma^{(k)}_{c,n} \Rightarrow H_6(\mathbf{x}) \geq \widehat{H}_6(\mathbf{x};l) \triangleq \prod_{c \in \mathcal{C}^{(k)}} \left(\frac{\gamma^{(k)}_{c,n} H_6[\mathbf{x}]^{l-1}}{[\gamma^{(k)}_{c,n}]^{l-1}}\right)^{\frac{[\gamma^{(k)}_{c,n}]^{l-1}}{H_6[\mathbf{x}]^{l-1}}}.
\end{equation}
Therefore, the constraint \eqref{cons:CSC_2} can be represented in the GP admitted form as follows:
\begin{tcolorbox}[ams align]
    & \sum_{c \in \mathcal{C}^{(k)}} \gamma^{(k)}_{c,n}\leq 1,
    \\& \frac{(A_{4})^{-1}}{\widehat{H}_6(\mathbf{x};l)} \leq 1,
    \\& (A_{4})^{-1} \geq 1.
\end{tcolorbox}

\textbullet \hspace{2mm} \textbf{Constraint \eqref{cons:CSC_1}}: The constraint $\sum_{n\in\mathcal{N}} \gamma^{(k)}_{c,n}\ge 2, \forall c\in\mathcal{C}^{(k)}$ is in the form of an inequality on a posynomial. However, the inequality is forcing the posynomial to be greater than or equal to $2$, which is not allowed in GP. We thus first transform the constraint as follows:
\begin{equation}
    \frac{2}{\sum_{n\in\mathcal{N}} \gamma^{(k)}_{c,n}}\leq 1.
\end{equation}
Note that the fraction on the L.H.S of the above inequality represents the division of a monomial on a posynomial, which would not admit the GP format. Thus, we use Lemma \ref{Lemma:ArethmaticGeometric} to condense its denominator to a monomial as follows:
\begin{equation}
     H_7(\mathbf{x}) \triangleq \sum_{n\in\mathcal{N}} \gamma^{(k)}_{c,n} \Rightarrow H_7(\mathbf{x}) \geq \widehat{H}_7(\mathbf{x};l) \triangleq \prod_{n \in \mathcal{N}} \left(\frac{\gamma^{(k)}_{c,n} H_7[\mathbf{x}]^{l-1}}{[\gamma^{(k)}_{c,n}]^{l-1}}\right)^{\frac{[\gamma^{(k)}_{c,n}]^{l-1}}{H_7[\mathbf{x}]^{l-1}}}.
\end{equation}
Therefore, the constraint \eqref{cons:CSC_1} can be represented in the GP admitted form as follows:
\begin{tcolorbox}[ams align]
\frac{2}{\widehat{H}_7(\mathbf{x};l)}\leq 1.
\end{tcolorbox}

\textbullet \hspace{2mm} \textbf{Constraint \eqref{eq:rootDef}}: The constraint
$r^{\mathsf{G}, (k)}=\sum_{n\in\mathcal{N}}  \pi_{n}^{\mathsf{G}, (k)} n$ can fisrt be rewritten as:
\begin{equation}
\frac{\sum_{n\in\mathcal{N}}  \pi_{n}^{\mathsf{G}, (k)} n}{r^{\mathsf{G}, (k)}}=1,
\end{equation}
which results in an equality on a posynomial that is a GP not admitted format. Using a similar approach to the previous constraints, we transform it as follows:
\begin{equation}
\begin{aligned}
     &\frac{\sum_{n\in\mathcal{N}}  \pi_{n}^{\mathsf{G}, (k)} n}{r^{\mathsf{G}, (k)}} \leq 1,
     \\& \frac{(A_{9})^{-1} r^{\mathsf{G}, (k)}}{\sum_{n\in\mathcal{N}}  \pi_{n}^{\mathsf{G}, (k)} n} \leq 1,
     \\& (A_{9})^{-1} \geq 1.
\end{aligned}
\end{equation}
Similar to previous constants, we use Lemma \ref{Lemma:ArethmaticGeometric} to condense the denominator as:
\begin{equation}
H_{14}(\mathbf{x}) \triangleq \sum_{n\in\mathcal{N}}  \pi_{n}^{\mathsf{G}, (k)}n  \Rightarrow H_{14}(\mathbf{x}) \geq \widehat{H}_{14}(\mathbf{x};l) \triangleq \prod_{n \in \mathcal{N}} \left(\frac{\pi_{n}^{\mathsf{G}, (k)} n H_{14}[\mathbf{x}]^{l-1}}{[\pi_{n}^{\mathsf{G}, (k)}n]^{l-1}}\right)^{\frac{[\pi_{n}^{\mathsf{G}, (k)}n]^{l-1}}{H_{14}[\mathbf{x}]^{l-1}}}.
\end{equation}
Thus, the constraint \eqref{eq:rootDef} can be shown in the GP admitted form as follows:
\begin{tcolorbox}[ams align]
    & \frac{\sum_{n\in\mathcal{N}}  \pi_{n}^{\mathsf{G}, (k)} n}{r^{\mathsf{G}, (k)}} \leq 1,
    \\& \frac{(A_{9})^{-1} r^{\mathsf{G}, (k)}}{\widehat{H}_{14}(\mathbf{x};l)} \leq 1,
    \\& (A_{9})^{-1} \geq 1.
\end{tcolorbox}

\textbullet \hspace{2mm} \textbf{Constraint \eqref{cons:SGD_1}}: The constraint $\sum_{n\in\mathcal{N}}  \pi_{n}^{\mathsf{G}, (k)} =1$ is also in the form of equality on a posynomial and to achieve a GP admitted format a similar approach can be exploited:
\begin{equation}
\begin{aligned}
     &\sum_{n\in\mathcal{N}} \pi_{n}^{\mathsf{G}, (k)}\leq 1,
     \\& \frac{(A_{10})^{-1}}{\sum_{n\in\mathcal{N}} \pi_{n}^{\mathsf{G}, (k)}} \leq 1,
     \\& (A_{10})^{-1} \geq 1.
\end{aligned}
\end{equation}
We further exploit Lemma \ref{Lemma:ArethmaticGeometric} to condense the denominator as:
\begin{equation}
H_{15}(\mathbf{x}) \triangleq \sum_{n\in\mathcal{N}} \pi_{n}^{\mathsf{G}, (k)} \Rightarrow H_{15}(\mathbf{x}) \geq \widehat{H}_{15}(\mathbf{x};l) \triangleq \prod_{n \in \mathcal{N}} \left(\frac{\pi_{n}^{\mathsf{G}, (k)} H_{15}[\mathbf{x}]^{l-1}}{[\pi_{n}^{\mathsf{G}, (k)}]^{l-1}}\right)^{\frac{[\pi_{n}^{\mathsf{G}, (k)}]^{l-1}}{H_{15}[\mathbf{x}]^{l-1}}}.
\end{equation}
Therefore, the constraint \eqref{cons:SGD_1} can be represented in the GP admitted form as follows:
\begin{tcolorbox}[ams align]
    &\sum_{n\in\mathcal{N}} \pi_{n}^{(k)}\leq 1,
    \\& \frac{(A_{10})^{-1}}{\widehat{H}_{15}(\mathbf{x};l)} \leq 1,
    \\& (A_{10})^{-1} \geq 1.
\end{tcolorbox}

\textbullet \hspace{2mm} \textbf{Constraint \eqref{eq:rootDef2}}: The constraint $r_{c}^{\mathsf{L}, (k)}=\sum_{n{\in}\mathcal{N}} \pi_{n}^{\mathsf{L}, (k)} \gamma^{(k)}_{c,n} \cdot n, \forall c \in \mathcal{C}^{(k)}$ can fisrt be rewritten as:
\begin{equation}
\frac{\sum_{n{\in}\mathcal{N}} \pi_{n}^{\mathsf{L}, (k)}  \gamma^{(k)}_{c,n}  n}{r_{c}^{\mathsf{L}, (k)}}=1,
\end{equation}
which is an equality on a posynomial, which is not admitted in GP. We transform it as below:
\begin{equation}
\begin{aligned}
     &\frac{\sum_{n{\in}\mathcal{N}} \pi_{n}^{\mathsf{L}, (k)}  \gamma^{(k)}_{c,n}  n}{r_{c}^{\mathsf{L}, (k)}} \leq 1,
     \\& \frac{(A_{11})^{-1} r_{c}^{\mathsf{L}, (k)}}{\sum_{n{\in}\mathcal{N}} \pi_{n}^{\mathsf{L}, (k)}  \gamma^{(k)}_{c,n}  n} \leq 1,
     \\& (A_{11})^{-1} \geq 1.
\end{aligned}
\end{equation}
We utilize Lemma \ref{Lemma:ArethmaticGeometric} in order to condense the denominator as:
\begin{equation}
H_{16}(\mathbf{x}) \triangleq \sum_{n{\in}\mathcal{N}} \pi_{n}^{\mathsf{L}, (k)}  \gamma^{(k)}_{c,n}  n  \Rightarrow H_{16}(\mathbf{x}) \geq \widehat{H}_{16}(\mathbf{x};l) \triangleq \prod_{n \in \mathcal{N}} \left(\frac{\pi_{n}^{\mathsf{L}, (k)}  \gamma^{(k)}_{c,n}  n H_{16}[\mathbf{x}]^{l-1}}{[\pi_{n}^{\mathsf{L}, (k)}  \gamma^{(k)}_{c,n}  n]^{l-1}}\right)^{\frac{[\pi_{n}^{\mathsf{L}, (k)}  \gamma^{(k)}_{c,n}  n]^{l-1}}{H_{16}[\mathbf{x}]^{l-1}}}.
\end{equation}
Thus, the constraint \eqref{eq:rootDef} can be shown in the GP admitted form as follows:
\begin{tcolorbox}[ams align]
    & \frac{\sum_{n{\in}\mathcal{N}} \pi_{n}^{\mathsf{L}, (k)}  \gamma^{(k)}_{c,n}  n}{r_{c}^{\mathsf{L}, (k)}} \leq 1,
    \\& \frac{(A_{11})^{-1} r_{c}^{\mathsf{L}, (k)}}{\widehat{H}_{16}(\mathbf{x};l)} \leq 1,
    \\& (A_{11})^{-1} \geq 1.
\end{tcolorbox}

\textbullet \hspace{2mm} \textbf{Constraint \eqref{cons:SCA_1}}: The constraint $\sum_{n{\in}\mathcal{N}^{(k)}_{c}} \pi_{c,n}^{(k,\ell)} =1$ is also alike to the previous constraint and a similar approach can be driven to achieve a GP admitted format as follows:
\begin{equation}
\begin{aligned}
     &\sum_{n{\in}\mathcal{N}^{(k)}_{c}} \pi_{c,n}^{(k,\ell)}\leq 1,
     \\& \frac{(A_{12})^{-1}}{\sum_{n{\in}\mathcal{N}^{(k)}_{c}} \pi_{c,n}^{(k,\ell)}} \leq 1,
     \\& (A_{12})^{-1} \geq 1.
\end{aligned}
\end{equation}
We further exploit Lemma \ref{Lemma:ArethmaticGeometric} to condense the denominator as:
\begin{equation}
H_{17}(\mathbf{x}) \triangleq \sum_{n{\in}\mathcal{N}^{(k)}_{c}} \pi_{c,n}^{(k,\ell)} \Rightarrow H_{17}(\mathbf{x}) \geq \widehat{H}_{17}(\mathbf{x};l) \triangleq \prod_{n \in \mathcal{N}^{(k)}_{(c)}} \left(\frac{\pi_{c,n}^{(k,\ell)} H_{17}[\mathbf{x}]^{l-1}}{[\pi_{c,n}^{(k,\ell)}]^{l-1}}\right)^{\frac{[\pi_{c,n}^{(k,\ell)}]^{l-1}}{H_{17}[\mathbf{x}]^{l-1}}}.
\end{equation}
Therefore, the constraint \eqref{cons:SCA_1} can be represented in the GP admitted form as follows:
\begin{tcolorbox}[ams align]
    &\sum_{n{\in}\mathcal{N}^{(k)}_{c}} \pi_{c,n}^{(k,\ell)}\leq 1,
    \\& \frac{(A_{12})^{-1}}{\widehat{H}_{17}(\mathbf{x};l)} \leq 1,
    \\& (A_{12})^{-1} \geq 1.
\end{tcolorbox}

\textbullet \hspace{2mm} \textbf{Constraint \eqref{cons:SGD_2}}: The constraint
$\sum_{n\in\mathcal{N}}\bm{\Gamma}^{\mathsf{GD}, (k)}_{n,n'} + \pi_{n'}^{\mathsf{G}, (k)}{=}1,~\forall n'\in\mathcal{N}$ is in the form of equality on posynomia;. Thus, it can be transformed to admit the GP format via the following steps:
\begin{equation}
\begin{aligned}
    &\sum_{n\in\mathcal{N}}\bm{\Gamma}^{\mathsf{GD}, (k)}_{n,n'} + \pi_{n'}^{\mathsf{G}, (k)} \leq 1,
    \\& \frac{(A_{13})^{-1}}{\sum_{n\in\mathcal{N}}\bm{\Gamma}^{\mathsf{GD}, (k)}_{n,n'} + \pi_{n'}^{\mathsf{G}, (k)}} \leq 1,
    \\& (A_{13})^{-1} \geq 1.
\end{aligned}
\end{equation}
The Lemma \ref{Lemma:ArethmaticGeometric} is used for the denominator condensation as follows:
\begin{equation}
H_{18}(\mathbf{x}) \triangleq \sum_{n\in\mathcal{N}}\bm{\Gamma}^{\mathsf{GD}, (k)}_{n,n'} + \pi_{n'}^{\mathsf{G}, (k)} \Rightarrow H_{18}(\mathbf{x}) \geq \widehat{H}_{18}(\mathbf{x};l) \triangleq \prod_{n\in\mathcal{N}} \left(\frac{\bm{\Gamma}^{\mathsf{GD}, (k)}_{n,n'} H_{18}[\mathbf{x}]^{l-1}}{[\bm{\Gamma}^{\mathsf{GD}, (k)}_{n,n'}]^{l-1}}\right)^{\frac{[\bm{\Gamma}^{\mathsf{GD}, (k)}_{n,n'}]^{l-1}}{H_{18}[\mathbf{x}]^{l-1}}} \left(\frac{\pi_{n'}^{\mathsf{G}, (k)} H_{18}[\mathbf{x}]^{l-1}}{[\pi_{n'}^{\mathsf{G}, (k)}]^{l-1}}\right)^{\frac{[\pi_{n'}^{\mathsf{G}, (k)}]^{l-1}}{H_{18}[\mathbf{x}]^{l-1}}}.
\end{equation}
Therefore, the constraint \eqref{cons:SGD_2} can be presented as GP format as follows:
\begin{tcolorbox}[ams align]
    &\sum_{n\in\mathcal{N}}\bm{\Gamma}^{\mathsf{GD}, (k)}_{n,n'} + \pi_{n'}^{\mathsf{G}, (k)} \leq 1,
    \\& \frac{(A_{13})^{-1}}{\widehat{H}_{18}(\mathbf{x};l)} \leq 1,
    \\& (A_{13})^{-1} \geq 1.
\end{tcolorbox}

\textbullet \hspace{2mm} \textbf{Constraints \eqref{eq:GM_dispatching_transmission_latency},\eqref{eq:GM_dispatching_latency},\eqref{cons:GM_dispatch_latency_1}}: For the constraint \eqref{cons:GM_dispatch_latency_1} as $\tau^{\mathsf{GD},(k)}\leq \tau^{\mathsf{GD},\mathsf{max}},\forall k \in\mathcal{K}$, via knowing that $\tau^{\mathsf{GD},(k)}=\max_{n'\in\mathcal{N}} \big\{ \tau_{n'}^{\mathsf{GD}, (k)}\big\}$, we utilize the constraint \eqref{eq:GM_dispatching_latency} as $\tau_{n'}^{\mathsf{GD}, (k)}{=}\sum_{n\in\mathcal{N}\setminus\{n'\}} \bm{\Gamma}^{\mathsf{GD}, (k)}_{n, n'} \tau_{n}^{\mathsf{GD}, (k)} +\tau_{n, n'}^{\mathsf{GD}, (k)}$, and also include \eqref{eq:GM_dispatching_transmission_latency}, in which $ \tau_{n,n'}^{\mathsf{GD}, (k)} = \bm{\Gamma}^{\mathsf{GD}, (k)}_{n,n'}\alpha^{\mathsf{Bit}} M^{\mathsf{Dim}}\big/{\mathfrak{R}^{\mathsf{GD}, (k)}_{n,n'}}$, which leads to the following formulation:
\begin{equation}
\begin{aligned}
    & \max_{n'\in\mathcal{N}} 
      \Big\{ \tau_{n'}^{\mathsf{GD}, (k)} \Big\} \leq \tau^{\mathsf{GD},\mathsf{max}},
    \\& \tau_{n'}^{\mathsf{GD}, (k)} = 
      \sum_{n\in\mathcal{N}\setminus\{n'\}}
       \bm{\Gamma}^{\mathsf{GD}, (k)}_{n, n'} 
      \tau_{n}^{\mathsf{GD}, (k)} 
      + \tau_{n,n'}^{\mathsf{GD}, (k)},
    \\& \tau_{n,n'}^{\mathsf{GD}, (k)} = \bm{\Gamma}^{\mathsf{GD}, (k)}_{n,n'}\alpha^{\mathsf{Bit}} M^{\mathsf{Dim}}\big/{\mathfrak{R}^{\mathsf{GD}, (k)}_{n,n'}}
\end{aligned}
\end{equation}
The first expression is constructed with a maximum function, which is not permitted in the GP format. Therefore, we use the following transformation, $\max\{A, B\}\approx (A^{p}+B^{p})^{\frac{1}{p}}$, where $p$ is a large positive constant. This transformation allows us to replace the maximum term with a GP-compatible form. Accordingly, the first expression can be reformulated as follows:
\begin{equation}
    \left(\sum_{n'\in\mathcal{N}} 
      \left( \tau_{n'}^{\mathsf{GD}, (k)} \right)^{p}\right)^{1/p}\leq \tau^{\mathsf{GD},\mathsf{max}},
\end{equation}
This expression is still not in the acceptable form of GP. We take a nested approach to condense this expression by introducing $Q^{(k)}_{(1)}(\mathbf{x})$ whcih is assumed as follows:
\begin{equation}
    Q^{(k)}_{(1)}(\mathbf{x}) = \sum_{n'\in\mathcal{N}} 
      \left( \tau_{n'}^{\mathsf{GD}, (k)} \right)^{p},
\end{equation}
which can be presented as follows to admit the GP format:
\begin{equation}
\begin{aligned}
    &\frac{Q^{(k)}_{(1)}(\mathbf{x})}{\sum_{n'\in\mathcal{N}} 
      \left( \tau_{n}^{\mathsf{GD}, (k)} \right)^{p}} \leq 1,
    \\& \frac{(A_{14})^{-1} \sum_{n'\in\mathcal{N}} 
      \left( \tau_{n}^{\mathsf{GD}, (k)} \right)^{p}}{Q^{(k)}_{(1)}(\mathbf{x})} \leq 1,
    \\& (A_{14})^{-1} \geq 1.
\end{aligned}
\end{equation}
This requires the Lemma \ref{Lemma:ArethmaticGeometric} to condense the denominator and transform it into GP admitted format as follows:
\begin{equation}
H_{19}(\mathbf{x}) \triangleq \sum_{n'\in\mathcal{N}} 
      \left( \tau_{n'}^{\mathsf{GD}, (k)} \right)^{p} \Rightarrow H_{19}(\mathbf{x}) \geq \widehat{H}_{19}(\mathbf{x};l) \triangleq \prod_{n'\in\mathcal{N}} \left(\frac{\left( \tau_{n'}^{\mathsf{GD}, (k)} \right)^{p} H_{19}[\mathbf{x}]^{l-1}}{\Bigg[\left( \tau_{n'}^{\mathsf{GD}, (k)} \right)^{p}\Bigg]^{l-1}}\right)^{\frac{\Bigg[\left( \tau_{n'}^{\mathsf{GD}, (k)} \right)^{p}\Bigg]^{l-1}}{H_{19}[\mathbf{x}]^{l-1}}}.
\end{equation}
Therefore, the first expression can be transformed into the following format:
\begin{equation}
\frac{\left(Q^{(k)}_{(1)}(\mathbf{x})\right)^{1/p}}{\tau^{\mathsf{GD},\mathsf{max}}} \leq 1.
\end{equation}
We also transform the second expression, which is in the form of an equality on a posynomial, as follows:
\begin{equation}
\begin{aligned}
& \frac{\sum_{n\in\mathcal{N}\setminus\{n'\}}\bm{\Gamma}^{\mathsf{GD}, (k)}_{n, n'} \tau_{n}^{\mathsf{GD}, (k)} + \tau_{n,n'}^{\mathsf{GD}, (k)}}{\tau_{n'}^{\mathsf{GD}, (k)}}\leq 1,
\\& \frac{(A_{15})^{-1} \tau_{n'}^{\mathsf{GD}, (k)}}{\sum_{n\in\mathcal{N}\setminus\{n'\}}\bm{\Gamma}^{\mathsf{GD}, (k)}_{n, n'}\tau_{n}^{\mathsf{GD}, (k)}  + \tau_{n,n'}^{\mathsf{GD}, (k)}}\leq 1,
\\& (A_{15})^{-1} \geq 1.
\end{aligned}
\end{equation}
Similarly, Lemma \ref{Lemma:ArethmaticGeometric} is used for denominator as follows:
\begin{equation}
\begin{aligned}
    & H_{20}(\mathbf{x}) \triangleq \sum_{n\in\mathcal{N}\setminus\{n'\}}\bm{\Gamma}^{\mathsf{GD}, (k)}_{n, n'} \tau_{n}^{\mathsf{GD}, (k)}  + \tau_{n,n'}^{\mathsf{GD}, (k)} 
\\& \Rightarrow H_{20}(\mathbf{x}) \geq \widehat{H}_{20}(\mathbf{x};l) \triangleq  \prod_{n\in\mathcal{N}\setminus\{n'\}}\left(\frac{\bm{\Gamma}^{\mathsf{GD}, (k)}_{n, n'}\tau_{n}^{\mathsf{GD}, (k)} H_{20}[\mathbf{x}]^{l-1}}{[\bm{\Gamma}^{\mathsf{GD}, (k)}_{n, n'}\tau_{n}^{\mathsf{GD}, (k)}]^{l-1}}\right)^{\frac{[\bm{\Gamma}^{\mathsf{GD}, (k)}_{n, n'}\tau_{n}^{\mathsf{GD}, (k)}]^{l-1}}{H_{20}[\mathbf{x}]^{l-1}}} \left(\frac{\tau_{n,n'}^{\mathsf{GD}, (k)} H_{20}[\mathbf{x}]^{l-1}}{[\tau_{n,n'}^{\mathsf{GD}, (k)}]^{l-1}}\right)^{\frac{[\tau_{n,n'}^{\mathsf{GD}, (k)}]^{l-1}}{H_{20}[\mathbf{x}]^{l-1}}}.
\end{aligned}
\end{equation}
The third expression is also in the form of an equality on a monomial, which is accepted in GP format; however, the nominator and denominator are summed with a constant to avoid a zero denominator, which leads to the following expression:
\begin{equation}
\begin{aligned}
        & \frac{\mathfrak{R}^{\mathsf{GD}, (k)}_{n,n'}\tau_{n,n'}^{\mathsf{GD}, (k)} + 1}{\bm{\Gamma}^{\mathsf{GD}, (k)}_{n,n'}\alpha^{\mathsf{Bit}} M^{\mathsf{Dim}} + 1} \leq 1,
        \\& \frac{(A_{16})^{-1} (\bm{\Gamma}^{\mathsf{GD}, (k)}_{n,n'}\alpha^{\mathsf{Bit}} M^{\mathsf{Dim}} + 1)}{\mathfrak{R}^{\mathsf{GD}, (k)}_{n,n'}\tau_{n,n'}^{\mathsf{GD}, (k)} + 1} \leq 1,
        \\& (A_{16})^{-1} \geq 1.
\end{aligned}
\end{equation}
The Lemma \ref{Lemma:ArethmaticGeometric} is used for denominators as follows:
\begin{equation}
 H_{21}(\mathbf{x}) \triangleq \bm{\Gamma}^{\mathsf{GD}, (k)}_{n,n'}\alpha^{\mathsf{Bit}} M^{\mathsf{Dim}} + 1 \Rightarrow H_{21}(\mathbf{x}) \geq \widehat{H}_{21}(\mathbf{x};l) \triangleq  \left(\frac{\bm{\Gamma}^{\mathsf{GD}, (k)}_{n,n'} H_{21}[\mathbf{x}]^{l-1}}{[\bm{\Gamma}^{\mathsf{GD}, (k)}_{n,n'}]^{l-1}}\right)^{\frac{[\bm{\Gamma}^{\mathsf{GD}, (k)}_{n,n'}\alpha^{\mathsf{Bit}} M^{\mathsf{Dim}}]^{l-1}}{H_{21}[\mathbf{x}]^{l-1}}} \left(H_{21}[\mathbf{x}]^{l-1}\right)^{\frac{1}{H_{21}[\mathbf{x}]^{l-1}}},
\end{equation}
\begin{equation}
 H_{22}(\mathbf{x}) \triangleq \mathfrak{R}^{\mathsf{GD}, (k)}_{n,n'}\tau_{n,n'}^{\mathsf{GD}, (k)} + 1 \Rightarrow H_{22}(\mathbf{x}) \geq \widehat{H}_{22}(\mathbf{x};l) \triangleq  \left(\frac{\mathfrak{R}^{\mathsf{GD}, (k)}_{n,n'}\tau_{n,n'}^{\mathsf{GD}, (k)} H_{22}[\mathbf{x}]^{l-1}}{[\mathfrak{R}^{\mathsf{GD}, (k)}_{n,n'}\tau_{n,n'}^{\mathsf{GD}, (k)}]^{l-1}}\right)^{\frac{[\mathfrak{R}^{\mathsf{GD}, (k)}_{n,n'}\tau_{n,n'}^{\mathsf{GD}, (k)}]^{l-1}}{H_{22}[\mathbf{x}]^{l-1}}} \left(H_{22}[\mathbf{x}]^{l-1}\right)^{\frac{1}{H_{22}[\mathbf{x}]^{l-1}}}.
\end{equation}
Therefore, constraints \eqref{eq:GM_dispatching_transmission_latency},\eqref{eq:GM_dispatching_latency},\eqref{cons:GM_dispatch_latency_1} can be presented in GP format as follows:
\begin{tcolorbox}[ams align]
    &\frac{\left(Q^{(k)}_{(1)}(\mathbf{x})\right)^{1/p}}{\tau^{\mathsf{GD},\mathsf{max}}} \leq 1
    \\& \frac{Q^{(k)}_{(1)}(\mathbf{x})}{\widehat{H}_{19}(\mathbf{x};l)} \leq 1,
    \\& \frac{(A_{14})^{-1} \sum_{n'\in\mathcal{N}} 
      \left( \tau_{n}^{\mathsf{GD}, (k)} \right)^{p}}{Q^{(k)}_{(1)}(\mathbf{x})} \leq 1,
    \\& \frac{\sum_{n\in\mathcal{N}\setminus\{n'\}}\bm{\Gamma}^{\mathsf{GD}, (k)}_{n, n'} \tau_{n}^{\mathsf{GD}, (k)} + \tau_{n,n'}^{\mathsf{GD}, (k)}}{\tau_{n'}^{\mathsf{GD}, (k)}}\leq 1,
    \\& \frac{(A_{15})^{-1} \tau_{n'}^{\mathsf{GD}, (k)}}{\widehat{H}_{20}(\mathbf{x};l)}\leq 1,
    \\& \frac{\mathfrak{R}^{\mathsf{GD}, (k)}_{n,n'}\tau_{n,n'}^{\mathsf{GD}, (k)} + 1}{\widehat{H}_{21}(\mathbf{x};l)} \leq 1,
    \\& \frac{(A_{16})^{-1} (\bm{\Gamma}^{\mathsf{GD}, (k)}_{n,n'}\alpha^{\mathsf{Bit}} M^{\mathsf{Dim}} + 1)}{\widehat{H}_{22}(\mathbf{x};l)} \leq 1,
    \\& (A_{14})^{-1} \geq 1,
    \\& (A_{15})^{-1} \geq 1,
    \\& (A_{16})^{-1} \geq 1.
\end{tcolorbox}

\textbullet \hspace{2mm} \textbf{Constraint \eqref{cons:SGD_4}}:
The constraint
$\sum_{t=1}^{\tau^{\mathsf{GD},(k)}} \Big(\bm{\Gamma}^{\mathsf{GD}, (k)} - \bm{\Gamma}\big(\bm{\Psi}(\mathcal{N},t^{\mathsf{GD},(k)}{+}t), r^{\mathsf{G}, (k)}\big)\Big) = \mathbf{0}$
is an equality on a posynomial with a negative sign. We first eliminate the negative sign as follows:
\begin{equation}
\sum_{t=1}^{\tau^{\mathsf{GD}, (k)}}
\bm{\Gamma}\big(\bm{\Psi}(\mathcal{N},t^{\mathsf{GD},(k)}{+}t), r^{\mathsf{G}, (k)}\big)
= \sum_{t=1}^{\tau^{\mathsf{GD}, (k)}} \bm{\Gamma}^{\mathsf{GD}, (k)},
\end{equation}
which is an equality on matrices, which is not applicable in GP format. We therefore make use of the Frobenius norm \cite{bottcher2008frobenius} defined as $|\mathbf{A}|_{F} = \sqrt{\sum_{i=1}^{m}\sum_{j=1}^{n} |a_{ij}|^{2}}$,
for any matrix $\mathbf{A} \in \mathbb{R}^{m \times n}$ with entries $a_{ij}$. Considering two matrices $\mathbf{A}, \mathbf{B} \in \mathbb{R}^{m \times n}$, the equality $\mathbf{A} = \mathbf{B}$ can equivalently be expressed as $ |\mathbf{A} - \mathbf{B}|_{F} = \sqrt{\sum_{i=1}^{m}\sum_{j=1}^{n} |a_{ij}-b_{ij}|^{2}} = 0$. We also know that when $\mathbf{A}, \mathbf{B}$ are real matrices we have $|a_{ij} - b_{i,j}|^{2} = (a_{ij} - b_{ij})^{2}$, thus, we can formulate real matrix equality constraints as $\sum_{i=1}^{m}\sum_{j=1}^{n} a_{ij}^{2} + b_{ij}^{2} = \sum_{i=1}^{m}\sum_{j=1}^{n} 2a_{ij}b_{ij}$. Applying the same principle to our problem, the matrix equality constraint can be written as follows:
\begin{equation}
\begin{aligned}
& \sum_{i=1}^{m}\sum_{j=1}^{n} \left[
\left( \tau^{\mathsf{GD}, (k)} \Gamma^{\mathsf{GD}, (k)}_{i,j}\right)^{2}
+ \left(\sum_{t=1}^{\tau^{\mathsf{GD}, (k)}} \Gamma \big(\bm{\Psi}(\mathcal{N},t^{\mathsf{GD},(k)}{+}t), r^{\mathsf{G}, (k)}\big)_{i,j}\right)^{2} \right]
\\ & = 2 \tau^{\mathsf{GD}, (k)} \sum_{i=1}^{m}\sum_{j=1}^{n} \Gamma^{\mathsf{GD}, (k)}_{i,j}  \sum_{t=1}^{\tau^{\mathsf{GD}, (k)}} \Gamma \big(\bm{\Psi}(\mathcal{N},t^{\mathsf{GD},(k)}{+}t), r^{\mathsf{G}, (k)}\big)_{i,j}.
\end{aligned}
\end{equation}
This reformulation eliminates the direct matrix equality and expresses the constraint as a valid equality over scalar sums. However, it requires a nested approach to comply with the GP format as follows:
\begin{equation}
  P^{\mathsf{GD},(k)}_{i,j}(\mathbf{x}) = \sum_{t=1}^{\tau^{\mathsf{GD}, (k)}} \Gamma \big(\bm{\Psi}(\mathcal{N},t^{\mathsf{GD},(k)}{+}t), r^{\mathsf{G}, (k)}\big)_{i,j},
\end{equation}
which then can be presented in the GP format as follows:
\begin{equation}
\begin{aligned}
  & \frac{\sum_{t=1}^{\tau^{\mathsf{GD}, (k)}} \Gamma \big(\bm{\Psi}(\mathcal{N},t^{\mathsf{GD},(k)}{+}t), r^{\mathsf{G}, (k)}\big)_{i,j}}{ P^{\mathsf{GD},(k)}_{i,j}(\mathbf{x})} \leq 1,
  \\& \frac{(A_{17})^{-1}  P^{\mathsf{GD},(k)}_{i,j}(\mathbf{x})}{\sum_{t=1}^{\tau^{\mathsf{GD}, (k)}} \Gamma \big(\bm{\Psi}(\mathcal{N},t^{\mathsf{GD},(k)}{+}t), r^{\mathsf{G}, (k)}\big)_{i,j}} \leq 1,
  \\& (A_{17})^{-1} \geq 1,
\end{aligned}
\end{equation}
where the denominator can be condensed as via Lemma~\ref{Lemma:ArethmaticGeometric} as follows:
\begin{equation}
\begin{aligned}
& H_{23}(\mathbf{x}) \triangleq
\sum_{t=1}^{\tau^{\mathsf{GD}, (k)}} \Gamma \big(\bm{\Psi}(\mathcal{N},t^{\mathsf{GD},(k)}{+}t), r^{\mathsf{G}, (k)}\big)_{i,j} 
\\& \Rightarrow H_{23}(\mathbf{x}) \geq \widehat{H}_{23}(\mathbf{x};l) \triangleq
\prod_{t=1}^{\tau^{\mathsf{GD}, (k)}}
\left(\frac{ \Gamma \big(\bm{\Psi}(\mathcal{N},t^{\mathsf{GD},(k)}{+}t), r^{\mathsf{G}, (k)}\big)_{i,j} H_{23}[\mathbf{x}]^{l-1}}
{\big[ \Gamma \big(\bm{\Psi}(\mathcal{N},t^{\mathsf{GD},(k)}{+}t), r^{\mathsf{G}, (k)}\big)_{i,j}\big]^{l-1}}\right)^{\frac{\big[ \Gamma \big(\bm{\Psi}(\mathcal{N},t^{\mathsf{GD},(k)}{+}t), r^{\mathsf{G}, (k)}\big)_{i,j}\big]^{l-1}}{H_{23}[\mathbf{x}]^{l-1}}}.
\end{aligned}
\end{equation}
This transformation leads to an equality on posynomials, which requires further steps to conform to the GP format as follows:
\begin{equation}
\begin{aligned}
& \frac{\sum_{i=1}^{m}\sum_{j=1}^{n} 
\left( \tau^{\mathsf{GD}, (k)} \Gamma^{\mathsf{GD}, (k)}_{i,j}\right)^{2}
+ \left( P^{\mathsf{GD},(k)}_{i,j}(\mathbf{x})\right)^{2}}{2 \tau^{\mathsf{GD}, (k)} \sum_{i=1}^{m}\sum_{j=1}^{n}
\Gamma^{\mathsf{GD}, (k)}_{i,j}
 P^{\mathsf{GD},(k)}_{i,j}(\mathbf{x})} \leq 1
\\& \frac{(A_{18})^{-1} 2 \tau^{\mathsf{GD}, (k)} \sum_{i=1}^{m}\sum_{j=1}^{n}
\Gamma^{\mathsf{GD}, (k)}_{i,j}
 P^{\mathsf{GD},(k)}_{i,j}(\mathbf{x})}{\sum_{i=1}^{m}\sum_{j=1}^{n} 
\left( \tau^{\mathsf{GD}, (k)} \Gamma^{\mathsf{GD}, (k)}_{i,j}\right)^{2}
+ \left( P^{\mathsf{GD},(k)}_{i,j}(\mathbf{x})\right)^{2}} \leq 1,
\\& (A_{18})^{-1} \geq 1.
\end{aligned}
\end{equation}
Applying Lemma~\ref{Lemma:ArethmaticGeometric}, both denominators can be condensed as follows:
\begin{equation}
\begin{aligned}
& H_{24}(\mathbf{x}) \triangleq
2 \tau^{\mathsf{GD}, (k)} \sum_{i=1}^{m}\sum_{j=1}^{n}
\Gamma^{\mathsf{GD}, (k)}_{i,j}
 P^{\mathsf{GD},(k)}_{i,j}(\mathbf{x}) \Rightarrow H_{24}(\mathbf{x}) \geq \widehat{H}_{24}(\mathbf{x};l) \\& \triangleq
 \prod_{i=1}^{m} \prod_{j=1}^{n}
\left(\frac{ 2 \tau^{\mathsf{GD}, (k)} \Gamma^{\mathsf{GD}, (k)}_{i,j}
 P^{\mathsf{GD},(k)}_{i,j}(\mathbf{x}) H_{24}[\mathbf{x}]^{l-1}}
{\big[ 2 \tau^{\mathsf{GD}, (k)} \Gamma^{\mathsf{GD}, (k)}_{i,j}
 P^{\mathsf{GD},(k)}_{i,j}(\mathbf{x})\big]^{l-1}}\right)^{\frac{\big[ 2 \tau^{\mathsf{GD}, (k)} \Gamma^{\mathsf{GD}, (k)}_{i,j}
 P^{\mathsf{GD},(k)}_{i,j}(\mathbf{x})\big]^{l-1}}{H_{24}[\mathbf{x}]^{l-1}}}.
\end{aligned}
\end{equation}
\begin{equation}
\begin{aligned}
& H_{25}(\mathbf{x}) \triangleq
\sum_{i=1}^{m}\sum_{j=1}^{n} 
\left( \tau^{\mathsf{GD}, (k)} \Gamma^{\mathsf{GD}, (k)}_{i,j}\right)^{2}
+ \left( P^{\mathsf{GD},(k)}_{i,j}(\mathbf{x})\right)^{2} \Rightarrow H_{25}(\mathbf{x}) \geq \widehat{H}_{25}(\mathbf{x};l) \\& \triangleq 
\prod_{i=1}^{m} \prod_{j=1}^{n}
\left(\frac{\left( \tau^{\mathsf{GD}, (k)} \Gamma^{\mathsf{GD}, (k)}_{i,j}\right)^{2} H_{25}[\mathbf{x}]^{l-1}}
{\left[ \left( \tau^{\mathsf{GD}, (k)} \Gamma^{\mathsf{GD}, (k)}_{i,j}\right)^{2}\right]^{l-1}}\right)^{\frac{\left[ \left( \tau^{\mathsf{GD}, (k)} \Gamma^{\mathsf{GD}, (k)}_{i,j}\right)^{2}\right]^{l-1}}{H_{25}[\mathbf{x}]^{l-1}}} \left(\frac{\left( P^{\mathsf{GD},(k)}_{i,j}(\mathbf{x})\right)^{2} H_{25}[\mathbf{x}]^{l-1}}
{\left[ \left( P^{\mathsf{GD},(k)}_{i,j}(\mathbf{x})\right)^{2}\right]^{l-1}}\right)^{\frac{\left[ \left( P^{\mathsf{GD},(k)}_{i,j}(\mathbf{x})\right)^{2}\right]^{l-1}}{H_{25}[\mathbf{x}]^{l-1}}}.
\end{aligned}
\end{equation}
Thus, the GP accepted form of \eqref{cons:SGD_4} can be expressed as follows:
\begin{tcolorbox}[ams align]
  & \frac{\sum_{t=1}^{\tau^{\mathsf{GD}, (k)}} \Gamma \big(\bm{\Psi}(\mathcal{N},t^{\mathsf{GD},(k)}{+}t), r^{\mathsf{G}, (k)}\big)_{i,j}}{ P^{\mathsf{GD},(k)}_{i,j}(\mathbf{x})} \leq 1,
  \\& \frac{(A_{17})^{-1}  P^{\mathsf{GD},(k)}_{i,j}(\mathbf{x})}{\widehat{H}_{23}(\mathbf{x};l)} \leq 1,
  \\& \frac{\sum_{i=1}^{m}\sum_{j=1}^{n} 
\left( \tau^{\mathsf{GD}, (k)} \Gamma^{\mathsf{GD}, (k)}_{i,j}\right)^{2}
+ \left( P^{\mathsf{GD},(k)}_{i,j}(\mathbf{x})\right)^{2}}{\widehat{H}_{24}(\mathbf{x};l)} \leq 1
\\& \frac{(A_{18})^{-1} 2 \tau^{\mathsf{GD}, (k)} \sum_{i=1}^{m}\sum_{j=1}^{n}
\Gamma^{\mathsf{GD}, (k)}_{i,j}
 P^{\mathsf{GD},(k)}_{i,j}(\mathbf{x})}{\widehat{H}_{25}(\mathbf{x};l)} \leq 1,
\\& (A_{17})^{-1} \geq 1,
\\& (A_{18})^{-1} \geq 1.
\end{tcolorbox}

\textbullet \hspace{2mm} \textbf{Constraint \eqref{GDenergy}}: The constraint $E^{\mathsf{GD}, (k)} = \sum_{n\in\mathcal{N}} \sum_{n'\in\mathcal{N}} \tau_{n,n'}^{\mathsf{GD}, (k)} P_{n}$ is an equality on a posynomial and requires the following transformation to adhere to the GP format:
\begin{equation}
\begin{aligned}
    & \frac{\sum_{n\in\mathcal{N}} \sum_{n'\in\mathcal{N}} \tau_{n,n'}^{\mathsf{GD}, (k)} P_{n}}{E^{\mathsf{GD}, (k)}} \leq 1,
    \\& \frac{(A_{19})^{-1} E^{\mathsf{GD}, (k)}}{\sum_{n\in\mathcal{N}} \sum_{n'\in\mathcal{N}} \tau_{n,n'}^{\mathsf{GD}, (k)} P_{n}} \leq 1,
    \\& (A_{19})^{-1} \geq 1,
\end{aligned}
\end{equation}
where the denominator is condensed using the Lemma \ref{Lemma:ArethmaticGeometric} as follows:
\begin{equation}
    H_{26}(\mathbf{x}) \triangleq \sum_{n\in\mathcal{N}} \sum_{n'\in\mathcal{N}} \tau_{n,n'}^{\mathsf{GD}, (k)} P_{n} 
    \Rightarrow H_{26}(\mathbf{x}) \geq \widehat{H}_{26}(\mathbf{x};l) \triangleq \prod_{n\in\mathcal{N}} \prod_{n'\in\mathcal{N}}\left(\frac{\tau_{n,n'}^{\mathsf{GD}, (k)} P_{n} H_{26}[\mathbf{x}]^{l-1}}{[\tau_{n,n'}^{\mathsf{GD}, (k)} P_{n}]^{l-1}}\right)^{\frac{[\tau_{n,n'}^{\mathsf{GD}, (k)} P_{n}]^{l-1}}{H_{26}[\mathbf{x}]^{l-1}}}.
\end{equation}
Ultimately, the constraint \eqref{GDenergy} can be transformed into the GP allowed format as follows:
\begin{tcolorbox}[ams align]
    & \frac{\sum_{n\in\mathcal{N}} \sum_{n'\in\mathcal{N}} \tau_{n,n'}^{\mathsf{GD}, (k)} P_{n}}{E^{\mathsf{GD}, (k)}} \leq 1,
    \\& \frac{(A_{19})^{-1} E^{\mathsf{GD}, (k)}}{\widehat{H}_{26}(\mathbf{x};l)} \leq 1,
    \\& (A_{19})^{-1} \geq 1,
\end{tcolorbox}

\textbullet \hspace{2mm} \textbf{Constraint \eqref{l-ch1}}: The constraint $\sum_{x\in\mathcal{X}}\lambda^{(k,\ell)}_{x} = 1, \forall \ell{\in}\mathcal{L}^{(k)}\setminus\{L^{(k)}\}, \forall k\in\mathcal{K}$ is an equality on a posynomial and can be transformed into GP admitted format as follows:
\begin{equation}
\begin{aligned}
    & \sum_{x\in\mathcal{X}}\lambda^{(k,\ell)}_{x} \leq 1,
    \\& \frac{(A_{20})^{-1}}{\sum_{x\in\mathcal{X}}\lambda^{(k,\ell)}_{x}} \leq 1,
    \\& (A_{20})^{-1} \geq 1,
\end{aligned}
\end{equation}
which the Lemma \ref{Lemma:ArethmaticGeometric} can condense its denominator as follows:
\begin{equation}
    H_{27}(\mathbf{x}) \triangleq \sum_{x\in\mathcal{X}}\lambda^{(k,\ell)}_{x} 
    \Rightarrow H_{27}(\mathbf{x}) \geq \widehat{H}_{27}(\mathbf{x};l) \triangleq \prod_{x\in\mathcal{X}} \left(\frac{\lambda^{(k,\ell)}_{x} H_{27}[\mathbf{x}]^{l-1}}{[\lambda^{(k,\ell)}_{x}]^{l-1}}\right)^{\frac{[\lambda^{(k,\ell)}_{x}]^{l-1}}{H_{27}[\mathbf{x}]^{l-1}}}.
\end{equation}
The constraint \eqref{l-ch1} can be transformed into the GP allowed format as follows:
\begin{tcolorbox}[ams align]
    & \sum_{x\in\mathcal{X}}\lambda^{(k,\ell)}_{x} \leq 1,
    \\& \frac{(A_{20})^{-1}}{\widehat{H}_{27}(\mathbf{x};l)} \leq 1,
    \\& (A_{20})^{-1} \geq 1.
\end{tcolorbox}

\textbullet \hspace{2mm} \textbf{Constraint \eqref{l-ch2}}: The constraint $\sum_{\ell \in\mathcal{L}^{(k)}\setminus\{L^{(k)}\}} \sum_{x\in\mathcal{X}}\lambda^{(k,\ell)}_{x} = L^{(k)}-1, \forall k\in\mathcal{K}$ can first be transformed into the following expression to eliminate the negative sign:
\begin{equation}
\frac{\sum_{\ell \in\mathcal{L}^{(k)}\setminus\{L^{(k)}\}} \sum_{x\in\mathcal{X}}\lambda^{(k,\ell)}_{x} + 1}{L^{(k)}} = 1,
\end{equation}
which leads to an equality on a posynomial that can be transformed into GP-admitted format as follows:
\begin{equation}
\begin{aligned}
    & \frac{\sum_{\ell \in\mathcal{L}^{(k)}\setminus\{L^{(k)}\}} \sum_{x\in\mathcal{X}}\lambda^{(k,\ell)}_{x} + 1}{L^{(k)}} \leq 1,
    \\& \frac{(A_{21})^{-1} L^{(k)}}{\sum_{\ell \in\mathcal{L}^{(k)}\setminus\{L^{(k)}\}} \sum_{x\in\mathcal{X}}\lambda^{(k,\ell)}_{x} + 1} \leq 1,
    \\& (A_{21})^{-1} \geq 1,
\end{aligned}
\end{equation}
where the denominator requires condensation through Lemma \ref{Lemma:ArethmaticGeometric} as follows:
\begin{equation}
\begin{aligned}
    & H_{28}(\mathbf{x}) \triangleq \sum_{\ell \in\mathcal{L}^{(k)}\setminus\{L^{(k)}\}} \sum_{x\in\mathcal{X}}\lambda^{(k,\ell)}_{x} + 1 
    \\& \Rightarrow H_{28}(\mathbf{x}) \geq \widehat{H}_{28}(\mathbf{x};l) \triangleq \prod_{\ell \in\mathcal{L}^{(k)}\setminus\{L^{(k)}\}} \prod_{x\in\mathcal{X}}\left(\frac{\lambda^{(k,\ell)}_{x} H_{28}[\mathbf{x}]^{l-1}}{[\lambda^{(k,\ell)}_{x}]^{l-1}}\right)^{\frac{[\lambda^{(k,\ell)}_{x}]^{l-1}}{H_{28}[\mathbf{x}]^{l-1}}} \left(H_{28}[\mathbf{x}]^{l-1}\right)^{\frac{1}{H_{28}[\mathbf{x}]^{l-1}}}.
\end{aligned}
\end{equation}
Consequently, the constraint \eqref{l-ch2} can be transformed to allow the GP format as follows:
\begin{tcolorbox}[ams align]
    & \frac{\sum_{\ell \in\mathcal{L}^{(k)}\setminus\{L^{(k)}\}} \sum_{x\in\mathcal{X}}\lambda^{(k,\ell)}_{x} + 1}{L^{(k)}} \leq 1,
    \\& \frac{(A_{21})^{-1} L^{(k)}}{\widehat{H}_{28}(\mathbf{x};l)} \leq 1,
    \\& (A_{21})^{-1} \geq 1.
\end{tcolorbox}

\textbullet \hspace{2mm} \textbf{Constraint \eqref{eq:timeselection}}: 
The constraint $t^{\mathsf{LT},(k,\ell)}=\sum_{x\in\mathcal{X}^{(k)}}\lambda^{(k,\ell)}_{x} t_{x}^{(k)},~\forall \ell{\in}\mathcal{L}^{(k)}\setminus\{L^{(k)}\},~\forall k\in\mathcal{K}$ 
is an equality on a posynomial and can be transformed into GP-admitted format as follows:
\begin{equation}
\begin{aligned}
    & \frac{\sum_{x\in\mathcal{X}^{(k)}}\lambda^{(k,\ell)}_{x} t_{x}^{(k)}}{t^{\mathsf{LT},(k,\ell)}} \leq 1,
    \\& \frac{(A_{22})^{-1} t^{\mathsf{LT},(k,\ell)}}{\sum_{x\in\mathcal{X}^{(k)}}\lambda^{(k,\ell)}_{x} t_{x}^{(k)}} \leq 1,
    \\& (A_{22})^{-1} \geq 1,
\end{aligned}
\end{equation}
where the denominator requires condensation through Lemma \ref{Lemma:ArethmaticGeometric} as follows:
\begin{equation}
    H_{29}(\mathbf{x}) \triangleq \sum_{x\in\mathcal{X}^{(k)}}\lambda^{(k,\ell)}_{x} t_{x}^{(k)}
    \Rightarrow H_{29}(\mathbf{x}) \geq \widehat{H}_{29}(\mathbf{x};l) \triangleq \prod_{x\in\mathcal{X}^{(k)}}\left(\frac{\lambda^{(k,\ell)}_{x} t_{x}^{(k)} H_{29}[\mathbf{x}]^{l-1}}{\big(\lambda^{(k,\ell)}_{x} t_{x}^{(k)}\big)^{l-1}}\right)^{\frac{\big(\lambda^{(k,\ell)}_{x} t_{x}^{(k)}\big)^{l-1}}{H_{29}[\mathbf{x}]^{l-1}}}.
\end{equation}
Consequently, the constraint \eqref{eq:timeselection} can be transformed to allow the GP format as follows:
\begin{tcolorbox}[ams align]
    & \frac{\sum_{x\in\mathcal{X}^{(k)}}\lambda^{(k,\ell)}_{x} t_{x}^{(k)}}{t^{\mathsf{LT},(k,\ell)}} \leq 1,
    \\& \frac{(A_{22})^{-1} t^{\mathsf{LT},(k,\ell)}}{\widehat{H}_{29}(\mathbf{x};l)} \leq 1,
    \\& (A_{22})^{-1} \geq 1.
\end{tcolorbox}

\textbullet \hspace{2mm} \textbf{Constraint \eqref{eq:training_latency}}: 
The constraint $\tau_{n}^{\mathsf{LT},{(k,\ell)}}= e^{(k)}_{n} a_{n} {B}^{(k,\ell)}_{n} {/}f^{(k,\ell)}_{n}$ is an equality on a monomial and is already in the GP admitted form. It can be expressed in the normalized form as:
\begin{tcolorbox}[ams align]
    \tau_{n}^{\mathsf{LT},{(k,\ell)}} f^{(k,\ell)}_{n} 
    (e^{(k)}_{n})^{-1} (a_{n})^{-1} ({B}^{(k,\ell)}_{n})^{-1} = 1.
\end{tcolorbox}

\textbullet \hspace{2mm} \textbf{Constraint \eqref{eq:EN_LC}}: 
The constraint $E_{n}^{\mathsf{LT},{(k,\ell)}}= \tfrac{\alpha^{\mathsf{chip}}_{n}}{2}(f^{(k)}_{n})^{3}\tau_{n}^{\mathsf{LT},{(k,\ell)}}$ 
is an equality on a monomial and also in the GP-admitted form. Similarly to the previous constraint, it can be expressed in the normalized form as:
\begin{tcolorbox}[ams align]
    & E_{n}^{\mathsf{LT},{(k,\ell)}} 
    \left(\tfrac{\alpha^{\mathsf{chip}}_{n}}{2}\right)^{-1} 
    (f^{(k)}_{n})^{-3} 
    (\tau_{n}^{\mathsf{LT},{(k,\ell)}})^{-1} = 1.
\end{tcolorbox}

\textbullet \hspace{2mm} \textbf{Constraint \eqref{eq:EN_LC_sum}}: 
The constraint $E^{\mathsf{LT},(k,\ell)}= \sum_{n \in \mathcal{N}} E_{n}^{\mathsf{LT},{(k,\ell)}}$ 
is an equality on a posynomial and can be transformed into GP-admitted format as follows:
\begin{equation}
\begin{aligned}
    & \frac{\sum_{n \in \mathcal{N}} E_{n}^{\mathsf{LT},{(k,\ell)}}}{E^{\mathsf{LT},(k,\ell)}} \leq 1,
    \\& \frac{(A_{23})^{-1} E^{\mathsf{LT},(k,\ell)}}{\sum_{n \in \mathcal{N}} E_{n}^{\mathsf{LT},{(k,\ell)}}} \leq 1,
    \\& (A_{23})^{-1} \geq 1,
\end{aligned}
\end{equation}
where the denominator requires condensation through Lemma \ref{Lemma:ArethmaticGeometric} as follows:
\begin{equation}
    H_{30}(\mathbf{x}) \triangleq \sum_{n \in \mathcal{N}} E_{n}^{\mathsf{LT},{(k,\ell)}}
    \Rightarrow H_{30}(\mathbf{x}) \geq \widehat{H}_{30}(\mathbf{x};l) \triangleq 
    \prod_{n \in \mathcal{N}} \left(\frac{E_{n}^{\mathsf{LT},{(k,\ell)}} H_{30}[\mathbf{x}]^{l-1}}
    {(E_{n}^{\mathsf{LT},{(k,\ell)}})^{l-1}}\right)^{\frac{(E_{n}^{\mathsf{LT},{(k,\ell)}})^{l-1}}{H_{30}[\mathbf{x}]^{l-1}}}.
\end{equation}
Consequently, the constraint \eqref{eq:EN_LC_sum} can be transformed to allow the GP format as follows:
\begin{tcolorbox}[ams align]
    & \frac{\sum_{n \in \mathcal{N}} E_{n}^{\mathsf{LT},{(k,\ell)}}}{E^{\mathsf{LT},(k,\ell)}} \leq 1,
    \\& \frac{(A_{23})^{-1} E^{\mathsf{LT},(k,\ell)}}{\widehat{H}_{30}(\mathbf{x};l)} \leq 1,
    \\& (A_{23})^{-1} \geq 1.
\end{tcolorbox}

\textbullet \hspace{2mm} \textbf{Constraint \eqref{cons:SCA_2}}: 
The constraint $\sum_{n'\in\mathcal{N}}\bm{\Gamma}^{\mathsf{LA}, (k,\ell)}_{n,n'} + \pi_{n}^{\mathsf{L}, (k)} = 1,~\forall n\in\mathcal{N}$ 
is an equality on a posynomial and can be transformed into GP-admitted format as follows:
\begin{equation}
\begin{aligned}
    & \sum_{n'\in\mathcal{N}}\bm{\Gamma}^{\mathsf{LA}, (k,\ell)}_{n,n'} + \pi_{n}^{\mathsf{L}, (k)} \leq 1,
    \\& \frac{(A_{24})^{-1}}{\sum_{n'\in\mathcal{N}}\bm{\Gamma}^{\mathsf{LA}, (k,\ell)}_{n,n'} + \pi_{n}^{\mathsf{L}, (k)}} \leq 1,
    \\& (A_{24})^{-1} \geq 1,
\end{aligned}
\end{equation}
where the denominator requires condensation through Lemma \ref{Lemma:ArethmaticGeometric} as follows:
\begin{equation}
    H_{31}(\mathbf{x}) \triangleq \sum_{n'\in\mathcal{N}}\bm{\Gamma}^{\mathsf{LA}, (k,\ell)}_{n,n'} + \pi_{n}^{\mathsf{L}, (k)}
    \Rightarrow H_{31}(\mathbf{x}) \geq \widehat{H}_{31}(\mathbf{x};l) \triangleq 
    \prod_{n'\in\mathcal{N}} \left(\frac{\bm{\Gamma}^{\mathsf{LA}, (k,\ell)}_{n,n'} H_{31}[\mathbf{x}]^{l-1}}
    {(\bm{\Gamma}^{\mathsf{LA}, (k,\ell)}_{n,n'})^{l-1}}\right)^{\frac{(\bm{\Gamma}^{\mathsf{LA}, (k,\ell)}_{n,n'})^{l-1}}{H_{31}[\mathbf{x}]^{l-1}}}
    \left(\frac{\pi_{n}^{\mathsf{L}, (k)} H_{31}[\mathbf{x}]^{l-1}}
    {(\pi_{n}^{\mathsf{L}, (k)})^{l-1}}\right)^{\frac{(\pi_{n}^{\mathsf{L}, (k)})^{l-1}}{H_{31}[\mathbf{x}]^{l-1}}}.
\end{equation}
Consequently, the constraint \eqref{cons:SCA_2} can be transformed to allow the GP format as follows:
\begin{tcolorbox}[ams align]
    & \sum_{n'\in\mathcal{N}}\bm{\Gamma}^{\mathsf{LA}, (k,\ell)}_{n,n'} + \pi_{n}^{\mathsf{L}, (k)} \leq 1,
    \\& \frac{(A_{24})^{-1}}{\widehat{H}_{31}(\mathbf{x};l)} \leq 1,
    \\& (A_{24})^{-1} \geq 1,
\end{tcolorbox}

\textbullet \hspace{2mm} \textbf{Constraints \eqref{eq:CM_aggregation_transmission_latency},\eqref{eq:CM_aggregation_latency},\eqref{cons:CM_aggregation_latency_1}}:
For the constraint \eqref{cons:CM_aggregation_latency_1} as 
$\tau^{\mathsf{LA}, (k,\ell)} \leq \tau^{\mathsf{LA},\mathsf{max}}, ~\forall \ell{\in}\mathcal{L}^{(k)}\setminus\{L^{(k)}\},~k\in\mathcal{K}$,
given that $\tau^{\mathsf{LA}, (k,\ell)}=\max_{n'\in\{r_{c}^{\mathsf{L}, (k)}\}_{c\in\mathcal{C}^{(k)}}} \{\tau_{n'}^{\mathsf{LA}, (k,\ell)}\}$,
we utilize \eqref{eq:CM_aggregation_latency} as 
$ \tau_{n'}^{\mathsf{LA}, (k,\ell)} = \max_{n\in\mathcal{N}\setminus\{n'\}}\Big\{ \bm{\Gamma}^{\mathsf{LA}, (k,\ell)}_{n, n'}\tau_{n}^{\mathsf{LA}, (k,\ell)} + \tau_{n,n'}^{\mathsf{LA}, (k,\ell)} \Big\}$
and include \eqref{eq:CM_aggregation_transmission_latency} as
$\tau_{n,n'}^{\mathsf{LA}, (k,\ell)} 
= \bm{\Gamma}^{\mathsf{LA}, (k,\ell)}_{n,n'}\alpha^{\mathsf{Bit}} M^{\mathsf{Dim}}\big/{\mathfrak{R}^{\mathsf{LA}, (k,\ell)}_{n,n'}}$, yielding the following formulation:
\begin{equation}
\begin{aligned}
    & \max_{n'\in\{r_{c}^{\mathsf{L}, (k)}\}_{c\in\mathcal{C}^{(k)}}} 
      \Big\{ \tau_{n'}^{\mathsf{LA}, (k,\ell)} \Big\} \leq \tau^{\mathsf{LA},\mathsf{max}},
    \\& \tau_{n'}^{\mathsf{LA}, (k,\ell)} = 
      \max_{n\in\mathcal{N}\setminus\{n'\}}\Big\{ \bm{\Gamma}^{\mathsf{LA}, (k,\ell)}_{n, n'}\tau_{n}^{\mathsf{LA}, (k,\ell)} + \tau_{n,n'}^{\mathsf{LA}, (k,\ell)} \Big\},
    \\& \tau_{n,n'}^{\mathsf{LA}, (k,\ell)} 
      = \bm{\Gamma}^{\mathsf{LA}, (k,\ell)}_{n,n'}\alpha^{\mathsf{Bit}} M^{\mathsf{Dim}}\big/{\mathfrak{R}^{\mathsf{LA}, (k,\ell)}_{n,n'}}.
\end{aligned}
\end{equation}
The first expression contains a maximum function, which is not admitted in GP. We therefore use the approximation $\max\{A,B\}\approx (A^{p}+B^{p})^{\frac{1}{p}}$ with a large $p$ to replace the maximum. Accordingly, the first expression is transformed as follows:
\begin{equation}
    \left(\sum_{n'\in\{r_{c}^{\mathsf{L}, (k)}\}_{c\in\mathcal{C}^{(k)}}} 
      \left( \tau_{n'}^{\mathsf{LA}, (k,\ell)} \right)^{p}\right)^{1/p}\leq \tau^{\mathsf{LA},\mathsf{max}},
\end{equation}
This expression is still not in the acceptable form of GP; thus, we take a nested approach by introducing $P^{(k,\ell)}_{(1)}(\mathbf{x})$, assumed as follows:
\begin{equation}
    P^{(k,\ell)}_{(1)}(\mathbf{x}) = \sum_{n'\in\{r_{c}^{\mathsf{L}, (k)}\}_{c\in\mathcal{C}^{(k)}}} 
      \left( \tau_{n'}^{\mathsf{LA}, (k,\ell)} \right)^{p}.
\end{equation}
We rewrite it in GP admitted format as follows:
\begin{equation}
\begin{aligned}
    & \frac{\sum_{n'\in\{r_{c}^{\mathsf{L}, (k)}\}_{c\in\mathcal{C}^{(k)}}} 
      \left( \tau_{n'}^{\mathsf{LA}, (k,\ell)} \right)^{p}}{P^{(k,\ell)}_{(1)}(\mathbf{x})} \leq 1,
    \\&\frac{(A_{25})^{-1} P^{(k,\ell)}_{(1)}(\mathbf{x})}{\sum_{n'\in\{r_{c}^{\mathsf{L}, (k)}\}_{c\in\mathcal{C}^{(k)}}} 
      \left( \tau_{n'}^{\mathsf{LA}, (k,\ell)} \right)^{p}} \leq 1,
    \\& (A_{25})^{-1} \geq 1.
\end{aligned}
\end{equation}
The denominator is condensed using Lemma \ref{Lemma:ArethmaticGeometric} as follows:
\begin{equation}
H_{32}(\mathbf{x}) \triangleq \sum_{n'\in\{r_{c}^{\mathsf{L}, (k)}\}_{c\in\mathcal{C}^{(k)}}} 
      \left( \tau_{n'}^{\mathsf{LA}, (k,\ell)} \right)^{p} \Rightarrow 
H_{32}(\mathbf{x}) \geq \widehat{H}_{32}(\mathbf{x};l) \triangleq 
\prod_{n'\in\{r_{c}^{\mathsf{L}, (k)}\}_{c\in\mathcal{C}^{(k)}}} \left(\frac{\left( \tau_{n'}^{\mathsf{LA}, (k,\ell)} \right)^{p} H_{32}[\mathbf{x}]^{l-1}}
{\Big[\left( \tau_{n'}^{\mathsf{LA}, (k,\ell)} \right)^{p}\Big]^{l-1}}\right)^{\frac{\Big[\left( \tau_{n'}^{\mathsf{LA}, (k,\ell)} \right)^{p}\Big]^{l-1}}{H_{32}[\mathbf{x}]^{l-1}}}.
\end{equation}
Hence, the first expression can be driven as follows:
\begin{equation}
\frac{\left(P^{(k,\ell)}_{(1)}(\mathbf{x})\right)^{1/p}}{\tau^{\mathsf{LA},\mathsf{max}}} \leq 1,
\end{equation}
The second expression also includes a maximum function, and we treat it similarly as follows:
\begin{equation}
    \tau_{n'}^{\mathsf{LA}, (k,\ell)} = 
      \left( \sum_{n\in\mathcal{N}\setminus\{n'\}} \left( \bm{\Gamma}^{\mathsf{LA}, (k,\ell)}_{n, n'}\tau_{n}^{\mathsf{LA}, (k,\ell)} + \tau_{n,n'}^{\mathsf{LA}, (k,\ell)} \right)^{p} \right)^{1/p},
\end{equation}
which also requires a nested approach by introducing $P^{(k,\ell)}_{(2)}(\mathbf{x})$ and $P^{(k,\ell)}_{(3),n}(\mathbf{x})$, assumed as follows:
\begin{equation}
\begin{aligned}
    &   \tau_{n'}^{\mathsf{LA}, (k,\ell)} = \left( P^{(k,\ell)}_{(2)}(\mathbf{x}) \right)^{1/p}
    \\& P^{(k,\ell)}_{(2)}(\mathbf{x}) =  \sum_{n\in\mathcal{N}\setminus\{n'\}} \left( P^{(k,\ell)}_{(3),n}(\mathbf{x}) \right)^{p},
    \\& P^{(k,\ell)}_{(3),n}(\mathbf{x}) =  \bm{\Gamma}^{\mathsf{LA}, (k,\ell)}_{n, n'}\tau_{n}^{\mathsf{LA}, (k,\ell)} + \tau_{n,n'}^{\mathsf{LA}, (k,\ell)}.
\end{aligned}
\end{equation}
These three expressions can be transformed into the GP allowed format as follows:
\begin{equation}
\begin{aligned}
    & \tau_{n'}^{\mathsf{LA}, (k,\ell)} \left( P^{(k,\ell)}_{(2)}(\mathbf{x}) \right)^{-1/p} = 1
  \\& \frac{\sum_{n\in\mathcal{N}\setminus\{n'\}} \left( P^{(k,\ell)}_{(3),n}(\mathbf{x}) \right)^{p}}{P^{(k,\ell)}_{(2)}(\mathbf{x})} \leq 1,
    \\& \frac{(A_{26})^{-1} P^{(k,\ell)}_{(2)}(\mathbf{x})}{\sum_{n\in\mathcal{N}\setminus\{n'\}} \left( P^{(k,\ell)}_{(3),n}(\mathbf{x}) \right)^{p}} \leq 1,
    \\& \frac{\bm{\Gamma}^{\mathsf{LA}, (k,\ell)}_{n, n'}\tau_{n}^{\mathsf{LA}, (k,\ell)} + \tau_{n,n'}^{\mathsf{LA}, (k,\ell)}}{P^{(k,\ell)}_{(3),n}(\mathbf{x})} \leq 1,
    \\& \frac{(A_{27})^{-1}  P^{(k,\ell)}_{(3),n}(\mathbf{x})}{\bm{\Gamma}^{\mathsf{LA}, (k,\ell)}_{n, n'}\tau_{n}^{\mathsf{LA}, (k,\ell)} + \tau_{n,n'}^{\mathsf{LA}, (k,\ell)}} \leq 1,
    \\& (A_{26})^{-1} \geq 1,
    \\& (A_{27})^{-1} \geq 1.
\end{aligned}
\end{equation}
For denominators, we can apply the Lemma \ref{Lemma:ArethmaticGeometric} for their condensation as follows:
\begin{equation}
    H_{33}(\mathbf{x}) \triangleq \sum_{n\in\mathcal{N}\setminus\{n'\}} \left( P^{(k,\ell)}_{(3),n}(\mathbf{x}) \right)^{p} \Rightarrow H_{33}(\mathbf{x}) \geq \widehat{H}_{33}(\mathbf{x};l) \triangleq 
\prod_{n\in\mathcal{N}\setminus\{n'\}}
\left(\frac{\left( P^{(k,\ell)}_{(3),n}(\mathbf{x}) \right)^{p} H_{33}[\mathbf{x}]^{l-1}}
{\left[\left( P^{(k,\ell)}_{(3),n}(\mathbf{x}) \right)^{p}\right]^{l-1}}\right)^{\frac{\left[\left( P^{(k,\ell)}_{(3),n}(\mathbf{x}) \right)^{p}\right]^{l-1}}{H_{33}[\mathbf{x}]^{l-1}}},
\end{equation}
\begin{equation}
\begin{aligned}
& H_{34}(\mathbf{x}) \triangleq \bm{\Gamma}^{\mathsf{LA}, (k,\ell)}_{n, n'}\tau_{n}^{\mathsf{LA}, (k,\ell)} + \tau_{n,n'}^{\mathsf{LA}, (k,\ell)} 
\\& \Rightarrow H_{34}(\mathbf{x}) \geq \widehat{H}_{34}(\mathbf{x};l) \triangleq 
\left(\frac{\bm{\Gamma}^{\mathsf{LA}, (k,\ell)}_{n, n'}\tau_{n}^{\mathsf{LA}, (k,\ell)} H_{34}[\mathbf{x}]^{l-1}}
{\left[ \bm{\Gamma}^{\mathsf{LA}, (k,\ell)}_{n, n'}\tau_{n}^{\mathsf{LA}, (k,\ell)} \right]^{l-1}}\right)^{\frac{\left[\bm{\Gamma}^{\mathsf{LA}, (k,\ell)}_{n, n'}\tau_{n}^{\mathsf{LA}, (k,\ell)}\right]^{l-1}}{H_{34}[\mathbf{x}]^{l-1}}} \left(\frac{\tau_{n,n'}^{\mathsf{LA}, (k,\ell)} H_{34}[\mathbf{x}]^{l-1}}
{\left[ \tau_{n,n'}^{\mathsf{LA}, (k,\ell)} \right]^{l-1}}\right)^{\frac{\left[\tau_{n,n'}^{\mathsf{LA}, (k,\ell)}\right]^{l-1}}{H_{34}[\mathbf{x}]^{l-1}}}.
\end{aligned}
\end{equation}
Finally, the third expression is also in the form of an equality on a monomial, which can be acceptable in GP format when the non-zero denominator is guaranteed as follows:
\begin{equation}
\begin{aligned}
    & \frac{\mathfrak{R}^{\mathsf{LA}, (k,\ell)}_{n,n'} \tau_{n,n'}^{\mathsf{LA}, (k,\ell)} + 1}{\bm{\Gamma}^{\mathsf{LA}, (k,\ell)}_{n,n'}\alpha^{\mathsf{Bit}} M^{\mathsf{Dim}} + 1} \leq 1,
    \\& \frac{(A_{28})^{-1}(\bm{\Gamma}^{\mathsf{LA}, (k,\ell)}_{n,n'}\alpha^{\mathsf{Bit}} M^{\mathsf{Dim}} + 1)}{\mathfrak{R}^{\mathsf{LA}, (k,\ell)}_{n,n'} \tau_{n,n'}^{\mathsf{LA}, (k,\ell)} + 1 } \leq 1,
    \\& (A_{28})^{-1} \geq 1,
\end{aligned}
\end{equation}
with denominators condensed by Lemma \ref{Lemma:ArethmaticGeometric} as follows:
\begin{equation}
 H_{35}(\mathbf{x}) \triangleq \bm{\Gamma}^{\mathsf{LA}, (k,\ell)}_{n,n'}\alpha^{\mathsf{Bit}} M^{\mathsf{Dim}} + 1 \Rightarrow 
 H_{35}(\mathbf{x}) \geq \widehat{H}_{35}(\mathbf{x};l) \triangleq  
 \left(\frac{\bm{\Gamma}^{\mathsf{LA}, (k,\ell)}_{n,n'}  H_{35}[\mathbf{x}]^{l-1}}{[\bm{\Gamma}^{\mathsf{LA}, (k,\ell)}_{n,n'}]^{l-1}}\right)^{\frac{[\bm{\Gamma}^{\mathsf{LA}, (k,\ell)}_{n,n'}\alpha^{\mathsf{Bit}} M^{\mathsf{Dim}}]^{l-1}}{H_{35}[\mathbf{x}]^{l-1}}}
 \left(H_{35}[\mathbf{x}]^{l-1}\right)^{\frac{1}{H_{35}[\mathbf{x}]^{l-1}}},
\end{equation}
\begin{equation}
\begin{aligned}
& H_{36}(\mathbf{x}) \triangleq \mathfrak{R}^{\mathsf{LA}, (k,\ell)}_{n,n'} \tau_{n,n'}^{\mathsf{LA}, (k,\ell)} + 1 
\\& \Rightarrow 
 H_{36}(\mathbf{x}) \geq \widehat{H}_{36}(\mathbf{x};l) \triangleq  
 \left(\frac{\mathfrak{R}^{\mathsf{LA}, (k,\ell)}_{n,n'} \tau_{n,n'}^{\mathsf{LA}, (k,\ell)}  H_{36}[\mathbf{x}]^{l-1}}{[\mathfrak{R}^{\mathsf{LA}, (k,\ell)}_{n,n'} \tau_{n,n'}^{\mathsf{LA}, (k,\ell)}]^{l-1}}\right)^{\frac{[\mathfrak{R}^{\mathsf{LA}, (k,\ell)}_{n,n'} \tau_{n,n'}^{\mathsf{LA}, (k,\ell)}]^{l-1}}{H_{36}[\mathbf{x}]^{l-1}}}
 \left(H_{36}[\mathbf{x}]^{l-1}\right)^{\frac{1}{H_{36}[\mathbf{x}]^{l-1}}}.
\end{aligned}
\end{equation}
Therefore, constraints \eqref{eq:CM_aggregation_transmission_latency},\eqref{eq:CM_aggregation_latency},\eqref{cons:CM_aggregation_latency_1} admit the GP format as:
\begin{tcolorbox}[ams align]
    & \frac{\left(P^{(k,\ell)}_{(1)}(\mathbf{x})\right)^{1/p}}{\tau^{\mathsf{LA},\mathsf{max}}} \leq 1,
    \\& \frac{\sum_{n'\in\{r_{c}^{\mathsf{L}, (k)}\}_{c\in\mathcal{C}^{(k)}}} 
      \left( \tau_{n'}^{\mathsf{LA}, (k,\ell)} \right)^{p}}{P^{(k,\ell)}_{(1)}(\mathbf{x})} \leq 1,
    \\&\frac{(A_{25})^{-1} P^{(k,\ell)}_{(1)}(\mathbf{x})}{\widehat{H}_{32}(\mathbf{x};l)} \leq 1,
    \\& \tau_{n'}^{\mathsf{LA}, (k,\ell)} \left( P^{(k,\ell)}_{(2)}(\mathbf{x}) \right)^{-1/p} = 1
  \\& \frac{\sum_{n\in\mathcal{N}\setminus\{n'\}} \left( P^{(k,\ell)}_{(3),n}(\mathbf{x}) \right)^{p}}{P^{(k,\ell)}_{(2)}(\mathbf{x})} \leq 1,
    \\& \frac{(A_{26})^{-1} P^{(k,\ell)}_{(2)}(\mathbf{x})}{\widehat{H}_{33}(\mathbf{x};l)} \leq 1,
    \\& \frac{\bm{\Gamma}^{\mathsf{LA}, (k,\ell)}_{n, n'}\tau_{n}^{\mathsf{LA}, (k,\ell)} + \tau_{n,n'}^{\mathsf{LA}, (k,\ell)}}{P^{(k,\ell)}_{(3),n}(\mathbf{x})} \leq 1,
    \\& \frac{(A_{27})^{-1}  P^{(k,\ell)}_{(3),n}(\mathbf{x})}{\widehat{H}_{34}(\mathbf{x};l)} \leq 1,
    \\& \frac{\mathfrak{R}^{\mathsf{LA}, (k,\ell)}_{n,n'} \tau_{n,n'}^{\mathsf{LA}, (k,\ell)} + 1}{\widehat{H}_{35}(\mathbf{x};l)} \leq 1,
    \\& \frac{(A_{28})^{-1}(\bm{\Gamma}^{\mathsf{LA}, (k,\ell)}_{n,n'}\alpha^{\mathsf{Bit}} M^{\mathsf{Dim}} + 1)}{\widehat{H}_{36}(\mathbf{x};l)} \leq 1,
    \\& (A_{25})^{-1} \geq 1,
    \\& (A_{26})^{-1} \geq 1,
    \\& (A_{27})^{-1} \geq 1,
    \\& (A_{28})^{-1} \geq 1.
\end{tcolorbox}

\textbullet \hspace{2mm} \textbf{Constraint \eqref{cons:SCA_4}}: 
The constraint 
$\sum_{t=1}^{\tau^{\mathsf{LA}, (k,\ell)}} \Big(\bm{\Gamma}^{\mathsf{LA}, (k,\ell)} - \sum_{c \in \mathcal{C}^{(k)}} \bm{\Gamma}\big(\bm{\Psi}(\mathcal{N}^{(k)}_{c},t^{\mathsf{LA},(k,\ell)}{+}t), r_{c}^{\mathsf{L}, (k)}\big)\Big) = \mathbf{0}$ is an equality on a posynomial with a negative sign. We first eliminate the negative sign as follows:
\begin{equation}
    \sum_{t=1}^{\tau^{\mathsf{LA}, (k,\ell)}} \sum_{c \in \mathcal{C}^{(k)}} \bm{\Gamma}\big(\bm{\Psi}(\mathcal{N}^{(k)}_{c},t^{\mathsf{LA},(k,\ell)}{+}t), r_{c}^{\mathsf{L}, (k)}\big) = \sum_{t=1}^{\tau^{\mathsf{LA}, (k,\ell)}} \bm{\Gamma}^{\mathsf{LA}, (k,\ell)},
\end{equation}
which is an equality on matrices, which is not applicable in GP format. We further utilized the matrix equality transformation through the Frobenius norm as follows:
\begin{equation}
\begin{aligned}
    & \sum_{i=1}^{m}\sum_{j=1}^{n} \left[
    \left( \tau^{\mathsf{LA}, (k,\ell)} \Gamma^{\mathsf{LA}, (k,\ell)}_{i,j}\right)^{2} 
    + \left( \sum_{t=1}^{\tau^{\mathsf{LA}, (k,\ell)}} \sum_{c \in \mathcal{C}^{(k)}} \Gamma \big(\bm{\Psi}(\mathcal{N}^{(k)}_{c},t^{\mathsf{LA},(k,\ell)}{+}t), r_{c}^{\mathsf{L}, (k)}\big)_{i,j}\right)^{2} \right] 
    \\ & 
    =  2 \tau^{\mathsf{LA}, (k,\ell)} \sum_{i=1}^{m}\sum_{j=1}^{n}
      \Gamma^{\mathsf{LA}, (k,\ell)}_{i,j}
     \sum_{t=1}^{\tau^{\mathsf{LA}, (k,\ell)}} \sum_{c \in \mathcal{C}^{(k)}} \Gamma \big(\bm{\Psi}(\mathcal{N}^{(k)}_{c},t^{\mathsf{LA},(k,\ell)}{+}t), r_{c}^{\mathsf{L}, (k)}\big)_{i,j}.
\end{aligned}
\end{equation}
Furthermore, this expression requires a nested approach to comply with the GP format as follows:
\begin{equation}
   P^{\mathsf{LA},(k,\ell)}_{i,j}(\mathbf{x}) = \sum_{t=1}^{\tau^{\mathsf{LA}, (k,\ell)}} \sum_{c \in \mathcal{C}^{(k)}} \Gamma \big(\bm{\Psi}(\mathcal{N}^{(k)}_{c},t^{\mathsf{LA},(k,\ell)}{+}t), r_{c}^{\mathsf{L}, (k)}\big)_{i,j},
\end{equation}
which then can be presented in the GP format as follows:
\begin{equation}
\begin{aligned}
  & \frac{\sum_{t=1}^{\tau^{\mathsf{LA}, (k,\ell)}} \sum_{c \in \mathcal{C}^{(k)}} \Gamma \big(\bm{\Psi}(\mathcal{N}^{(k)}_{c},t^{\mathsf{LA},(k,\ell)}{+}t), r_{c}^{\mathsf{L}, (k)}\big)_{i,j}}{ P^{\mathsf{LA},(k,\ell)}_{i,j}(\mathbf{x})} \leq 1,
  \\& \frac{(A_{29})^{-1}  P^{\mathsf{LA},(k,\ell)}_{i,j}(\mathbf{x})}{\sum_{t=1}^{\tau^{\mathsf{LA}, (k,\ell)}} \sum_{c \in \mathcal{C}^{(k)}} \Gamma \big(\bm{\Psi}(\mathcal{N}^{(k)}_{c},t^{\mathsf{LA},(k,\ell)}{+}t), r_{c}^{\mathsf{L}, (k)}\big)_{i,j}} \leq 1,
  \\& (A_{29})^{-1} \geq 1.
\end{aligned}
\end{equation}
The denominator can be condensed as via Lemma~\ref{Lemma:ArethmaticGeometric} as follows:
\begin{equation}
\begin{aligned}
& H_{37}(\mathbf{x}) \triangleq
\sum_{t=1}^{\tau^{\mathsf{LA}, (k,\ell)}} \sum_{c \in \mathcal{C}^{(k)}} \Gamma \big(\bm{\Psi}(\mathcal{N}^{(k)}_{c},t^{\mathsf{LA},(k,\ell)}{+}t), r_{c}^{\mathsf{L}, (k)}\big)_{i,j}
\\& \Rightarrow H_{37}(\mathbf{x}) \geq \widehat{H}_{37}(\mathbf{x};l) \triangleq
\prod_{t=1}^{\tau^{\mathsf{LA}, (k,\ell)}} \prod_{c \in \mathcal{C}^{(k)}}
\left(\frac{ \Gamma \big(\bm{\Psi}(\mathcal{N}^{(k)}_{c},t^{\mathsf{LA},(k,\ell)}{+}t), r_{c}^{\mathsf{L}, (k)}\big)_{i,j} H_{37}[\mathbf{x}]^{l-1}}
{\big[ \Gamma \big(\bm{\Psi}(\mathcal{N}^{(k)}_{c},t^{\mathsf{LA},(k,\ell)}{+}t), r_{c}^{\mathsf{L}, (k)}\big)_{i,j}\big]^{l-1}}\right)^{\frac{\big[ \Gamma \big(\bm{\Psi}(\mathcal{N}^{(k)}_{c},t^{\mathsf{LA},(k,\ell)}{+}t), r_{c}^{\mathsf{L}, (k)}\big)_{i,j}\big]^{l-1}}{H_{37}[\mathbf{x}]^{l-1}}}.
\end{aligned}
\end{equation}
This transformation leads to an equality on posynomials, which requires further steps to conform to the GP format as follows:
\begin{equation}
\begin{aligned}
    & \frac{\sum_{i=1}^{m}\sum_{j=1}^{n}
    \left( \tau^{\mathsf{LA}, (k,\ell)} \Gamma^{\mathsf{LA}, (k,\ell)}_{i,j}\right)^{2} 
    + \left(  P^{\mathsf{LA},(k,\ell)}_{i,j}(\mathbf{x})\right)^{2} }{2 \tau^{\mathsf{LA}, (k,\ell)} \sum_{i=1}^{m}\sum_{j=1}^{n}
      \Gamma^{\mathsf{LA}, (k,\ell)}_{i,j}  P^{\mathsf{LA},(k,\ell)}_{i,j}(\mathbf{x})} \leq 1
    \\& \frac{(A_{30})^{-1} 2 \tau^{\mathsf{LA}, (k,\ell)} \sum_{i=1}^{m}\sum_{j=1}^{n}
      \Gamma^{\mathsf{LA}, (k,\ell)}_{i,j}  P^{\mathsf{LA},(k,\ell)}_{i,j}(\mathbf{x})}{\sum_{i=1}^{m}\sum_{j=1}^{n} 
    \left( \tau^{\mathsf{LA}, (k,\ell)} \Gamma^{\mathsf{LA}, (k,\ell)}_{i,j}\right)^{2} 
    + \left(  P^{\mathsf{LA},(k,\ell)}_{i,j}(\mathbf{x})\right)^{2} } \leq 1,
    \\& (A_{30})^{-1} \geq 1.
\end{aligned}
\end{equation}
Applying Lemma~\ref{Lemma:ArethmaticGeometric}, both denominators can be condensed as follows:
\begin{equation}
\begin{aligned}
    & H_{38}(\mathbf{x}) \triangleq 
    2 \tau^{\mathsf{LA}, (k,\ell)} \sum_{i=1}^{m}\sum_{j=1}^{n}
      \Gamma^{\mathsf{LA}, (k,\ell)}_{i,j}  P^{\mathsf{LA},(k,\ell)}_{i,j}(\mathbf{x}) \Rightarrow H_{38}(\mathbf{x}) \geq \widehat{H}_{38}(\mathbf{x};l) \\&  \triangleq
      \prod_{i=1}^{m} \prod_{j=1}^{n} 
    \left(\frac{2 \tau^{\mathsf{LA}, (k,\ell)} \Gamma^{\mathsf{LA}, (k,\ell)}_{i,j}  P^{\mathsf{LA},(k,\ell)}_{i,j}(\mathbf{x}) H_{38}[\mathbf{x}]^{l-1}}
    {\big[2 \tau^{\mathsf{LA}, (k,\ell)} \Gamma^{\mathsf{LA}, (k,\ell)}_{i,j}  P^{\mathsf{LA},(k,\ell)}_{i,j}(\mathbf{x})\big]^{l-1}}\right)^{\frac{\big[2 \tau^{\mathsf{LA}, (k,\ell)} \Gamma^{\mathsf{LA}, (k,\ell)}_{i,j}  P^{\mathsf{LA},(k,\ell)}_{i,j}(\mathbf{x})\big]^{l-1}}{H_{38}[\mathbf{x}]^{l-1}}},
\end{aligned}
\end{equation}
\begin{equation}
\begin{aligned}
    & H_{39}(\mathbf{x}) \triangleq 
     \sum_{i=1}^{m}\sum_{j=1}^{n} 
    \left( \tau^{\mathsf{LA}, (k,\ell)} \Gamma^{\mathsf{LA}, (k,\ell)}_{i,j}\right)^{2} 
    + \left(  P^{\mathsf{LA},(k,\ell)}_{i,j}(\mathbf{x})\right)^{2} \Rightarrow H_{39}(\mathbf{x}) \geq \widehat{H}_{39}(\mathbf{x};l) 
    \\&  \triangleq
    \prod_{i=1}^{m} \prod_{j=1}^{n}
    \left(\frac{\left( \tau^{\mathsf{LA}, (k,\ell)} \Gamma^{\mathsf{LA}, (k,\ell)}_{i,j}\right)^{2} H_{39}[\mathbf{x}]^{l-1}}
    {\left[ \left( \tau^{\mathsf{LA}, (k,\ell)} \Gamma^{\mathsf{LA}, (k,\ell)}_{i,j}\right)^{2}\right]^{l-1}}\right)^{\frac{\left[ \left( \tau^{\mathsf{LA}, (k,\ell)} \Gamma^{\mathsf{LA}, (k,\ell)}_{i,j}\right)^{2}\right]^{l-1}}{H_{39}[\mathbf{x}]^{l-1}}} \left(\frac{\left(  P^{\mathsf{LA},(k,\ell)}_{i,j}(\mathbf{x})\right)^{2} H_{39}[\mathbf{x}]^{l-1}}
    {\left[\left(  P^{\mathsf{LA},(k,\ell)}_{i,j}(\mathbf{x})\right)^{2}\right]^{l-1}}\right)^{\frac{\left[\left(  P^{\mathsf{LA},(k,\ell)}_{i,j}(\mathbf{x})\right)^{2}\right]^{l-1}}{H_{39}[\mathbf{x}]^{l-1}}}.
\end{aligned}
\end{equation}
Therefore, the GP admissible form of \eqref{cons:SCA_4} is presentable as follows:
\begin{tcolorbox}[ams align]
  & \frac{\sum_{t=1}^{\tau^{\mathsf{LA}, (k,\ell)}} \sum_{c \in \mathcal{C}^{(k)}} \Gamma \big(\bm{\Psi}(\mathcal{N}^{(k)}_{c},t^{\mathsf{LA},(k,\ell)}{+}t), r_{c}^{\mathsf{L}, (k)}\big)_{i,j}}{ P^{\mathsf{LA},(k,\ell)}_{i,j}(\mathbf{x})} \leq 1,
  \\& \frac{(A_{29})^{-1}  P^{\mathsf{LA},(k,\ell)}_{i,j}(\mathbf{x})}{\widehat{H}_{37}(\mathbf{x};l) } \leq 1,
      \\& \frac{\sum_{i=1}^{m}\sum_{j=1}^{n}
    \left( \tau^{\mathsf{LA}, (k,\ell)} \Gamma^{\mathsf{LA}, (k,\ell)}_{i,j}\right)^{2} 
    + \left(  P^{\mathsf{LA},(k,\ell)}_{i,j}(\mathbf{x})\right)^{2} }{\widehat{H}_{38}(\mathbf{x};l) } \leq 1
    \\& \frac{(A_{30})^{-1} 2 \tau^{\mathsf{LA}, (k,\ell)} \sum_{i=1}^{m}\sum_{j=1}^{n}
      \Gamma^{\mathsf{LA}, (k,\ell)}_{i,j}  P^{\mathsf{LA},(k,\ell)}_{i,j}(\mathbf{x})}{\widehat{H}_{39}(\mathbf{x};l) } \leq 1,
    \\& (A_{29})^{-1} \geq 1,
    \\& (A_{30})^{-1} \geq 1.
\end{tcolorbox}

\textbullet \hspace{2mm} \textbf{Constraint \eqref{forestenergyAG}}: 
The constraint $E^{\mathsf{LA},(k,\ell)} = \sum_{n\in\mathcal{N}} \sum_{n'\in\mathcal{N}} \tau_{n,n'}^{\mathsf{LA},(k,\ell)} P_{n}$ is an equality on a posynomial and requires the following transformation:
\begin{equation}
\begin{aligned}
    & \frac{\sum_{n\in\mathcal{N}} \sum_{n'\in\mathcal{N}}
    \tau_{n,n'}^{\mathsf{LA},(k,\ell)} P_{n}}{E^{\mathsf{LA},(k,\ell)}} \leq 1, \\
    & \frac{(A_{31})^{-1} E^{\mathsf{LA},(k,\ell)}}{\sum_{n\in\mathcal{N}} \sum_{n'\in\mathcal{N}}
    \tau_{n,n'}^{\mathsf{LA},(k,\ell)} P_{n}} \leq 1, \\
    & (A_{31})^{-1} \geq 1.
\end{aligned}
\end{equation}
Using the lemma~\ref{Lemma:ArethmaticGeometric}, the denominator is condensed as follows:
\begin{equation}
     H_{40}(\mathbf{x}) \triangleq 
    \sum_{n\in\mathcal{N}} \sum_{n'\in\mathcal{N}} 
    \tau_{n,n'}^{\mathsf{LA},(k,\ell)} P_{n} \Rightarrow H_{40}(\mathbf{x}) \geq \widehat{H}_{40}(\mathbf{x};l) \triangleq 
    \prod_{n\in\mathcal{N}} \prod_{n'\in\mathcal{N}}
    \left(\frac{\tau_{n,n'}^{\mathsf{LA},(k,\ell)} P_{n}  H_{40}[\mathbf{x}]^{l-1}}
    {[\tau_{n,n'}^{\mathsf{LA},(k,\ell)} P_{n}]^{l-1}}\right)^{\frac{[\tau_{n,n'}^{\mathsf{LA},(k,\ell)} P_{n}]^{l-1}}{H_{40}[\mathbf{x}]^{l-1}}}.
\end{equation}
Therefore, the GP admissible version of \eqref{forestenergyAG} is shown as follows:
\begin{tcolorbox}[ams align]
    & \frac{\sum_{n\in\mathcal{N}} \sum_{n'\in\mathcal{N}}
    \tau_{n,n'}^{\mathsf{LA},(k,\ell)} P_{n}}{E^{\mathsf{LA},(k,\ell)}} \leq 1, \\
    & \frac{(A_{31})^{-1} E^{\mathsf{LA},(k,\ell)}}{\widehat{H}_{40}(\mathbf{x};l)} \leq 1, \\
    & (A_{31})^{-1} \geq 1.
\end{tcolorbox}

\textbullet \hspace{2mm} \textbf{Constraint \eqref{cons:SCD_1}}: The constraint
$\sum_{n\in\mathcal{N}}\bm{\Gamma}^{\mathsf{LD}, (k,\ell)}_{n,n'} + \pi_{n'}^{\mathsf{L}, (k)}{=}1,~\forall n'\in\mathcal{N}$ is also a similar version of previous constraints and the following transformations leads to GP admitted format:
\begin{equation}
\begin{aligned}
    &\sum_{n\in\mathcal{N}}\bm{\Gamma}^{\mathsf{LD}, (k,\ell)}_{n,n'} + \pi_{n'}^{\mathsf{L}, (k)} \leq 1,
    \\& \frac{(A_{32})^{-1}}{\sum_{n\in\mathcal{N}}\bm{\Gamma}^{\mathsf{LD}, (k,\ell)}_{n,n'} + \pi_{n'}^{\mathsf{L}, (k)}} \leq 1,
    \\& (A_{32})^{-1} \geq 1.
\end{aligned}
\end{equation}
Thus, exploiting Lemma \ref{Lemma:ArethmaticGeometric} for both denominators would result in the following condensations:
\begin{equation}
H_{41}(\mathbf{x}) \triangleq \sum_{n\in\mathcal{N}}\bm{\Gamma}^{\mathsf{LD}, (k,\ell)}_{n,n'} + \pi_{n'}^{\mathsf{L}, (k)} \Rightarrow H_{41}(\mathbf{x}) \geq \widehat{H}_{41}(\mathbf{x};l) \triangleq \prod_{n\in\mathcal{N}} \left(\frac{\bm{\Gamma}^{\mathsf{LD}, (k,\ell)}_{n,n'} H_{41}[\mathbf{x}]^{l-1}}{[\bm{\Gamma}^{\mathsf{LD}, (k,\ell)}_{n,n'}]^{l-1}}\right)^{\frac{[\bm{\Gamma}^{\mathsf{LD}, (k,\ell)}_{n,n'}]^{l-1}}{H_{41}[\mathbf{x}]^{l-1}}}.
\end{equation}
Therefore, the constraint \eqref{cons:SCD_1} can be presented as GP format as follows:
\begin{tcolorbox}[ams align]
    &\sum_{n\in\mathcal{N}}\bm{\Gamma}^{\mathsf{LD}, (k,\ell)}_{n,n'} + \pi_{n'}^{\mathsf{L}, (k)} \leq 1,
    \\& \frac{(A_{32})^{-1}}{\widehat{H}_{41}(\mathbf{x};l)} \leq 1,
    \\& (A_{32})^{-1} \geq 1.
\end{tcolorbox}

\textbullet \hspace{2mm} \textbf{Constraints \eqref{eq:CM_dispatching_transmission_latency}, \eqref{eq:CM_dispatching_latency}, \eqref{cons:CM_dispatching_latency_1}}: 
For the constraint \eqref{cons:CM_dispatching_latency_1} as 
$\tau^{\mathsf{LD},(k,\ell)} \leq \tau^{\mathsf{LD},\mathsf{max}},~\forall \ell{\in}\mathcal{L}^{(k)}\setminus\{L^{(k)}\},~k\in\mathcal{K}$, 
via knowing that $\tau^{\mathsf{LD},(k,\ell)}=\max_{n'\in\mathcal{N}} \big\{ \tau_{n'}^{\mathsf{LD}, (k,\ell)}\big\}$, 
we utilize \eqref{eq:CM_dispatching_latency} as $\tau_{n'}^{\mathsf{LD}, (k,\ell)} 
= \sum_{n\in\mathcal{N}\setminus\{n'\}} 
\bm{\Gamma}^{\mathsf{LD}, (k,\ell)}_{n, n'} \tau_{n}^{\mathsf{LD}, (k,\ell)} 
+ \tau_{n, n'}^{\mathsf{LD}, (k,\ell)}$ and also include \eqref{eq:CM_dispatching_transmission_latency} defined as 
$\tau_{n,n'}^{\mathsf{LD}, (k,\ell)} 
= \bm{\Gamma}^{\mathsf{LD}, (k,\ell)}_{n,n'}\alpha^{\mathsf{Bit}} M^{\mathsf{Dim}}\big/{\mathfrak{R}^{\mathsf{LD}, (k,\ell)}_{n,n'}}$
.
This leads to the following formulation:
\begin{equation}
\begin{aligned}
    & \max_{n'\in\mathcal{N}} 
      \Big\{ \tau_{n'}^{\mathsf{LD}, (k,\ell)} \Big\} \leq \tau^{\mathsf{LD},\mathsf{max}},
    \\& \tau_{n'}^{\mathsf{LD}, (k,\ell)} = 
      \sum_{n\in\mathcal{N}\setminus\{n'\}}
       \bm{\Gamma}^{\mathsf{LD}, (k,\ell)}_{n, n'} 
      \tau_{n}^{\mathsf{LD}, (k,\ell)} 
      + \tau_{n,n'}^{\mathsf{LD}, (k,\ell)} ,
    \\& \tau_{n,n'}^{\mathsf{LD}, (k,\ell)}  
      = \bm{\Gamma}^{\mathsf{LD}, (k,\ell)}_{n,n'}\alpha^{\mathsf{Bit}} M^{\mathsf{Dim}}\big/{\mathfrak{R}^{\mathsf{LD}, (k,\ell)}_{n,n'}}.
\end{aligned}
\end{equation}
The first expression is constructed with a maximum function, which is not permitted in the GP format. Similarly to previous constraints, we reformulated it as follows:
\begin{equation}
    \left(\sum_{n'\in\mathcal{N}} 
      \left( \tau_{n'}^{\mathsf{LD}, (k,\ell)} \right)^{p}\right)^{1/p}
      \leq \tau^{\mathsf{LD},\mathsf{max}},
\end{equation}
This expression is still not in the acceptable form of GP, and we take a similar nested approach by introducing $S^{(k,\ell)}_{(1)}(\mathbf{x})$ as follows:
\begin{equation}
    S^{(k,\ell)}_{(1)}(\mathbf{x}) = \sum_{n'\in\mathcal{N}} 
      \left( \tau_{n'}^{\mathsf{LD}, (k,\ell)} \right)^{p},
\end{equation}
which can be presented as follows to admit the GP format:
\begin{equation}
\begin{aligned}
    & \frac{\sum_{n'\in\mathcal{N}} 
      \left( \tau_{n'}^{\mathsf{LD}, (k,\ell)} \right)^{p}}
      {S^{(k,\ell)}_{(1)}(\mathbf{x})} \leq 1,
    \\& \frac{(A_{33})^{-1}  S^{(k,\ell)}_{(1)}(\mathbf{x})}{\sum_{n'\in\mathcal{N}} 
      \left( \tau_{n'}^{\mathsf{LD}, (k,\ell)} \right)^{p}} \leq 1,
    \\& (A_{33})^{-1} \geq 1.
\end{aligned}
\end{equation}
This also requires Lemma~\ref{Lemma:ArethmaticGeometric} to condense the denominator for a GP-admitted format, as follows:
\begin{equation}
H_{42}(\mathbf{x}) \triangleq \sum_{n'\in\mathcal{N}} 
      \left( \tau_{n'}^{\mathsf{LD}, (k,\ell)} \right)^{p} \Rightarrow 
      H_{42}(\mathbf{x}) \geq \widehat{H}_{42}(\mathbf{x};l) \triangleq 
      \prod_{n'\in\mathcal{N}} \left(\frac{\left( \tau_{n'}^{\mathsf{LD}, (k,\ell)} \right)^{p} H_{42}[\mathbf{x}]^{l-1}}
      {\left[\left( \tau_{n'}^{\mathsf{LD}, (k,\ell)} \right)^{p}\right]^{l-1}}\right)^{\frac{\big[\left( \tau_{n'}^{\mathsf{LD}, (k,\ell)} \right)^{p}\big]^{l-1}}{H_{42}[\mathbf{x}]^{l-1}}}.
\end{equation}
Therefore, the first expression can be transformed into the following format:
\begin{equation}
\frac{\left(S^{(k,\ell)}_{(1)}(\mathbf{x})\right)^{1/p}}{\tau^{\mathsf{LD},\mathsf{max}}} \leq 1.
\end{equation}
We also transform the second expression, which is in the form of an equality on a posynomial, as follows:
\begin{equation}
\begin{aligned}
& \frac{\sum_{n\in\mathcal{N}\setminus\{n'\}}\bm{\Gamma}^{\mathsf{LD}, (k,\ell)}_{n, n'} \tau_{n}^{\mathsf{LD}, (k,\ell)} + \tau_{n,n'}^{\mathsf{LD}, (k,\ell)} }{\tau_{n'}^{\mathsf{LD}, (k,\ell)}}\leq 1,
\\& \frac{(A_{34})^{-1} \tau_{n'}^{\mathsf{LD}, (k,\ell)}}{\sum_{n\in\mathcal{N}\setminus\{n'\}}\bm{\Gamma}^{\mathsf{LD}, (k,\ell)}_{n, n'}\tau_{n}^{\mathsf{LD}, (k,\ell)}  + \tau_{n,n'}^{\mathsf{LD}, (k,\ell)} }\leq 1,
\\& (A_{34})^{-1} \geq 1.
\end{aligned}
\end{equation}
The Lemma~\ref{Lemma:ArethmaticGeometric} is also required for the denominator as follows:
\begin{equation}
\begin{aligned}
    & H_{43}(\mathbf{x}) \triangleq \sum_{n\in\mathcal{N}\setminus\{n'\}}\bm{\Gamma}^{\mathsf{LD}, (k,\ell)}_{n, n'} \tau_{n}^{\mathsf{LD}, (k,\ell)}  + \tau_{n,n'}^{\mathsf{LD}, (k,\ell)}  
\\& \Rightarrow H_{43}(\mathbf{x}) \geq \widehat{H}_{43}(\mathbf{x};l) \triangleq  
\prod_{n\in\mathcal{N}\setminus\{n'\}}\left(\frac{\bm{\Gamma}^{\mathsf{LD}, (k,\ell)}_{n, n'}\tau_{n}^{\mathsf{LD}, (k,\ell)} H_{43}[\mathbf{x}]^{l-1}}{\left[\bm{\Gamma}^{\mathsf{LD}, (k,\ell)}_{n, n'}\tau_{n}^{\mathsf{LD}, (k,\ell)}\right]^{l-1}}\right)^{\frac{\left[\bm{\Gamma}^{\mathsf{LD}, (k,\ell)}_{n, n'}\tau_{n}^{\mathsf{LD}, (k,\ell)}\right]^{l-1}}{H_{43}[\mathbf{x}]^{l-1}}} 
\left(\frac{\tau_{n,n'}^{\mathsf{LD}, (k,\ell)}  H_{43}[\mathbf{x}]^{l-1}}{\left[\tau_{n,n'}^{\mathsf{LD}, (k,\ell)} \right]^{l-1}}\right)^{\frac{\left[\tau_{n,n'}^{\mathsf{LD}, (k,\ell)} \right]^{l-1}}{H_{43}[\mathbf{x}]^{l-1}}}.
\end{aligned}
\end{equation}
The third expression is also in the form of an equality on a monomial and with ensuring a non-zero denominator, leading to the following expression:
\begin{equation}
\begin{aligned}
    &\frac{\mathfrak{R}^{\mathsf{LD}, (k,\ell)}_{n,n'}\tau_{n,n'}^{\mathsf{LD}, (k,\ell)}  + 1}{\bm{\Gamma}^{\mathsf{LD}, (k,\ell)}_{n,n'}\alpha^{\mathsf{Bit}} M^{\mathsf{Dim}} + 1} \leq 1,
    \\& \frac{(A_{35})^{-1}(\bm{\Gamma}^{\mathsf{LD}, (k,\ell)}_{n,n'}\alpha^{\mathsf{Bit}} M^{\mathsf{Dim}} + 1)}{\mathfrak{R}^{\mathsf{LD}, (k,\ell)}_{n,n'}\tau_{n,n'}^{\mathsf{LD}, (k,\ell)} + 1} \leq 1,
    \\& (A_{35})^{-1} \geq 1,
\end{aligned}
\end{equation}
where the Lemma \ref{Lemma:ArethmaticGeometric} can be utilized for denominators as follows:
\begin{equation}
 H_{44}(\mathbf{x}) \triangleq \bm{\Gamma}^{\mathsf{LD}, (k,\ell)}_{n,n'}\alpha^{\mathsf{Bit}} M^{\mathsf{Dim}} + 1 \Rightarrow H_{44}(\mathbf{x}) \geq \widehat{H}_{44}(\mathbf{x};l) \triangleq  
 \left(\frac{\bm{\Gamma}^{\mathsf{LD}, (k,\ell)}_{n,n'} H_{44}[\mathbf{x}]^{l-1}}{[\bm{\Gamma}^{\mathsf{LD}, (k,\ell)}_{n,n'}]^{l-1}}\right)^{\frac{[\bm{\Gamma}^{\mathsf{LD}, (k,\ell)}_{n,n'}\alpha^{\mathsf{Bit}} M^{\mathsf{Dim}}]^{l-1}}{H_{44}[\mathbf{x}]^{l-1}}} 
 \left(H_{44}[\mathbf{x}]^{l-1}\right)^{\frac{1}{H_{44}[\mathbf{x}]^{l-1}}},
\end{equation}
\begin{equation}
\begin{aligned}
    & H_{45}(\mathbf{x}) \triangleq \mathfrak{R}^{\mathsf{LD}, (k,\ell)}_{n,n'}\tau_{n,n'}^{\mathsf{LD}, (k,\ell)} + 1  
  \\& \Rightarrow H_{45}(\mathbf{x}) \geq \widehat{H}_{45}(\mathbf{x};l) \triangleq  
 \left(\frac{\mathfrak{R}^{\mathsf{LD}, (k,\ell)}_{n,n'}\tau_{n,n'}^{\mathsf{LD}, (k,\ell)} H_{45}[\mathbf{x}]^{l-1}}{[\mathfrak{R}^{\mathsf{LD}, (k,\ell)}_{n,n'}\tau_{n,n'}^{\mathsf{LD}, (k,\ell)}]^{l-1}}\right)^{\frac{[\mathfrak{R}^{\mathsf{LD}, (k,\ell)}_{n,n'}\tau_{n,n'}^{\mathsf{LD}, (k,\ell)}]^{l-1}}{H_{45}[\mathbf{x}]^{l-1}}} 
 \left(H_{45}[\mathbf{x}]^{l-1}\right)^{\frac{1}{H_{45}[\mathbf{x}]^{l-1}}}.
\end{aligned}
\end{equation}
Therefore, constraints 
\eqref{eq:CM_dispatching_transmission_latency}, 
\eqref{eq:CM_dispatching_latency}, 
\eqref{cons:CM_dispatching_latency_1} 
can be presented in GP format as follows:
\begin{tcolorbox}[ams align]
    &\frac{\left(S^{(k,\ell)}_{(1)}(\mathbf{x})\right)^{1/p}}{\tau^{\mathsf{LD},\mathsf{max}}} \leq 1,
    \\& \frac{\sum_{n'\in\mathcal{N}} 
      \left( \tau_{n'}^{\mathsf{LD}, (k,\ell)} \right)^{p}}
      {S^{(k,\ell)}_{(1)}(\mathbf{x})} \leq 1,
    \\& \frac{(A_{33})^{-1}  S^{(k,\ell)}_{(1)}(\mathbf{x})}{\widehat{H}_{42}(\mathbf{x};l)} \leq 1, 
    \\& \frac{\sum_{n\in\mathcal{N}\setminus\{n'\}}\bm{\Gamma}^{\mathsf{LD}, (k,\ell)}_{n, n'} \tau_{n}^{\mathsf{LD}, (k,\ell)} + \tau_{n,n'}^{\mathsf{LD}, (k,\ell)} }{\tau_{n'}^{\mathsf{LD}, (k,\ell)}}\leq 1,
    \\& \frac{(A_{34})^{-1} \tau_{n'}^{\mathsf{LD}, (k,\ell)}}{\widehat{H}_{43}(\mathbf{x};l)}\leq 1,
    \\&\frac{\mathfrak{R}^{\mathsf{LD}, (k,\ell)}_{n,n'}\tau_{n,n'}^{\mathsf{LD}, (k,\ell)}  + 1}{\widehat{H}_{44}(\mathbf{x};l)} \leq 1,
    \\& \frac{(A_{35})^{-1}(\bm{\Gamma}^{\mathsf{LD}, (k,\ell)}_{n,n'}\alpha^{\mathsf{Bit}} M^{\mathsf{Dim}} + 1)}{\widehat{H}_{45}(\mathbf{x};l)} \leq 1,
    \\& (A_{33})^{-1} \geq 1,
    \\& (A_{34})^{-1} \geq 1,
    \\& (A_{35})^{-1} \geq 1.
\end{tcolorbox}

\textbullet \hspace{2mm} \textbf{Constraint \eqref{cons:SCD_3}}:  
The constraint 
$\sum_{t=1}^{\tau^{\mathsf{LD}, (k,\ell)}} \Big(\bm{\Gamma}^{\mathsf{LD}, (k,\ell)} - \sum_{c \in \mathcal{C}^{(k)}} \bm{\Gamma}\big(\bm{\Psi}(\mathcal{N}^{(k)}_{c},t^{\mathsf{LD},(k,\ell)}{+}t), r_{c}^{\mathsf{L}, (k)}\big)\Big) = \mathbf{0}$ is an equality on a posynomial including a negative sign. We eliminate the negative sign as follows:
\begin{equation}
    \sum_{t=1}^{\tau^{\mathsf{LD}, (k,\ell)}} \sum_{c \in \mathcal{C}^{(k)}} \bm{\Gamma}\big(\bm{\Psi}(\mathcal{N}^{(k)}_{c},t^{\mathsf{LD},(k,\ell)}{+}t), r_{c}^{\mathsf{L}, (k)}\big) = \sum_{t=1}^{\tau^{\mathsf{LD}, (k,\ell)}} \bm{\Gamma}^{\mathsf{LD}, (k,\ell)},
\end{equation}
which is an equality on matrices and not admitted in GP format. We further utilized the matrix equality transformation through the Frobenius norm as follows:
\begin{equation}
\begin{aligned}
    & \sum_{i=1}^{m}\sum_{j=1}^{n} \left[
    \left( \tau^{\mathsf{LD}, (k,\ell)} \Gamma^{\mathsf{LD}, (k,\ell)}_{i,j}\right)^{2} 
    + \left( \sum_{t=1}^{\tau^{\mathsf{LD}, (k,\ell)}} \sum_{c \in \mathcal{C}^{(k)}} \Gamma \big(\bm{\Psi}(\mathcal{N}^{(k)}_{c},t^{\mathsf{LD},(k,\ell)}{+}t), r_{c}^{\mathsf{L}, (k)}\big)_{i,j}\right)^{2} \right] 
    \\ & 
    =  2 \tau^{\mathsf{LD}, (k,\ell)} \sum_{i=1}^{m}\sum_{j=1}^{n}
      \Gamma^{\mathsf{LD}, (k,\ell)}_{i,j}
     \sum_{t=1}^{\tau^{\mathsf{LD}, (k,\ell)}} \sum_{c \in \mathcal{C}^{(k)}} \Gamma \big(\bm{\Psi}(\mathcal{N}^{(k)}_{c},t^{\mathsf{LD},(k,\ell)}{+}t), r_{c}^{\mathsf{L}, (k)}\big)_{i,j}.
\end{aligned}
\end{equation}
Furthermore, this expression requires a nested approach to comply with the GP format as follows:
\begin{equation}
   P^{\mathsf{LD},(k,\ell)}_{i,j}(\mathbf{x}) = \sum_{t=1}^{\tau^{\mathsf{LD}, (k,\ell)}} \sum_{c \in \mathcal{C}^{(k)}} \Gamma \big(\bm{\Psi}(\mathcal{N}^{(k)}_{c},t^{\mathsf{LD},(k,\ell)}{+}t), r_{c}^{\mathsf{L}, (k)}\big)_{i,j},
\end{equation}
which then can be presented in the GP format as follows:
\begin{equation}
\begin{aligned}
  & \frac{\sum_{t=1}^{\tau^{\mathsf{LD}, (k,\ell)}} \sum_{c \in \mathcal{C}^{(k)}} \Gamma \big(\bm{\Psi}(\mathcal{N}^{(k)}_{c},t^{\mathsf{LD},(k,\ell)}{+}t), r_{c}^{\mathsf{L}, (k)}\big)_{i,j}}{P^{\mathsf{LD},(k,\ell)}_{i,j}(\mathbf{x})} \leq 1,
  \\& \frac{(A_{36})^{-1} P^{\mathsf{LD},(k,\ell)}_{i,j}(\mathbf{x})}{\sum_{t=1}^{\tau^{\mathsf{LD}, (k,\ell)}} \sum_{c \in \mathcal{C}^{(k)}} \Gamma \big(\bm{\Psi}(\mathcal{N}^{(k)}_{c},t^{\mathsf{LD},(k,\ell)}{+}t), r_{c}^{\mathsf{L}, (k)}\big)_{i,j}} \leq 1,
  \\& (A_{36})^{-1} \geq 1,
\end{aligned}
\end{equation}
and the denominator can be condensed as via Lemma~\ref{Lemma:ArethmaticGeometric} as follows:
\begin{equation}
\begin{aligned}
& H_{46}(\mathbf{x}) \triangleq
\sum_{t=1}^{\tau^{\mathsf{LD}, (k,\ell)}} \sum_{c \in \mathcal{C}^{(k)}} \Gamma \big(\bm{\Psi}(\mathcal{N}^{(k)}_{c},t^{\mathsf{LD},(k,\ell)}{+}t), r_{c}^{\mathsf{L}, (k)}\big)_{i,j}
\\& \Rightarrow H_{46}(\mathbf{x}) \geq \widehat{H}_{46}(\mathbf{x};l) \triangleq
\prod_{t=1}^{\tau^{\mathsf{LD}, (k,\ell)}} \prod_{c \in \mathcal{C}^{(k)}}
\left(\frac{ \Gamma \big(\bm{\Psi}(\mathcal{N}^{(k)}_{c},t^{\mathsf{LD},(k,\ell)}{+}t), r_{c}^{\mathsf{L}, (k)}\big)_{i,j} H_{46}[\mathbf{x}]^{l-1}}
{\big[ \Gamma \big(\bm{\Psi}(\mathcal{N}^{(k)}_{c},t^{\mathsf{LD},(k,\ell)}{+}t), r_{c}^{\mathsf{L}, (k)}\big)_{i,j}\big]^{l-1}}\right)^{\frac{\big[ \Gamma \big(\bm{\Psi}(\mathcal{N}^{(k)}_{c},t^{\mathsf{LD},(k,\ell)}{+}t), r_{c}^{\mathsf{L}, (k)}\big)_{i,j}\big]^{l-1}}{H_{46}[\mathbf{x}]^{l-1}}}.
\end{aligned}
\end{equation}
This transformation leads to an equality on posynomials, which requires further steps to conform to the GP format as follows:
\begin{equation}
\begin{aligned}
    & \frac{\sum_{i=1}^{m}\sum_{j=1}^{n}
    \left( \tau^{\mathsf{LD}, (k,\ell)} \Gamma^{\mathsf{LD}, (k,\ell)}_{i,j}\right)^{2} 
    + \left( P^{\mathsf{LD},(k,\ell)}_{i,j}(\mathbf{x})\right)^{2} }{2 \tau^{\mathsf{LD}, (k,\ell)} \sum_{i=1}^{m}\sum_{j=1}^{n}
      \Gamma^{\mathsf{LD}, (k,\ell)}_{i,j} P^{\mathsf{LD},(k,\ell)}_{i,j}(\mathbf{x})} \leq 1
    \\& \frac{(A_{37})^{-1} 2 \tau^{\mathsf{LD}, (k,\ell)} \sum_{i=1}^{m}\sum_{j=1}^{n}
      \Gamma^{\mathsf{LD}, (k,\ell)}_{i,j} P^{\mathsf{LD},(k,\ell)}_{i,j}(\mathbf{x})}{\sum_{i=1}^{m}\sum_{j=1}^{n} 
    \left( \tau^{\mathsf{LD}, (k,\ell)} \Gamma^{\mathsf{LD}, (k,\ell)}_{i,j}\right)^{2} 
    + \left( P^{\mathsf{LD},(k,\ell)}_{i,j}(\mathbf{x})\right)^{2} } \leq 1,
    \\& (A_{37})^{-1} \geq 1
\end{aligned}
\end{equation}
Applying Lemma~\ref{Lemma:ArethmaticGeometric}, both denominators can be condensed as follows:
\begin{equation}
\begin{aligned}
    & H_{47}(\mathbf{x}) \triangleq 
    2 \tau^{\mathsf{LD}, (k,\ell)} \sum_{i=1}^{m}\sum_{j=1}^{n}
      \Gamma^{\mathsf{LD}, (k,\ell)}_{i,j} P^{\mathsf{LD},(k,\ell)}_{i,j}(\mathbf{x}) \Rightarrow H_{47}(\mathbf{x}) \geq \widehat{H}_{47}(\mathbf{x};l) \\&  \triangleq
      \prod_{i=1}^{m} \prod_{j=1}^{n} 
    \left(\frac{2 \tau^{\mathsf{LD}, (k,\ell)} \Gamma^{\mathsf{LD}, (k,\ell)}_{i,j} P^{\mathsf{LD},(k,\ell)}_{i,j}(\mathbf{x}) H_{47}[\mathbf{x}]^{l-1}}
    {\big[2 \tau^{\mathsf{LD}, (k,\ell)} \Gamma^{\mathsf{LD}, (k,\ell)}_{i,j} P^{\mathsf{LD},(k,\ell)}_{i,j}(\mathbf{x})\big]^{l-1}}\right)^{\frac{\big[2 \tau^{\mathsf{LD}, (k,\ell)} \Gamma^{\mathsf{LD}, (k,\ell)}_{i,j} P^{\mathsf{LD},(k,\ell)}_{i,j}(\mathbf{x})\big]^{l-1}}{H_{47}[\mathbf{x}]^{l-1}}}.
\end{aligned}
\end{equation}
\begin{equation}
\begin{aligned}
    & H_{48}(\mathbf{x}) \triangleq 
     \sum_{i=1}^{m}\sum_{j=1}^{n} 
    \left( \tau^{\mathsf{LD}, (k,\ell)} \Gamma^{\mathsf{LD}, (k,\ell)}_{i,j}\right)^{2} 
    + \left( P^{\mathsf{LD},(k,\ell)}_{i,j}(\mathbf{x})\right)^{2} \Rightarrow H_{48}(\mathbf{x}) \geq \widehat{H}_{48}(\mathbf{x};l) 
    \\&  \triangleq
    \prod_{i=1}^{m} \prod_{j=1}^{n}
    \left(\frac{\left( \tau^{\mathsf{LD}, (k,\ell)} \Gamma^{\mathsf{LD}, (k,\ell)}_{i,j}\right)^{2} H_{48}[\mathbf{x}]^{l-1}}
    {\left[ \left( \tau^{\mathsf{LD}, (k,\ell)} \Gamma^{\mathsf{LD}, (k,\ell)}_{i,j}\right)^{2}\right]^{l-1}}\right)^{\frac{\left[ \left( \tau^{\mathsf{LD}, (k,\ell)} \Gamma^{\mathsf{LD}, (k,\ell)}_{i,j}\right)^{2}\right]^{l-1}}{H_{48}[\mathbf{x}]^{l-1}}} \left(\frac{\left( P^{\mathsf{LD},(k,\ell)}_{i,j}(\mathbf{x})\right)^{2} H_{48}[\mathbf{x}]^{l-1}}
    {\left[\left( P^{\mathsf{LD},(k,\ell)}_{i,j}(\mathbf{x})\right)^{2}\right]^{l-1}}\right)^{\frac{\left[\left( P^{\mathsf{LD},(k,\ell)}_{i,j}(\mathbf{x})\right)^{2}\right]^{l-1}}{H_{48}[\mathbf{x}]^{l-1}}}.
\end{aligned}
\end{equation}
Therefore, the GP admissible form of \eqref{cons:SCD_3} is presentable as follows:
\begin{tcolorbox}[ams align]
  & \frac{\sum_{t=1}^{\tau^{\mathsf{LD}, (k,\ell)}} \sum_{c \in \mathcal{C}^{(k)}} \Gamma \big(\bm{\Psi}(\mathcal{N}^{(k)}_{c},t^{\mathsf{LD},(k,\ell)}{+}t), r_{c}^{\mathsf{L}, (k)}\big)_{i,j}}{P^{\mathsf{LD},(k,\ell)}_{i,j}(\mathbf{x})} \leq 1,
  \\& \frac{(A_{36})^{-1} P^{\mathsf{LD},(k,\ell)}_{i,j}(\mathbf{x})}{\widehat{H}_{46}(\mathbf{x};l)} \leq 1,
    \\& \frac{\sum_{i=1}^{m}\sum_{j=1}^{n}
    \left( \tau^{\mathsf{LD}, (k,\ell)} \Gamma^{\mathsf{LD}, (k,\ell)}_{i,j}\right)^{2} 
    + \left( P^{\mathsf{LD},(k,\ell)}_{i,j}(\mathbf{x})\right)^{2} }{\widehat{H}_{47}(\mathbf{x};l)} \leq 1
    \\& \frac{(A_{37})^{-1} 2 \tau^{\mathsf{LD}, (k,\ell)} \sum_{i=1}^{m}\sum_{j=1}^{n}
      \Gamma^{\mathsf{LD}, (k,\ell)}_{i,j} P^{\mathsf{LD},(k,\ell)}_{i,j}(\mathbf{x})}{\widehat{H}_{48}(\mathbf{x};l) } \leq 1,
    \\& (A_{36})^{-1} \geq 1,
    \\& (A_{37})^{-1} \geq 1
\end{tcolorbox}

\textbullet \hspace{2mm} \textbf{Constraint \eqref{LDenergy}}:  
The constraint $E^{\mathsf{LD},(k,\ell)} = \sum_{n\in\mathcal{N}} \sum_{n'\in\mathcal{N}} \tau_{n,n'}^{\mathsf{LD},(k,\ell)} P_{n}$ is an equality on a posynomial and requires the following transformation:
\begin{equation}
\begin{aligned}
    & \frac{\sum_{n\in\mathcal{N}} \sum_{n'\in\mathcal{N}}
    \tau_{n,n'}^{\mathsf{LD},(k,\ell)} P_{n}}{E^{\mathsf{LD},(k,\ell)}} \leq 1, \\
    & \frac{(A_{38})^{-1} E^{\mathsf{LD},(k,\ell)}}{\sum_{n\in\mathcal{N}} \sum_{n'\in\mathcal{N}}
    \tau_{n,n'}^{\mathsf{LD},(k,\ell)} P_{n}} \leq 1, \\
    & (A_{38})^{-1} \geq 1.
\end{aligned}
\end{equation}
Using Lemma~\ref{Lemma:ArethmaticGeometric}, the denominator is condensed as follows:
\begin{equation}
\begin{aligned}
     H_{49}(\mathbf{x}) \triangleq 
    \sum_{n\in\mathcal{N}} \sum_{n'\in\mathcal{N}} 
    \tau_{n,n'}^{\mathsf{LD},(k,\ell)} P_{n}
    \;\Rightarrow\;
    H_{49}(\mathbf{x}) \geq \widehat{H}_{49}(\mathbf{x};l) \triangleq 
    \prod_{n\in\mathcal{N}} \prod_{n'\in\mathcal{N}}
    \left(\frac{\tau_{n,n'}^{\mathsf{LD},(k,\ell)} P_{n} H_{49}[\mathbf{x}]^{l-1}}
    {[\tau_{n,n'}^{\mathsf{LD},(k,\ell)} P_{n}]^{l-1}}\right)^{\frac{[\tau_{n,n'}^{\mathsf{LD},(k,\ell)} P_{n}]^{l-1}}{H_{49}[\mathbf{x}]^{l-1}}}.
\end{aligned}
\end{equation}
Therefore, the GP admissible version of \eqref{LDenergy} is shown as follows:
\begin{tcolorbox}[ams align]
    & \frac{\sum_{n\in\mathcal{N}} \sum_{n'\in\mathcal{N}}
    \tau_{n,n'}^{\mathsf{LD},(k,\ell)} P_{n}}{E^{\mathsf{LD},(k,\ell)}} \leq 1, \\
    & \frac{(A_{38})^{-1} E^{\mathsf{LD},(k,\ell)}}{\widehat{H}_{49}(\mathbf{x};l)} \leq 1, \\
    & (A_{38})^{-1} \geq 1.
\end{tcolorbox}

\textbullet \hspace{2mm} \textbf{Constraint \eqref{cons:SGA_1}}: The constraint $\sum_{n'\in\mathcal{N}}\bm{\Gamma}^{\mathsf{GA}, (k)}_{n,n'} + \pi_{n}^{\mathsf{A}, (k)}{=}1,~\forall n\in\mathcal{N}$ is another version of the previous constraints and a similar approach can be taken to admit the GP admitted format:
\begin{equation}
\begin{aligned}
    &\sum_{n'\in\mathcal{N}}\bm{\Gamma}^{\mathsf{GA}, (k)}_{n,n'} + \pi_{n}^{\mathsf{A}, (k)}\leq 1,
    \\& \frac{(A_{39})^{-1}}{\sum_{n'\in\mathcal{N}}\bm{\Gamma}^{\mathsf{GA}, (k)}_{n,n'} + \pi_{n}^{\mathsf{A}, (k)}} \leq 1,
    \\& (A_{39})^{-1} \geq 1.
\end{aligned}
\end{equation}
The Lemma \ref{Lemma:ArethmaticGeometric} is used for denominator condensation as follows:
\begin{equation}
\begin{aligned}
&H_{50}(\mathbf{x}) \triangleq \sum_{n'\in\mathcal{N}}\bm{\Gamma}^{\mathsf{GA}, (k)}_{n,n'} + \pi_{n}^{\mathsf{A}, (k)} \Rightarrow H_{50}(\mathbf{x}) \geq \widehat{H}_{50}(\mathbf{x};l) \triangleq \prod_{n'\in\mathcal{N}} \left(\frac{\bm{\Gamma}^{\mathsf{GA}, (k)}_{n,n'} H_{50}[\mathbf{x}]^{l-1}}{[\bm{\Gamma}^{\mathsf{GA}, (k)}_{n,n'}]^{l-1}}\right)^{\frac{[\bm{\Gamma}^{\mathsf{GA}, (k)}_{n,n'}]^{l-1}}{H_{50}[\mathbf{x}]^{l-1}}} \left(\frac{\pi_{n}^{\mathsf{A}, (k)} H_{50}[\mathbf{x}]^{l-1}}{[\pi_{n}^{\mathsf{A}, (k)}]^{l-1}}\right)^{\frac{[\pi_{n}^{\mathsf{A}, (k)}]^{l-1}}{H_{50}[\mathbf{x}]^{l-1}}}.
\end{aligned}
\end{equation}
Therefore, the constraint \eqref{cons:SGA_1} can be presented in the GP format as follows:
\begin{tcolorbox}[ams align]
    &\sum_{n'\in\mathcal{N}}\bm{\Gamma}^{\mathsf{GA}, (k)}_{n,n'} + \pi_{n}^{\mathsf{A}, (k)}\leq 1,
    \\& \frac{(A_{39})^{-1}}{\widehat{H}_{50}(\mathbf{x};l)} \leq 1,
    \\& (A_{39})^{-1} \geq 1.
\end{tcolorbox}

\textbullet \hspace{2mm} \textbf{Constraints \eqref{eq:GM_aggregation_transmission_latency},\eqref{eq:GM_aggregation_latency},\eqref{cons:GM_aggregation_latency_1}}:  
For the constraint \eqref{cons:GM_aggregation_latency_1} as $\tau^{\mathsf{GA},(k)} \leq \tau^{\mathsf{GA},\mathsf{max}},\forall k \in\mathcal{K}$, via knowing that $\tau^{\mathsf{GA},(k)}=\tau_{n'}^{\mathsf{GA}, (k)}$ where $n' = r^{\mathsf{A}, (k)}$ and using \eqref{eq:GM_aggregation_latency} as
$ \tau_{n'}^{\mathsf{GA}, (k)} 
  = \max_{n\in\mathcal{N}\setminus\{n'\}} \Big\{ \bm{\Gamma}^{\mathsf{GA}, (k)}_{n, n'} \big(\tau_{n,n'}^{\mathsf{GA}, (k)} + \tau_{n}^{\mathsf{GA}, (k)}\big)\Big\}$
and by using \eqref{eq:GM_aggregation_transmission_latency} we also have $\tau_{n,n'}^{\mathsf{GA}, (k)} = \bm{\Gamma}^{\mathsf{GA}, (k)}_{n,n'}\alpha^{\mathsf{Bit}} M^{\mathsf{Dim}}\big/{\mathfrak{R}^{\mathsf{GA}, (k)}_{n,n'}}$. This leads to the following formulation:
\begin{equation}
\begin{aligned}
    & \tau_{n'}^{\mathsf{GA}, (k)} \leq \tau^{\mathsf{GA},\mathsf{max}},
    \\& \tau_{n'}^{\mathsf{GA}, (k)} 
    = \max_{n\in\mathcal{N}\setminus\{n'\}}
      \Big\{\bm{\Gamma}^{\mathsf{GA}, (k)}_{n, n'} \tau_{n}^{\mathsf{GA}, (k)}
      + \tau_{n,n'}^{\mathsf{GA}, (k)} \Big\},
      \\& \tau_{n,n'}^{\mathsf{GA}, (k)}  = \bm{\Gamma}^{\mathsf{GA}, (k)}_{n,n'}\alpha^{\mathsf{Bit}} M^{\mathsf{Dim}}\big/{\mathfrak{R}^{\mathsf{GA}, (k)}_{n,n'}}.
\end{aligned}
\end{equation}
The first expression can be transformed accordingly:
\begin{equation}
     \frac{\tau_{n'}^{\mathsf{GA}, (k)}}{\tau^{\mathsf{GA},\mathsf{max}}} \leq 1,
\end{equation}
For the second expression, since it contains a maximum function, we further utilize the corresponding approximation as follows:
\begin{equation}
 \tau_{n'}^{\mathsf{GA}, (k)} 
    = \left(\sum_{n\in\mathcal{N}\setminus\{n'\}}
      \left(\bm{\Gamma}^{\mathsf{GA}, (k)}_{n, n'} \tau_{n}^{\mathsf{GA}, (k)}
      + \tau_{n,n'}^{\mathsf{GA}, (k)} \right)^{p}\right)^{1/p},
\end{equation}
where $p$ is a large value. In order to adhere to the GP format, we further utilize the nested approach and introduce $K^{(k)}_{(1)}(\mathbf{x})$ and $K^{(k)}_{(2),n}(\mathbf{x})$, which are defined as follows:
\begin{equation}
\begin{aligned}
& \tau_{n'}^{\mathsf{GA}, (k)} = \left(K^{(k)}_{(1)}(\mathbf{x})\right)^{1/p}
 \\& K^{(k)}_{(1)}(\mathbf{x}) 
    = \sum_{n\in\mathcal{N}\setminus\{n'\}}
      \left(K^{(k)}_{(2),n}(\mathbf{x})\right)^{p},
\\& K^{(k)}_{(2),n}(\mathbf{x}) 
    = \bm{\Gamma}^{\mathsf{GA}, (k)}_{n, n'} \tau_{n}^{\mathsf{GA}, (k)}
      + \tau_{n,n'}^{\mathsf{GA}, (k)}.
\end{aligned}
\end{equation}
These constraints can be transformed into GP allowed format with the following transformations:
\begin{equation}
\begin{aligned}
& \tau_{n'}^{\mathsf{GA}, (k)} \left(K^{(k)}_{(1)}(\mathbf{x})\right)^{-1/p} =1
\\& \frac{\sum_{n\in\mathcal{N}\setminus\{n'\}}
      \left(K^{(k)}_{(2),n}(\mathbf{x})\right)^{p}}{K^{(k)}_{(1)}(\mathbf{x})} \leq 1,
 \\& \frac{(A_{40})^{-1} K^{(k)}_{(1)}(\mathbf{x})}{\sum_{n\in\mathcal{N}\setminus\{n'\}}
      \left(K^{(k)}_{(2),n}(\mathbf{x})\right)^{p}} \leq 1
\\& \frac{\bm{\Gamma}^{\mathsf{GA}, (k)}_{n, n'} \tau_{n}^{\mathsf{GA}, (k)}
      + \tau_{n,n'}^{\mathsf{GA}, (k)} }{K^{(k)}_{(2),n}(\mathbf{x})} \leq 1,
\\& \frac{ (A_{41})^{-1}  K^{(k)}_{(2),n}(\mathbf{x})}{\bm{\Gamma}^{\mathsf{GA}, (k)}_{n, n'} \tau_{n}^{\mathsf{GA}, (k)}
      + \tau_{n,n'}^{\mathsf{GA}, (k)} } \leq 1,
\\& (A_{40})^{-1} \geq 1,
\\& (A_{41})^{-1} \geq 1,
\end{aligned}
\end{equation}
where denominators must be condensed via Lemma  \ref{Lemma:ArethmaticGeometric} as follows:
\begin{equation}
    H_{51}(\mathbf{x}) \triangleq \sum_{n\in\mathcal{N}\setminus\{n'\}}
      \left(K^{(k)}_{(2),n}(\mathbf{x})\right)^{p} \Rightarrow H_{51}(\mathbf{x}) \geq \widehat{H}_{51}(\mathbf{x};l) \triangleq  
 \prod_{n\in\mathcal{N}\setminus\{n'\}} \left(\frac{\left(K^{(k)}_{(2),n}(\mathbf{x})\right)^{p} H_{51}[\mathbf{x}]^{l-1}}{\left[\left(K^{(k)}_{(2),n}(\mathbf{x})\right)^{p}\right]^{l-1}}\right)^{\frac{[\left(K^{(k)}_{(2),n}(\mathbf{x})\right)^{p}]^{l-1}}{H_{51}[\mathbf{x}]^{l-1}}},
\end{equation}
\begin{equation}
\begin{aligned}
& H_{52}(\mathbf{x}) \triangleq \bm{\Gamma}^{\mathsf{GA}, (k)}_{n, n'} \tau_{n}^{\mathsf{GA}, (k)}
      + \tau_{n,n'}^{\mathsf{GA}, (k)}  
\\& \Rightarrow H_{52}(\mathbf{x}) \geq \widehat{H}_{52}(\mathbf{x};l) \triangleq  
 \left(\frac{\bm{\Gamma}^{\mathsf{GA}, (k)}_{n, n'} \tau_{n}^{\mathsf{GA}, (k)} H_{52}[\mathbf{x}]^{l-1}}{[\bm{\Gamma}^{\mathsf{GA}, (k)}_{n, n'} \tau_{n}^{\mathsf{GA}, (k)}]^{l-1}}\right)^{\frac{[\bm{\Gamma}^{\mathsf{GA}, (k)}_{n, n'} \tau_{n}^{\mathsf{GA}, (k)}\alpha^{\mathsf{Bit}} M^{\mathsf{Dim}}]^{l-1}}{H_{52}[\mathbf{x}]^{l-1}}} 
 \left( \frac{\tau_{n,n'}^{\mathsf{GA}, (k)}  H_{52}[\mathbf{x}]^{l-1}}{[\tau_{n,n'}^{\mathsf{GA}, (k)} ]^{l-1}}\right)^{\frac{[\tau_{n,n'}^{\mathsf{GA}, (k)}]^{l-1}}{H_{52}[\mathbf{x}]^{l-1}}}.
 \end{aligned}
 \end{equation}
Finally, the third expression is an equality on a monomial, and the GP format is admissible after adding a constant to avoid zero denominators:
\begin{equation}
\begin{aligned}
    &\frac{\mathfrak{R}^{\mathsf{GA}, (k)}_{n,n'} \tau_{n,n'}^{\mathsf{GA}, (k)}  + 1}{\bm{\Gamma}^{\mathsf{GA}, (k)}_{n,n'} \alpha^{\mathsf{Bit}} M^{\mathsf{Dim}} + 1} \leq 1,
    \\& \frac{(A_{42})^{-1}(\bm{\Gamma}^{\mathsf{GA}, (k)}_{n,n'} \alpha^{\mathsf{Bit}} M^{\mathsf{Dim}} + 1)}{\mathfrak{R}^{\mathsf{GA}, (k)}_{n,n'} \tau_{n,n'}^{\mathsf{GA}, (k)} + 1} \leq 1,
    \\& (A_{42})^{-1} \geq 1.
\end{aligned}
\end{equation}
The denominators can be condensed with Lemma \ref{Lemma:ArethmaticGeometric} as follows:
\begin{equation}
 H_{53}(\mathbf{x}) \triangleq \bm{\Gamma}^{\mathsf{GA}, (k)}_{n,n'} \alpha^{\mathsf{Bit}} M^{\mathsf{Dim}} + 1 \Rightarrow H_{53}(\mathbf{x}) \geq \widehat{H}_{53}(\mathbf{x};l) \triangleq  
 \left(\frac{\bm{\Gamma}^{\mathsf{GA}, (k)}_{n,n'} H_{53}[\mathbf{x}]^{l-1}}{[\bm{\Gamma}^{\mathsf{GA}, (k)}_{n,n'}]^{l-1}}\right)^{\frac{[\bm{\Gamma}^{\mathsf{GA}, (k)}_{n,n'}\alpha^{\mathsf{Bit}} M^{\mathsf{Dim}}]^{l-1}}{H_{53}[\mathbf{x}]^{l-1}}} 
 \left(H_{53}[\mathbf{x}]^{l-1}\right)^{\frac{1}{H_{53}[\mathbf{x}]^{l-1}}},
\end{equation}
\begin{equation}
 H_{54}(\mathbf{x}) \triangleq \mathfrak{R}^{\mathsf{GA}, (k)}_{n,n'} \tau_{n,n'}^{\mathsf{GA}, (k)} + 1 \Rightarrow H_{54}(\mathbf{x}) \geq \widehat{H}_{54}(\mathbf{x};l) \triangleq  
 \left(\frac{\mathfrak{R}^{\mathsf{GA}, (k)}_{n,n'} \tau_{n,n'}^{\mathsf{GA}, (k)} H_{54}[\mathbf{x}]^{l-1}}{[\mathfrak{R}^{\mathsf{GA}, (k)}_{n,n'} \tau_{n,n'}^{\mathsf{GA}, (k)}]^{l-1}}\right)^{\frac{[\mathfrak{R}^{\mathsf{GA}, (k)}_{n,n'} \tau_{n,n'}^{\mathsf{GA}, (k)}]^{l-1}}{H_{54}[\mathbf{x}]^{l-1}}} 
 \left(H_{54}[\mathbf{x}]^{l-1}\right)^{\frac{1}{H_{54}[\mathbf{x}]^{l-1}}}.
\end{equation}
Therefore, constraints \eqref{eq:GM_aggregation_transmission_latency},\eqref{eq:GM_aggregation_latency},\eqref{cons:GM_aggregation_latency_1} can be transformed into GP allowed format as follows:
\begin{tcolorbox}[ams align]
     & \frac{\tau_{n'}^{\mathsf{GA}, (k)}}{\tau^{\mathsf{GA},\mathsf{max}}} \leq 1,
\\& \tau_{n'}^{\mathsf{GA}, (k)} \left(K^{(k)}_{(1)}(\mathbf{x})\right)^{-1/p} =1
\\& \frac{\sum_{n\in\mathcal{N}\setminus\{n'\}}
      \left(K^{(k)}_{(2),n}(\mathbf{x})\right)^{p}}{K^{(k)}_{(1)}(\mathbf{x})} \leq 1,
 \\& \frac{(A_{40})^{-1} K^{(k)}_{(1)}(\mathbf{x})}{\widehat{H}_{51}(\mathbf{x};l)} \leq 1
\\& \frac{\bm{\Gamma}^{\mathsf{GA}, (k)}_{n, n'} \tau_{n}^{\mathsf{GA}, (k)}
      + \tau_{n,n'}^{\mathsf{GA}, (k)} }{K^{(k)}_{(2),n}(\mathbf{x})} \leq 1,
\\& \frac{(A_{41})^{-1}  K^{(k)}_{(2),n}(\mathbf{x})}{\widehat{H}_{52}(\mathbf{x};l)} \leq 1,
\\&\frac{\mathfrak{R}^{\mathsf{GA}, (k)}_{n,n'} \tau_{n,n'}^{\mathsf{GA}, (k)}  + 1}{\widehat{H}_{53}(\mathbf{x};l)} \leq 1,
    \\& \frac{(A_{42})^{-1}(\bm{\Gamma}^{\mathsf{GA}, (k)}_{n,n'} \alpha^{\mathsf{Bit}} M^{\mathsf{Dim}} + 1)}{\widehat{H}_{54}(\mathbf{x};l)} \leq 1,
    \\& (A_{40})^{-1} \geq 1,
    \\& (A_{41})^{-1} \geq 1,
    \\& (A_{42})^{-1} \geq 1.
\end{tcolorbox}

\textbullet \hspace{2mm} \textbf{Constraint \eqref{cons:SGA_3}}:  
The constraint 
$\sum_{t=1}^{\tau^{\mathsf{GA}, (k)}} \Big(\bm{\Gamma}^{\mathsf{GA}, (k)} - \bm{\Gamma}\big(\bm{\Psi}(\mathcal{N},t^{\mathsf{GA},(k)}{+}t), r^{\mathsf{A}, (k)}\big)\Big) = \mathbf{0}$ is an equality on a posynomial with a negative sign. We first eliminate the negative sign as follows:
\begin{equation}
\sum_{t=1}^{\tau^{\mathsf{GA}, (k)}}
\bm{\Gamma}\big(\bm{\Psi}(\mathcal{N},t^{\mathsf{GA},(k)}{+}t), r^{\mathsf{A}, (k)}\big)
= \sum_{t=1}^{\tau^{\mathsf{GA}, (k)}} \bm{\Gamma}^{\mathsf{GA}, (k)},
\end{equation}
which is an equality on matrices, which is not applicable in GP format. We further utilized the matrix equality transformation through the Frobenius norm as follows:
\begin{equation}
\begin{aligned}
& \sum_{i=1}^{m}\sum_{j=1}^{n} \left[
\left( \tau^{\mathsf{GA}, (k)} \Gamma^{\mathsf{GA}, (k)}_{i,j}\right)^{2}
+ \left(\sum_{t=1}^{\tau^{\mathsf{GA}, (k)}} \Gamma \big(\bm{\Psi}(\mathcal{N},t^{\mathsf{GA},(k)}{+}t), r^{\mathsf{A}, (k)}\big)_{i,j}\right)^{2} \right]
\\ & = 2 \tau^{\mathsf{GA}, (k)} \sum_{i=1}^{m}\sum_{j=1}^{n} \Gamma^{\mathsf{GA}, (k)}_{i,j}  \sum_{t=1}^{\tau^{\mathsf{GA}, (k)}} \Gamma \big(\bm{\Psi}(\mathcal{N},t^{\mathsf{GA},(k)}{+}t), r^{\mathsf{A}, (k)}\big)_{i,j}.
\end{aligned}
\end{equation}
This reformulation eliminates the direct matrix equality and expresses the constraint as a valid equality over scalar sums. However, it requires a nested approach to comply with the GP format as follows:
\begin{equation}
  P^{\mathsf{GA},(k)}_{i,j}(\mathbf{x}) = \sum_{t=1}^{\tau^{\mathsf{GA}, (k)}} \Gamma \big(\bm{\Psi}(\mathcal{N},t^{\mathsf{GA},(k)}{+}t), r^{\mathsf{A}, (k)}\big)_{i,j},
\end{equation}
which then can be presented in the GP format as follows:
\begin{equation}
\begin{aligned}
  & \frac{\sum_{t=1}^{\tau^{\mathsf{GA}, (k)}} \Gamma \big(\bm{\Psi}(\mathcal{N},t^{\mathsf{GA},(k)}{+}t), r^{\mathsf{A}, (k)}\big)_{i,j}}{ P^{\mathsf{GA},(k)}_{i,j}(\mathbf{x})} \leq 1,
  \\& \frac{(A_{43})^{-1}  P^{\mathsf{GA},(k)}_{i,j}(\mathbf{x})}{\sum_{t=1}^{\tau^{\mathsf{GA}, (k)}} \Gamma \big(\bm{\Psi}(\mathcal{N},t^{\mathsf{GA},(k)}{+}t), r^{\mathsf{A}, (k)}\big)_{i,j}} \leq 1,
  \\& (A_{43})^{-1} \geq 1,
\end{aligned}
\end{equation}
where the denominator can be condensed as via Lemma~\ref{Lemma:ArethmaticGeometric} as follows:
\begin{equation}
\begin{aligned}
& H_{55}(\mathbf{x}) \triangleq
\sum_{t=1}^{\tau^{\mathsf{GA}, (k)}} \Gamma \big(\bm{\Psi}(\mathcal{N},t^{\mathsf{GA},(k)}{+}t), r^{\mathsf{A}, (k)}\big)_{i,j} 
\\& \Rightarrow H_{55}(\mathbf{x}) \geq \widehat{H}_{55}(\mathbf{x};l) \triangleq
\prod_{t=1}^{\tau^{\mathsf{GA}, (k)}}
\left(\frac{ \Gamma \big(\bm{\Psi}(\mathcal{N},t^{\mathsf{GA},(k)}{+}t), r^{\mathsf{A}, (k)}\big)_{i,j} H_{55}[\mathbf{x}]^{l-1}}
{\big[ \Gamma \big(\bm{\Psi}(\mathcal{N},t^{\mathsf{GA},(k)}{+}t), r^{\mathsf{A}, (k)}\big)_{i,j}\big]^{l-1}}\right)^{\frac{\big[ \Gamma \big(\bm{\Psi}(\mathcal{N},t^{\mathsf{GA},(k)}{+}t), r^{\mathsf{A}, (k)}\big)_{i,j}\big]^{l-1}}{H_{55}[\mathbf{x}]^{l-1}}}.
\end{aligned}
\end{equation}
This transformation leads to an equality on posynomials, which requires further steps to conform to the GP format as follows:
\begin{equation}
\begin{aligned}
& \frac{\sum_{i=1}^{m}\sum_{j=1}^{n} 
\left( \tau^{\mathsf{GA}, (k)} \Gamma^{\mathsf{GA}, (k)}_{i,j}\right)^{2}
+ \left( P^{\mathsf{GA},(k)}_{i,j}(\mathbf{x})\right)^{2}}{2 \tau^{\mathsf{GA}, (k)} \sum_{i=1}^{m}\sum_{j=1}^{n}
\Gamma^{\mathsf{GA}, (k)}_{i,j}
 P^{\mathsf{GA},(k)}_{i,j}(\mathbf{x})} \leq 1,
\\& \frac{(A_{44})^{-1} 2 \tau^{\mathsf{GA}, (k)} \sum_{i=1}^{m}\sum_{j=1}^{n}
\Gamma^{\mathsf{GA}, (k)}_{i,j}
 P^{\mathsf{GA},(k)}_{i,j}(\mathbf{x})}{\sum_{i=1}^{m}\sum_{j=1}^{n} 
\left( \tau^{\mathsf{GA}, (k)} \Gamma^{\mathsf{GA}, (k)}_{i,j}\right)^{2}
+ \left( P^{\mathsf{GA},(k)}_{i,j}(\mathbf{x})\right)^{2}} \leq 1,
\\& (A_{44})^{-1} \geq 1.
\end{aligned}
\end{equation}
Applying Lemma~\ref{Lemma:ArethmaticGeometric}, both denominators can be condensed as follows:
\begin{equation}
\begin{aligned}
& H_{56}(\mathbf{x}) \triangleq
2 \tau^{\mathsf{GA}, (k)} \sum_{i=1}^{m}\sum_{j=1}^{n}
\Gamma^{\mathsf{GA}, (k)}_{i,j}
 P^{\mathsf{GA},(k)}_{i,j}(\mathbf{x}) \Rightarrow H_{56}(\mathbf{x}) \geq \widehat{H}_{56}(\mathbf{x};l) \\& \triangleq
 \prod_{i=1}^{m} \prod_{j=1}^{n}
\left(\frac{ 2 \tau^{\mathsf{GA}, (k)} \Gamma^{\mathsf{GA}, (k)}_{i,j}
 P^{\mathsf{GA},(k)}_{i,j}(\mathbf{x}) H_{56}[\mathbf{x}]^{l-1}}
{\big[ 2 \tau^{\mathsf{GA}, (k)} \Gamma^{\mathsf{GA}, (k)}_{i,j}
 P^{\mathsf{GA},(k)}_{i,j}(\mathbf{x})\big]^{l-1}}\right)^{\frac{\big[ 2 \tau^{\mathsf{GA}, (k)} \Gamma^{\mathsf{GA}, (k)}_{i,j}
 P^{\mathsf{GA},(k)}_{i,j}(\mathbf{x})\big]^{l-1}}{H_{56}[\mathbf{x}]^{l-1}}}.
\end{aligned}
\end{equation}
\begin{equation}
\begin{aligned}
& H_{57}(\mathbf{x}) \triangleq
\sum_{i=1}^{m}\sum_{j=1}^{n} 
\left( \tau^{\mathsf{GA}, (k)} \Gamma^{\mathsf{GA}, (k)}_{i,j}\right)^{2}
+ \left( P^{\mathsf{GA},(k)}_{i,j}(\mathbf{x})\right)^{2} \Rightarrow H_{57}(\mathbf{x}) \geq \widehat{H}_{57}(\mathbf{x};l) \\& \triangleq 
\prod_{i=1}^{m} \prod_{j=1}^{n}
\left(\frac{\left( \tau^{\mathsf{GA}, (k)} \Gamma^{\mathsf{GA}, (k)}_{i,j}\right)^{2} H_{57}[\mathbf{x}]^{l-1}}
{\left[ \left( \tau^{\mathsf{GA}, (k)} \Gamma^{\mathsf{GA}, (k)}_{i,j}\right)^{2}\right]^{l-1}}\right)^{\frac{\left[ \left( \tau^{\mathsf{GA}, (k)} \Gamma^{\mathsf{GA}, (k)}_{i,j}\right)^{2}\right]^{l-1}}{H_{57}[\mathbf{x}]^{l-1}}} \left(\frac{\left( P^{\mathsf{GA},(k)}_{i,j}(\mathbf{x})\right)^{2} H_{57}[\mathbf{x}]^{l-1}}
{\left[ \left( P^{\mathsf{GA},(k)}_{i,j}(\mathbf{x})\right)^{2}\right]^{l-1}}\right)^{\frac{\left[ \left( P^{\mathsf{GA},(k)}_{i,j}(\mathbf{x})\right)^{2}\right]^{l-1}}{H_{57}[\mathbf{x}]^{l-1}}}.
\end{aligned}
\end{equation}
Thus, the GP accepted form of \eqref{cons:SGA_3} can be expressed as follows:
\begin{tcolorbox}[ams align]
  & \frac{\sum_{t=1}^{\tau^{\mathsf{GA}, (k)}} \Gamma \big(\bm{\Psi}(\mathcal{N},t^{\mathsf{GA},(k)}{+}t), r^{\mathsf{A}, (k)}\big)_{i,j}}{ P^{\mathsf{GA},(k)}_{i,j}(\mathbf{x})} \leq 1,
  \\& \frac{(A_{43})^{-1}  P^{\mathsf{GA},(k)}_{i,j}(\mathbf{x})}{\widehat{H}_{55}(\mathbf{x};l)} \leq 1,
  \\& \frac{\sum_{i=1}^{m}\sum_{j=1}^{n} 
\left( \tau^{\mathsf{GA}, (k)} \Gamma^{\mathsf{GA}, (k)}_{i,j}\right)^{2}
+ \left( P^{\mathsf{GA},(k)}_{i,j}(\mathbf{x})\right)^{2}}{\widehat{H}_{56}(\mathbf{x};l)} \leq 1,
\\& \frac{(A_{44})^{-1} 2 \tau^{\mathsf{GA}, (k)} \sum_{i=1}^{m}\sum_{j=1}^{n}
\Gamma^{\mathsf{GA}, (k)}_{i,j}
 P^{\mathsf{GA},(k)}_{i,j}(\mathbf{x})}{\widehat{H}_{57}(\mathbf{x};l)} \leq 1,
\\& (A_{43})^{-1} \geq 1,
\\& (A_{44})^{-1} \geq 1.
\end{tcolorbox}

\textbullet \hspace{2mm} \textbf{Constraint \eqref{GAenergy}}:  
The constraint $E^{\mathsf{GA}, (k)} = \sum_{n\in\mathcal{N}} \sum_{n'\in\mathcal{N}} \tau_{n,n'}^{\mathsf{GA}, (k)} P_{n}$ is an equality on a posynomial and requires the following transformation:
\begin{equation}
\begin{aligned}
    & \frac{\sum_{n\in\mathcal{N}} \sum_{n'\in\mathcal{N}}
    \tau_{n,n'}^{\mathsf{GA}, (k)} P_{n}}{E^{\mathsf{GA}, (k)}} \leq 1, \\
    & \frac{(A_{45})^{-1} E^{\mathsf{GA}, (k)}}{\sum_{n\in\mathcal{N}} \sum_{n'\in\mathcal{N}}
    \tau_{n,n'}^{\mathsf{GA}, (k)} P_{n}} \leq 1, \\
    & (A_{45})^{-1} \geq 1.
\end{aligned}
\end{equation}
Further utilizing Lemma~\ref{Lemma:ArethmaticGeometric}, the denominator can be condensed as follows:
\begin{equation}
\begin{aligned}
     H_{58}(\mathbf{x}) \triangleq 
    \sum_{n\in\mathcal{N}} \sum_{n'\in\mathcal{N}} 
    \tau_{n,n'}^{\mathsf{GA}, (k)} P_{n}
    \;\Rightarrow\;
    H_{58}(\mathbf{x}) \geq \widehat{H}_{58}(\mathbf{x};l) \triangleq 
    \prod_{n\in\mathcal{N}} \prod_{n'\in\mathcal{N}}
    \left(\frac{\tau_{n,n'}^{\mathsf{GA}, (k)} P_{n} H_{58}[\mathbf{x}]^{l-1}}
    {[\tau_{n,n'}^{\mathsf{GA}, (k)} P_{n}]^{l-1}}\right)^{\frac{[\tau_{n,n'}^{\mathsf{GA}, (k)} P_{n}]^{l-1}}{H_{58}[\mathbf{x}]^{l-1}}}.
\end{aligned}
\end{equation}
Therefore, the GP admissible form of \eqref{GAenergy} is shown as:
\begin{tcolorbox}[ams align]
    & \frac{\sum_{n\in\mathcal{N}} \sum_{n'\in\mathcal{N}}
    \tau_{n,n'}^{\mathsf{GA}, (k)} P_{n}}{E^{\mathsf{GA}, (k)}} \leq 1, \\
    & \frac{(A_{45})^{-1} E^{\mathsf{GA}, (k)}}{\widehat{H}_{58}(\mathbf{x};l)} \leq 1, \\
    & (A_{45})^{-1} \geq 1.
\end{tcolorbox}

\textbullet \hspace{2mm} \textbf{Constraint \eqref{eq:idletime}}: The constraint $\Omega^{(k,\ell)} {=} \tau^{\mathsf{LT}, \mathsf{max}} - \tau^{\mathsf{LT},{(k,\ell)}},~\ell{\in}\mathcal{L}^{(k)}, k \in\mathcal{K}$
can be rewritten with knowing that $\Omega^{(k)} = \sum_{\ell \in \mathcal{L}^{(k)}} \Omega^{(k,\ell)}$ and to avoid negative sign as follows:
\begin{equation}
\sum_{\ell \in \mathcal{L}^{(k)}} \Omega^{(k,\ell)} + \tau^{\mathsf{LT},{(k,\ell)}} = \tau^{\mathsf{LT}, \mathsf{max}},
\end{equation}
which leads to an equality of posynomial and is not permitted by the GP format. Therefore, we take the following transformation steps:
\begin{equation}
\begin{aligned}
    & \frac{\sum_{\ell \in \mathcal{L}^{(k)}} \Omega^{(k,\ell)} + \tau^{\mathsf{LT},{(k,\ell)}}}{\tau^{\mathsf{LT}, \mathsf{max}}} \leq 1, 
    \\ & \frac{(A_{46})^{-1} \tau^{\mathsf{LT}, \mathsf{max}}}{\sum_{\ell \in \mathcal{L}^{(k)}} \Omega^{(k,\ell)} + \tau^{\mathsf{LT},{(k,\ell)}}} \leq 1, 
    \\ & (A_{46})^{-1} \geq 1.
\end{aligned}
\end{equation}
We require the denominator to be condensed with Lemma~\ref{Lemma:ArethmaticGeometric} as follows:
\begin{equation}
\begin{aligned}
     & H_{59}(\mathbf{x}) \triangleq 
    \sum_{\ell \in \mathcal{L}^{(k)}} \Omega^{(k,\ell)} + \tau^{\mathsf{LT},{(k,\ell)}}
    \\& \Rightarrow
    H_{59}(\mathbf{x}) \geq \widehat{H}_{59}(\mathbf{x};l) \triangleq 
    \prod_{\ell \in \mathcal{L}^{(k)}}
    \left(\frac{\Omega^{(k,\ell)} H_{59}[\mathbf{x}]^{l-1}}
    {[\Omega^{(k,\ell)}]^{l-1}}\right)^{\frac{[\Omega^{(k,\ell)}]^{l-1}}{H_{59}[\mathbf{x}]^{l-1}}} \left(\frac{\tau^{\mathsf{LT},{(k,\ell)}} H_{59}[\mathbf{x}]^{l-1}}
    {[\tau^{\mathsf{LT},{(k,\ell)}}]^{l-1}}\right)^{\frac{[\tau^{\mathsf{LT},{(k,\ell)}}]^{l-1}}{H_{59}[\mathbf{x}]^{l-1}}}.
\end{aligned}
\end{equation}
Therefore, constraint \eqref{eq:idletime} can be driven in the GP allowed format as follows:
\begin{tcolorbox}[ams align]
    & \frac{\sum_{\ell \in \mathcal{L}^{(k)}} \Omega^{(k,\ell)} + \tau^{\mathsf{LT},{(k,\ell)}}}{\tau^{\mathsf{LT}, \mathsf{max}}} \leq 1, 
    \\ & \frac{(A_{46})^{-1} \tau^{\mathsf{LT}, \mathsf{max}}}{\widehat{H}_{59}(\mathbf{x};l)} \leq 1, 
    \\ & (A_{46})^{-1} \geq 1.
\end{tcolorbox}
\label{gpend}

\newpage
\subsection{Pseudo Code of Our GP Optimization Solver}\label{app:cons:sudo}

\begin{algorithm}[H]
\caption{Signomial programming and GP-based optimization solver proposed for problem~$\bm{\mathcal{P}}$}
\label{alg:cent}
{\footnotesize
\begin{algorithmic}[1]
 \STATE \textbf{Input:} Predefined convergence criterion.
 \STATE Initialize the iteration index {\scriptsize$l=0$}.
 \STATE Select an initial feasible point $\mathbf{x}^{[0]}$ as the beginning iteration.
 \STATE Extract the approximations (highlighted in green boxes, pages \pageref{gpstart}–\pageref{gpend}) based on the present solution {\scriptsize$\mathbf{x}^{[l]}$}.
 \STATE Replace the original constraints in Problem~$\bm{\mathcal{P}}$ by these approximations to obtain a GP approximation.
 \STATE Apply logarithmic variable transformations to translate the resulting GP into the convex problem (see the beginning of Appendix~\ref{app:optTransform}).
\STATE Increment iteration index {\scriptsize$l \leftarrow l+1$}.
 \STATE Solving the GP problem via a state-of-the-art convex optimization solver (e.g., CVXPY~\cite{diamond2016cvxpy}) and obtaining the solutions  {\scriptsize$\mathbf{x}^{[l]}$}.
 \IF{successive solutions {\scriptsize$\mathbf{x}^{[l-1]}$} and {\scriptsize$\mathbf{x}^{[l]}$} fail to satisfy the convergence criterion}
   \STATE Go to Step 4 and re-iterate via the most recent solution {\scriptsize$\mathbf{x}^{[l]}$}.
 \ELSE
   \STATE Set the most recent solution as the final solution {\scriptsize$\mathbf{x}^{\star} = \mathbf{x}^{[l]}$}.
 \ENDIF
\end{algorithmic}
}
\vspace{-.1mm}
\end{algorithm}

 \begin{proposition}[Convergence of the Optimization Solver]\label{propo:KKT}
 The resulting solution (i.e., $\mathbf{x}^\star$) of Algorithm \ref{alg:cent} converges to the Karush–Kuhn–Tucker (KKT) conditions of problem $\bm{\mathcal{P}}$.
\end{proposition}
\begin{proof}
 Algorithm \ref{alg:cent} generates a sequence of solutions for $\bm{\mathcal{P}}$ by employing two approximations: (a-i) arithmetic-geometric mean inequality \eqref{eq:approxPosMonMain} and (a-ii) sum-power approximation $\max\{A, B\}\approx (A^{p}+B^{p})^{-\frac{1}{p}}$.
It is straightforward to verify that Algorithm \ref{alg:cent} yields an \textit{inner approximation} \cite{marks1978general} of~$\bm{\mathcal{P}}$ under the tight approximation of (a-ii). To establish convergence, it suffices to verify the three characteristics outlined in~\cite{marks1978general} for our solver, under (a-i), which can be shown using the techniques of~\cite{xu2014global} and are omitted here for brevity.
\end{proof}
\subsection{Complexity Analysis} \label{subsec:Complexity}
As demonstrated in this appendix, our optimization problem $\bm{\mathcal{P}}$ can be transformed to a GP optimization problem, which can then be further reformulated as a convex optimization problem via a logarithmic transformation ~\cite{nesterov1994interior, chiang2005geometric}.  CVXPY iteratively updates the solution of the convex reformulation of our GP problem by leveraging interior-point methods (IPMs) and the KKT conditions (see Proposition \ref{propo:KKT}). This approach entails iteratively solving a Newton system derived through KKT conditions, which requires solving a linear system with an approximate complexity of \(O(n^2)\), where \(n\) is the size of the problem, i.e., the sum of the number of constraints and variables (for details refer to \cite{nesterov1994interior}).
As IPMs in convex optimization require \(O\left(\sqrt{n} \log\left(\frac{1}{\epsilon}\right)\right)\) iterations for achieving \(\epsilon > 0\) desired solution accuracy \cite{nesterov1994interior}, the GP optimization problem using CVXPY has a solution complexity of \(O\left(n^{2.5} \log\left(\frac{1}{\epsilon}\right)\right)\), which is polynomial with respect to $n$ (ensuring computational tractability and making it a suitable choice for large-scale problem instances). Furthermore, it should be noted that the implementation of {Fed-Span} does not require executing all of the intermediate mathematical derivations outlined in Appendix~\ref{app:optTransform}. Instead, it leverages only the closed-form highlighted expressions in green boxes throughout the appendix, thereby ensuring the mathematical tractability of {Fed-Span}.

\end{document}